\documentclass[letterpaper,12pt,titlepage,oneside,final]{book}
\usepackage{fullpage}

\usepackage{subcaption}
\usepackage{enumitem}
\usepackage{wrapfig}
\usepackage{xspace}
\usepackage{algcompatible}
\usepackage{algcompatible}
\usepackage{slantsc}
\usepackage{algorithm}
\usepackage{algpseudocode}
\usepackage{listings}
\usepackage{verbatim}
\usepackage{multicol}
\usepackage{float}
\usepackage[utf8]{inputenc}             
\usepackage[T1]{fontenc}                
\usepackage{booktabs}                  
\usepackage{amsfonts}                  
\usepackage{nicefrac}                  
\usepackage{microtype}                 
\usepackage{listings}
\usepackage{xcolor}
\usepackage{pifont}
\usepackage{subfiles}
\usepackage{multirow}
\usepackage{colortbl}
\usepackage{amsthm}
\usepackage{array, makecell} 
\usepackage{tabularx} 

\usepackage{xpatch}

\usepackage{tikz}
\usetikzlibrary{chains,shapes, arrows.meta, positioning, calc}

\usepackage[commentColor=blue]{algpseudocodex}

\newcommand{\href}[1]{#1} 

\usepackage{ifthen}
\newboolean{PrintVersion}
\setboolean{PrintVersion}{false}

\usepackage{amsmath,amssymb,amstext} 

\usepackage[pdftex,pagebackref=false]{hyperref} 
\hypersetup{
    plainpages=false,       
    unicode=false,          
    pdftoolbar=true,        
    pdfmenubar=true,        
    pdffitwindow=false,     
    pdfstartview={FitH},    
    pdfnewwindow=true,      
    colorlinks=true,        
    linkcolor=blue,         
    citecolor=green,        
    filecolor=magenta,      
    urlcolor=cyan           
}
\ifthenelse{\boolean{PrintVersion}}{   
\hypersetup{	
    citecolor=black,%
    filecolor=black,%
    linkcolor=black,%
    urlcolor=black}
}{} 

\usepackage[automake,toc,abbreviations]{glossaries-extra} 

\let\origdoublepage\cleardoublepage
\newcommand{\clearemptydoublepage}{%
  \clearpage{\pagestyle{empty}\origdoublepage}}
\let\cleardoublepage\clearemptydoublepage

\makeglossaries

\graphicspath{{figs/}}

\begin{document}

\pagestyle{empty}
\pagenumbering{roman}


\begin{titlepage}
        \begin{center}
        \vspace*{1.0cm}

        \Huge
        {\bf Safe Memory Reclamation Techniques}

        \vspace*{1.0cm}

        \normalsize
        by \\

        \vspace*{1.0cm}

        \Large
        Ajay Singh \\

        \vspace*{3.0cm}

        \normalsize
        A thesis \\
        presented to the University of Waterloo \\ 
        in fulfillment of the \\
        thesis requirement for the degree of \\
        Doctor of Philosophy \\
        in \\
        Computer Science \\

        \vspace*{2.0cm}

        Waterloo, Ontario, Canada, 2024 \\

        \vspace*{1.0cm}

        \copyright\ Ajay Singh 2024 \\
        \end{center}
\end{titlepage}

\pagestyle{plain}
\setcounter{page}{2}

\cleardoublepage 
\phantomsection    
 
\addcontentsline{toc}{chapter}{Examining Committee}
\begin{center}\textbf{Examining Committee Membership}\end{center}
  \noindent
The following served on the Examining Committee for this thesis. The decision of the Examining Committee is by majority vote.
  \bigskip
  
\noindent
\begin{tabbing}
Internal-External Member: \=  \kill 
External Examiner: \>  Erez Petrank \\ 
\> Professor, \\
\> Dept. of Computer Science, \\
\> Technion — Israel Institute of Technology \\
\end{tabbing} 
  
\noindent
\begin{tabbing}
Internal-External Member: \=  \kill 
Supervisor(s): \> Trevor Brown \\
\> Associate Professor, \\ 
\> Cheriton School of Computer Science, University of Waterloo \\
\\
\> Peter Buhr \\
\> Associate Professor, \\
\> Cheriton School of Computer Science, University of Waterloo \\
\end{tabbing}
  
\noindent
\begin{tabbing}
Internal-External Member: \=  \kill 
Internal Member: \> Ali Mashtizadeh \\
\> Associate Professor, \\
\> Cheriton School of Computer Science, University of Waterloo \\
\\
\> Sihang Liu \\
\> Assistant Professor, \\
\> Cheriton School of Computer Science, University of Waterloo \\
\end{tabbing}
  
\noindent
\begin{tabbing}
Internal-External Member: \=  \kill 
Internal-External Member: \> Wojciech Golab \\
\> Professor, \\
\> Dept. of Electrical and Computer Engineering, \\
\> University of Waterloo \\
\end{tabbing}


\cleardoublepage
\phantomsection    

 \addcontentsline{toc}{chapter}{Author's Declaration}
 \begin{center}\textbf{Author's Declaration}\end{center}

 \noindent
I hereby declare that I am the sole author of this thesis. This is a true copy of the thesis, including any required final revisions, as accepted by my examiners.
 \noindent  
  \bigskip
  
  \noindent
I understand that my thesis may be made electronically available to the public.

\cleardoublepage
\phantomsection    


\addcontentsline{toc}{chapter}{Abstract}
\begin{center}\textbf{Abstract}\end{center}

This dissertation presents three paradigms to address the challenge of concurrent memory reclamation, manifesting as use-after-free errors that arise in concurrent data structures using non-blocking techniques. Each paradigm aligns with one of our three objectives for practical and safe memory reclamation algorithms.

\noindent
\textbf{Objective 1:} Design memory reclamation algorithms that are fast, have a bounded memory footprint, and are easy to use — requiring neither intrusive changes to data structures nor specific architecture or compiler support. These algorithms should also deliver consistent performance across various workloads and be applicable to a wide range of data structures. To achieve this, we introduce the neutralization paradigm with the NBR (Neutralization-Based Reclamation) algorithm and its enhanced version, NBR+ (Optimized Neutralization-Based Reclamation). These algorithms use POSIX signals and a lightweight handshaking mechanism to facilitate safe memory reclamation among threads. By relying solely on atomic reads and writes, they achieve bounded garbage and high performance with minimal overhead compared to existing algorithms. They are straightforward to implement, similar in reasoning and programming effort to two-phased locking, and compatible with numerous data structures.

\noindent
\textbf{Objective 2:} Eliminate the asymmetric synchronization overhead in existing reclamation algorithms, which often incur costly memory fences while eagerly publishing reservations, as seen in algorithms like hazard pointers and hazard eras. We propose the reactive synchronization paradigm, implemented through deferred memory reclamation and POSIX signals. This mechanism enables threads to privately track memory references (or reservations) and share this information on demand, using the publish-on-ping algorithm. This approach serves as a drop-in replacement for hazard pointers and hazard eras and includes a variant (EpochPOP) that combines epochs with the robustness of hazard pointers to approach the performance of epoch-based reclamation.

\noindent
\textbf{Objective 3:} Completely eliminate the batching common in current reclamation algorithms to allow immediate memory reclamation, similar to sequential data structures, while maintaining high performance. 
We introduce Conditional Access, a hardware-software co-design paradigm implemented in a graphite multi-core simulator. This paradigm leverages cache coherence to enable efficient detection of potential use-after-free errors without explicit shared-memory communication or additional coherence traffic. Conditional Access provides programmers with hardware instructions for immediate memory reclamation with minimal overhead in optimistic data structures.

To validate our claims, we designed and conducted extensive benchmark tests to evaluate all proposed algorithms on high-end machines under various scenarios. We paired these algorithms with several real-world concurrent data structures, representing various memory access patterns, and compared their time and space efficiency against numerous state-of-the-art memory reclamation algorithms, demonstrating significant improvements.



\cleardoublepage
\phantomsection    

\addcontentsline{toc}{chapter}{Acknowledgements}
\begin{center}\textbf{Acknowledgements}\end{center}

My Saashtaang Dandvat Pranaam (signifies complete reverence and surrender of ego to a revered figure) to my Gurus (the one who dispels ignorance) who have shown me the way to good research.

I express my deepest gratitude to my Guru, Trevor Brown, whose insights and guidance have profoundly shaped my development as a researcher. His ability to offer fresh perspectives during our meetings was truly inspiring. I recall moments of uncertainty that transformed into clarity, a plethora of ideas, and renewed enthusiasm thanks to Trevor's mentorship. His knack for concluding brainstorming sessions and discussions on a positive and creative note greatly enhanced my experience.
I feel fortunate to be one of Trevor's disciples, benefitting from his gentle mentorship and the freedom and motivation he provided to explore my ideas. His guidance was instrumental in steering me towards impactful work, and his support was a source of reassurance, even during challenging times. Trevor's emphasis on automation and efficiency in research has been invaluable, saving me considerable time. His lightning-fast solutions to bottlenecks and live coding sessions made me affectionately nickname him "the Flash."
I am deeply grateful for the additional term of financial support that allowed me to focus solely on writing my thesis. Beyond our direct interactions, Trevor’s work ethic and approach to research provided significant passive learning opportunities. I also thoroughly enjoyed TAing for his CS798 course. Beyond identifying my own strengths and weaknesses with his support, one of the most enduring lessons I have learned from him is how to truly enjoy the research process.

Holding Trevor’s hands, I had the fortune of collaborating with Guru Michael Spear from Lehigh University on our hardware-software co-design project, which evolved into Conditional Access, as well as another project that was ongoing at the time of writing. Mike’s approach to idea generation and evaluation, coupled with his talent for infusing humor into our discussions, made our collaboration both educational and enjoyable. I am grateful to him for teaching me how to explore ideas, scrutinize them, and extract their essence.
In the early days of our collaboration, observing the debates between Mike and Trevor, as well as engaging in discussions with them, was pivotal in my development as a researcher. This experience was especially valuable when navigating challenges with the Graphite simulator and writing manuscripts.

I am also deeply thankful to Guru Ali Mashtizadeh for introducing me to the system research community through our work on NBR and CLpush. His advice on connecting the dots and prioritizing information while reading research papers and generating new ideas has been invaluable.
The vibrant atmosphere he creates during paper submissions, with the added bonus of free food, is particularly memorable. I appreciate his enthusiasm in visiting the labs to assist with code, simulator, and kernel-level concepts, as well as his patience while recording my first conference presentation video for PPOPP 2021.

Guru Peter Buhr deserves special mention for his constructive feedback on my presentations and writing. I had the privilege of observing him during reading sessions and while TAing for his concurrency course. I aspire to embody even a fraction of his qualities as an educator. From him, I have learned how complex concepts can be distilled into simpler forms and the importance of asking the most profound questions in the simplest possible ways.

Although vastly experienced and knowledgeable, working with Trevor Brown, Michael Spear, Ali Mashtizadeh, and Peter Buhr felt like collaborating with peers. Their approachability and willingness to engage in discussions on equal footing greatly enriched my experience.

I sincerely thank the members of my committee — Professor Wojciech Golab, Professor Sihang Liu, and Professor Erez Petrank — for reviewing my work and providing invaluable feedback. Their insights not only helped refine the writing of this thesis but also broadened my understanding of various aspects of my research and potential future directions.

I also want to thank my lab mates and friends who made my time here memorable: William Sigouin for being a great TA, Gaetano Coccimiglio for his support in running a reading group, and Rosina Kharal for the home-cooked meals every time we met. Gautam Pathak provided companionship on biking trails and rock climbing adventures; Daewoo Kim was there for lunches during deadlines, accompanied me to the gym whenever we felt the need to work on our fitness, and offered chewing gum during stressful times, all while dealing with his own work in the lab. Mohammad Khalaji, Emil Tsalapatis, and Daewoo engaged me in enriching discussions about the PhD journey, the world, and beyond. I greatly appreciate Kenneth R. Hancock's willingness to read my TPDS draft on short notice. I am also grateful to Sean Ovens and Quan Pham for supporting me on the ongoing distributed memory project, which allowed me to focus more effectively on writing my thesis.
Lastly, special thanks to the friends with whom I spent most of my time outside the lab, Navya Nair from the faculty of environments, and Harish Mendu from Actuarial Science, for giving me a significant share of your time and for all the cooking and fun together at the student residence.

Special thanks to Lori Paniak and the other support staff at the Cheriton School of Computer Science for making my life easier by handling queries and administrative work. Additionally, I acknowledge the unexpected experience brought about by COVID-19, which added an unforeseen dimension to my journey.

I express my gratitude to Professor Sathya Peri, my Master's degree supervisor, for introducing me to the field of concurrency research.

\textbf{I would be remiss if I did not acknowledge the work of researchers who came before me, whose leftover ideas I picked up, whose code bases I repurposed, listening and watching whom I got inspiration. If at all there is something worthy of being called a novelty would be my perspective to the solutions presented in the dissertation. }

\cleardoublepage
\phantomsection    

\addcontentsline{toc}{chapter}{Dedication}
\begin{center}\textbf{Dedication}\end{center}

\begin{itemize}
    \item I dedicate this work to my mother, who, despite not being able to read or write, and my father, who could only study up to high school, both of whom supported my education in every possible way. I also dedicate it to my loving sisters, Chandni and Pooja.
    My heartfelt gratitude extends to my grandparents, who encouraged my parents to support my studies to whatever extent I wished.

    \item I dedicate this to my teachers and to my village Birahimpur, where every individual is part of my clan, for instilling in me a strong sense of purpose.


    \item Lastly, I cannot complete this without acknowledging the unwavering support of my partner, Swati Yadav, who, despite her own challenges, embraced my purpose as her own, sacrificed many personal moments and stood by me throughout these years.
\end{itemize}


\cleardoublepage
\phantomsection    


\renewcommand\contentsname{Table of Contents}
\tableofcontents
\cleardoublepage
\phantomsection    

\addcontentsline{toc}{chapter}{List of Figures}
\listoffigures
\cleardoublepage
\phantomsection		

\addcontentsline{toc}{chapter}{List of Tables}
\listoftables
\cleardoublepage
\phantomsection		




\pagenumbering{arabic}

\newcommand{\punt}[1]{}
\newcommand{\cmnt}[1]{}

\definecolor{xxxcolor}{rgb}{0.8,0,0}
\newcommand{\XXX}[1]{{\color{xxxcolor} XXX: #1}\xspace}

\newcommand{\func}[1]{\texttt{#1}}
\newcommand{\var}[1]{\texttt{#1}}

\newcommand{\nosplit}{\linebreak}

\def\nohyphens{\hyphenpenalty=10000\exhyphenpenalty=10000}

\newcommand{\tilda}{\symbol{126}}

\newcommand{\ang}[1]{\langle #1 \rangle}
\newcommand{\Ang}[1]{\Big\langle #1 \Big\rangle}
\newcommand{\ceil}[1]{\lceil #1 \rceil}
\newcommand{\floor}[1]{\lfloor #1 \rfloor}

\newtheorem{theorem}{Theorem}
\newtheorem{lemma}[theorem]{Lemma}
\newtheorem{corollary}[theorem]{Corollary}
\newtheorem{proposition}[theorem]{Proposition}
\newtheorem{property}[theorem]{Property}
\newtheorem{claim}[theorem]{Claim}
\newtheorem{definition}{Definition}
\newtheorem{guarantee}{Guarantee}
\newtheorem{requirement}[theorem]{Requirement}
\newcounter{history}
\newcommand{\hist}[1]{\refstepcounter{history} {#1}}

\newcommand{\trevor}[1]{\textbf{[[[Trevor: #1]]]}}
\newcommand{\ajay}[1]{\textcolor{red}{#1}}


\newtheorem{acknowledgement}[theorem]{Acknowledgement}
\newtheorem{observation}[theorem]{Observation}
\newtheorem{assumption}[theorem]{Assumption}

\newcommand{\chapref}[1]{Chapter~\ref{chap:#1}}
\newcommand{\secref}[1]{Section~\ref{sec:#1}}
\newcommand{\figref}[1]{Figure~\ref{fig:#1}}
\newcommand{\tabref}[1]{Table~\ref{tab:#1}}
\newcommand{\stref}[1]{step~\ref{step:#1}}
\newcommand{\thmref}[1]{Theorem~\ref{thm:#1}}
\newcommand{\lemref}[1]{Lemma~\ref{lem:#1}}
\newcommand{\insref}[1]{line~\ref{ins:#1}}
\newcommand{\corref}[1]{Corollary~\ref{cor:#1}}
\newcommand{\axmref}[1]{Proposition~\ref{axm:#1}}
\newcommand{\defref}[1]{Definition~\ref{def:#1}}
\newcommand{\eqnref}[1]{Eqn(\ref{eq:#1})}
\newcommand{\eqvref}[1]{Equivalence~(\ref{eqv:#1})}
\newcommand{\ineqref}[1]{Inequality~(\ref{ineq:#1})}
\newcommand{\exref}[1]{Example~\ref{ex:#1}}
\newcommand{\propref}[1]{Property~\ref{prop:#1}}
\newcommand{\clmref}[1]{Claim~\ref{clm:#1}}
\newcommand{\obsref}[1]{Observation~\ref{obs:#1}}
\newcommand{\asmref}[1]{Assumption~\ref{asm:#1}}
\newcommand{\thref}[1]{Thread~\ref{th:#1}}
\newcommand{\trnref}[1]{Transaction~\ref{trn:#1}}
\newcommand{\lstref}[1]{listing~\ref{lst:#1}}
\newcommand{\Lstref}[1]{Listing~\ref{lst:#1}}
\newtheorem{question}{Question}

\newcommand{\subsecref}[1]{SubSection{\ref{subsec:#1}}}

\newcommand{\histref}[1]{\ref{hist:#1}}

\newcommand{\apnref}[1]{Appendix~\ref{apn:#1}}
\newcommand{\invref}[1]{Invariant~\ref{inv:#1}}

\newcommand{\Chapref}[1]{Chapter~\ref{chap:#1}}
\newcommand{\Secref}[1]{Section~\ref{sec:#1}}
\newcommand{\Figref}[1]{Figure~\ref{fig:#1}}
\newcommand{\Tabref}[1]{Table~\ref{tab:#1}}
\newcommand{\Stref}[1]{Step~\ref{step:#1}}
\newcommand{\Thmref}[1]{Theorem~\ref{thm:#1}}
\newcommand{\Lemref}[1]{Lemma~\ref{lem:#1}}
\newcommand{\Corref}[1]{Corollary~\ref{cor:#1}}
\newcommand{\Axmref}[1]{Proposition~\ref{axm:#1}}
\newcommand{\Defref}[1]{Definition~\ref{def:#1}}
\newcommand{\Eqref}[1]{eq(\ref{eq:#1})}
\newcommand{\Eqvref}[1]{Equivalence~(\ref{eqv:#1})}
\newcommand{\Ineqref}[1]{Inequality~(\ref{ineq:#1})}
\newcommand{\Exref}[1]{Example~\ref{ex:#1}}
\newcommand{\Propref}[1]{Property~\ref{prop:#1}}
\newcommand{\Obsref}[1]{Observation~\ref{obs:#1}}
\newcommand{\Asmref}[1]{Assumption~\ref{asm:#1}}
\newcommand{\reqref}[1]{Requirement~\ref{req:#1}}
\newcommand{\guarref}[1]{Guarantee~\ref{guar:#1}}

\newcommand{\Lineref}[1]{Line~\ref{lin:#1}}
\newcommand{\lineref}[1]{line~\ref{lin:#1}}
\newcommand{\algoref}[1]{Algorithm~\ref{algo:#1}}
\newcommand{\Algoref}[1]{{\sf Algorithm$_{\ref{algo:#1}}$}}

\newcommand{\Apnref}[1]{Section~\ref{apn:#1}}
\newcommand{\Invref}[1]{Invariant~\ref{inv:#1}}
\newcommand{\Confref}[1]{Conflict~\ref{conf:#1}}

\newcommand{\theqed}{$\Box$}
\newcommand{\nsqed}{\hspace*{\fill} \theqed}

\renewcommand{\thefootnote}{\alph{footnote}}
\newcommand{\ignore}[1]{}
\newcommand{\myparagraph}[1]{\noindent\textbf{#1}}

%

\newcommand{\lastup} {lastUpdt}
\newcommand{\lupdt}[2] {#2.lastUpdt(#1)}
\newcommand{\fkmth}[3] {#3.firstKeyMth(#1, #2)}

%
\newcommand{\nz}{\texttt{restartable}\xspace}
\newcommand{\cas}[3] {CAS(#1, #2, #3)}

\newcommand{\qp}{\emph{quiescent-phase}\xspace}
\newcommand{\rdp}{\emph{$\Phi_{read}$}\xspace}
\newcommand{\wtp}{\emph{$\Phi_{write}$}\xspace}

\newcommand{\rd}{\emph{reader}\xspace}
\newcommand{\wt}{\emph{writer}\xspace}
\newcommand{\rl}{\emph{reclaimer}\xspace}
\newcommand{\rrc}{\emph{reader-reclaimer}\xspace}
\newcommand{\wrc}{\emph{writer-reclaimer}\xspace}

\newcommand{\knbr}{\emph{NBR}\xspace}
\newcommand{\ds}{\emph{data structure}\xspace}
\newcommand{\rb}{\emph{retireBag}\xspace}
\newcommand{\lb}{\texttt{limboBag}\xspace}
\newcommand{\hw}{\emph{HiWatermark}\xspace}
\newcommand{\lw}{\emph{LoWatermark}\xspace}
\newcommand{\smr}{\emph{safe memory reclamation}\xspace}

\newcommand{\tid}[1]{$T_#1$}
\newcommand{\pred}{\texttt{pred}\xspace}
\newcommand{\curr}{\texttt{curr}\xspace}
\newcommand{\suc}{\texttt{succ}\xspace}

\newcommand{\nbr}{NBR\xspace}
\newcommand{\nbrp}{NBR+\xspace}
\newcommand{\rcu}{RCU\xspace}
\newcommand{\qsbr}{QSBR\xspace}
\newcommand{\debra}{DEBRA\xspace}
\newcommand{\ibr}{IBR\xspace}
\newcommand{\geibr}{2GEIBR\xspace}
\newcommand{\cl}{crystallineL\xspace}
\newcommand{\cw}{crystallineW\xspace}
\newcommand{\hp}{HP\xspace}
\newcommand{\he}{HE\xspace}
\newcommand{\wfe}{WFE\xspace}
\newcommand{\ebr}{EBR\xspace}
\newcommand{\rgp}{\textit{neutralization event}\xspace}

\newcommand{\ssj}{\texttt{sigsetjmp()}}
\newcommand{\slj}{\texttt{siglongjmp()}}

\definecolor{light-gray}{gray}{0.80}
\lstdefinestyle{ajstyle}{ %
	language=C++,
	numbers=left,                    
	morekeywords={*, startOp,...},   
	breaklines=true,                 
	frame=single,
belowcaptionskip=1\baselineskip,
showstringspaces=false,
basicstyle=\ttfamily,
keywordstyle=\bfseries\color{green!40!black},
commentstyle=\itshape\color{purple!40!black},
identifierstyle=\color{blue},
stringstyle=\color{orange},	
}

\lstdefinestyle{ajstyleds}{ %
	language=C++,
	morekeywords={*, startOp,...},   
	breaklines=true,                 
	frame=single,
belowcaptionskip=1\baselineskip,
xleftmargin=\parindent,	
showstringspaces=false,
basicstyle=\scriptsize\ttfamily,
keywordstyle=\bfseries\color{green!40!black},
commentstyle=\itshape\color{purple!40!black},
identifierstyle=\color{blue},
stringstyle=\color{orange},	
numberstyle=\scriptsize,
}
\lstset{escapechar=|,style=ajstyle}


\algnewcommand{\IfThenElse}[3]{
  \State \algorithmicif\ #1\ \algorithmicthen\ #2\ \algorithmicelse\ #3}
\algnewcommand{\IfThen}[2]{
  \State \algorithmicif\ #1\ \algorithmicthen\ #2}  
\algnewcommand{\IfThenNoS}[2]{
 \algorithmicif\ #1\ \algorithmicthen\ #2}  
\algnewcommand{\IfNoS}[2]{
 \algorithmicif\ #1\ \algorithmicthen\ #2}

\newcommand{\crdp}[1]{\texttt{cread}(#1)}
\newcommand{\cwrp}[2]{\texttt{cwrite}(#1, #2)}
\newcommand{\cwr}{\texttt{cwrite}\xspace}
\newcommand{\crd}{\texttt{cread}\xspace}
\newcommand{\utag}[1]{\texttt{untagOne} #1}
\newcommand{\utagp}[1]{\texttt{untagOne}(#1)}
\newcommand{\utagall}{\texttt{untagAll}\xspace}
\newcommand{\hbrbit}{\textit{accessRevokedBit}\xspace}
\newcommand{\hbrset}{\textit{tagSet}\xspace}
\newcommand{\head}{\textit{head}\xspace}
\newcommand{\ca}{\textit{Conditional Access}\xspace}
\newcommand{\uaf}{\textit{use-after-free}\xspace}

\newcommand{\aj}[1]{\textcolor{red}{#1}}


\chapter{Introduction}
This dissertation focuses on the design and implementation of several practical
safe memory reclamation algorithms for concurrent data structures.

\section{Motivation}
The non-blocking (or lock-free) programming paradigm for designing concurrent data structures has emerged as an alternative to the blocking paradigm on ubiquitous multicore systems. This is primarily due to its advantages of scalability, fault tolerance (preemption tolerance in practice), and freedom from issues typical of the blocking paradigm, such as deadlock, priority inversion, and convoying.
Key building blocks of the non-blocking programming paradigm are non-blocking data structures, for example lists, stacks, queues, trees, dequeues, priority queues, and snapshot objects. 
Given their benefits, these data structures are being progressively incorporated into a variety of concurrent software systems. Examples include open-source standard libraries such as Facebook's Folly, Intel TBB, libcds, and Concurrency Kit; operating system kernels, including Massalin and Pu's lock-free OS kernel~\cite{massalin1992lock} and RCU-protected lists, hash tables, and radix trees in other operating system kernels~\cite{mckenney2013rcu}), and index structures within database systems.

The advantages of non-blocking data structures stem from a key design characteristic:
threads can read or update shared memory locations without blocking
each other. This means that threads can concurrently read a shared memory location while another thread updates it.
This is facilitated by efficient hardware-level atomic read-modify-write primitives, such as compare-and-swap instruction. Apart from being non-blocking, these primitives are fine-grained, reducing contention granularity to a few words (32 or 64 bits) and permitting all possible interleavings of thread executions, unlike their blocking counterparts.
Consequently, at least one thread can make progress independently of other threads, ensuring system-wide forward progress, and multiple threads can work on independent parts of a data structure in parallel.
This enables all threads to leverage multiple cores, providing scalability, especially in real-world scenarios of oversubscription and fault tolerance during thread preemption.

However, this same design characteristic makes designing and implementing non-blocking data structures and algorithms more difficult than their blocking counterparts.
Data structure designers must ensure the correctness of individual reads and updates to shared memory locations.
Since reads are concurrent with updates, programmers must establish that the result of reads is correct. Similarly, multiple threads may attempt updates to a shared memory location concurrently, but only one of them succeeds at a time, such as when using the compare-and-swap instruction. In both cases, if the reads are incorrect or if the compare-and-swap fails, the programmer may need to take corrective action.

In addition to the complexity of the design, practitioners face another synchronization problem of similar complexity when memory recycling is involved. Since threads can access memory locations when they are being updated, threads may access a location that has been freed (returned to the operating system for reuse), resulting in undefined behavior as also mentioned in the C/C++ standard 2020. Such accesses cause program crashes due to segmentation faults and are known as \textit{use-after-free} errors. This problem of concurrent memory reclamation is the central topic in this dissertation.

One might consider using garbage collectors to solve the concurrent memory reclamation
problem. Nevertheless, several programming languages like C/C++ do not implement garbage collection. Moreover, even in languages that incorporate garbage collection for heap memory, there is a growing tendency to resort to off-heap memory to alleviate memory pressure in big data applications. For example, Oracle's Java supports unmanaged off-heap memory allocation in addition to garbage-collected heap memory allocations~\cite{java22}. This is used for in-memory data structures in the Cassandra and HBase database systems~\cite{fakhoury2024nova}.
Consequently, data structures in these cases require concurrent memory reclamation algorithms.  

There has been significant progress in handling the complexity of designing time-- and space-- efficient non-blocking data structures.
For example, Brown, Ellen, and Ruppert proposed LLX, VLX, and SCX primitives to simplify programming non-blocking tree-based data structures~\cite{brown2014general}. More recently, Naama Ben-David, Blelloch and Wei built on prior work by Turek, Shasha, Prakash and Barnes on lock-free locks~\cite{turek1992locking} to produce lock-free data structures from their lock-based implementations~\cite{ben2022lock}. 

Similarly, various memory reclamation algorithms have been designed to solve the problem of concurrent memory reclamation~\cite{hart2007performance,fraser2004practical,mckenney1998read,brown2015reclaiming, michael2004hazard, gidenstam2008efficient, valois1995lock, wen2018interval, ramalhete2017brief, nikolaev2020universal, nikolaev2019hyaline, nikolaev2021crystalline, balmau2016fast, alistarh2017forkscan,alistarh2018threadscan, alistarh2014stacktrack,dragojevic2011power, hart2007performance, dice2016fast, zhou2017hand, morrison2015temporally, kang2020marriage, braginsky2013drop, cohen2015efficient, cohen2015automatic, cohen2018every, moreno2023releasing, cohen2018every, sheffi2021vbr}. These algorithms are commonly referred to as \textbf{safe memory reclamation algorithms}.
Popularized by the seminal work of Maged Michael~\cite{michael2004hazard}, the task of a safe memory reclamation algorithm is to synchronize the freeing of a shared memory location with threads that can concurrently access it. 


Designing and implementing a memory reclamation algorithm is as complex as designing and implementing non-blocking data structures.
This complexity is the main reason why memory reclamation algorithms are designed independently of the non-blocking data structures and algorithms they are used with.
Consequently, many reclamation algorithms appear to conflict with the design goals of non-blocking data structures when practitioners use them. 
For instance, some reclamation algorithms add overhead at every access to a shared memory location, undermining the goal of achieving faster searches. Others permit unbounded garbage accumulation, which jeopardizes the progress guarantees of the associated data structures.

For the non-blocking data structures to be suitable for real applications in production, they must be paired with a memory reclamation algorithm that maintains scalability, non-blocking progress, and ease of programmability. Additionally, the algorithm should be widely applicable, have a low memory footprint, and perform consistently across various workloads and scenarios.
This is the current focus of research on concurrent memory reclamation algorithms, leading to a growing body of work in this area. 
Therefore, in this thesis, we revisit the paradigms of designing concurrent memory reclamation algorithms and propose new approaches that address the shortcomings of previous approaches.

\section{Contribution}

The primary contribution of this dissertation is the introduction of several innovative safe memory reclamation algorithms utilizing three paradigms aimed at addressing use-after-free errors in data structures designed with non-blocking techniques. 
These paradigms include: Neutralization Paradigm, Reactive Synchronization Paradigm, and Hardware-Software Co-Design Paradigm. Using these paradigms, we have created six safe memory reclamation algorithms: NBR (neutralization-based reclamation), NBR+ (optimized neutralization-based reclamation), CA (conditional access), HazardPOP (Hazard Pointers with Publish on Ping), EpochPOP (Robust Epoch-Based Reclamation with Publish on Ping), and HazardEraPOP (Hazard Era with Publish on Ping).

\begin{itemize}
    \item 
\textbf{Neutralization Paradigm}. (\chapref{chapnbr},~\chapref{chapnbrp}, and~\chapref{chapappuse}).
In the first paradigm, accessing threads and reclaiming threads cooperate to eliminate use-after-free errors on a shared memory location. Essentially, all threads access a shared memory location in such a way that they can either discard their access at any time or reserve the location beforehand and announce their reservation when needed. A reclaiming thread communicates with all other threads in the system that it is about to reclaim a memory location. In response, other threads either discard their access to let the reclaiming thread reclaim the location or announce their reservation, indicating to the reclaiming thread to skip freeing unless the reservation is released. This cooperative mechanism is termed neutralization and is implemented using POSIX signals available in commodity operating systems, built on top of hardware Inter-Processor Interrupts.  

The two new algorithms using this paradigm, NBR and its signal-optimized version NBR+, only use atomic reads, writes, and read-modify-write instructions such as compare-and-swap and fetch-and-add and POSIX signals. These algorithms are non-blocking with a non-blocking operating system kernel. Interestingly, these algorithms significantly improve over the state of the art in safe memory reclamation: they are fast (NBR+ is faster than best known epoch based reclamation algorithms), achieve an upper bound on the amount of unreclaimed memory, are straightforward to use with many different data structures, and in most cases, require similar reasoning and programmer effort to two-phase locking. The algorithms were implemented, with results reproduced in Rust by other researchers.

\item \textbf{Reactive Synchronization Paradigm}.(\chapref{chaprsp}).

It is surprising that most of the well-known methods for safe memory reclamation algorithms impose uneven synchronization overhead. Although an individual thread seldom performs reclamation, other threads eagerly log and communicate costly bookkeeping related to reclamation, even if they are not accessing the shared memory being reclaimed. The process-wide memory barrier (\texttt{sys\_membarrier}~\cite{membarrierSystemwide} on Linux) used in Folly's hazard pointers do help reduce this uneven overhead, However, the availability and implementation of \texttt{sys\_membarrier} (which can be blocking) vary across kernel versions and architectures. When unavailable, the system may fall back to \texttt{mprotect}, which can degrade performance~\cite{mprotectoverhead}. Additionally, \texttt{sys\_membarrier} has been linked to security vulnerabilities~\cite{CVE202426602}.

We turn this paradigm on its head and mitigate the asymmetric overhead by allowing threads to silently maintain the reclamation-related bookkeeping information and communicate the information only right before a thread reclaims. We term this the Reactive Synchronization Paradigm for designing safe memory reclamation algorithms. Reactive synchronization is enabled with the help of POSIX signals. Although this paradigm utilizes signals similarly to the neutralization paradigm, it triggers a different behavior from threads.

We apply this paradigm to multiple state of the art safe memory reclamation algorithms, including hazard pointers (proposed for inclusion in C++26), epoch-based reclamation, and hazard-eras, thereby introducing a family of reactive synchronization-based algorithms. 
To evaluate these algorithms, we design a comprehensive benchmark suite that compares them with various state of the art memory reclamation methods and data structures.
The results show that the reactive synchronization-based algorithms achieve up to 5$\times$ improvement over the original algorithms while retaining their original desirable properties, such as ease of use, wide applicability and bounded unreclaimed garbage. 

\item \textbf{Co-design Reclamation Paradigm}.(\chapref{chapirp}).

The fourth algorithm we developed, named Conditional Access, is a co-design method that combines hardware and software to address use-after-free issues in concurrent data structures. 
At its core, it utilizes hardware cache coherence to offer programmers efficient memory access instructions. Using these instructions, programmers can safely access their program's memory and immediately release all memory locations.
It realizes an ideal case of having a concurrent data structure with the memory efficiency of a sequential data structure while retaining the scalability of a concurrent implementation. This is particularly useful in a system where a low free-to-use (time between free and reuse of a memory location) latency is desirable. 
For example, in the Universal Memory Allocator in the FreeBSD kernel, low free-to-use latency could help reduce memory pressure in data structures (e.g. radix zone) where memory is frequently allocated and freed~\cite{umasmrfreebsd} and also improve cache efficiency.

This hardware-software co-design paradigm pushes the boundary of the safe memory reclamation algorithm space by allowing implicit memory management, leveraging hardware to eliminate complex and suboptimal bookkeeping mechanisms at the data structure level.
The immediate reclamation of memory without compromising on speedup makes it ideal in modern data centers, to save costs related to memory over-allocation and to facilitate memory overcommitment~\cite{vmwareunderstanding}. For instance, in virtualized data centers, host machines can take advantage of memory overcommitment to allocate the available dynamic memory across several virtual machine instances.


We prototyped Conditional Access on the Open Source Graphite Multicore Simulator~\cite{miller2010graphite} and benchmark it against multiple state-of-the-art memory reclamation techniques. Our extensive evaluation showcases the practical feasibility, efficiency, and superiority of Conditional Access, strengthening the case for more exploration in the hardware-software co-design based approaches.

\end{itemize}

\section{Outline}

\chapref{chapbg} introduces the system model used to implement all four safe memory reclamation algorithms, along with the terminology employed throughout the thesis. It covers the categories of concurrent data structures that benefit from our reclamation algorithms and discusses the challenge of reclaiming memory, particularly the problem of use-after-free errors when threads free memory locations that others may still be using.

\chapref{chapsurvey} provides a survey of the latest safe memory reclamation algorithms, highlighting their characteristics and common issues. 
We identify three main problems with current approaches: the lack of several desirable properties in a single algorithm, the prevalence of asymmetric synchronization overhead in existing reclamation algorithms, and the neglect of lower-level system events, such as those in the allocator and cache subsystems.
These issues serve as motivation for proposing our three new paradigms, which underpin the design of the six safe memory reclamation algorithms presented in the subsequent chapters.

\chapref{chapnbr} introduces the first design paradigm, neutralization, to address the weaknesses of existing algorithms. 
We present the Neutralization-Based Reclamation (NBR) algorithm, including its pseudocode, correctness proof, and evidence of its ability to bound unreclaimed garbage.
Although NBR is theoretically efficient, preliminary evaluations reveal scalability issues on large NUMA machines, leading to the development of the optimized NBR+ algorithm, discussed in the next chapter.

\chapref{chapnbrp} introduces \nbrp, a highly scalable algorithm based on the neutralization paradigm.
We present its pseudocode, proof of correctness, and demonstrate that it effectively bounds the amount of unreclaimed garbage.
The chapter includes an extensive evaluation of both neutralization paradigm algorithms -- NBR and NBR+ -- paired with popular data structures, comparing them with multiple state-of-the-art safe memory reclamation techniques across various workloads.
This evaluation highlights the superior scalability of both NBR and NBR+.  

\chapref{chapappuse} examines the applicability of NBR and NBR+ across different data structures, categorizing them into compatible, semi-compatible, and incompatible types.
We review a selection of popular data structures to assess their compatibility with NBR and NBR+, as well as other popular state-of-the-art safe memory reclamation algorithms.
Furthermore, we compare the ease of use of NBR and NBR+ with other reclamation algorithms such as hazard pointers and DEBRA, using a data structure as an example.
The chapter concludes with a primer on the POSIX signal API, providing guidance for programmers to effectively utilize our algorithms.

\chapref{chapnbr},~\chapref{chapnbrp}, and~\chapref{chapappuse} are based on the following published manuscripts.
\begin{itemize}[noitemsep]
    \item \cite{singh2021nbr} NBR: Neutralization Based Reclamation. Ajay Singh, Trevor Brown, and Ali Mashtizadeh. In Principles and Practice of Parallel Programming (PPoPP), 2021. (received best artifact award and was among the top 4 manuscripts).
    \item \cite{singhTPDS2023NBRP} Simple, Fast and Widely Applicable Concurrent Memory Reclamation via Neutralization. Ajay Singh, Trevor Brown, and Ali Mashtizadeh. In Transactions on Parallel and Distributed Systems (TPDS), Feb. 2024.  Full version in ~\cite{singh2020nbr}.
\end{itemize}

\chapref{chaprsp} addresses the second issue of asymmetric synchronization overhead by proposing a reactive synchronization paradigm. This paradigm aims to reduce overhead by allowing threads to maintain and communicate reclamation-related information only when necessary.

\chapref{chapirp} addresses the third issue of overlooking lower-level system interactions in existing algorithms by introducing Conditional Access, a hardware-software co-design paradigm.
This chapter demonstrates how using lower-level system events can help achieve a memory footprint similar to sequential data structures while maintaining superior scalability. 
We propose new hardware instructions and implement a prototype on the Graphite multicore simulator. Additionally, we design a benchmark suite to evaluate Conditional Access, which confirms its efficiency and ideal memory footprint. 

This chapter is based on the following manuscript.
\begin{itemize}[noitemsep]
\item \cite{singh2023IPDPS} Efficient Hardware Primitives for Immediate Memory Reclamation in Optimistic Data Structures. Ajay Singh, Trevor Brown, and Michael Spear. International Parallel and Distributed Processing Symposium (IPDPS), 2023. Full version in ~\cite{singh2023efficient}. 
\end{itemize}

\chapref{chapconclusion} concludes the dissertation along with discussion on the limitations of the proposed algorithms and suggesting directions for future research.

\chapter{Background}
\label{chap:chapbg}

In this chapter, we discuss the assumed system model for all the safe memory reclamation algorithms presented in this dissertation.
We first discuss concurrent data structures, primarily non-blocking data structures and optimistic blocking data structures, and explain the challenge of concurrent memory reclamation associated with their use.
This challenge, which manifests itself as use-after-free errors, arises when threads access a memory location that is concurrently being freed. We then introduce a class of algorithms --- safe memory reclamation algorithms--- that address this challenge.
In addition, this chapter defines several key terms to help the reader understand the exposition of the algorithms in the subsequent chapters.

\section{System Model}

We consider the asynchronous shared memory model described by Herlihy~\cite{herlihy1991wait, herlihy1990linearizability}, which consists of a fixed number of \textit{processes}.
The term "asynchronous" implies that no assumptions are made about the underlying scheduler.  
Processes can be arbitrarily delayed or experience halting failures. Consequently, a delayed process is indistinguishable from a halted process. 
Each process has a private memory only accessible to itself and a shared memory accessible to all processes.

Shared memory is divided into primitive objects, typically up to 64 bits in size on current systems, and accessed using primitive hardware operations.
These include atomic \func{read} and \func{write} operations, as well as read-modify-write operations such as \func{compare-and-swap} and \func{fetch-and-add}.
A \func{read} operation atomically retrieves the last value written in a memory location, while a \func{write} operation atomically updates the content of a memory location.   
\func{Compare-and-swap} takes a memory address, an expected content of the memory address, and a new content to be written to the memory address as input. Atomically compares the current content at the memory address with the expected content, and if they match, updates the memory address with the new content, returning true. Otherwise, if the comparison fails, it returns false. A \func{fetch-and-add} takes a memory location and a number as input and atomically adds the input number to the current value at the memory location. 

Additionally, for the algorithms discussed in ~\chapref{chapnbr} through \chapref{chaprsp}, we assume that threads can communicate using POSIX signals.
The advantage of using POSIX signals is that we can relax the traditional thread failure model in the asynchronous shared memory setting. In theory, a thread can run arbitrarily slowly or halt altogether, and one cannot distinguish between a slow thread and a halted one. In practice, threads that appear to be delayed are either busy with other work, trapped in infinite loops, or descheduled. The operating system knows which threads are running, descheduled, or have terminated or become zombies (awaiting events to terminate). The \func{pthread\_kill} function used to signal threads returns an error code if a thread is a zombie, or terminated, allowing a signaling thread to learn about such threads. A programmer can use this information to ignore a signal sent to such threads, assuming that thread ids are not reused. 
This leaves running threads and descheduled threads.
Running threads will be interrupted by a signal to execute a signal handler. 
Notably, our model assumes that the signal receiving thread is interrupted before \func{pthread\_kill} returns, although, in practice, there is a slight delay, which is examined in \chapref{chapnbr}.
As for descheduled threads, modern schedulers are (somewhat) fair, and ensure each thread is scheduled to run periodically, at which point a thread will eventually execute a signal handler. In particular, the signal handler is run when the thread is context switched back, before any regular program instructions are executed.

\section{Concurrent Data Structures} 
In our context, a primitive object is called a \textit{field}, and a collection of fields forms a complex object, referred to as a \textit{node}. 
A linked concurrent data structure is an implementation of an abstract data type in which each node has one or more fields that point to other nodes.
Henceforth, whenever we refer to a concurrent data structure, we mean a linked concurrent data structure.
These data structures have a fixed set of \textit{entry points} that are pointers to nodes and support high-level operations, such as \textit{insert}, \textit{delete} and \textit{lookup}, among others.

\begin{itemize}
    \item \textbf{Insert} operation adds a node to the data structure if it is not already present.
    \item \textbf{Delete} operation removes a node from the data structure if it is present.
    \item \textbf{Contains} operation tests if a node is present in the data structure. 
\end{itemize}

Each \textit{operation} in turn is a sequence of steps (hardware primitives that work on primitive objects). 
We assume that these operations never return a pointer to a node.

\noindent
\textbf{States of a Node:}
We assume existence of a memory allocator that provides operations to allocate and free nodes.
A data structure object (or node) typically goes through several states from its allocation to the time it is reclaimed to be allocated again. These states are as follows:

\begin{itemize}
    \item \textbf{Unallocated:} A node is termed \textit{unallocated} when its memory is inaccessible to executing threads~\cite{sheffi2023era}.
    \item \textbf{Allocated:} The node is referred to as \textit{allocated} when the request for a memory region successfully returns from the underlying memory allocator, but the node has not yet been inserted into the data structure.
    \item \textbf{Reachable:} The node is inserted into the data structure and threads could access it by accessing a finite set of pointers from the entry point of the data structure.
    \item \textbf{Deleted:} The node is logically marked for removal from the data structure, but threads could still reach it.
    \item \textbf{Unreachable:} The node is physically removed from the data structure, so no threads starting from the entry point could reach it.
    \item \textbf{Free:} The node has been returned to the allocator and could be recycled (reclaimed) for subsequent allocation requests.
\end{itemize}

An \textit{unreachable} but not yet \textit{free} node is said to be in the \textit{retired} state, and is called \textit{garbage}. 
For our neutralization based algorithms (\chapref{chapnbr} - \chapref{chapnbrp}) and publish on ping algorithms (\chapref{chaprsp}), a node is considered \textit{safe for reclamation} if it is in the \textit{retired} state and no thread can access it. 
Otherwise, it is \textit{unsafe for reclamation}.
A node is \textit{safe to be accessed} if it has not been previously freed to the operating system since the last time it was allocated.
Additionally, Sheffi and Petrank~\cite{sheffi2023era} state that a node can be considered \textit{safe for access} after being freed, provided there is a guarantee that no threads will use its contents.

In the following sections, we discuss approaches to design concurrent data structures which are relevant to the reclamation algorithms we propose.

\subsection{Blocking Concurrent Data Structures}

The blocking paradigm for designing concurrent data structures uses mutual exclusion to synchronize access to shared memory locations. Mutual exclusion is implemented using locking primitives that provide acquire and release operations.
Essentially, before a thread $T1$ executes a read or update step on a shared memory location, it has to atomically acquire a lock and after it has completed the step, it has to release the lock.
Any other threads running concurrently must wait for $T1$ to release the lock before they can acquire it and perform a read or update operation on the memory location.

This locking paradigm introduces scalability issues in concurrent data structures.
In more detail, locks prevent other threads from taking their own steps in the program to complete their operations (even if they are on disjoint parts of the data structure).
Therefore, the available cores remain idle as the threads wait for the locks to be released.
This limitation prevents data structures from taking advantage of the parallelism provided by the increasing number of processing cores in the hardware, resulting in poor scalability (esp. in an oversubscribed system).
For example, in a hand-over-hand locking implementation of a linked list, a thread $T1$ intending to update a node near the end of the list could be stuck while waiting for a lock to be released on a node at the beginning of the list. If another thread $T2$ has already acquired the lock at the beginning of the list and is in the middle of its update, $T1$ must wait because locks do not allow concurrent reads during an update, which negatively impacts \textit{scalability}. 

Another factor that contributes to slow performance in lock-based data structures is the cache contention that arises due to the exclusive ownership of cache lines when a lock is acquired (essentially a write operation) on modern cache-coherent multicore processors~\cite{moir2018concurrent}. These issues, to some extent, can be addressed by reducing the lock granularity (number of instructions executed while holding a lock) and frequency of locks (as in optimistically synchronized data structures, which are discussed in \secref{obds}).
New advances like transactional lock elision have attempted to address the scalability issues but have their own downsides due to their fastpath/slowpath approach~\cite{rajwar2001speculative, dice2016refined}.


\subsection{Non-Blocking Concurrent Data Structures}

The non-blocking paradigm for designing concurrent data structures
allows threads to read memory locations even though other threads might be concurrently updating them. 
From a programmer's perspective, 
unlike the blocking paradigm, threads do not use lock based mutual exclusion to block a runnable thread to read a shared memory location that is concurrently being updated.
This enables threads to take advantage of all available cores. 

Consider the same example of the linked list. The thread $T2$ can read the earlier nodes and proceed without waiting for thread $T1$ to complete its update.
Note that programmers might need to ensure that reads are correct and may require corrective actions, like helping, but the main purpose here is to convey that threads do not block (spin or sleep) for other threads to finish and could concurrently access memory locations. As a result, the two threads could execute concurrently on two different cores, leveraging the underlying parallelism in the hardware.
This 
helps in improving \textit{scalability}, especially under oversubscription, where threads working on disjoint memory locations can make use of allocated core cycles. Additionally, this paradigm provides \textit{fault tolerance} because a blocked thread cannot indefinitely prevent other threads from taking steps in their operations, thus ensuring system-wide progress. 

The progress guarantees of a non-blocking data structure further fall into three categories, namely obstruction-freedom, lock-freedom and wait-freedom~\cite{scott2013shared}. Depending on which one of the above progress guarantee a non-blocking data structure provides, it can be categorized as obstruction-free, lock-free and wait-free. An obstruction-free data structure ensures that a thread executing alone will complete its operation in a bounded number of steps. 
A lock-free data structure design ensures that at least one thread completes executing its operation on the data structure in a bounded number of steps.
A wait-free data structure design ensures that all threads will complete executing their operation on the data structure in a bounded number of steps.
The lock-free data structures are more popular amongst the three and are the main focus of the techniques developed in this thesis. 
Non-blocking data structures utilize read-modify-write hardware instructions, such as compare-and-swap, fetch-and-add, and test-and-set.

\subsubsection{Lock-free Data Structures}
Implementations of a lock-free data structure may vary widely depending on the data structure at hand and type of synchronization primitives available~\cite{scott2013shared}.
However, what remains common is that threads can read concurrently with updates and do not have to block (spin or sleep) waiting for other threads to complete. Additionally, amongst multiple threads that attempt to update a shared memory location, for example changing the next pointer of a field to insert or delete a node in a list, at least one thread is always guaranteed to complete its operation~\cite{treiber1986systems, michael2002high, harris2001pragmatic, timnat2014practical, sundell2004efficient, michael1996simple, ellen2010non}. 

\subsection{Optimistic Blocking Data Structures}
\label{sec:obds}
Optimistic concurrency control, introduced to optimize transactions in databases~\cite{kung1981optimistic},  inspired the design of optimistic concurrent data structures~\cite{guerraoui2016optimistic, heller2005lazy, herlihy2007simple, david2015asynchronized, cha2001cache, wang2018building, shi2023optiql, leis2016art}. In these data structures, operations consist of an optimistic search phase where threads read shared memory locations in a synchronization-free manner, followed by a validation phase using locks to ensure the correctness of the searches, possibly executing updates, in update phase, before returning. 

Similar to non-blocking data structures, a thread can read shared memory locations in search phase while they are concurrently being updated by another thread in update phase.
Consequently, these data structures also permit high concurrency, as multiple threads can work concurrently on different parts of the data structure, thereby enhancing scalability.
Examples include the lazy list~\cite{heller2005lazy, herlihy2007simple} by Heller, Herlihy, Luchangco, Moir, Scherer III and Shavit, the binary search trees~\cite{bronson2010practical} by Bronson, Casper, Chafi and Olukotun,  and the ticket lock-based binary search trees~\cite{david2015asynchronized} by David, Guerraoui, and Trigonakis.

\ignore{However, this improvement in performance and scalability introduces a challenge, as explained below.}

\section{Challenge: Concurrent Memory Reclamation}

In non-blocking data structures and during the search phase of optimistic blocking data structures, threads optimistically access shared memory locations, assuming conflicts are rare. This contributes to scalability, but the downside is that while executing a delete operation, if a thread reclaims a node, other threads could be concurrently accessing that node and might incur a \textbf{use-after-free} error.
This can lead to two primary issues:
\begin{enumerate}
    \item \textbf{Segmentation Faults:} Other threads might crash when they access a memory location that has been reclaimed.
    \item \textbf{Correctness Bugs:} Other threads might continue execution but work on incorrect data, introducing correctness bugs. This can manifest itself as the more complex ABA problem. 
\end{enumerate}

A well-known example is the \textit{lazy} linked list~\cite{heller2005lazy}, where threads perform searches without locking and only obtain locks on a select few nodes to confirm the search or execute updates.

\begin{figure}
\centering
\begin{tikzpicture}[
  node distance=1.5cm,
  every node/.style={draw, rectangle, minimum size=1cm},
  arrow/.style={-{Latex[scale=1.5]}, thick},
]

\node (n1) {Node 1};
\node[right=of n1] (n2) {Node 2};
\node[right=of n2] (n3) {Node 3};

\draw[arrow] (n1.east) -- (n2.west);
\draw[arrow] (n2.east) -- (n3.west);

\node[above=1cm of n2, , draw=none] (t1) {1. \{access a pointer to n2\}$_{T1}$};
\draw[arrow, dashed] (t1) -- (n2);
\node[below=1cm of n2, draw=none] (t2) {2. \{access a pointer to n2\}$_{T2}$};
\draw[arrow, dashed] (t2) -- (n2);

\end{tikzpicture}
\hfill
\begin{tikzpicture}[
  node distance=1.5cm,
  every node/.style={draw, rectangle, minimum size=1cm},
  arrow/.style={-{Latex[scale=1.5]}, thick}
]

\node (n1) {Node 1};
\node[right=3cm of n1] (n3) {Node 3};
\node[below=0.5cm of $(n1)!0.5!(n3)$] (n2) {Node 2};

\draw[arrow] (n1.east) -- (n3.west);
\draw[arrow, dashed] (n2.east) -> (n3.west);

\node[below=1cm of n2, draw=none] (t2) {3. \{unlink n2\}$_{T2}$ \hspace{0.5cm} 4. \{free n2\}$_{T2}$};
\draw[arrow, dashed] (t2.north) -- (n2.south);

\node[above=1cm of n2, draw=none] (t1) {5. \{deref pointer to n2\}$_{T1}$};
\draw[arrow, dashed, red] (t1.south) -- (n2.north);

\end{tikzpicture}
\caption{Example of steps leading to use-after-free error on a lazy list~\cite{heller2005lazy}.}
\label{fig:uafex}
\end{figure}
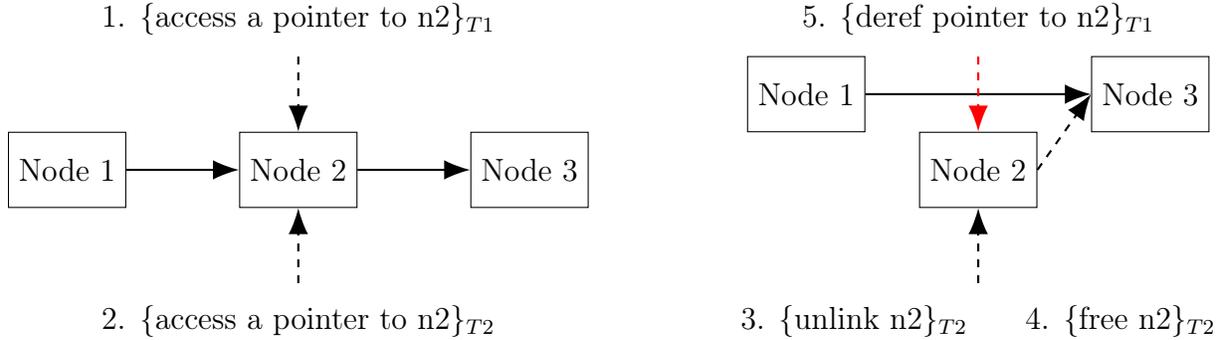
Suppose a thread $T1$, in the middle of its search, acquires a reference to a node (step 1 in the \Figref{uafex}). Another thread $T2$, executing a delete operation, reads the node (step 2), unlinks and subsequently frees (or reclaims) the node (step 3 and step 4 in figure). If the former thread, unaware of the concurrent free operation, attempts to dereference the node (step 5 in figure), it will encounter a use-after-free error.
Such accesses are unsafe, leading to undefined behavior (see C/C++ standard 2020~\cite{ISOCPP2020}), and can cause program crashes
due to segmentation faults.
In a type stable memory system, the crash might not occur, but this can lead to the more subtle ABA bug~\cite{scott2013shared}.

In summary, unlike in the blocking programming paradigm, 
in the non-blocking paradigm, threads may access reclaimed shared memory locations.
Fundamentally, this is a special type of read-write conflict on a shared memory node between two threads. One thread could \textit{reclaim} (or free) the node, while another thread could still \textit{access} it, unaware that the node was freed since the last access, leading to the use-after-free error.

Resolving this error requires the conflicting threads to coordinate their access to the shared memory nodes.
In the literature, this is the task of a class of concurrent memory reclamation algorithms, also called safe memory reclamation algorithms. 
These algorithms typically determine when a shared location is eligible to be safely freed~\cite{michael2004hazard}.
Conversely, it is equally the task to determine when a shared location can be safely accessed. 
Both views of looking at what a concurrent memory reclamation algorithm is supposed to do are captured when we see the problem as a special type of read-write conflict.

\section{Solution: Safe Memory Reclamation Algorithms}

Safe Memory Reclamation algorithms implement a synchronization mechanism to resolve use-after-free errors in concurrent data structures. These algorithms provide an interface to programmers, who wrap their data structure operations and individual memory access instructions using the provided interface. 
Depending on a particular safe memory reclamation algorithm, the exact interface may vary. However, the common aspect is that programmers, instead of directly invoking \texttt{free} on a pointer to a node, have to invoke a \texttt{retire} function provided by the interface. In addition, the use of the interface can be automated. In this dissertation, we focus on the manual use of safe memory reclamation algorithms. These algorithms can be divided into two categories: automatic and manual. 
Garbage collection, in general, can be seen as a form of automated safe memory reclamation that operates across the entire system. While safe memory reclamation algorithms primarily help recycle memory in concurrent data structures, garbage collection takes care of the memory recycling for the whole system, without limiting itself to the specific application type.

\Figref{smrlandscape} illustrates the interaction between a safe memory reclamation (SMR) algorithm, a concurrent data structure, and a memory allocator. In the absence of an SMR algorithm, a programmer calls \texttt{allocate} to obtain the memory location for a new node and \texttt{free} to deallocate the memory of a node upon deletion. Conversely, with an SMR algorithm in place, instead of invoking \texttt{free} directly, the programmer sends the deleted node to the SMR module, which leverages its synchronization mechanism to safely return the node to the allocator.

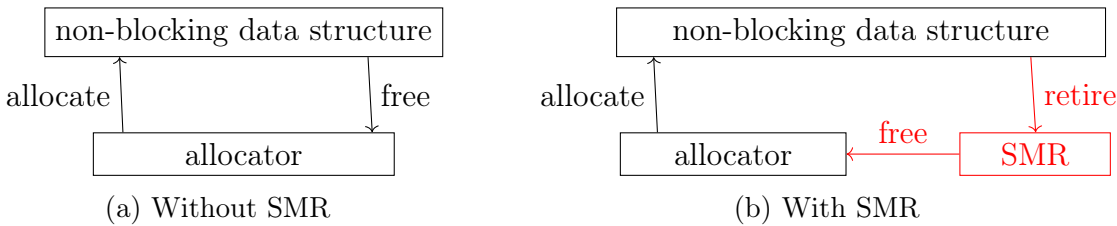
\begin{figure}[H]
\centering
\begin{subfigure}{.5\textwidth}
\centering

\begin{tikzpicture}[node distance=1cm]
\def\rectwidth{4cm}
\node (ds) [draw, rectangle, minimum width=\rectwidth, align=center] {non-blocking data structure};
\node (allocator) [draw, rectangle, below=of ds, minimum width=\rectwidth, align=center] {allocator};
\draw[->] ([xshift=0.4cm]allocator.north west) -- node[left] {allocate} ([xshift=1cm]ds.south west);
\draw[->] ([xshift=-1cm]ds.south east) -- node[right] {free} ([xshift=-0.3cm]allocator.north east);
\end{tikzpicture}
\caption{Without SMR}
\label{fig:wosmr}
\end{subfigure}\hfill
\begin{subfigure}{.5\textwidth}
\centering

\begin{tikzpicture}[node distance=1cm]
\def\rectwidth{3cm}
\def\combinedwidth{6.5cm} 
\node (ds) [draw, rectangle, minimum width=\combinedwidth, align=center] {non-blocking data structure};
\node (allocate) [draw, rectangle, minimum width=3cm, align=center, below=of ds, xshift=-1.7cm] {allocator};
\node (smr) [draw, red, rectangle, minimum width=2cm, align=center, right=of allocate, xshift=0.5cm] {SMR};
\draw[->] ([xshift=0.5cm]allocate.north west) -- node[left] {allocate} ([xshift=0.5cm]ds.south west);
\draw[->, red] ([xshift=-1cm]ds.south east) -- node[right] {retire} ([xshift=-1cm]smr.north east);
\draw[->, red] ([xshift=-0cm]smr.west) -- node[above] {free} ([xshift=-0cm]allocate.east);
\end{tikzpicture}
\caption{With SMR}
\label{fig:withsmr}
\end{subfigure}

\caption{A programmer's view on how a data structure interacts with a memory allocator.}
\label{fig:smrlandscape}
\end{figure}

\subsection{Assumptions}
\label{sec:introasmp}
Safe memory reclamation algorithms make the following assumptions about the concurrent data structures to which they are added.
\begin{enumerate}
    \item A node that can be reclaimed may only be accessed by the operations of the data structure.
    \item A node is only retired after it has been unlinked from the data structure. Other threads that started accessing the data structure before the node was unlinked can access the node, but any thread that started accessing the data structure from its entry points after the node was unlinked cannot access it.
    \item At any time, only one thread may retire a node. This is required to prevent double-freeing. Additionally, a programmer should use retire for all the nodes that are deleted to avoid memory leaks.
    
    \item Similar to most safe memory reclamation algorithms in literature, we assume a fixed number of threads in each of our algorithms, meaning that threads do not enter or leave dynamically during execution of our benchmarks. 
\end{enumerate}

\subsection{Key Terms}
\label{sec:keyterms}

We use the term \textit{reader} or \textit{writer} for a thread that holds a reference to a node in the data structure and can read (\textit{resp.} write) to the node.
A thread that invokes \textit{retire} on a node and can free the node is called \textit{reclaimer}. Note that we interchangeably use \textit{reclaim} and \textit{free} for a node to indicate that it will be returned to the operating system for later reuse.
Threads in the system alternate roles between being a \textit{reader}, \textit{writer}, or \textit{reclaimer} depending on their current task.
We also use the term \textit{reader} to describe a thread that holds a reference to a node, regardless of whether it plans to read from or write to the node. The context will make this distinction clear, or we will specify as needed.

Typically, a safe memory reclamation algorithm maintains the retired nodes in a local list, termed a \textit{limbo bag} or \textit{retire list}. These nodes in the lists are freed after the synchronization mechanism of a safe memory reclamation algorithm, say at a time $t$, indicates that no thread holds a reference to the retired nodes. The retired nodes after time $t$ are referred to as \textit{safe} to free nodes, and before time $t$ the references to the retired nodes are \textit{hazardous} or \textit{unsafe} and can not be freed.

Threads discard all references to reclaimable shared nodes between consecutive data structure operations and can be referred to as \textit{quiescent}. When a thread turns quiescent, it helps a \textit{reclaimer} learn that it cannot hold a reference to a \textit{retired} shared node.

We now describe terms relevant to a concurrent data structure that uses a safe memory reclamation technique.

\begin{enumerate} 
\item [P1] \textit{Performance}: Refers to throughput in terms of the number of operations per second for a data structure used with a safe memory reclamation algorithm. 

\item [P2] \textit{Garbage}: Refers to the number of unlinked but unreclaimed nodes, i.e., retired nodes, in a data structure. It provides an estimate of the memory footprint of a data structure. 

\item [P3] \textit{Usability}: Refers to the ease with which a programmer can use a safe memory reclamation technique with a data structure. Intrusive changes to data structure layout, code modifications, and reliance on specialized hardware instructions,  compilers or custom allocators are often a hindrance to the ease of adoption and integration of a safe memory reclamation technique.

\item [P4] \textit{Consistent Performance}: Refers to the stability of the performance of a concurrent data structure to which safe memory reclamation is added.
Ideally, a safe memory reclamation algorithm should help a data structure's performance remain stable and consistent across different workloads, such as shifting between read-intensive and update-intensive scenarios, and when the system is oversubscribed with more threads than available cores.

\item [P5] \textit{Applicability}: Refers to the generality of a safe memory reclamation algorithm and measures how many types of data structures the safe memory reclamation algorithm can be used with. Ideally, an algorithm should be applicable to a wide range of data structures.
\end{enumerate}



\chapter{Survey of Concurrent Reclamation Algorithms}
\label{chap:chapsurvey}
In this chapter, we survey state-of-the-art concurrent reclamation algorithms, focusing on their key traits, advantages and disadvantages.

\section{Reference Counting}
One of the earliest and fairly general concurrent memory reclamation algorithms is lock-free reference counting~\cite{detlefs2002lock, michael1995correction, valois1995lock}.
This algorithm employs a count field per node to track the number of active references across all threads.
Before accessing a node, threads increment the count field to indicate that the node should not be freed, and decrement the count when they no longer will use the node. 
In this way, a thread can reclaim a node whenever its count drops to zero, avoiding use-after-free errors.

However, reference-counting-based algorithms introduce significant overhead due to cache invalidations caused by frequent updates to reference counts. Specifically, each node access requires incrementing the associated reference count, which is later decremented when the location is no longer needed.   
In modern optimistic data structures, in which threads search without locking, reference counting can diminish or negate any performance gains over traditional locking algorithms.
Techniques such as update coalescing~\cite{levanoni2001fly}, deferred reference counting~\cite{anderson2021concurrent, correia2021orcgc}, and immediate reference counting~\cite{jung2024concurrent} have demonstrated improved performance.

Additionally, reference counting also introduces other complications. For example, the count is often stored alongside the node which complicates any advanced pointer arithmetic or implicit pointers, and can require changes to node layouts (or the use of a custom allocator) as well as size. 
Reference counting typically requires a programmer to invoke a \textit{deref} operation to dereference a pointer (and sometimes to explicitly invoke operations for read, write, and CAS)~\cite{detlefs2002lock, herlihy2005nonblocking, blelloch2020concurrent, nikolaev2019hyaline, gidenstam2008efficient}, adding significant overhead and programmer effort.
Programmer intervention is also needed to identify and break pointer \textit{cycles} in garbage nodes, for example using weak references~\cite{anderson2022turning,jones2023garbage} or setting the pointer fields to null.

\section{Pointer Reservations}
Pointer reservation-based reclamation algorithms require each thread to \textit{reserve} pointers before dereferencing them.
Reservations enable a thread attempting to reclaim a pointer to determine whether it is safe to free, and a thread that reserves the pointer to learn that the pointer will not be reclaimed while it is reserved~\cite{michael2004hazard, herlihy2005nonblocking, dice2016fast}. Therefore, avoiding use-after-free errors.
Reserving a pointer typically involves atomically writing it to a shared memory location and subsequently performing a \textit{validation}. The validation ensures that the reserved pointer has not been already deleted, and it can be specific to a data structure implementation.
Think of the reservation mechanism as a form of bookkeeping, which a thread about to reclaim a node (shared memory location) uses to establish whether the location can be freed without incurring a use-after-free error (or safely freed).


Similarly to reference counting, reserving a shared memory location (or pointer to nodes in data structures) every time it is accessed introduces significant overhead.
Additionally, programmers are responsible for correctly releasing previously reserved pointers.
Moreover, if the validation of a reservation fails, threads need to take alternative actions, such as restarting an operation.
This may require modifying the original data structure in potentially complex ways that may necessitate re-proving correctness or progress.

An important factor to consider when applying pointer-reservation-based algorithms is their incompatibility with data structures that involve traversing marked (logically deleted) nodes~\cite{heller2005lazy, Brown:2014, afek2014cb, drachsler2014practical, ellen2010non, natarajan2014fast, shafiei2013non, bronson2010practical, fatourou2019persistent, prokopec2012concurrent, harris2001pragmatic}. In these data structures, such as the lazy list~\cite{heller2005lazy} or the Harris list~\cite{harris2001pragmatic}, the validation fails and the operation is retried when either the node from which a new pointer is read is marked for deletion, or the node to which the pointer points is marked, or both are marked. 
Since these data structures allow executions where a thread could reach a node from a sequence of logically deleted nodes, 
a thread might read a pointer to an already freed node from a logically deleted reserved node. 
This could cause a crash when the thread attempts to access the freed node.
A detailed explanation of the problem with the applicability of pointer reservation techniques to such data structures appears in Brown's DEBRA paper~\cite{brown2015reclaiming}.
Consequently, pointer-reservation-based reclamation algorithms have limited applicability because they do not directly apply to the aforementioned data structures.
Some reservation-based algorithms, such as the hybrid of hazard pointer~\cite{michael2004hazard} and reference counting~\cite{valois1995lock} by Gidenstam et al.~\cite{gidenstam2008efficient}, sidestep these issues at the cost of considerably higher programmer effort due to intrusive changes to data structures. In addition, it is not clear how effective the technique is.

Jung, Lee, Kim and Kang~\cite{jung2023applying} proposed HP++ : a technique to extend the applicability of the hazard pointer technique to data structures that allow traversal of logically deleted nodes by ensuring that threads that unlink nodes also protect the concurrent traversals from use-after-free. This requires adding an invalidate bit to a node's memory layout and extending the interface for the hazard pointers with a function for unlinking and another function for correctly figuring out the nodes to be invalidated.

\section{Epoch Based Reclamation}
Epoch Based Reclamation (EBR) algorithms~\cite{hart2007performance,fraser2004practical,mckenney1998read,brown2015reclaiming, desnoyers2009low} offer a straightforward, high-performance approach.
In these algorithms, execution is divided into \textit{epochs}, and each thread must participate regularly in the reclamation algorithm for the epoch to change. The Epoch is advanced when all threads report that they are executing in the latest epoch.
Typically, this imposes a small amount of overhead on each data structure operation.
Classically, each thread maintains three batches of objects it is waiting to reclaim, one for each of the last three epochs. Whenever the epoch changes, each thread can reclaim its oldest batch.
This method results in minimal overhead; however, if a thread gets delayed and does not report the most recent epoch, it can hinder the progression of the epoch, thereby obstructing the reclamation of all memory.

Brown introduced an optimized variant of EBR called DEBRA, as well as an enhanced variant called DEBRA+\cite{brown2015reclaiming} that bounds garbage in compatible non-blocking data structures.
However, when a thread is signalled in DEBRA+, it restarts its operation, even if it had already begun modifying the data structure.
Thus, data structure-specific recovery code is required to deal with the inconsistent states that can result. 
It is not clear how DEBRA+ could be applied to any lock-based (or optimistic) algorithm~\cite{heller2005lazy, david2015asynchronized}, since restarting a thread after it has acquired a lock could cause deadlocks.

\section{Epoch Reservations}
Following DEBRA and DEBRA +, other algorithms such as Hazard Eras (HE)~\cite{ramalhete2017brief} and Interval Based Reclamation (IBR)~\cite{wen2018interval} emerged,  aiming to limit the amount of garbage that cannot be reclaimed to individual epochs (or specific ranges of epochs) where a thread is stuck.
These algorithms augment nodes with additional metadata --- the epochs when a node was allocated (birth epoch) and unlinked (retire epoch) from the data structure. These epochs essentially denote the lifetime of a node. The idea is that, while accessing nodes within data structure operations, threads will reserve current epochs. When reclaiming nodes, threads will check if the lifetime of a node being freed intersects with reserved epochs of all the threads in the system. If the lifetime of a node intersects with any thread's reserved epochs, the node is not reclaimed, as the thread could be accessing the node. Otherwise, the reclaiming thread can infer that no thread could access the node and that the node can be safely reclaimed.

In IBR, a thread reserves an epoch interval which starts from the epoch that a thread saw when its operation started (i.e., before it accessed any node in the data structure) up to the latest epoch which the thread reserves while it last accessed a node in its operation. In this way, a thread that is stalled could only prevent nodes whose lifetime intersects with the reserved interval, and all nodes retired before the reserved interval or allocated after the reserved interval can still be reclaimed. In practice, this reduces the chances of unbounded garbage.
Unfortunately, in certain scenarios involving a thread that executes a long-running operation, memory consumption can still be extremely high, proportional to the number of allocations.
Theoretically speaking, with pathological thread scheduling, where threads stall while holding the maximum possible epoch interval in IBR, unbounded garbage may occur because threads never reclaim. 

HE, on the other hand, reserves more precise epochs. Essentially, it is like hazard pointers, but instead of reserving each pointer, it reserves the global epoch at the time it accesses a node. The reclaiming thread frees only those nodes whose lifetime does not intersect with epochs reserved by all threads.
In this way, a stalled thread could only reserve as many epochs as it can access pointers to nodes at a time. Thus, HE has bounded garbage.

\textbf{Wait-Free-Eras}~\cite{nikolaev2020universal} Nikolaev and Ravindran's Wait-Free-Eras scheme adds wait-freedom to hazard eras and interval based memory reclamation schemes. They note that most of the interface of these techniques is wait-free except for the function that is supposed to read the next pointer.
The read function in hazard eras contains a loop that attempts to reserve the current global epoch when a pointer to a node is accessed.
This reservation is crucial to determine if a node can be safely freed: if the reserved epoch overlaps with the node's lifetime, the reclaiming thread skips freeing the node unless the reservation is released.
To successfully reserve the latest epoch, a thread executing the read function validates, after loading the pointer, that the epoch has remained unchanged during the read operation. If the epoch has changed, the thread retries until the validation is successful. Theoretically, in an unfortunate scheduling, the global epoch could keep changing perpetually, causing the thread to starve indefinitely.

To prevent indefinite starvation of a thread during pointer reading, the authors suggest a fast-path slow-path strategy. In this approach, a thread that repeatedly fails to read transitions to a slow path.
In the slow path, the reading thread declares its intention to read (through a descriptor), and before increasing the global epoch, other threads check each thread's descriptor to determine if helping is needed. 
If help is required, the helping thread reads the pointer to the node on behalf of the starving thread before updating the global epoch.
This method enables a reading thread to reserve the latest global epoch while reading the pointer, ensuring that the global epoch remains unchanged during the read operation.
Specifically, the helping thread holds the current global epoch and updates the descriptor with the node pointer for the reading (which was starving) thread, allowing it to leave the slow path. To prevent repeated helping, a tag field (akin to a version number) is used in both the reservation slots and descriptor field, necessitating the use of WCAS (Wide-compare-and-swap) instruction.

These algorithms require per node metadata, necessitating changes to the memory layout of nodes.
Moreover, it appears that most epoch reservation algorithms, including HE, IBR, and WFE, cannot be used with a wide variety of data structures in which threads traverse chains of unlinked nodes, as explained in the pointer reservation reclamation algorithms above. This can be mitigated by using the technique proposed in ~\cite{jung2023applying}.
Additionally, these algorithms exhibit nontrivial overhead due to the bookkeeping required at every read of a node.

\section{Hybrid and OS/Hardware Primitives}
Popular hybrid memory reclamation algorithms combine the approaches of previously mentioned reclamation algorithms, including leveraging operating system features like context switches~\cite{balmau2016fast} and POSIX signals~\cite{alistarh2017forkscan,alistarh2018threadscan}, as well as hardware primitives such as hardware transactional memory~\cite{alistarh2014stacktrack,dragojevic2011power}.

In \cite{balmau2016fast}, Balmau et al. employ context switches triggered by the periodic scheduling of auxiliary processes per core to timely flush hazard pointers.
In this algorithm, reclaimers wait for a set interval of time to pass since a node (or object) is retired to ensure that reservations to the object (if any) are published, reclaiming the object if it is not reserved. 
These algorithms incur overhead to periodically publish reservations, even when threads might not be reclaiming.
The same paper has another variant, Qsense, which optimizes for the common case by employing an EBR-like approach as a fast code path~\cite{hart2007performance}.
However, to ensure bounded garbage when threads experience delays, Qsense switches to a hazard pointer inspired slow path.
Qsense provides bounded garbage, but has been shown to be slower than EBR~\cite{balmau2016fast}.

Dice et al.~\cite{dice2016fast} explored using the write-protection feature of memory pages that triggers a global barrier to facilitate timely publication of hazard pointers before threads reclaim. However, it could block threads trying to reserve an object if a reclaimer stalls after write protecting the page on which the object resides. 
The techniques in this work mainly build upon system properties that are not always well specified.

Another algorithm, StackTrack~\cite{alistarh2014stacktrack} by Alistarh et. al., focuses on automating memory reclamation through compiler assistance.
It utilizes hardware transactional memory's ability to abort when one transaction's data set conflicts with that of another. It divides a data structure operation into a series of small hardware transactions. Each subtransaction before committing publishes its current stack and register values. This is necessary to ensure that the pointer to nodes at the end of a sub-transaction are not reclaimed before another transaction begins.
A reclaiming thread before freeing its list of retired nodes scans stack and register contents of all other threads and skips freeing the nodes that are found in the stack or registers of the other thread. To ensure progress, the technique deploys a hazard pointer-like fallback mode.
However, StackTrack requires significant programmer intervention and assumes that programs will not hide pointers (via, e.g., special pointer arithmetic), which implies data structures that use pointer hiding or bit stealing cannot use it, limiting its applicability.
Furthermore, due to their reliance on HP, StackTrack and Qsense have the same issues with respect to applicability as other HP and epoch-reservation-based approaches.

Alistarh et. al. in ForkScan~\cite{alistarh2017forkscan} uses OS-level copy-on-write support for memory pages and POSIX signaling to automatically reclaim memory for concurrent data structures.
Like StackTrack, ForkScan does not work if users perform special pointer arithmetic.
Due to the way that ForkScan identifies objects that are safe to reclaim, it has the same applicability issues as StackTrack and Qsense.
ForkScan's successor, ThreadScan~\cite{alistarh2018threadscan} offers numerous improvements over the original but suffers from the same applicability issue.

Dragojevic, Herlihy et al. in \cite{dragojevic2011power} proposed HOHRC and FastCollect which mainly aims to make the scanning of reservations in hazard pointer based algorithms efficient and adaptive \textemdash independent of the prior knowledge of the number of processes in the system. At heart, these mechanisms implement a collect object as a list of reservations and exports register, deregister and collect interfaces to threads to reserve unreserve or scan reservations. Scanning occurs using a sequence of short HTM transactions that appears to collect all the nodes that are reserved when a thread attempts to reclaim its retired objects.

Zhou, Luchangco and Spear\cite{zhou2017hand} make use of the sequence of short hardware transactions which execute in hand over hand fashion to design concurrent data structures that retain the property of immediate memory reclamation. 
The execution of transactions in hand over hand manner ensures that a thread executes as if it were running a single transaction, thus utilizing the implicit bookkeeping mechanism of HTM to safely access memory locations. 

Though it has the property of immediate reclamation, it incurs the bookkeeping overhead of maintaining optimal sized transactions. Moreover, frequent start and commit of several transactions consumes a significant compute time. Because these mechanisms use HTM, they inherit its restrictions. For instance, they prohibit data structure operations from using instructions that do not support hardware transactions.

Another approach by Morrison and Afek assumes an alternative memory model called temporally bounded total store order (TBTSO)~\cite{morrison2015temporally} and applies it to HPs to guarantee that reserved pointers will be published to all threads within a bounded time, facilitating safe memory reclamation.

PEBR~\cite{kang2020marriage} technique by Kang and Jung is another hybrid of HP and EBR.
The key idea is to \textit{eject} a stalled thread from the epoch advancement mechanism and to use hazard pointers to protect the references of the ejected thread. 
Every read needs to reserve objects as in HP, but the reservation is made more efficient using algorithms from~\cite{dice2016fast}.
The implementation of PEBR is available in Rust and it
has been shown to be slower than Rust's Crossbeam implementation of EBR.

Anderson et. al. ~\cite{anderson2021concurrent} propose an efficient wait-free automatic technique which mitigates the per-access overhead associated with modification of reference counters. The key idea in their technique is to use hazard pointers to defer the increment and decrement of reference counters for a managed node rather than deferring the reclamation of the node.
A similar automated technique was proposed by Correia et. al. in OrcGC~\cite{correia2021orcgc}.
Subsequently, in ~\cite{anderson2022turning}, Anderson et. al. show that the deferred reference counting can be implemented with other SMR techniques like IBR and Hyaline replacing hazard pointer based protection of reference counts.

In Drop the Anchor (DTA)~\cite{braginsky2013drop}, Braginsky, Kogan and Petrank integrate epoch-based reclamation (EBR) and Hazard Pointers (HP). DTA relies primarily on EBR and enhances robustness against rare stalled threads by employing a recovery mechanism. This mechanism utilizes periodically published hazard pointers. During recovery, the process duplicates the range of objects reachable from the stalled thread into the data structure, enabling other threads to resume reclamation without having to worry about freeing objects reachable from the stalled threads.

In 2024, Kim, Jung and Kang~\cite{kim2024expediting} published a hybrid of HP and RCU to help reduce the per node overhead in HP. The technique proposes dividing data structure operations in alternating sequences of HP protected region followed by RCU protected region. Thus, the overhead of reserving nodes at every access is reduced to every $n$ node. For example, in a list, the first two nodes can be HP protected and then an RCU protected region can be started to protect access of the next $n$ nodes. The traversal alternates between these two protected regions until the operation is complete.
The HP protected regions serve as checkpoints, which allows a subsequent RCU region to think of the last HP protected region as the entry point of the data structure. 

To ensure bounded garbage, the technique takes inspiration from neutralization based reclamation (which is a contribution of this thesis discussed in \chapref{chapnbr} and \ref{chap:chapnbrp}), where a reclaiming thread can send a neutralizing signal selectively to delayed threads in RCU protected regions. 
The preceding HP protected region serves as a checkpoint from which a neutralized thread can restart, thus restarting from an entry point of the data structure is not required (which is different from neutralization based reclamation where threads restart from the entry point of a data structure).
As a result, it extends the applicability of neutralization based reclamation (presented in ~\chapref{chapnbr}) to data structures where threads do not restart from an entry point after an auxiliary update (for example, after unlinking a marked node encountered during traversals).
Essentially, in their technique before starting an update, threads use HP to reserve all the nodes required to enable starting a subsequent RCU protected region, banking on a data structure's ability to verify that the previously protected nodes are not logically deleted for safety.
This allows a thread in the middle of an update to defer neutralizing to the end of the update and let a reclaiming thread to reclaim its garbage, excluding any reserved nodes. 
Their technique allows signal based neutralizing mechanism~\cite{singh2021nbr, singhTPDS2023NBRP} to be applied to data structures with auxiliary updates~\cite{harris2001pragmatic, michael2002high, natarajan2014fast} and at the same time maintains their benefit of low overhead on data structure traversals.

\ignore{
TREV----
}

\section{Alternative Paradigms for Safe Memory Reclamation}

\subsection{Optimistic Access Paradigm}
\label{sec:oaparadigm}

Cohen and Petrank~\cite{cohen2015efficient} introduced an innovative paradigm that allows threads to optimistically access potentially deleted objects and roll back to a safe point if validation (after access) indicates that the object may have been deleted. Their technique enables threads to continue allocating and retiring nodes from a pre-allocated pool of memory until threads can no longer allocate new nodes. This process is referred to as one phase or round of reclamation. To start a new round, the retired memory nodes are reclaimed without waiting for threads that might still hold references to them. This makes the retired nodes available for new allocation requests, marking the beginning of the next round. Threads that might still hold references to a retired node are notified of the reclamation of retired memory nodes (the beginning of a new round) by setting their warning bits. Before reading the contents of a memory location, each thread checks if its warning bit is set. If the bit is set, the thread learns that the node it is about to access might have been reclaimed. Upon learning about a concurrent reclamation event (the beginning of the next reclamation round, in the words of the authors), the reading threads reset their reclamation bits and restart their data structure operations.

The nodes retired in one round can only be reclaimed in subsequent rounds. This requirement ensures the safety of reads because threads become aware of concurrent reclamations (by observing that their warning bits are set) only when a new reclamation round starts. If a node retired in one phase is reclaimed in the same phase, a reading thread might access stale content, leading to the ABA problem or, more generally, breaking correctness. To keep track of the reclamation rounds, each thread maintains a local round number. To ensure that a retired node is only reclaimed in a subsequent round, each retired node is also tagged with the current reclamation round number. These round numbers act as versions. Additionally, while a thread starts a new reclamation round by recycling the retired nodes, any other thread trying to allocate or retire nodes executes a helping function to ensure lock-freedom. Recycling of a node about to be modified is prevented by protecting it with hazard pointers. In the original algorithm (Algorithm 2 in ~\cite{cohen2015efficient}), after protecting the nodes involved in an update instruction (e.g., all pointer parameters of a CAS), a thread restarts if the warning bit is set.

The \textbf{Optimistic Access (OA)} algorithm involves instrumenting reads, writes, and CAS instructions and has been shown to work with the data structures in their normalized form~\cite{timnat2014practical} (a pattern which makes restarting an operation from an arbitrary point in code easier). 
In an effort to improve the ease of programming, the authors introduced \textbf{Automatic Optimistic Access (AOA)}~\cite{cohen2015automatic}, which relieves programmers from figuring out the placement of retire calls and automates the transformation into normalized form. The technique works similarly to a mark-and-sweep garbage collector but only for lock-free data structures in normalized form. Although automatic, programmers still need to do some initial setup to use AOA. For example, they need to provide all the entry points of a data structure, invoke the AOA initialization function, and provide the description of nodes within the data structure.

Subsequently, the authors proposed \textbf{Free Access (FA)}~\cite{cohen2018every} which applies to arbitrary data structures (except those that use instructions such as SWAP), eliminating the requirement for data structures to be presented in a normalized form\cite{timnat2014practical}.
Free Access is based on the mark-and-sweep style reclamation mechanism of AOA.

At a high level, Free Access considers a data structure operation to consist of a sequence of read-only and write-only phases. At the start of each read-only phase, a checkpoint is established, from which threads can restart their read phase if a reclamation event has begun since the phase started. The checkpoint involves protecting the current local pointers and the program counter.

Restart is required to ensure threads do not access a reclaimed node. Similarly, at the start of a write-only phase, threads use hazard pointers to protect all pointers that will be used in the phase so that threads do not modify the reclaimed pointers. To guarantee that the hazard pointers have not been reclaimed concurrently, threads, after saving hazard pointers at a shared memory location, also verify using a warning bit (similar to OA and AOA) that a reclamation event has not occurred, implying that the saved hazard pointers cannot be reclaimed.

Free Access utilizes a compiler pass to annotate data structure operations with reclamation specific code. Specifically, a compiler pass identifies the read-only and write-only phases of a data structure and automates the protection of local pointers at the beginning of the read and write phases and the checkpoint required to restart threads during the read-only phase. It is also required to ensure that the local pointers from which a thread could restart are reachable from an entry point of the data structure; otherwise, threads could access a potentially unreachable node, whose content could be arbitrary.

Improving the Optimistic Access paradigm of safe memory reclamation further, Sheffi and Petrank~\cite{sheffi2021vbr} introduced \textbf{version based reclamation (VBR)}--- an optimistic access-based safe memory reclamation that allows access to reclaimed nodes for reads as well as writes. Notably, previous reclamation schemes, in order to prevent modifying a reclaimed node (which might be in use as a different node in a data structure), used hazard pointers to delay their reclamation. This incurred memory fence overhead associated with the acquisition of hazard pointers. In contrast, VBR allows for reclamation of retired nodes without waiting for other threads and increments a global epoch variable, which each thread reads to learn about the reclamation event, instead of using warning bits like those in previous approaches. 
Every time a thread reads the content of a node, it compares its last read epoch value with the current global epoch to learn whether a reclamation event has occurred. If the global epoch has changed since it was last read, it implies that the node might have been reclaimed and the content is stale, so the reading thread rolls back to a safe checkpoint in the program. Similarly, any updates to a field in a node are done using wide-compare-swap (WCAS), which uses version numbers to detect if the node has potentially been reclaimed. If so, the WCAS fails and the operation rolls back to a previous safe checkpoint. Each node must be instrumented with a birth epoch and a retire epoch. Additionally, every mutable field of the node has an associated version number, which is essentially the birth epoch used along with the birth and retire epochs to decide if the accesses are valid (i.e. the node was not reclaimed).


While the OA, AOA, FA, and VBR algorithms offer an alternative low overhead paradigm for safe memory reclamation, they rely on a strategy where \textit{unsafe accesses} to a potentially reclaimed memory are performed first, followed by a check to ensure that the accessed memory has not been reclaimed. Originally, this approach introduced challenges for developers: they either need to handle segmentation faults (potentially losing a useful debugging tool) or employ a type-stable allocator that retains all memory pages, which limits the system’s ability to free memory in long-running applications. These trade-offs can impact programmability. However, in 2023, Moreno and Rocha~\cite{moreno2023releasing} proposed a lock-free allocator that addresses this issue by allowing Optimistic Access line of algorithms to safely returning freed memory by exploiting virtual memory's ability to remap pages to a single physical frame. 


\subsection{Reserve to Free Paradigm}

Nikolaev and Ravindran proposed ~\cite{nikolaev2019hyaline} Hyaline, which flips the traditional approach of reserving nodes before access. Instead of reserving the nodes threads intend to access, threads reserve the nodes they intend to free and require other threads to respond, stating that they don't need these nodes. 

A 2D grid of linked lists mediates this communication, where rows represent threads and columns represent batches of retired nodes. Threads executing data structure operations are considered active as they can hold a reference to a retired node; otherwise, a thread is inactive.

A thread that intends to free a batch of nodes links this batch into the grid, making the batch accessible from each active thread's row. The idea is that when active threads become inactive, they traverse their lists, unreserving the batches they encounter, enabling them to be freed (when the last thread with access to them unreserves). The algorithm uses a limited form of reference counting for batches in this grid. A batch can be reclaimed as soon as all threads communicate that they no longer can access nodes in the batch (i.e., when the reference counter of the batch drops to zero).

However, similar to EBR, if a thread stalls, it may never mark itself as inactive by decrementing a reference counter of a batch attached to its row. Other threads thinking that the thread is still active could continue reserving their batches by attaching to this thread's row. This can prevent the reclamation of an unbounded number of batches.


This problem is addressed in a subsequent version, where threads also announce the last era in which they were active, and nodes record the era in which they were allocated (birth era), similar to IBR. Therefore, a stalled thread could only prevent nodes of batches that were allocated in the era on or before the thread stalled. Nodes allocated after the era announced by the stalled thread continue to be reclaimed.

However, like IBR, a stalled thread can prevent reclamation of all objects that were allocated up until the thread stalled. For a long-running operation, a thread may continuously keep reserving the latest era, essentially reserving all epochs since the beginning of the execution, preventing an unbounded number of nodes from being reclaimed.

This issue with robustness is fixed in the subsequent lock-free reclamation scheme named Crystalline-L~\cite{nikolaev2024family}. Crystalline-L follows a hazard era-like fine-grained reservation approach where only the eras corresponding to the objects being accessed are reserved, and previous reserved eras are overwritten so that objects corresponding to those now unreserved eras can be freed.

Briefly, instead of one row per thread in the grid, Crystalline-L has NUM\_HE rows per thread, where NUM\_HE is the number of hazard eras a data structure operation needs to reserve at once. And, like the previous versions, batches of retired nodes are inserted into all rows in the grid that represent an active thread in the current era (so that rows representing old eras do not grow). Since threads can only reserve a limited number of eras at a time, the amount of garbage that a stalled thread can prevent from being freed is reduced, similar to that in hazard eras.

The authors improved the progress guarantees of Crystalline-L by making it wait-free (Crystalline-W). Essentially, the wait-free version of Crystalline-L replaces a CAS loop that links batches into the grid with an unconditional SWAP to prevent a thread from repeatedly failing to attach its retired batch because other threads succeed in attaching their own batches. In doing so, the authors solve interesting race conditions by tainting the next pointers of the attached nodes. Another potential unbounded loop occurs in Crystalline-L when a thread, during its protect call, is unable to converge on a value of the era in which it accessed a reference due to continuously changing global era. To enable wait-freedom, the scheme uses a fast-path slow-path approach to eliminate the unbounded loop.

Performance-wise and memory-efficiency-wise, Crystalline is also shown to have improved over Hyaline.

\section{SMR in Garbage Collected Languages like Java}

Garbage collected languages like Java provide access to non garbage collected off-heap memory~\cite{java22}. This capability allows programmers to execute system calls and access libraries written in C/C++, and helps to avoid the scalability issues of garbage collection, which worsen with an increasing amount of on-heap memory footprint in multithreaded data structures. Several big data systems use off-heap memory support in Java; for example, Cassandra~\cite{cassandra} and HBase~\cite{apachehbase}.

Solving the problem of concurrent memory reclamation for unmanaged off-heap memory in these programming languages requires safe memory reclamation algorithms. Fakhoury et al.\cite{fakhoury2024nova} draw inspiration from SMR techniques to propose an abstraction that allows safe and efficient allocation, access, and deallocation of off-heap memory. Specifically, they utilize Optimistic Access\cite{cohen2015efficient} to protect each read or write operation to an off-heap memory location. They augment each off-heap \textit{memory location} and an on-heap \textit{shared pointer} pointing to that memory location with a special header consisting of a version number field and a boolean \textit{deleted} field.

When a pointer is created to an off-heap memory location, the current era is assigned to the headers of both the memory location and the shared pointer. During each subsequent read of the shared pointer, the version numbers in the headers and the deleted bit of the memory location are compared. If these have not changed since the last access of the shared pointer, the read is allowed; otherwise, it fails.

As in Optimistic Access, reading a deleted memory location is allowed for validation. During writes, a hazard-pointer-like technique is used. Specifically, upon writing to an off-heap memory location, the memory location is reserved (similar to a hazard pointer reservation). It is then verified that the memory location has not been deleted before the write executes, and afterward the memory location is unreserved (similar to a hazard pointer).

To implement this technique, they proposed an interface named SOMAR, which provides speculative read and write functions to access off-heap memory locations. These operations can fail if the deleted bit indicates that the off-heap memory location was deleted (or reused by allocation to some other shared pointer) or if version-based validation fails. The SOMAR interface differs from a typical SMR interface in that it does not assume that operations cannot return a shared memory location. The implementation of this interface is referred to by the authors as NOVA. This technique has been shown to be useful for data structures where only partial data reside in off-heap memory, though its generalization remains unclear.

\section{Summary}

In this chapter, we provide an extensive survey of concurrent safe memory reclamation algorithms, focusing on their design aspects. We identify three categories of issues with existing approaches to designing safe memory reclamation algorithms and address each in the chapters of this dissertation.

The \textbf{first issue} we uncover is that no existing reclamation technique satisfies all the key properties desirable for a practical safe reclamation algorithm. Some algorithms ensure safety by requiring threads to \textit{wait} for other threads to \textit{discard} their references to deleted shared memory locations before reclaiming them. These algorithms assume an asynchronous shared memory model, where threads may experience arbitrarily long delays. If a thread holding a reference to a memory location encounters a significant delay, no thread may ever reclaim deleted memory locations, potentially leading to an out-of-memory situation as newer memory locations continue to be allocated.

In some algorithms, to ensure that a reference to a shared memory location is not freed during access, a thread eagerly executes reclamation-related steps before accessing a memory location. This approach involves communicating with other threads that might free the memory location. Such measures incur non-trivial overhead even when no reclamation is occurring. The overhead can be \textit{coarse-grained}, incurred at every data structure operation (e.g., lookup, delete, or insert), or very \textit{fine-grained}, occurring whenever a new shared memory location is accessed within data structure operations. As a result, threads sacrifice the performance and scalability benefits of the original data structure.

Other algorithms attempt to avoid the approaches of waiting for threads to discard references or eagerly executing reclamation steps, but these may sacrifice wide applicability to different data structures or ease of programming.

The \textbf{second issue} with existing concurrent memory reclamation algorithms is that they introduce asymmetric synchronization overhead on non-blocking data structures. This asymmetry arises because, even though reclamation events are infrequent, all threads perform relatively frequent reclamation-related steps in their common-case execution path to synchronize with reclaiming threads. This results in non-trivial coarse-grained or fine-grained overhead. Such an approach is counterproductive to the goal of optimizing performance in non-blocking data structures, especially during searches.

The \textbf{third issue} is that concurrent reclamation algorithms often fail to account for the nuances of the broader system in which they are deployed. For instance, they may overlook how memory locations are managed or how access to them is synchronized at lower layers of the system stack, such as allocators and cache subsystems. Designing reclamation algorithms in isolation from these lower layers can lead to inefficiencies. In fact, as detailed in ~\cite{kim2024token}, freeing retired nodes in batches can negatively interact with memory allocators, leading to poor scalability.

To address these three issues, we introduce three paradigms for designing safe memory reclamation algorithms. To address the first issue of missing several desirable properties in a single algorithm, we propose a neutralization paradigm that is discussed in \chapref{chapnbr}, \ref{chap:chapnbrp}, and \ref{chap:chapappuse}. To address the second issue of asymmetric overhead we propose a reactive synchronization based paradigm which is detailed in \chapref{chaprsp}. To address the third issue, we propose a hardware-software co-design paradigm, which we leverage to propose a novel Conditional Access safe memory reclamation algorithm in \chapref{chapirp}.

\chapter{Neutralization Paradigm}
\label{chap:chapnbr}

In \chapref{chapsurvey}, we presented an extensive survey of current safe memory reclamation algorithms (SMR) and identified that existing SMR algorithms fall short of satisfying five key properties defined in \secref{keyterms}: \textit{high performance}, \textit{bounded garbage}, \textit{usability}, \textit{consistent performance}, and \textit{applicability}.
Note, subsequently Sheffi and Petrank~\cite{sheffi2023era} formally defined easy integration which may appear similar to usability, but the key difference is that in our definition we do not allow changing a node's memory layout.
In particular, achieving \textit{bounded garbage} and \textit{high performance} simultaneously is quite challenging. In this chapter, we address this limitation by proposing an alternative approach to safely reclaim memory, called the neutralization paradigm. Using the neutralization paradigm, we propose a safe memory reclamation algorithm called Neutralization-Based Reclamation (NBR).

The outline of this chapter is as follows.
\begin{itemize}
\item \secref{nbrintro} builds upon the survey presented in \chapref{chapsurvey} to emphasize that achieving both \textit{speed} and \textit{bounded garbage} with existing approaches is difficult. It is even harder to achieve other desirable properties, such as \textit{wide applicability}, \textit{usability}, and \textit{consistent performance}, along with these two properties simultaneously.

\item 
\secref{nzparadigm} Explains the neutralization paradigm as a candidate for achieving these properties.

\item \secref{nbr} presents the design and implementation of our safe memory reclamation algorithm: NBR, which leverages POSIX signals to implement the neutralization paradigm.

\item \secref{nbrcorr} proves that NBR prevents use-after-free errors and bounds garbage when used with a data structure.

\item \secref{nbreval} highlights that, in practice, the signals used to implement the neutralization mechanism incur overhead, particularly when the number of threads is high. This chapter presents only a preliminary experiment to illustrate the signaling overhead issue; a more detailed evaluation is provided in \chapref{chapnbrp}.
\end{itemize}

The chapter concludes with a brief summary of the content.

\section{Introduction}
\label{sec:nbrintro}

Researchers have developed a rich variety of manual concurrent memory reclamation algorithms with a diverse spectrum of desirable properties, idiosyncrasies, and limitations.

Perhaps the most popular family of SMR algorithms, Hazard Pointer Based (HP) algorithms~\cite{michael2004hazard, herlihy2005nonblocking, dice2016fast, ramalhete2017brief, wen2018interval, nikolaev2020universal} ensure safe reclamation and \textit{bounded garbage} by requiring readers to explicitly reserve nodes before accessing them.
This allows a thread attempting to reclaim a set of nodes to identify which ones it can safely reclaim (i.e., the ones not reserved by any thread). 
Since a reader must reserve each node before accessing it, HP-based SMR algorithms impose significant overhead, especially in data structures where threads repeatedly follow pointers (e.g., lists and trees).
Reference counting based reclamation (RCBR) techniques~(e.g., \cite{detlefs2002lock}) typically offer similar overheads and can have unbounded garbage if unlinked nodes have cycles.

In contrast, Epoch-Based Reclamation (EBR) algorithms~\cite{brown2015reclaiming, hart2007performance,fraser2004practical,mckenney1998read} ensure safety in a highly efficient way by dividing execution into \textit{epochs} and reclaiming nodes in large batches.
Intuitively, the epoch changes when all threads have ``forgotten'' pointers to nodes that were reclaimed in a previous epoch.
Thus, whenever the epoch changes, all threads can reclaim a batch of retired nodes. 
The synchronization needed to track epoch changes imposes significantly less overhead than HP-based algorithms, but a stalled thread can prevent the epoch from advancing, so garbage is not bounded. 

Hybrid SMR algorithms attempt to blend the best attributes of these approaches, combining ideas from each~\cite{gidenstam2008efficient, braginsky2013drop, brown2015reclaiming, nikolaev2019hyaline, anderson2021concurrent}.
Such algorithms typically offer \textit{high performance} and some notion of \textit{bounded garbage}, but often compromise significantly on \textit{applicability} and \textit{usability}.

One such hybrid algorithm, DEBRA+~\cite{brown2015reclaiming}, a variant of EBR with a restricted form of hazard pointers, is both efficient and bounds garbage.
It achieves bounded garbage using POSIX signals and data structure specific recovery code.
Specifically, a thread whose reclamation is delayed by a slow thread will send a signal to the slow thread.
Upon receipt of the signal, the thread executes its recovery code and then restarts its data structure operation, 
allowing reclamation to continue and ultimately guaranteeing a bound on the number of unreclaimed nodes. 
However, the bound on garbage comes at the cost of both \textit{usability} and \textit{applicability}, as users need to write data structure specific recovery code that is not always straightforward or even possible.
Moreover, it is not clear how DEBRA+ could be used for lock-based data structures, as neutralizing a thread that holds a lock could cause deadlock.

In summary, HP bounds garbage but incurs high overhead, EBR offers low overhead but allows unbounded garbage, and hybrid approaches compromise significantly on applicability and usability.

\subsubsection{Goal}
\textit{``To design concurrent memory reclamation algorithms that match or outperform existing algorithms, bound garbage, are simple to use, exhibit consistent performance and are applicable to a large class of data structures"}

\subsubsection{Observation}
We find that the majority of existing solutions are pessimistic. Meaning, each thread either protects every reference or a set of references using epochs. The former approach, taken by hazard pointers and reference counting, incurs per access overhead but allows objects that are not protected to be immediately reclaimed. In contrast, the latter approach, as in epoch based reclamation, significantly reduces per access overhead to per operation but leads to unbounded garbage impacting progress when a thread gets delayed. This pessimistic approach, both while accessing objects and while freeing them, makes achieving speed and bounded garbage simultaneously a hard task.

Our solution to this challenge takes a contrary approach. We believe that a solution to concurrent memory reclamation lies in answers to the following two broad challenges:
\begin{enumerate}
    \item Can we eliminate the need to reserve every object beforehand and instead allow threads to freely access all references without any synchronization? Our hope is that solving this issue will address the main source of overhead and hindrance in usability.
    \item Can we, instead of passively waiting for references to become eligible for reclamation, actively force threads to communicate to determine references that can be reclaimed? 
\end{enumerate}

\section{Neutralization to Enforce Coordination}
\label{sec:nzparadigm}

Taking hints from Brown's DEBRA+~\cite{brown2015reclaiming}, a reclaiming thread (termed as reclaimer) could enforce communication by signaling all threads to give up their current references to objects and restart their operation.
This way the reclaimer could continue reclamation.
However, this gives rise to another problem. Restarting an operation in the middle of an update could render the underlying data structure incorrect. DEBRA+ resolves this by requiring programmers to provide a HP-based recovery code which can be used to restart the operation or clean up inconsistencies that may arise due to restarting in the middle of updates. However, that limits DEBRA+’s applicability to specific lock-free data structures and usability. For instance, certain algorithms that use locks, like a lock-based binary search tree featuring lock-free search operations as in \cite{david2015asynchronized}(DGT), are incompatible with DEBRA+. This is because restarting after lock acquisition could lead to deadlock.

Fortunately, many concurrent data structures have a typical pattern: long read-only phases followed by short write phases\cite{timnat2014practical}. This observation helps us design an alternative fast and correct restart strategy. 
Simply put, a thread in the read phase (termed as reader) can restart safely, while a thread in the write phase (termed as writer) should protect all references it will use and thus does not need to restart. This helps the reclaimer to assume that all references it intends to free are either safe to free or are protected, thus facilitating the immediate determination of references that can be reclaimed. 
In NBR, a thread using neutralization can enforce reclamation without waiting for other threads, similar to DEBRA+, but it differs in that it does not require recovery code and is widely applicable to a variety of data structures.

We define a reference as hazardous when the memory it points to \textit{can} be concurrently freed, according to the terminology from Maged Michael's hazard pointers paper~\cite{michael2004hazard}.

\begin{definition}
     In NBR, a thread $T$ is said to have been neutralized in response to a POSIX signal if one of the following conditions hold true:
    \begin{enumerate}
        \item $T$ discards all hazardous references and restarts. or
        \item $T$ doesn't restart but protects all the hazardous references.
    \end{enumerate} 
\end{definition}

The neutralization paradigm avoids waiting for stalled threads to reclaim memory (unlike epoch based approaches) and does not require readers to repeatedly reserve nodes during searches (unlike pointer reservation or epoch reservation approaches).
More specifically, readers do not reserve nodes, and writers reserve nodes exactly once, after they have finished gathering pointers to all the nodes they intend to modify, and before they have begun modifying these nodes.
Synchronization between reclaiming threads and threads accessing the data structure is mediated via \textit{neutralizing}.
This allows the neutralization paradigm to guarantee bounded garbage with the speed of epoch based algorithms, while remaining applicable to many data structures that are incompatible with existing pointer reservation and hybrid SMR algorithms.

Neutralization in NBR, as in DEBRA+, is achieved by using POSIX signals. However, unlike DEBRA+, during a neutralizing event, participating threads, based on their roles as readers or writers, take specific actions to ensure that all accesses are safe. Readers discard all references they acquired before the event and retry their operation from the beginning, effectively resetting their state. On the other hand, writers continue to operate unimpeded as long as they promise only to access nodes that they reserved before the neutralization event. If a writer has not yet gathered pointers to all the nodes it intends to modify, it behaves like a reader.

\section{NBR: Neutralization Based Reclamation}
\label{sec:nbr}
In this section we discuss Neutralization Based Reclamation--- a safe memory reclamation algorithm designed using neutralization paradigm. 
By neutralizing other threads, a reclaimer can effectively force readers (or writers that have not yet begun writing) to drop their pointers to any nodes that the reclaimer might be trying to reclaim.
Crucially, this allows a reclaimer to reclaim nodes without having to wait for an epoch to advance (unlike EBR), and readers can avoid the overhead of reserving individual nodes (unlike HP).
NBR retains the key benefit of HP: There is an upper bound on the number of garbage nodes other threads could prevent from being reclaimed, which is a function of the number of nodes that threads can reserve at a given time and the number of threads in the system. 
In practise, this bound is small as the number of threads and the number of reserved nodes (hazard pointers at a time) per thread are typically constant, e.g., lists require only 2 hazard pointers per thread.
In NBR, only the nodes that active writers intend to modify are reserved.
This allows NBR to guarantee bounded garbage without introducing per-node overheads for reservations.

In more detail, each thread places
unlinked objects in a thread-local buffer, and when the buffer’s size exceeds a predetermined threshold, the thread sends a neutralizing signal to all other threads. Upon receiving such a signal, a thread checks whether its current data structure operation has already done any writes to shared memory (nodes), and if not, restarts its operation (using the C/C++ procedures \func{sigsetjmp} and \func{siglongjmp}). Otherwise, it finishes executing its operation, ensuring that it must only access reserved nodes.


\subsection{Assumptions on data structures}
\label{sec:dsassmp}

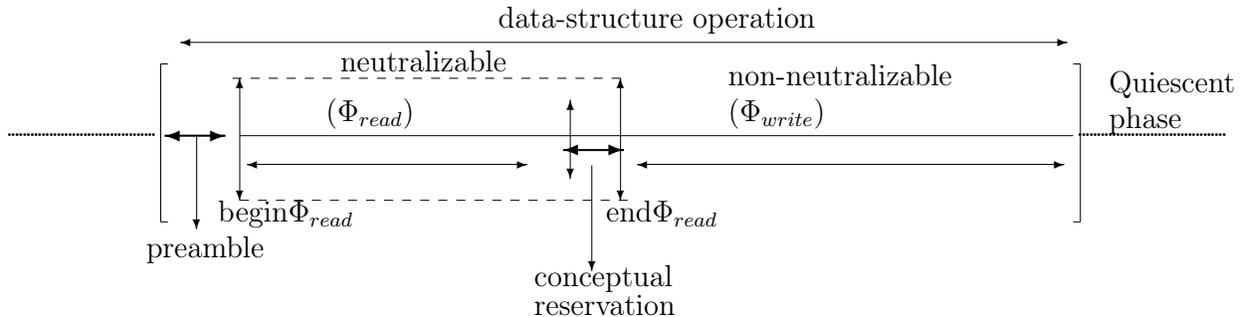
\begin{figure}[h]
\centering
\resizebox{1\linewidth}{!}{\setlength{\unitlength}{4144sp}%
\begingroup\makeatletter\ifx\SetFigFont\undefined%
\gdef\SetFigFont#1#2#3#4#5{%
  \reset@font\fontsize{#1}{#2pt}%
  \fontfamily{#3}\fontseries{#4}\fontshape{#5}%
  \selectfont}%
\fi\endgroup%
\begin{picture}(7629,1974)(2869,-3001)
\thinlines
{\color[rgb]{0,0,0}\put(9541,-2401){\line( 1, 0){ 45}}
\put(9586,-2401){\line( 0, 1){990}}
\put(9586,-1411){\line(-1, 0){ 45}}
}%
{\color[rgb]{0,0,0}\put(4366,-2041){\vector(-1, 0){  0}}
\put(4366,-2041){\vector( 1, 0){1755}}
}%
{\color[rgb]{0,0,0}\multiput(10486,-1861)(-23.68421,0.00000){39}{\makebox(1.5875,11.1125){\small.}}
}%
{\color[rgb]{0,0,0}\put(4321,-1501){\vector( 0, 1){  0}}
\put(4321,-1501){\vector( 0,-1){765}}
}%
{\color[rgb]{0,0,0}\put(4051,-1861){\vector( 0,-1){585}}
}%
\thicklines
{\color[rgb]{0,0,0}\put(3850,-1861){\vector(-1, 0){  0}}
\put(3850,-1861){\vector( 1, 0){381}}
}%
\thinlines
{\color[rgb]{0,0,0}\put(3871,-1411){\line(-1, 0){ 45}}
\put(3826,-1411){\line( 0,-1){990}}
\put(3826,-2401){\line( 1, 0){ 45}}
}%
{\color[rgb]{0,0,0}\multiput(3781,-1861)(-23.68421,0.00000){39}{\makebox(1.5875,11.1125){\small.}}
}%
{\color[rgb]{0,0,0}\put(4321,-1861){\line( 1, 0){5220}}
}%
\thicklines
{\color[rgb]{0,0,0}\put(6346,-1951){\vector(-1, 0){  0}}
\put(6346,-1951){\vector( 1, 0){381}}
}%
\thinlines
{\color[rgb]{0,0,0}\put(6391,-1635){\vector( 0, 1){  0}}
\put(6391,-1635){\vector( 0,-1){496}}
}%
{\color[rgb]{0,0,0}\put(6706,-1501){\vector( 0, 1){  0}}
\put(6706,-1501){\vector( 0,-1){765}}
}%
{\color[rgb]{0,0,0}\put(6810,-2041){\vector(-1, 0){  0}}
\put(6810,-2041){\vector( 1, 0){2675}}
}%
{\color[rgb]{0,0,0}\multiput(4276,-1501)(113.02326,0.00000){22}{\line( 1, 0){ 56.512}}
}%
{\color[rgb]{0,0,0}\multiput(4321,-2266)(113.02326,0.00000){22}{\line( 1, 0){ 56.512}}
}%
{\color[rgb]{0,0,0}\put(3949,-1276){\vector(-1, 0){  0}}
\put(3949,-1276){\vector( 1, 0){5556}}
}%
{\color[rgb]{0,0,0}\put(6526,-2047){\vector( 0,-1){663}}
}%
\put(9766,-1591){\makebox(0,0)[lb]{\smash{{\SetFigFont{12}{14.4}{\rmdefault}{\mddefault}{\updefault}{\color[rgb]{0,0,0}Quiescent}%
}}}}
\put(9766,-1816){\makebox(0,0)[lb]{\smash{{\SetFigFont{12}{14.4}{\rmdefault}{\mddefault}{\updefault}{\color[rgb]{0,0,0}phase}%
}}}}
\put(4186,-2401){\makebox(0,0)[lb]{\smash{{\SetFigFont{12}{14.4}{\rmdefault}{\mddefault}{\updefault}{\color[rgb]{0,0,0}begin$\Phi_{read}$}%
}}}}
\put(6616,-2401){\makebox(0,0)[lb]{\smash{{\SetFigFont{12}{14.4}{\rmdefault}{\mddefault}{\updefault}{\color[rgb]{0,0,0}end$\Phi_{read}$}%
}}}}
\put(4951,-1456){\makebox(0,0)[lb]{\smash{{\SetFigFont{12}{14.4}{\rmdefault}{\mddefault}{\updefault}{\color[rgb]{0,0,0}neutralizable}%
}}}}
\put(4861,-1771){\makebox(0,0)[lb]{\smash{{\SetFigFont{12}{14.4}{\rmdefault}{\mddefault}{\updefault}{\color[rgb]{0,0,0}($\Phi_{read}$)}%
}}}}
\put(5941,-1186){\makebox(0,0)[lb]{\smash{{\SetFigFont{12}{14.4}{\rmdefault}{\mddefault}{\updefault}{\color[rgb]{0,0,0}data-structure operation}%
}}}}
\put(7381,-1546){\makebox(0,0)[lb]{\smash{{\SetFigFont{12}{14.4}{\rmdefault}{\mddefault}{\updefault}{\color[rgb]{0,0,0}non-neutralizable}%
}}}}
\put(7381,-1771){\makebox(0,0)[lb]{\smash{{\SetFigFont{12}{14.4}{\rmdefault}{\mddefault}{\updefault}{\color[rgb]{0,0,0}($\Phi_{write}$)}%
}}}}
\put(6166,-2806){\makebox(0,0)[lb]{\smash{{\SetFigFont{12}{14.4}{\rmdefault}{\mddefault}{\updefault}{\color[rgb]{0,0,0}conceptual}%
}}}}
\put(6166,-2986){\makebox(0,0)[lb]{\smash{{\SetFigFont{12}{14.4}{\rmdefault}{\mddefault}{\updefault}{\color[rgb]{0,0,0}reservation}%
}}}}
\put(3736,-2626){\makebox(0,0)[lb]{\smash{{\SetFigFont{12}{14.4}{\rmdefault}{\mddefault}{\updefault}{\color[rgb]{0,0,0}preamble}%
}}}}
\end{picture}%
}
\caption{
Typical \nbr compatible data structure operation with preamble (optional), read-phase~(\rdp) that ends with a conceptual reservation phase followed by write-phase~(\wtp). 
}
\label{fig:3phase}
\end{figure}

\nbr can be applied to data structure operations that follow a simple template.
Each operation should consist of a sequence of \textit{phases}: preamble, read-phase (\rdp), (conceptual) reservation phase, and write-phase (\wtp), as illustrated in \figref{3phase}.
The read-phase and write-phase can be repeated any number of times as long as every read-phase starts from an entry point in the data structure operation.
As we shall see, this template ensures that it is safe for a thread to restart its \rdp from any arbitrary point in the \rdp.
Note that outside a data structure operation threads are in {\em quiescent phase}, which is like preamble phase.

Immediately before begin\rdp a checkpoint is set so that threads in \rdp can be neutralized (discard all local references to shared memory) and restart.

\subsection{Requirements}
\label{sec:nbrreq}

To integrate NBR correctly with a data structure, a programmer should ensure that the following requirements are satisfied:
\begin{enumerate}[label=\Roman*]
    \item Each \rdp in a data structure starts from an entry point (e.g. the head in lists or root in trees). This means shared node references from previous phases are not retained.
    \item All references to shared pointers that could be accessed within a \wtp are reserved before entering the \wtp.
\end{enumerate}

Moreover, given that NBR's restarting mechanism employs the \func{sigsetjmp} and \func{siglongjmp} functions from the POSIX signal API, it is important for programmers to be aware of key guidelines for their use throughout each phase.
More detailed discussion on these requirements and their impact on applicability to various data structures, along with a discussion on how to integrate NBR with the Harris list~\cite{harris2001pragmatic} and lazylist~\cite{heller2005lazy} appears in \chapref{chapappuse} (\secref{comptds} and \secref{usb}, respectively).


\begin{enumerate}
    \item \textbf{Preamble.} Accesses (reads/writes/CASs) to program global variables are permitted. System calls (heap allocation/deallocation, file I/O, network I/O, etc.) are permitted. Access to shared records, for example, nodes of a shared data structure, is \textit{not} permitted. 
    We use the term global to mean program-wide variables (different from the data structure nodes), for example, these global variables may hold configuration values, counters, or thread local data structures that are meant to persist between operations.

    \item \textbf{{\em Read phase} (\rdp).} Reading program global variables is permitted and reading shared records is permitted if pointers to them were obtained during this phase (e.g., by traversing a sequence of shared objects by following pointers starting from a global variable---i.e., a \textit{root}). Writes/CASs to shared records, writes/CASs to shared globals, and system calls, are \textit{not} permitted. To understand the latter restriction, suppose an operation allocates a node using \func{malloc} during its \rdp, and before it uses the node, the thread performing the operation is neutralized. This would cause a memory leak.
    
    Additionally, writes to thread local data structures are not recommended. To see why, suppose a thread maintains a thread local doubly-linked list, and also updates this list as part of the \rdp of some operation on the shared data structure. If the thread is neutralized in middle of its update to this local list, it might corrupt the structure of the list. Note that in some cases, it is okay to write to a thread local variable or data structure. For example, if a thread wants to count how many times it gets neutralized while executing a \rdp then neutralization of this thread is not a problem. 
    
    More generally, any idempotent updates to thread local data structures should remain correct, even if the \rdp is restarted. The programmer must simply proceed with caution.
    These restrictions exist because \nbr uses \func{sigsetjmp} to implement a checkpoint from which a thread restarts and neutralization fires a POSIX signal that (sometimes) causes a thread to \func{siglongjmp} to its last checkpoint. For readers not familiar with the caveats of using these subroutines, we explain further the rules in \secref{nbrrules}. \emph{In summary, a programmer should be vary of doing anything in the read-phase which when undone (due to restarting) may cause incorrect behaviour.}
    
    \item \textbf{Reservation phase.} This is not a concrete phase in a data structure operation, but rather a \textit{conceptual stage} where shared nodes to be modified in the next phase are identified in a data structure operations. At this stage the identified nodes can be reserved by a programmer. We denote these nodes as \textit{reserved} nodes.

    \item \textbf{{\em Write phase} (\wtp ).}  Accesses (reads/writes/CASs) to global variables, and system calls, are permitted. Accesses (including write/CAS) to shared records are permitted only if the records are \textit{reserved}. To understand what could go wrong if this restriction is violated, we need to better understand \nbr, so we will return to this restriction with an example in \secref{exwtp}.
\end{enumerate}


\subsection{NBR interface}


NBR provides the following essential functions. Detailed information regarding the interface and its implementation can be found in \algoref{nbr} in \secref{nbrimpl}.

\begin{itemize}
    \item \Call{Checkpoint}{ }, invoked immediately before the start of the read phase in a data structure operation.
    \item \Call{begin\rdp}{ }, marks the beginning of the read phase in an operation and also this is used to clear reservations from previous write-phase, if any. 
    \item \Call{end\rdp}{...}, Invoked at the end of the read phase to reserve the set of nodes identified in this phase, indicating the start of the write phase. The set of nodes to be reserved are passed as parameter to this function.
    \item \Call{retire}{ }, Invoked after a node is unlinked to temporarily buffer it in a per-thread retire list.
\end{itemize}

\subsection{Design}
\label{sec:nbrdesign}

In \nbr, each thread accumulates nodes that it has unlinked in a private buffer (or \textit{limbo bag}).
When the size of a thread $T$'s buffer exceeds a predetermined threshold, the thread sends a neutralizing signal to all other threads. 
Upon receipt of such a signal, the behavior of a thread $T1$ depends on the phase in which it is executing.

\textbf{If $T1$ is in a quiescent phase, or preamble}, it holds no pointers to shared nodes, and does not prevent $T$ from reclaiming nodes in its buffer.
$T1$ simply continues executing (effectively ignoring the signal).

On the other hand, \textbf{if $T1$ is in \rdp}, it may hold pointers to nodes in $T$'s buffer.
If $T1$ were to continue executing, it would have to prevent $T$ from reclaiming nodes.
However, note that $T1$ has not yet performed any modifications to any shared nodes (since it is still in \rdp).
So, $T1$ can simply discard all of its pointers (that are in its private memory), and jump back to the start of its \rdp, without leaving any shared data structures in an inconsistent state.
To implement this jump, every data structure operation invokes \func{sigsetjmp} at the start of its \rdp, which creates a \textit{checkpoint} (saving the values of all stack variables).
Subsequently, a thread can invoke \func{siglongjmp} to return to the last place it performed \func{sigsetjmp} (and restore the values of all stack variables).
Technically \func{sigsetjmp} saves the current stack frame. Stack variables defined deeper on the stack will not necessarily be saved or restored.
The thread can then retry executing its \rdp, traversing a new sequence of nodes, starting from the \textit{root}, without any risk of accessing any nodes \func{free}d by $T$ (since those are no longer reachable).

The subtlety in \nbr arises when \textbf{$T1$ is in \wtp}.
In this case, $T1$ may hold pointers to nodes in the buffer of $T$.
Thus, if it continues executing, $T1$ must prevent $T$ from reclaiming these nodes.
Moreover, since $T1$ may have modified some shared nodes (but not completed its operation yet), we cannot simply restart its data structure operation, or we may leave the data structure in an inconsistent state.
So, $T1$ will \textit{not} restart its operation.
Instead, it will simply \textit{continue executing} wherever it was when it received the signal (effectively ignoring the signal).
At this point, the reader might wonder how we \textit{simultaneously avoid}: 
\begin{itemize}
    \item [\textbf{A.}] {Blocking the reclamation of $T$, and}
    \item [\textbf{B.}] {The possibility that $T1$ continues executing and one of the nodes it is about to access is concurrently \func{free}d by $T$.}
\end{itemize}

The solution lies in the \textit{reservation} phase of $T1$.
During the reservation phase of $T1$, just before it begins its \wtp, $T1$ \textit{reserves} all shared nodes it will access in its \wtp by \textit{announcing} pointers to them in a shared array.
These reservations serve a similar purpose to hazard pointers (HP), but are quite different from HP in terms of performance and safety guarantees. For example, there is no need for validation after the reservation step. This is mainly because by the time $T1$ is in its \wtp (so it will ignore any neutralization signals), its reservations are guaranteed to be visible to all threads,
and $T$ can refer to these reservations to avoid reclaiming any of those nodes. More detailed discussion appears in \secref{nbrimpl} where we discuss the implementation that will explain why it is safe not to validate and in \secref{usb} where we discuss how programmers use the NBR interface.

In short, operations in the \rdp discard their pointers and restart, and operations in the \wtp must have reserved them. This empowers the \textit{reclaimers} to assume that \rd{s} lose all of their pointers in response to neutralizing, and the \wt{s} lose all pointers that are not reserved. As a result, once a thread sends a neutralizing signal to all other threads, it can scan all reservations, and free any nodes in its buffer (limbo bag) that are not reserved.

\begin{algorithm} 
\small
    \caption{Neutralization Based Reclamation (\nbr) Interface. 
    }\label{algo:nbr}

    \begin{algorithmic}[1]

\Statex \textbf{thread local variable:} 
    \State int \texttt{tid;} \Comment{\texttt{current thread id}}
    \State node *\texttt{limboBag;} \label{lin:rb}\Comment{\texttt{retired nodes. Maxsize:S}}
    \State atomic<bool> \texttt{restartable;} \label{lin:rst}\Comment{\texttt{tracks \rdp/\wtp}}
    \State node *\texttt{tail;} \label{lin:tl}\Comment{\texttt{last node in limboBag}}
\Statex \textbf{shared variable:\texttt{n:\#threads, r:max reserved nodes.}} 
\State atomic<node*> \texttt{reservations[n][r];} \label{lin:resv} \Comment{\texttt{Assume:r$<<$S}}
\Statex
\LComment{Programmers create a checkpoint just before invoking begin\rdp to enable restart of a thread from this checkpoint.}     
        \Procedure{CHECKPOINT}{ } \Comment{\texttt{must be INLINED.}} \label{lin:chkp}
            \While{\texttt{sigsetjmp(...))}} \Comment{\texttt{signal mask not saved.}} \label{lin:sigset}
            \State {unblock neutralizing signal.} \Comment{\texttt{post siglongjmp.}}
            \EndWhile
        \EndProcedure
\Statex

        \Procedure{signalHandler}{ } \label{lin:sigh}
            \If{\texttt{!restartable}} 
            \State {\texttt{return;}} \Comment{\texttt{in \wtp, ignore signal and return.}}
            \EndIf
            \State \texttt{siglongjmp(...);}\Comment{\texttt{in \rdp, jump to checkpoint.}}\label{lin:sigjmp}            
        \EndProcedure
\Statex
\LComment{This function clears previous reservations and sets the restartable flag to true, marking beginning of a read-phase.}
        \Procedure{begin\rdp}{ }\label{lin:brp}
            \State \texttt{reservations[tid].clear();} \label{lin:arpc}
            \State \texttt{restartable = True;}\label{lin:nz01}
            
        \EndProcedure
\Statex
\LComment{programmer invokes this function to reserve nodes identified in read-phase and marks the beginning of subsequent write phase.}
        \Procedure{end\rdp}{rec = \{$rec_1,\cdot\cdot\cdot, rec_R$\}}\label{lin:bwp}
            \State \texttt{reservations[tid]=rec; }\label{lin:arp_add}
            \State \texttt{restartable = False;} \label{lin:nz10}
        \EndProcedure
\Statex
        \Procedure{retire}{rec} \label{lin:func_retire}
            \If{\texttt{isLimboBagTooLarge()}} \label{lin:ioop}
                \State \texttt{\Call{signalAll}{ };}\label{lin:sa} 
                \State \texttt{\Call{reclaimFreeable}{tail};}\label{lin:rf}
            \EndIf \label{lin:ioop_eif}
            \State \texttt{limboBag[tid].append(rec);}\label{lin:re_app}
        \EndProcedure
\Statex        
        \Procedure{reclaimFreeable}{tail} \label{lin:func_rf}
            \State \texttt{$A$=\Call{collectReservations()}{}; }\label{lin:carp} 
            \State \texttt{$R$=limboBag[tid].remove(A, tail);}\label{lin:retg} 
            \State \texttt{free(\{$R$\});} \label{lin:free}
        \EndProcedure
    \end{algorithmic}
\end{algorithm}

\subsection{Implementation}
\label{sec:nbrimpl}
In this section we present the interface of NBR and its implementation. The pseudocode for \nbr\ appears in \algoref{nbr}.
We assume the C++ memory model.
In particular, this means \texttt{atomic} accesses are by default sequentially consistent. 

\subsubsection{Description of Variables}
Each thread has an ID denoted by \texttt{tid}.
The thread collects retired nodes in its \texttt{limboBag}, which is implemented as an array.
The \texttt{tail} variable points to the last retired node added to the \texttt{limboBag}.
Each thread also has a local \texttt{restartable} flag, which determines whether the thread will restart if it receives a neutralizing signal.
This flag is set when a thread enters its read-phase (\rdp), and reset 
when the thread exits the \rdp. 
Prior to exiting the \rdp, the thread enter a conceptual reservation phase.
During this phase, it announces all the nodes it will access in its subsequent write-phase (if any), using a single-writer multi-reader (SWMR) \texttt{reservations} array. 
We assume the maximum number \texttt{r} of reserved nodes is strictly less than the maximum size \texttt{s} of a limbo bag.

\subsubsection{CHECKPOINT() and SIGNALHANDLER()}
To enable thread restarts, \Call{checkpoint()}{} (\lineref{chkp}) is invoked just before entering the \rdp.
The checkpoint indicates a location in code from which the executing thread can safely restart after it receives a neutralizing signal.
When \texttt{sigsetjmp} (\lineref{sigset}) is executed, it saves the current execution context (stack frame pointer, program counter and register contents) and returns false when executed for the first time, effectively creating a checkpoint from which the thread can restart.

A thread invokes \Call{begin\rdp()}{} to initiate a \rdp.
While in the \rdp, if a neutralizing signal is received, the thread executes a custom signal handler (\lineref{sigh}).
From within the signal handler, it uses \texttt{siglongjmp} (\lineref{sigjmp}) to jump back to the \texttt{checkpoint} (restoring the saved context) and re-executes the \texttt{sigsetjmp}, which this time returns true.
It also unblocks the neutralizing signal (which was automatically blocked by \texttt{siglongjmp} to avoid recursive signal handler invocations).
This effectively restarts the \rdp from scratch with a new checkpoint.

\subsubsection{BEGIN\rdp()}
As mentioned above, a thread invokes \Call{begin\rdp()}{} to start its \rdp (\lineref{brp}).
The thread first clears its previous reservations so that it can reserve new nodes in the reservation phase.
It then sets \texttt{restartable} to true (using a sequentially consistent store; \lineref{nz01}).
This ensures that the thread becomes restartable before it accesses any shared nodes in its \rdp.

\subsubsection{END\rdp(...)}
In order to enter a \wtp, a thread invokes \Call{end\wtp()}{} (\lineref{bwp}).
The thread begins by reserving a sequence of nodes (\lineref{arp_add}) in its designated slots within the \texttt{reservations} array (implementing the conceptual reservation phase).
Then, it resets the \texttt{restartable} flag (with a sequentially consistent store; \lineref{nz10}). 
This store ensures that when this thread, in its \wtp, executes a signal handler and sees restartable is false, it has already reserved all of the nodes it will need in its \wtp, so any reclaimers that read these reservations will see all of the nodes the thread will access. 
Note that these reservations are different from classic hazard pointers based reservations, both because nodes are reserved \textit{after} the nodes have been accessed, and because no special validation is required after the stores to \texttt{reservations}.

\subsubsection{RETIRE()}
Whenever a thread unlinks a node from a data structure, it invokes \Call{retire}, which adds the node (\textit{rec}) to the thread's \texttt{limboBag} (\lineref{re_app}).
If the size of the \texttt{limboBag} exceeds a predefined threshold (\lineref{ioop}), 16k in our experiments, neutralizing signals are sent to all other threads using the \Call{signalAll}{} function (\lineref{sa}).
In our implementation, \Call{signalAll}{} employs \textit{pthread\_kill()} with the standard Linux signal \textit{SIGQUIT}.
To simplify determining the number and IDs of threads that should receive a neutralizing signal, we assume a fixed number of threads in the system, with no dynamic thread entry or exit (see assumption in \secref{introasmp}).

Upon receipt of the signal, a thread immediately executes \Call{signalHandler}{} (\lineref{sigh}).
If a thread is restartable (i.e, in a \rdp), it restarts from its previous \texttt{checkpoint}.
Otherwise, it returns and continues where it left off.

After, a reclaiming thread has sent neutralizing signals to all other threads (\lineref{sa}), it executes \Call{reclaimfreeable}{} (\lineref{rf}).
Within \Call{reclaimfreeable}{} (\lineref{func_rf}), the thread collects all reservations in a set A. 
These reservations are used to determine which nodes are safe to free (\lineref{retg}), and these nodes are then deallocated (\lineref{free}).

\subsection{Coordination for Reclamation}
\label{sec:nzimpl}

In the previous section we discussed the design and implementation of \nbr. In this section we elaborate more on how \rd{s}, \wt{s} and \rl{s} collaborate to achieve safety of reclamation.

\paragraph{Reader-reclaimer handshake.}\label{sec:rrhandshake}
Each thread $T1$ at the time of \textsc{begin\rdp} saves its execution state (program counter and stack frame) using \textit{sigsetjmp} so that when it becomes restartable it can jump back to this state upon receiving a neutralizing signal.
When a \rl $T$ sends a neutralization signal to thread $T1$, the operating system causes the control flow of $T1$ to be interrupted, so that $T1$ will immediately execute a \textit{signal handler} if $T1$ is currently running. Otherwise, if $T1$ is not currently running, the next time it is scheduled to run it will execute the signal handler before any other steps.
The signal handler determines whether $T1$ is restartable by reading the local \nz variable.
If the thread is restartable, then the signal handler will invoke \func{siglongjmp} and jump back to the start of the \rdp (so it is as if $T1$ never started the \rdp{}). 

This behaviour represents a sort of two-step \textit{handshake} between \textit{readers} (threads in \rdp) and \textit{reclaimers} (threads executing lines 16 and 17 in \func{retire}) to avoid scenarios where a reader might access a \func{free}d node. 
A \rl guarantees that before \emph{reclaiming} any of its \emph{unlinked} nodes, it will signal all threads, and all \rd{s} guarantee that they will relinquish any reference to unsafe nodes when they receive a neutralization signal.

\paragraph{Writer-reclaimer handshake.} \label{sec:whandshake}
\textbf{(1)} Each \rl signals all threads before starting to reclaim any nodes.
When a \wt receives a signal, it executes a \emph{signalHandler} that determines the thread is non-restartable, and immediately returns. 
The \rl then goes on to reclaim its \emph{limboBag} (\lineref{rf}), except for any reserved nodes contained therein, independently from the actions of the \wt.
This is safe because a \wt, before entering into the \wtp, \textit{reserves} all of the shared nodes it might access in its \wtp 
(\lineref{arp_add}).
Thus, \textbf{(2)} the \wt guarantees to the \rl that, although it will not restart its data structure operation, it will only access \emph{reserved} nodes.
The \textbf{(3)} \rl, in turn, guarantees it will scan all announcements after signaling and before reclaiming the contents of its \emph{limboBag}, and will consequently avoid reclaiming any nodes that will be accessed by the \wt in its \wtp.

This three-step handshake formed by (1), (2) and (3) avoids scenarios where a \wt might access a \func{free}d node. 
Crucially, all \wt{s} atomically ensure that their reserved nodes are visible to the \rl at the moment they become non-restartable.
In turn, \rl{s} scan reservations \textit{after} sending neutralization signals at which point any thread that does not restart has already made its reservations visible.

\subsection{Reservations in Write-Phase}
\label{sec:exwtp}

In this section, we will trace an incorrect execution that could occur if a thread accesses any node that is not reserved \textit{before} entering the \wtp.

Suppose a thread $T$ is in a \wtp, and sleeps just before it accesses a shared node $rec$, which it has \textit{not} reserved.
Then, another thread $T1$ sends a neutralization signal to $T$ using \func{signalAll}.
Next, $T1$ scans the \emph{reservations} array of the thread $T$.
$T$ did not reserve $rec$ so $T1$ will not find $rec$ in $T$'s reserved nodes (which violates the writer-reclaimer handshake, \secref{whandshake}).
Therefore, $T1$ will assume that $rec$ can be freed safely, and will do so.
Finally, $T$ wakes up and proceeds with its \textit{unsafe} access of $rec$.

 Note, despite the use of HP like reservations, \nbr does not inherit the drawbacks of HP, namely, (D1) high overhead of announcing an HP after every read of a pointer to a shared node and (D2) and inability to be used with data structures that allow traversal over marked nodes~\cite{brown2015reclaiming}. Firstly, \nbr only requires to publish its reservations once per operation right before entering the \wtp and not after every read of a shared node. Thus, it avoids the high overhead of frequent publishing of reservations in HPs.
 Secondly, unlike HP, NBR is safe with data structures that may require traversal over marked ( or logically deleted ) nodes due to \rd-\rl and \wt-\rl handshakes.


\subsection{Optimization: Relaxing the Memory Model}
\label{sec:nbr_implrelx}

In the previous section, we presented the theoretical algorithm (\algoref{nbr}) assuming sequential consistency.
However, in practice, there is considerable interest in relaxing memory models to obtain higher performance. 
An implementation of \nbr\ on a relaxed memory model should provide two essential guarantees.
\begin{guarantee}
\label{guar:gurwtp}
When a thread $T_1$ executing a signal handler reads its \texttt{restartable} variable and sees $false$, $T_1$ must have already written its reservations, and any other thread $T_2$ that subsequently reads $T_{1}'s$ reservations must see those reservations that were written by $T_1$.
\end{guarantee}

\begin{guarantee}
\label{guar:gurrdp}
A thread $T$ must store true to its \texttt{restartable} variable before it discovers any new nodes in its \rdp.
\end{guarantee}

On modern Intel and AMD x86/64 systems, which implement total store order (TSO), one possible optimization is to set the C++ \texttt{atomic} memory order of all loads and stores in the algorithm to \texttt{memory\_order\_relaxed}, except for the stores to the \texttt{restartable} variable.
As long as the stores to \texttt{restartable} are sequentially consistent the algorithm remains correct. Precisely, the implementation should ensure \guarref{gurwtp} and \guarref{gurrdp}.


If \guarref{gurwtp} is violated, the following incorrect execution may occur on x86/64: a thread $T$ reserves node $rec$ and writes $0$ to \texttt{restartable}.
Suppose the reservations of thread $T$ remain in the processor's store buffer, and are not visible to other threads yet.
Then, another thread $T1$ sends a neutralizing signal to $T$, scans the reservations and does not see $rec$, and consequently frees $rec$.
Upon receiving the signal, $T1$ will \textit{not} restart since it has already written 0 to restartable.\footnote{A detailed explanation of the behaviour of store buffers and serializing instructions on modern processor architectures is out of scope. But, briefly, if $T$ and $T1$ are executing on different processors, then $T$ will not see the effects of any pending writes in the store buffer of $T1$, but $T1$ \textit{will} see the effects its own pending writes in order to maintain sequential consistency.}
Instead, it continues executing, and dereferences $rec$ (accessing a freed node).

Similarly, \guarref{gurrdp} ensures that \texttt{restartable} will be set to $true$ before $T_1$ discovers any new nodes in its \rdp.
In a more relaxed memory model, it could be violated as it would be possible for $T$'s reads of shared nodes to be reordered before this write. In other words, some reads of shared nodes in \rdp may appear to occur in \textit{preamble} (or previous \wtp{}) due to instruction reordering.
This could end up breaking the rule that says access to shared nodes is not permitted in \textit{preamble} as discussed in \secref{dsassmp}.
As a result, the thread, which is not yet restartable, might ignore a neutralization signal and access a \func{free}d node.

Note that on modern Intel and AMD x86/64 systems, these sequentially consistent stores to \texttt{restartable} typically compile to a \texttt{mov} instruction followed by an \texttt{mfence}, which is much slower than an \texttt{xchg} instruction (\texttt{atomic\_exchange} in C++) which would also offer the required semantics.
(In fact, some modern compilers will emit \texttt{xchg} for these stores with high optimization compiler flags.)


\section{Correctness}
\label{sec:nbrcorr}
\begin{assumption}
\label{asm:sigg}
If \xspace \tid{i} in order to reclaim its retired nodes sends a signal to \tid{j}, then by the time \tid{i} finishes sending the signal, \tid{j} is guaranteed to receive it and execute a signal handler before taking further steps in its program.
\end{assumption}

Note, in theory, we simplify by assuming signals are delivered immediately, but in practice there is some delay, and some understanding of the hardware and operating system stack is necessary to understand how large the delay can be. (We experimentally measure this delay for the machines on which we evaluate NBR in~\secref{practicalsigdelivery}.) With some reasonable assumptions on the hardware, for example, the APIC taking a few hundred cycles to deliver an IPI, and approximating the time it takes to retire currently executing instruction in a processor pipeline, we can reason about what a safe waiting period would be in practice.

\begin{property}
\label{prop:rdp}
A thread \tid{j} in \rdp
\begin{enumerate}[label=\ref{prop:rdp}.\arabic*]
    \item \label{prop:rdp1} {upon receiving a signal executes a signal handler and restarts from the entry point of a data structure (i.e gets neutralized).}
    \item \label{prop:rdp2} {is permitted to dereference reference fields of shared nodes to discover new shared nodes (unless it is neutralized).}
\end{enumerate}
\end{property}

\begin{property}
\label{prop:wtp}
A thread \tid{j} in \wtp
\begin{enumerate}[label=\ref{prop:wtp}.\arabic*]
    \item \label{prop:wtp1} {reserves all nodes to be used in \wtp before entering \wtp.}
    \item \label{prop:wtp12} {upon receiving a neutralizing signal simply continues its execution as if no signal was received. }
    \item \label{prop:wtp2} {does not access any nodes it did not reserve prior to entering \wtp.}
\end{enumerate}
\end{property}

\begin{property}
Every \rl thread \tid{r} does the following in order:
\label{prop:rl}
\begin{enumerate}[label=\ref{prop:rl}.\arabic*]
    \item \label{prop:rl1} {sends signals to all participating threads.}
    \item \label{prop:rl2} {scans all reserved nodes of each participating thread \tid{j}.}
    \item \label{prop:rl3} {reclaims nodes in its bag that are not reserved.}
\end{enumerate}
\end{property}

\begin{lemma}
\label{lem:resscan}
A reclaiming thread $T_r$ is guaranteed to scan all the reservations of a thread $T_w$ if $T_w$ enters its \wtp before it is signalled by $T_r$.  
\end{lemma}
\begin{proof}
A thread \tid{w} reserves nodes by announcing them in a single writer multi reader array. On x86/64, this announcement must be followed by a memory ordering instruction to ensure that the reservations are visible when the thread enters its \wtp. Specifically, \nbr uses a CAS to reset a thread local variable \textit{restartable} which simultaneously begins \wtp and ensures that the preceding reservations are visible to any other thread that sees \tid{w} is in \wtp. 

Let $t_{res}$ be the time at which the last reservation was made and $t_{wp}$ be the time at which \wtp began. As we just explained, this means \textbf{(A)} $t_{res}$ $<$ $t_{wp}$. 

To show that \tid{r} observes all of \tid{w}'s reservations, we show that 
$t_{scan}$ $<$ $t_{res}$.

From, \propref{rl} we know that \textbf{(B)} $t_{sig}$ $<$ $t_{scan}$, where $t_{sig}$ is when \tid{r} sent its last signal, and $t_{scan}$ is when it reads the first reservation slot of \tid{w}.
And, since \tid{w} was in \wtp when it received the signal, we have \textbf{(C)} $t_{wp}$ $<$ $t_{sig}$.

By \textbf{(A)}, \textbf{(B)} and \textbf{(C)}, $t_{res} < t_{wp} < t_{sig} < t_{scan}$. 
\end{proof}

\begin{lemma}[NBR is Safe]
\label{lem:nbr}
No \rl thread in \nbr reclaims an unsafe node.
\end{lemma}
\begin{proof}
Suppose to obtain a contradiction that some \rl \tid{r} reclaims an unsafe node $rec$.
This can occur in only two ways: (1) a \wt accesses a node $rec$ that it did not reserve, or (2) a \rd accesses a node $rec$ in the limboBag of \tid{r} that is being reclaimed. 

It is easy to argue that (1) does not happen.
By \propref{rl1}, \tid{r} must have sent a signal to the \wt before reclaiming $rec$.
By \propref{wtp}, if the \wt accesses $rec$ it must reserve $rec$ before entering \wtp.

Next we show (2) does not happen.
As above, by \propref{rl}, \tid{r} must have sent signal to the \rd before reclaiming $rec$.
Once \tid{r} has sent this signal, \asmref{sigg} and \propref{rdp1} imply that the \rd will execute its signal handler as its next step in its execution, at which point it will discard any private reference to $rec$. 
Consequently, by the time \tid{r} begins reclaiming nodes in its limboBag, the \rd will no longer have access to $rec$.


\end{proof}

\begin{lemma}[\nbr is robust]
The number of nodes that are retired but not yet reclaimed is bounded.
\end{lemma}
\begin{proof}
Let $k$ be an upper bound on the number of nodes a thread reserves per operation, $p$ be the number of processes, and $h$ be the maximum limboBag size at which a thread decides to reclaim (i.e., when a predetermined limboBag size is exceeded).
%
Let, \tid{r} be a \rl and \tid{j} be an arbitrary thread.
If \tid{j} is delayed (or crashes), it can reserve at most $k$ nodes in \tid{r}'s limboBag.
A retired node can be present in only one limboBag, so in the worst case a single thread can prevent only $k$ nodes from being reclaimed (across all limboBags).
It follows that, in a system where $p-1$ threads can crash, those $p-1$ threads can prevent at most $k(p-1)$ nodes in total from being reclaimed.
Moreover, a \rl always reclaims when it exceeds its limboBag size. The limboBag size is much larger than the number of nodes a thread can reserve so a limboBag can contain at most $h$ nodes after reclamation.

\end{proof}

\begin{corollary}
A thread can prevent only the nodes it reserves in a single operation from being reclaimed.
\end{corollary}

In NBR, we assume the number of nodes that can be reserved by a single data structure operation is smaller than the limboBag size (otherwise a thread could prevent all nodes in a limboBag from being reclaimed).
Data structures typically require only a small number of reservations per operation.
For example in our experiments, the lazy linked list~\cite{heller2005lazy} required a maximum of two reservations per operation, and the harris list~\cite{harris2001pragmatic}, DGT binary search tree~\cite{david2015asynchronized}, and relaxed (a,b)~tree~\cite{brown2017techniques} needed to reserve a maximum of three nodes per operation.

\subsection{Immediacy of signal delivery in practice}
\label{sec:practicalsigdelivery}

\nbr assumes signal delivery is immediate for safety (\asmref{sigg}), so that neutralizing signals are delivered to target threads before the sender initiates reclamation. 
We examine the signal handling code of two open source operating systems, FreeBSD and Linux, and find that signal delivery is immediate (i.e. occurs in finite time), unless the receiving thread has masked the signal or has exited. 
Additionally, our experiments, spanning over 10 hours, confirm low latency between signal generation and delivery, supporting~\asmref{sigg}.

To understand POSIX signal guarantees better, we analyze the steps involved in the signal delivery process (see \figref{sighandle}). First, the signaling thread calls \texttt{pthread\_kill}, which invokes a system call and transfers control to the operating system (kernel switch).
Second, the kernel processes the call by marking the signal as pending in the target thread's per-thread structure (e.g., thread in FreeBSD). 
Third, if the thread is currently running, the kernel sends an Inter-Processor Interrupt (IPI) to asynchronously interrupt the target thread and transfer control to the operating system.
Fourth, the destination core is interrupted (assuming an IPI is sent) through the local APIC, in case of x86 systems.
Lastly, when the thread is resumed, the signal is delivered to the user space, followed by the execution of the user space signal handler.

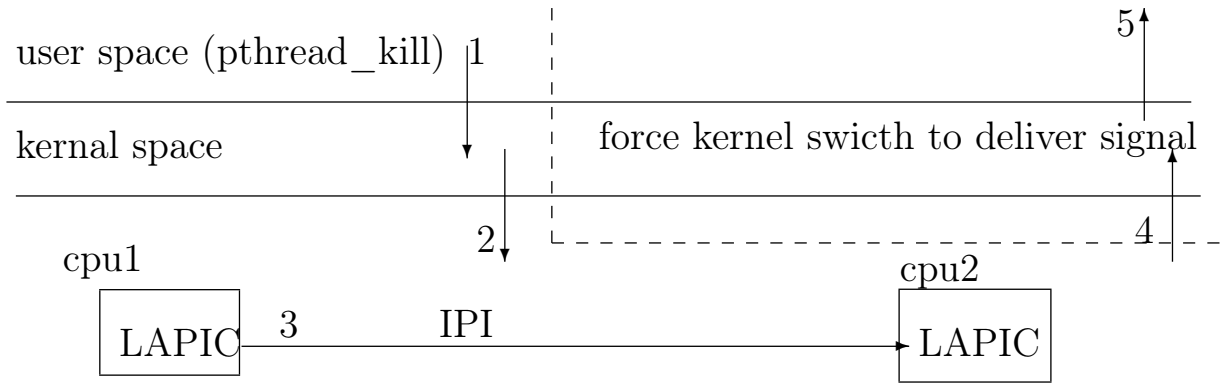
\begin{figure}
\centering
\resizebox{1\linewidth}{!}{\setlength{\unitlength}{4144sp}%
\begingroup\makeatletter\ifx\SetFigFont\undefined%
\gdef\SetFigFont#1#2#3#4#5{%
  \reset@font\fontsize{#1}{#2pt}%
  \fontfamily{#3}\fontseries{#4}\fontshape{#5}%
  \selectfont}%
\fi\endgroup%
\begin{picture}(5829,1815)(889,-2311)
\thinlines
{\color[rgb]{0,0,0}\put(2026,-2131){\vector( 1, 0){3195}}
}%
{\color[rgb]{0,0,0}\put(1351,-2266){\framebox(662,403){}}
}%
{\color[rgb]{0,0,0}\put(3286,-1186){\vector( 0,-1){540}}
}%
{\color[rgb]{0,0,0}\put(6346,-1051){\vector( 0, 1){540}}
}%
{\color[rgb]{0,0,0}\put(3106,-691){\vector( 0,-1){540}}
}%
{\color[rgb]{0,0,0}\put(5176,-2299){\framebox(720,438){}}
}%
{\color[rgb]{0,0,0}\put(6481,-1726){\vector( 0, 1){540}}
}%
{\color[rgb]{0,0,0}\multiput(3511,-511)(0.00000,-118.42105){10}{\line( 0,-1){ 59.211}}
\multiput(3511,-1636)(116.18182,0.00000){28}{\line( 1, 0){ 58.091}}
}%
{\color[rgb]{0,0,0}\put(901,-961){\line( 1, 0){5670}}
}%
{\color[rgb]{0,0,0}\put(946,-1411){\line( 1, 0){5670}}
}%
\put(1441,-2176){\makebox(0,0)[lb]{\smash{{\SetFigFont{12}{14.4}{\rmdefault}{\mddefault}{\updefault}{\color[rgb]{0,0,0}LAPIC}%
}}}}
\put(5266,-2176){\makebox(0,0)[lb]{\smash{{\SetFigFont{12}{14.4}{\rmdefault}{\mddefault}{\updefault}{\color[rgb]{0,0,0}LAPIC}%
}}}}
\put(2206,-2086){\makebox(0,0)[lb]{\smash{{\SetFigFont{12}{14.4}{\rmdefault}{\mddefault}{\updefault}{\color[rgb]{0,0,0}3}%
}}}}
\put(2971,-2086){\makebox(0,0)[lb]{\smash{{\SetFigFont{12}{14.4}{\rmdefault}{\mddefault}{\updefault}{\color[rgb]{0,0,0}IPI}%
}}}}
\put(3151,-1681){\makebox(0,0)[lb]{\smash{{\SetFigFont{12}{14.4}{\rmdefault}{\mddefault}{\updefault}{\color[rgb]{0,0,0}2}%
}}}}
\put(3106,-781){\makebox(0,0)[lb]{\smash{{\SetFigFont{12}{14.4}{\rmdefault}{\mddefault}{\updefault}{\color[rgb]{0,0,0}1}%
}}}}
\put(946,-781){\makebox(0,0)[lb]{\smash{{\SetFigFont{12}{14.4}{\rmdefault}{\mddefault}{\updefault}{\color[rgb]{0,0,0}user space (pthread\_kill)}%
}}}}
\put(946,-1231){\makebox(0,0)[lb]{\smash{{\SetFigFont{12}{14.4}{\rmdefault}{\mddefault}{\updefault}{\color[rgb]{0,0,0}kernal space}%
}}}}
\put(3736,-1186){\makebox(0,0)[lb]{\smash{{\SetFigFont{12}{14.4}{\rmdefault}{\mddefault}{\updefault}{\color[rgb]{0,0,0}force kernel swicth to deliver signal}%
}}}}
\put(1171,-1771){\makebox(0,0)[lb]{\smash{{\SetFigFont{12}{14.4}{\rmdefault}{\mddefault}{\updefault}{\color[rgb]{0,0,0}cpu1}%
}}}}
\put(6211,-646){\makebox(0,0)[lb]{\smash{{\SetFigFont{12}{14.4}{\rmdefault}{\mddefault}{\updefault}{\color[rgb]{0,0,0}5}%
}}}}
\put(5176,-1816){\makebox(0,0)[lb]{\smash{{\SetFigFont{12}{14.4}{\rmdefault}{\mddefault}{\updefault}{\color[rgb]{0,0,0}cpu2}%
}}}}
\put(6301,-1636){\makebox(0,0)[lb]{\smash{{\SetFigFont{12}{14.4}{\rmdefault}{\mddefault}{\updefault}{\color[rgb]{0,0,0}4}%
}}}}
\end{picture}%
}
\caption{Break down of signal transmission. \textbf{Step 1}: \func{pthread\_kill} is invoked from user space. \textbf{Step 2}: kernel switch is triggered by the corresponding syscall. \textbf{Step 3}: as a result of the signal sending routines CPU triggers an IPI using its local APIC (advanced programmable interrupt controller). After step 3 cpu1 may return to execute its program without waiting for the serialization of its IPI. \textbf{Step 4}: cpu2 triggers the kernel switch to deliver the signal immediately. \textbf{Step5}: when the control return to userspace the signal is delivered.}
\label{fig:sighandle}
\end{figure}

For our purposes, we are primarily concerned with the point at which, after the signal is sent, we can guarantee that the target thread stops executing new user space instructions.
For non-running target threads, this is true when the sender returns from its \texttt{pthread\_kill} because the \texttt{pthread\_kill} notifies the target thread that a signal is pending by setting a flag in a pending signal vector in the kernel thread structure.
So, when the target thread is scheduled to run again, it checks its signal pending mask and immediately delivers the signal.
In the case of running threads, the kernel employs an IPI to interrupt the core where the target thread is executing.
This is an \textit{asynchronous} process, which may need to wait for a certain amount of time to confirm the delivery of the IPI.

We have observed that in practice this waiting time for IPI delivery is reasonably low. To determine the upper bound of this waiting time, we design several benchmarks involving two threads placed on consecutive CPU cores within the same socket.
The first thread sends a signal to the second thread by invoking \texttt{pthread\_kill}. We measure the end-to-end times (CPU cycles required to send and deliver a signal) using the \texttt{rdtsc} instruction and further breakdown the measured times by using DTrace to analyze internal operating system functions which allows us to identify signals that induce an IPI. 

The latency is divided into two parts: the latency of transmitting the IPI to the target core and how the processor drains or aborts the remaining in-flight instructions. The IPI is sent at the end of the \texttt{pthread\_kill} operation, and the local APIC completes the delivery (hardware portion) asynchronously. 
To establish an upper bound on the end-to-end signal delivery cost, we measured the local delivery cost using the \texttt{int3} instruction~\cite{intelSDM}. This is one byte instruction (typically used in debuggers to invoke debug exception handler or set software breakpoints) which triggers an exception that is delivered faster than an interrupt routed through the local APIC. This measurement represents the lower bound of the signal delivery cost on the destination core, as externally generated interrupts will take longer than an interrupt generated by the instruction stream. 

It is worth noting that if multiple interrupts are pending, the delivery of the IPI may be delayed. However, for the purpose of our wait time analysis, once the first IPI is received and delivered at the target core, all other queued IPIs will be delivered immediately. Hence, it is safe to assume that no further instructions will be executed.

\tabref{sigdeliverystat} presents the measurements for the number of CPU cycles required for signal sending and delivery, including end-to-end delivery, the cost of sending an IPI, the local signal delivery cost, and the time to execute a \texttt{pthread\_kill} syscall. We have experimented with multiple hardware platforms, but we report these measurements on an Intel Xeon Gold 6342 CPU running at 2.80 GHz (Ice Lake) with FreeBSD 13.0.

We approximate the maximum waiting time for the IPI to complete by subtracting the local signal delivery time from the end-to-end time. This computation provides an empirical upper bound on the moment when the processor stops executing user-space instructions. Additionally, we subtract the time it takes to complete \texttt{pthread\_kill} since the source cannot initiate reclamation until control is returned from the operating system. This calculation results in unaccounted cycles that vary from zero to 283 cycles, on multiple machines we tested. 

The machines used for evaluation of NBR in our experiments section had no unaccounted for cycles. This may be in part due to the extra overhead in returning from kernel mode when Kernel Page Table Isolation (KPTI, kernel patch to mitigate Meltdown like security vulnerabilities) is enabled, which masks the signal delivery latency, giving the effect of immediate signal delivery. 
However, on the Intel’s Icelake machine we calculated roughly 283 cycles that we can not account for.
The unaccounted cycles represent the latency required for the two Local APICs to communicate the IPI request between the cores and any additional cost for interrupt delivery. 
Thus, on this machine threads should wait for 283 cycles before initiating reclamation to ensure \asmref{sigg}’s validity. 
To get a sense of how long this is, note that this latency is within a small factor of the cost of a cache line transfer between L2 caches.

In general, we have observed that signaling mechanism is designed to be fast and the signal delivery cost is low. However, it is important to acknowledge that for some specific processor designs the signal delivery cost may vary. While the exact variations are dependent on the specific architecture, we believe that it is still feasible to approximate the wait times by employing the profiling and benchmarking technique similar to ours. 
By carefully analyzing the characteristics of the architecture at hand, it is possible to infer reliable wait time for threads to initiate reclaiming after sending signals to ensure validity of \asmref{sigg}, on such architectures. Alternatively, fast user-space interrupts could also be explored to address this~\cite{useripilinux}.


\ignore{
\nbr, assumes (\asmref{sigg}) that signals are delivered immediately. By analysing the signal handling code of two open source OSes, FreeBSD\footnote{version 13.1, \url{https://github.com/freebsd/freebsd-src/blob/main/sys/kern/kern_sig.c}} and Linux\footnote{version 5.18, \url{https://elixir.bootlin.com/linux/latest/source/kernel/sched/core.c}}, we find that signal delivery is immediate unless the receiving thread has masked the signal or has exited.  Furthermore, experiments that have run for more than 10 hours confirm that 
the signals delivery latency is negligible for our purposes, i.e, the latency between the signal generation and delivery phase is low, supporting \asmref{sigg}.

To better understand the guarantees provided by POSIX let us examine how signals are delivered.  First, the signalling thread calls \func{pthread\_kill} that invokes a system call and transfers control to the OS.  Second, the kernel will process the call by marking the signal as pending in the target thread's per-thread structure (e.g., \func{thread} in FreeBSD).  Third, if the thread is running the kernel sends an IPI to asynchronously interrupt the the target thread and transfer control to the OS.  Fourth, the destination core will be interrupted, assuming an IPI is sent, through the local APIC in the case of x86.  Fifth, when the thread is resumed the signal will be delivered to the user space.  The final step is to execute the user space signal handler.



For our purposes, we only care at what point, after the signal was sent, we can guarantee that the target thread stops executing instructions.
For threads that are not running, this is true upon return from \func{pthread\_kill}.  \func{pthread\_kill} notifies the target thread that a signal is pending by setting a flag in a pending signal vector in the kernel thread structure.  The target thread will check its signal pending mask before resuming execution and immediately deliver the signal.

For threads that are running, the kernel uses an inter-processor interrupt (IPI) to interrupt the core where the target thread is running.  This process is asynchronous and we may need to wait a certain amount of time to confirm that the IPI was delivered.  In this section, we use several benchmarks to approximate a upper bound on number of cycles a thread needs to wait to be certain that the signal was delivered.

The benchmark consists of two threads placed on consecutive CPU cores on the same socket. The first thread sends a signal to the second thread by calling \func{pthread\_kill}.  We measure the end-to-end times using the {\tt rdtsc} instruction.  We break down the numbers by measuring internal operating system functions using DTrace, which allows us to confirm which signals induced an IPI.
We experimented with several hardware platforms, but we report measurements from an Intel Xeon Gold 6342 CPU running at 2.80~GHz running FreeBSD 13.0.

The latency is broken into two parts: the latency of transmitting the IPI to the target core and how the processor drains or aborts the remaining in-flight instructions.
Also, note that if there are multiple pending interrupts it may delay the delivery of the IPI, but this is of no concern as in either case no further instructions should be executed once the thread is interrupted and the IPI is received at the target core.

Table~\ref{tab:sigdeliverystat} shows the measurements for the number of CPU cycles required to send and deliver a signal (end to end delivery), the cost of sending an IPI, the local signal delivery cost, time to execute a \func{pthread\_kill} syscall.
The IPI is sent at the end of the {\tt pthread\_kill} operation and the local APIC may complete the hardware portion asynchronously.

The end to end signal delivery costs can be upper bounded by examining the local signal delivery costs.
The local delivery cost was measured using the {\tt int3} instruction to trigger an exception that is delivered faster than an interrupt that goes through the local APIC.
This measures the lower bound of the signal delivery cost on the destination core, because externally generated interrupts will take longer than an interrupt generated by the instruction stream (specifically \func{int3} and \func{into}).

We approximate the maximum amount of time that we must wait for the IPI to complete (meaning the target thread will not execute any additional instructions), as the end-to-end time minus the local signal delivery time.
This computes an upper bound on the point in time where the processor will stop executing instructions.
We can also subtract the time it takes to complete \func{pthread\_kill} as that as source will not be able to start reclamation until control is returned from the OS.
This results in roughly 283 cycles that we can not account for on the Icelake machine.  The machine used in the evaluation had no unaccounted for cycles this may be in part due to the extra overhead in returning from kernel mode when KPTI is enabled.

The unaccounted for cycles is the latency required for the two Local APICs to communicate the IPI request between the cores and any additional cost for interrupt delivery.  The total communication and interrupt delivery costs are partially masked by the \func{pthread\_kill} syscall return path.  The unaccounted latency is within a small factor of the cost of a cache line transfer between L2s, which serves as a good approximation of the local APIC communication costs.

}
\ignore{
\nbr, as stated in \asmref{sigg}, assumes that the neutralizing signals from a reclaiming thread will be delivered immediately to all the threads in their \rdp before starting reclamation.
We analyzed the signal handling code of two open source OSes, FreeBSD\footnote{version 13.1, \url{https://github.com/freebsd/freebsd-src/blob/main/sys/kern/kern_sig.c}} and Linux\footnote{version 5.18, \url{https://elixir.bootlin.com/linux/latest/source/kernel/sched/core.c}} to understand their behavior.
Signal delivery is immediate unless the receiving thread has masked the signal or has exited.
Furthermore, all of our experiments that have run for 10 hours or more confirm that the signals delivery latency is negligible for our purposes, i.e, the latency between the signal generation and delivery phase is low, supporting the \asmref{sigg}.

The signaling thread calls {\tt pthread\_kill} to send a signal to one other thread.  The kernel will update the target thread's {\tt thread structure} to mark the signal as pending.  If the thread is descheduled, this is enough, as the thread will check the signal pending mask before resuming execution and immediately deliver it to the thread.  Running threads are immediately notified through an interprocessor interrupt (IPI).  This will trigger an interrupt that will hand control back to the operating system at which point the signal is guaranteed to be delivered before executing any more instructions.

The \figref{sighandle} shows a breakdown of the signal transmission to a target thread across user space, kernel space, and hardware. Steps 1 to 3 are part of the \textit{signal generation} and steps 4 to 5 constitute the \textit{signal delivery}.
The signal generation phase is synchronous, i.e., before the call to \func{pthread\_kill} returns, the {\tt thread} data structure at the target thread is updated.
The signal delivery phase appears to be asynchronous and it is not clear how much time it takes to deliver the signal to the target thread.  This includes scheduling latency for threads that are not running and the IPI latency for running threads.

For our purposes, we only care at what point we can guarantee no further instructions can be executed.
For threads that are not running this is true upon return from \func{pthread\_kill}, but for threads that are running we need to wait a certain amount of time after.
The remaining latency is broken into two parts: the latency of transmitting the IPI to the target local core and how the processor drains or aborts the remaining instructions in the reorder window within the processor.
Also, note that if there are multiple pending interrupts it may delay the delivery of the IPI, but this is of no concern as in either case no further instructions should be executed once the thread is interrupted and the IPI is received at the target core.
}
\ignore{
In NBR the reclaiming thread sends $N-1$ signals one after the other, where $N$ is the total number of threads. Meaning, that after the $i^{th}$ signal has been sent the thread can be certain that its ${i-1}^{th}$ signal has been delivered. Therefore, after sending all $N-1$ signals the reclaimer thread only needs to wait for the amount of time it takes to send one more signal to guarantee that all its signals have been delivered.
Any latency experienced during signal delivery will be a constant factor and ensure that \asmref{sigg} is valid on modern machines.

In this section, we dissect the signal transmission process and provide an approximate upper bound on the latency (number of CPU cycles) of the signal delivery.  We experimented with several hardware platforms, but we conducted these measurements on an Intel Xeon Gold 6342 CPU running at 2.80~GHz running FreeBSD 13.0.

The benchmark consists of two threads placed on consecutive CPU cores on the same socket. The first thread sends a signal to the second thread by calling $pthread\_kill$.  We measure the end-to-end times using the {\tt rdtsc} instruction and shared variable.  We further broke down the numbers by then analyzing the internal operating system functions using DTrace.

Table~\ref{tab:sigdeliverystat} shows the measurements for the number of CPU cycles required to complete the {\tt pthread\_kill} call, the time to receive the signal, the cost of sending an IPI, and the local signal delivery cost.  The IPI is sent almost at the very end of the {\tt pthread\_kill} operation.

The exception delivery and kernel signal delivery costs can be lower bounded by the local signal delivery costs.  The local delivery cost was measured using the {\tt int3} instruction to trigger an exception (generated by the instruction stream) that should be delivered faster than an interrupt (external to the core) that goes through the local APIC.

We can approximate the maximum amount of time that we must wait for the IPI to complete, as the {\tt pthread\_kill} + delivery time less {\tt pthread\_kill} system call time and signal delivery time.  This results in roughly 283 cycles that we can not account for on the Icelake machine.  The machine used in the evaluation had no unaccounted for cycles this may be in part due to the extra overhead in returning from kernel mode when KPTI is enabled.

The unaccounted for cycles is the latency required for the two Local APICs to communicate the IPI request between the cores and any additional cost for interrupt delivery.  The total communication and interrupt delivery costs are partially masked by the system call return path.  The unaccounted latency is within a small factor of the cost of a cache line transfer between L2s, which likely approximates the APIC communication costs.
}

\begin{table}
\centering
\begin{tabular}{lrr}
\toprule
\textbf{Operation}                       & \textbf{CPU Cycles} \\
\midrule
end-to-end signal delivery               & 6527    \\
sending ipi                              & 393     \\
local signal delivery                    & 4324    \\
{\tt pthread\_kill}                      & 1920    \\
\bottomrule
\end{tabular}
\caption{Breakdown of the approximate signal delivery latencies in cycles.  We break down the {\tt pthread\_kill} into the call cost and the actual sending of the IPI.  Signals are delivered using pthread\_kill to a thread in another core in the same socket.}
\label{tab:sigdeliverystat}
\end{table}


In the following section we present a preliminary evaluation of the technique to highlight a performance problem that occurs in the algorithm presented in this chapter which arises due to overhead of excessive signals and present an extensive evaluation later in ~\chapref{chapnbrp}. 

\section{Preliminary Performance Evaluation}
\label{sec:nbreval}
In this section we evaluate throughput of \nbr when integrated with an external binary search tree of David et~al (DGT)~\cite{david2015asynchronized} to highlight a performance issue. (More extensive experimentation appears in \chapref{chapnbrp}.)
For comparison, we additionally implemented the tree with the best known EBR algorithm \texttt{debra}, hazard eras (\texttt{he}) and hazard pointers (\texttt{hp}) and a leaky implementation which we denote as \texttt{none}. 

We ran the experiment on a quad-socket Intel Xeon Platinum 8160 machine with 192 threads, 384GB of RAM,  33MB of L3 cache per socket (total 132MB). The machine ran Ubuntu 20.04 with kernel 5.8 with GCC 9.3.0 with \texttt{-O3} optimization flag and utilized \emph{jemalloc} as the memory allocator~\cite{evans2006scalable}. 

The experiment measured throughput across varying number of threads for the tree with a maximum size of 2 million nodes with 10\%, 50\% and 100\% update operations.  
The reported results are obtained by averaging data from 5 timed trials, each lasting 5 seconds. The 
experiments were conducted with thread counts ranging from 1 to 384 threads on a system with 192 hardware threads and more than 91\% of data points had less than 5\% variance.
Before each execution, the tree was pre-filled to half of its maximum size.

\begin{figure}
\centering
            \includegraphics[width=0.33\linewidth, height=6cm, keepaspectratio]{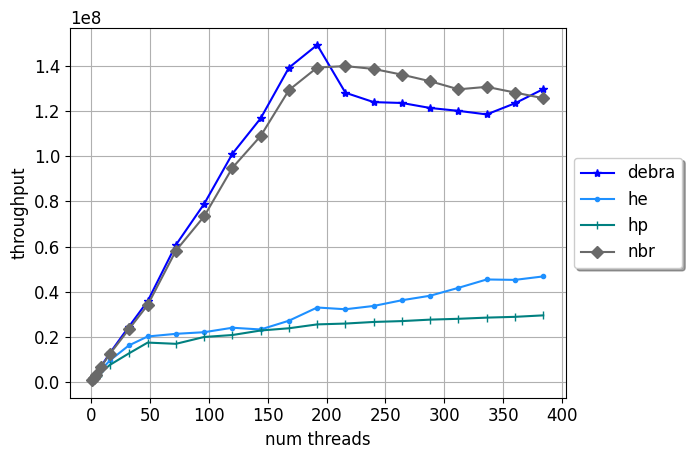}\hfill
            \includegraphics[width=0.33\linewidth, height=6cm, keepaspectratio]{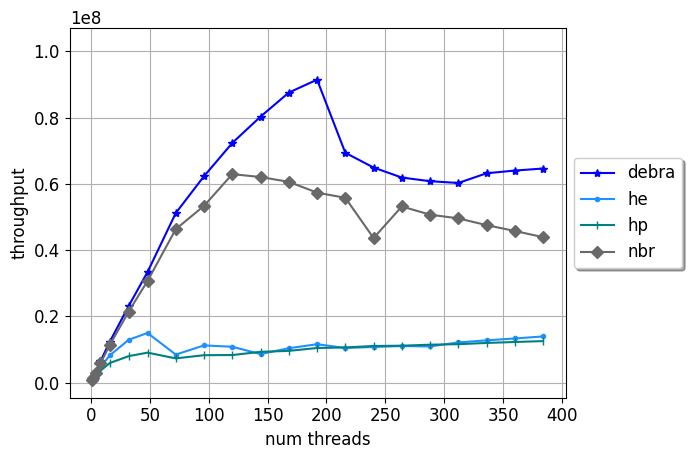}\hfill
            \includegraphics[width=0.33\linewidth, height=6cm, keepaspectratio]{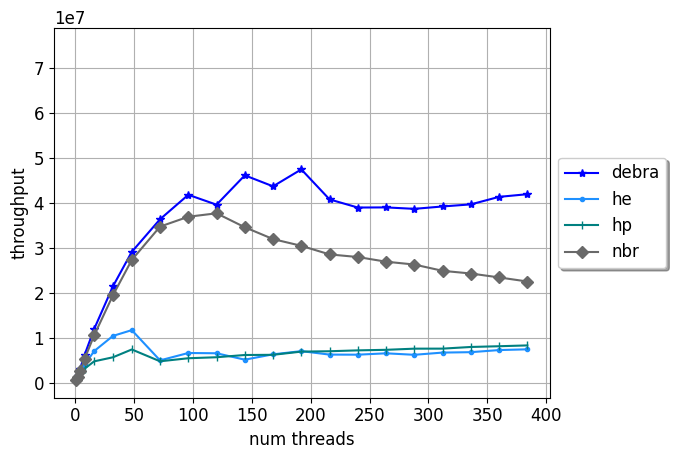}\hfill
    \caption{DGT throughput. Updates: Left: 10\%. Middle: 50\%. Right: 100\%. Max size:2000000. Y axis: throughput in million operations per second. X-axis: \#threads.}
    \label{fig:e1nbrdgt}
\end{figure}


\Figref{e1nbrdgt}, show the throughput of the tree with varying workloads. In general, it can be observed that especially with 50\% and 100\% update workloads, \nbr is similar to \debra (a high performing EBR style SMR) at lower thread counts, i.e. up to 48 threads. Then \textbf{at higher thread counts, \nbr is significantly slower}. We analyze this reason for this slowdown in \nbr using the perf tool on Linux and found that most of the slowdown occurred due to system activity of signals. Note that with a higher number of updates the reclamation occurs more frequently, which implies that signals are sent frequently to other threads by \nbr. Furthermore, as the thread count increases the number of signals sent at each reclamation event increase in similar proportion, leading to higher signalling overhead. 

\newpage
\section{Summary}
\label{sec:nbrsummary}

In this chapter, 
we described the design and implementation of our first safe memory reclamation algorithm: NBR.
Under the hood, NBR leverages POSIX signals in existing operating systems.
We presented proof that when applied to a data structure, NBR ensures the safety of reclamation and \textit{bounds garbage}. The safety of reclamation depends on the immediacy of signal delivery. Thus, we investigated the signal delivery timing in commodity operating systems in popular architectures and established that signals are delivered almost immediately in real systems. Finally, we showed that although theoretically NBR is \textit{efficient}, in practice, due to signal overhead, it can be slow at high thread counts on large NUMA systems.

We started with the goal of designing a fast, bounded garbage, widely applicable, simple to use, and consistently performing safe memory reclamation algorithm. Among these goals, we have shown that a safe memory reclamation algorithm designed with the neutralization paradigm, i.e., NBR, can be theoretically efficient and bound garbage. 

In the upcoming chapters, we establish the remaining properties.
Specifically, we first demonstrate that \nbr is fast and exhibits consistent performance in practice by introducing \nbrp in \chapref{chapnbrp}, which addresses the issue of signal overhead.
Later, in \chapref{chapappuse}, we show that \nbr is widely applicable and easy to use, thus fulfilling the objectives we set out for our first contribution using the neutralization paradigm in this dissertation.


\chapter{Optimizing Neutralization Based Reclamation}
\label{chap:chapnbrp}

In the previous \chapref{chapnbr}, we presented \nbr and showed that it is theoretically fast and bounds the amount of garbage. However, in practice, it experienced slowdown due to excess signalling overheads.
In this chapter, we propose NBR+, a signal-optimized safe memory reclamation algorithm that is also fast in practice. In a nutshell, this chapter focuses on establishing that the neutralization paradigm yields reclamation algorithms that are fast, consistent, and bound garbage. 
The remaining desirable properties of wide applicability and usability are the subject of subsequent \chapref{chapappuse}.

The outline of this chapter is as follows. 
In \secref{nbrpintro}, we closely examine the neutralizing mechanism in NBR and discuss how the average number of signals sent by each thread can be reduced. This leads us to develop \nbrp, a faster safe memory reclamation algorithm using the neutralization paradigm. We present the design and implementation of \nbrp in \secref{nbrp}, followed by proof of its safety and bounded garbage property in \secref{nbrpcorr}. 
Finally, we describe our benchmark in \secref{nbrpeval}, which we use to extensively evaluate NBR and NBR+ for multiple data structures and compare them with various state-of-the-art safe memory reclamation techniques. We study both the throughput and the memory consumption behavior of our algorithms.

\section{Introduction}
\label{sec:nbrpintro}
The \nbr algorithm, discussed in the preceding chapter, presents a simple solution to the safe memory reclamation problem.
However, every thread is required to send neutralizing signals to every other thread in order to initiate the reclamation of its limbo bag. 
These signals on Linux trigger page fault routines and a switch from user to kernel mode. As a result, the use of POSIX signals in \nbr introduces significant overhead.
In fact, in a system with $n$ threads, $O(n^2)$ signals are required for all threads to reclaim their limbo bags at least once. For larger number of threads the overhead drastically increases which hinders the scalability of the data structure the neutralization based algorithm is used with.

In NBR, when a thread wants to reclaim its limbo bag, it sends $n-1$ signals to other threads.
This causes all the other $n-1$ threads to discard any unreserved references to the shared nodes. Meaning, not only the unreserved nodes in the reclaimer's limbo bag but the unreserved nodes in all other thread's limbo bags can also be safely freed.

More specifically, a reclaimer in order to reclaim nodes in its limbo bag starts sending neutralizing signals to all $n-1$ threads at time $t$ and completes sending signals at time $t'$. Then, at time $t'$ the reclaimer can safely reclaim all unreserved nodes in its limbo bag. The nodes in the limbo bag of the reclaimer are essentially the nodes that were retired before time $t$. In addition,  at time $t'$ other $n-1$ threads can also safely free the nodes in their limbo bags that were retired before $t$.
In NBR, every time a thread tries to reclaim its limbo bag, it creates a time interval [t, t'] during which each thread is neutralized. We term it as a \textit{neutralization event}.

Therefore, if somehow we could propagate this information that a thread $T$ has started and finished sending signals, i.e a \rgp has occurred, then all other threads could piggyback on $T$ to partially or completely reclaim their own limbo bags without sending signals of their own.
In other words, in the best case, all $n$ participating threads could reclaim memory after detecting exactly one reclamation event, induced by sending a total of $n-1$ signals. As a result, in the best case, the number of signals required for all threads to reclaim their limbo bags at least once reduces from $O(n^2)$ to $O(n)$ which in turn translates into reduction of signal overhead. 
\section{NBR+}
\label{sec:nbrp}

\subsection{Design}
The key insight in \nbrp is that when a reclaimer sends neutralization signals to \textit{all} threads, \textit{all} threads discard their pointers to unreserved nodes, and thus \textit{all} threads can potentially reclaim some nodes in their limbo bags. 
\nbrp proposes a mechanism to take advantage of this information and significantly reduce signal overhead, making the neutralization technique highly efficient.

This suggests a design wherein each thread: 
\begin{enumerate}
    \item Passively detects a \rgp by observing signals sent by another thread, and
    \item Determines which nodes in its limbo bag were unlinked \textit{before} the \rgp (i.e., are safe to reclaim).
\end{enumerate}

In \nbrp, each thread passively monitors for a \rgp starting from a predetermined lower limbo bag size known as the \textit{LoWatermark}, before reaching the maximum limbo bag size referred to as the \textit{HiWatermark}.
When a thread reaches the \textit{HiWatermark}, it announces the beginning of its \rgp, sends neutralizing signals, and announces the end of its \rgp. Other threads, which have passed the \textit{LoWatermark} but not yet reached the \textit{HiWatermark}, observe these announcements to infer that a \rgp has occurred. 
They can then safely reclaim all nodes in their limbo bags retired before reaching the \textit{LoWatermark}. If a thread has not reached \lw or \hw, it simply continues to add the retired node to its limbo bag.

\subsection{Implementation}

\label{sec:nbrp_impl}
The pseudocode for the algorithm appears in \algoref{nbrp}. It is built on top of \nbr in \algoref{nbr} and the part that is different from \nbr is highlighted in {yellow}.

\subsubsection{Variable Description}
The variable \texttt{otid} depicts a logical thread id of another thread other than the thread that is currently executing the code. \texttt{scanTS} is a thread-local array of slots where each thread saves the timestamps announced by another thread. The thread local variable \texttt{firstLoWmEntryFlag} is used to depict that the thread's limbo bag has crossed its \textit{LoWatermark} and is reset after the limbo bag size is reduced below the LoWatermark size. The \texttt{bookmarkTail} is used to save the pointer to the last retired node in the limbo bag of a thread when it passes the \textit{LoWatermark} size. The \textit{announceTS} is an array of $n$ slots. Each thread writes a timestamp in its dedicated slot and all other threads can read from it.  

We explain the design of \nbrp by building our exposition around three main design challenges.
\begin{enumerate}
    \item [(C1)] {When should a thread start tracking other threads' signals to detect a \rgp?}
    \item [(C2)] {How can a thread \textit{recognize} that a \rgp has occurred?}
    \item [(C3)] {Once a thread recognises that a \rgp has occurred how should it determine which nodes in its limbo bag are safe to reclaim?}
\end{enumerate}

\begin{algorithm}
\small
\caption{\nbrp incorporates all variables and procedures from \nbr in \algoref{nbr} while introducing a modified \func{RETIRE()} and additional helper procedures with self explanatory names. All changes different from \nbr are highlighted.}\label{algo:nbrp}

\begin{algorithmic}[1]

\Statex \textbf{thread local variable:}
    \State int \texttt{tid;} \Comment{\texttt{current thread id}}
    \State node *\texttt{limboBag;} \label{lin:p_rb}\Comment{\texttt{Maxsize:S}}
    \State atomic<bool> \texttt{restartable;} \label{lin:p_rst}\Comment{\texttt{tracks \rdp/\wtp}}
    \State node *\texttt{tail;} \label{lin:p_tl}\Comment{\texttt{last node in limboBag}}
\BeginBox[fill=yellow]
        \State int \texttt{otid;} \Comment{\texttt{other thread's ids, excluding tid.}}
        \State int \texttt{scanTS[n];} \label{lin:tlvsta} \Comment{\texttt{n:\#threads}}
        \State bool \texttt{firstLoWmEntryFlag=true;}
        \State node* \texttt{bookmarkTail;} 
\EndBox
\Statex
\Statex \textbf{shared variable:} 
\State atomic<node*> \texttt{reservations[n][r];} \label{lin:p_resv} \Comment{\texttt{r:max reserved nodes. Assume:r$<<$S}}

\BeginBox[fill=yellow]
    \State atomic<int> \texttt{announceTS[n];}\label{lin:shvalg2}
\EndBox
\Statex     
        \Procedure{CHECKPOINT}{ } \Comment{\texttt{must be INLINED.}} \label{lin:p_chkp}
            \While{\texttt{sigsetjmp(...))}} \Comment{\texttt{signal mask not saved.}} \label{lin:p_sigset}
            \State {unblock neutralizing signal.} \Comment{\texttt{post siglongjmp.}}
            \EndWhile
        \EndProcedure
\Statex
        \Procedure{signalHandler}{ } \label{lin:p_sigh}
            \If{\texttt{!restartable}} 
            \State {\texttt{return;}} \Comment{\texttt{in \wtp, ignore signal and return.}}
            \EndIf
            \State \texttt{siglongjmp(...);}\Comment{\texttt{in \rdp, jump to checkpoint.}}\label{lin:p_sigjmp}            
        \EndProcedure
\Statex
        \Procedure{begin\rdp}{ }\label{lin:p_brp}
            \State \texttt{reservations[tid].clear();} \label{lin:p_arpc}
            \State \texttt{restartable = True;}\label{lin:p_nz01}
        \EndProcedure
\Statex
        \Procedure{end\rdp}{rec = \{$rec_1,\cdot\cdot\cdot, rec_R$\}}\label{lin:p_bwp}
            \State \texttt{reservations[tid]=rec; }\label{lin:p_arp_add}
            \State \texttt{restartable = False;} \label{lin:p_nz10}
        \EndProcedure
\Statex        
        \Procedure{reclaimFreeable}{tail} \label{lin:p_func_rf}
            \State \texttt{$A$=\Call{collectReservations()}{}; }\label{lin:p_carp} 
            \State \texttt{$R$=limboBag[tid].remove(A, tail);}\label{lin:p_retg} 
            \State \texttt{free(\{$R$\});} \label{lin:p_free}
        \EndProcedure
        \algstore{part1}
    \end{algorithmic}
\end{algorithm}

\begin{algorithm}
\small
\begin{algorithmic}[1]
\algrestore{part1}

\BeginBox[fill=yellow]
    \Procedure{retire}{rec}
        \If{isAtHiWm()} \label{lin:iahw}
            \State \texttt{\Call{FAA}{\&announceTS[tid],1};}\Comment{neutralization event begin } \label{lin:afaa1}
            \State \texttt{\Call{signalAll()}{}} \label{lin:sigall}
            \State \texttt{\Call {FAA}{\&announceTS[tid],1};}\Comment{neutralization event end } \label{lin:afaa2}                
            \State \texttt{\Call{reclaimFreeable}{tail};} \label{lin:recfr}
            \State \texttt{\Call{cleanUp()}{};}\label{lin:iahwend}
        \ElsIf{isAtLoWm()} \label{lin:ialw}
            \If{firstLoWmEntryFlag}
                \State \texttt{bookmarkedTail $=$ tail;} \label{lin:bmtail}
                \State \texttt{scanTS[tid] $=$ \Call {$scanAnnounceTS()$}{}} \label{lin:readscants}
                \State \texttt{firstLoWmEntryFlag$=$0;} 
            \EndIf
            \For{each otid }\label{lin:attemptfreebeg} 
                \If{announceTS[otid]$\geq$scanTS[tid][otid]+2 } \label{lin:pback}
                    \State \texttt{\Call{reclaimFreeable}{bookmarkTail};} \label{lin:prf}
                    \State \texttt{\Call{cleanUp()}{};} \label{lin:cluplw}
                    \State \texttt{break;}
                \EndIf            
            \EndFor \label{lin:attemptfreeend}
        \EndIf \label{lin:ialw_end}
        \State \texttt{limboBag[tid].append(rec);} \label{lin:apprec}
    \EndProcedure
    \Statex        
    \Procedure{cleanUp}{} \label{lin:clup}
        \State \texttt{firstLoWmEntryFlag $=$ 1;}
    \EndProcedure
\EndBox
\end{algorithmic}
\end{algorithm}

As a solution to \textbf{(C1)}, each thread in \nbrp, in addition to watching the \texttt{limboBag} size to determine when it becomes \textit{too large} (triggering neutralization), also determines when the \texttt{limboBag} size crosses a predetermined threshold called the \lw (e.g., one half full or one quarter full). 
If a thread's \texttt{limboBag} is full, we say that the thread is at the \hw. 
If a thread's \lb keeps growing without reclamation, it will first cross the \lw and then hit the \hw. 
As shown in \algoref{nbrp}, a thread determines whether it has passed \hw or \lw using procedures {isAtHiWm()} (\lineref{iahw}) and {isAtLoWm()} (\lineref{ialw}). 
Once a thread has passed the \lw, it begins recording and analyzing information about signals sent by other threads to detect \textit{neutralization event}.

To tackle \textbf{(C2)}, a \rl at the \lw (who wants to detect a \rgp) must perform a sort of handshake with another \rl at the \hw (who triggers a \rgp).
\nbrp implements this handshake using per-thread single-writer multi-reader timestamps (similar to vector clocks).

Whenever a \rl{} hits the \hw, it first increments its timestamp (to an odd value) to indicate that it is \textit{currently broadcasting signals} (\lineref{afaa1}).
This denotes the \textit{beginning} of a \rgp.
It then sends signals to all threads, and increments its timestamp again (to an even value) to indicate that it has \textit{finished broadcasting signals} (\lineref{afaa2}).
This denotes the \textit{end} of the \rgp. 

Whenever a \rl $T$ passes the \lw, it 
collects and saves the current timestamps of all threads in its local \texttt{scanTS} array (\lineref{readscants}), as well as the current \textit{tail} pointer of its \texttt{limboBag}  in its local \texttt{bookmarkTail} variable(\lineref{bmtail}), so it can remember precisely which nodes it had unlinked \textit{before} it reached its \lw.
$T$ then periodically collects the timestamps of all threads, comparing the new values it sees to the original values it saw when it passed the \lw (\lineref{attemptfreebeg} - \lineref{attemptfreeend}).
Note, scanning \texttt{announceTS[]} would incur overhead in terms of cache misses.
This cost is amortized over multiple \func{retire} operations by scanning \texttt{announceTS[]} after a fixed number of calls to \func{retire} have been made. 
It continues to do this until it detects a \rgp or hits the \hw itself (and sends signals to induce its own \rgp).
Observe that, after $T$ hits its \lw, if the timestamp of any thread changes from one even number to another even number, then that thread has both \textit{begun and finished} sending signals to \textit{all} threads since $T$ hit the \lw.
Thus, $T$ can identify that a \rgp has occurred since it hit its \lw, solving (C2). 

Finally, to tackle \textbf{(C3)}, observe that T saves the last node ($tail$ of its \textit{limboBag}) it had retired before entering \lw at \lineref{bmtail}.
If T successfully observes a \rgp as explained in the solution to (C2), then all threads would either have discarded or reserved all their private references to the nodes in $T$'s \texttt{limboBag} up to the saved \texttt{bookmarkTail}. 
Thus, $T$ can invoke \Call{reclaimFreeable()}{} to free all unreserved nodes up to the \texttt{bookmarkTail} (\lineref{prf}). solving (C3).

\Call{cleanUp()}{} (\lineref{clup}) method is used to set \emph{firstLoWmEntryFlag} after a thread reclaims either at \lw (\lineref{cluplw}) or at \hw (\lineref{iahwend}) to prepare it for subsequent reclamation. 

A thread that has not reached the \lw or the \hw simply continues to append any retired nodes to its \lb (\lineref{apprec}).

\ignore{
A reclaimer that has reached the \hw reclaims exactly as it would in \nbr, meaning it broadcasts a neutralization signal to all threads (\lineref{sigall}), causing an \rgp, and then 
reclaims safe nodes (skipping any reserved nodes, (\lineref{recfr})).

Execution on the \lw path (\lineref{ialw}--\ref{lin:ialw_end}) monitors the signals that have been sent in the system to try to avoid sending signals in the future, incurring some overhead that is amortized over multiple data structure operations.

To solve \textbf{(C2)}, a thread that passes \lw has to periodically check for a \emph{relaxed grace period}.
In order to be able to reclaim memory without sending any signals, the reclaimer at the \lw must perform a 
handshake with another reclaimer at the \hw (who \textit{is} sending signals). 
}

\subsubsection{Why not reclaim every time a signal is received?}

At first, it may appear that a thread $T$ can reclaim its \texttt{limboBag} as soon as it receives a neutralizing signal from a \rl thread $T'$.
However, the receipt of a single signal is not enough for $T$ to safely reclaim memory.
To safely reclaim the set \texttt{R} of nodes in its \lb up to its $bookmarkTail$, $T$ needs to know that \textit{all} threads have been neutralized \textit{since $T$ retired the nodes in \texttt{R}}.
Otherwise, some other thread may still have a pointer to a node in \texttt{R}.

Let us discuss an example of what can go wrong if a thread reclaims its limbo bag after it receives a single signal.
Consider a system with three threads $T1$, $T2$ and $T3$.
Suppose $T1$ is at its \hw, $T2$ is at its \lw and $T3$ holds a private reference to a node $rec$ that is in $T2$'s limbo bag.
$T1$, being at its \hw, begins neutralizing all threads one by one.
First, it sends a neutralizing signal to $T2$ (starting a \rgp).
$T2$, upon receiving the signal, reclaims its \lb including \texttt{rec}.
Note, that $T1$ hasn't neutralized $T3$ yet, meaning a \rgp has not yet occurred.
Now, if $T3$ accesses \texttt{rec}, a \textit{use-after-free} error would occur.
To prevent this, $T2$ should not reclaim the contents of its limbo bag unless $T1$ \textit{completes} the \rgp by neutralizing $T3$ (preventing $T3$ from doing this unsafe access). 
The crucial point is that $T2$ must detect the \textit{start} and  \textit{end} of a \rgp to know that it can safely reclaim nodes in its limbo bag.

\section{Correctness}
\label{sec:nbrpcorr}

\begin{lemma}[\nbrp is safe]
No \rl thread in \nbrp reclaims an unsafe record.
\end{lemma}
\begin{proof}
In \nbrp, threads can reclaim at the \lw or at the \hw. 
Reclamation at the \hw is similar to reclamation in \nbr, and the argument that reclaimers at the \hw do not reclaim unsafe records is similar to the lemma~\ref{lem:nbr}.

It remains to prove that 
reclamation at the \lw is safe.
We argue that any record $rec$ reclaimed by a thread \tid{{lw}} at the \lw\ must be safe.
In other words, $rec$ must not be reserved by any thread, and no thread should have a private reference to $rec$.

In \nbrp, \tid{{lw}} reclaims only up to its \emph{bookmarkTail}, which means it only reclaims records that it had retired before the time $t$ when \tid{{lw}} reached the \lw and scanned the timestamps of all threads.
And, \tid{{lw}} reclaims these records only at a later time $t'$ when it sees that at least one thread has incremented its timestamps at least twice. 
These timestamp increments indicate that, between times $t$ and $t'$, all threads received signals and discarded their pointers to unreserved records.
Since $rec$ is retired before time $t < t'$, it follows that at time $t'$ any thread that has a reference to $rec$ must have reserved $rec$ before $t'$.
If $rec$ is still reserved when \tid{{lw}} scans the reservations of all threads (after $t'$) then $rec$ will not be reclaimed.
Otherwise, $rec$ is safe to reclaim.
\end{proof}

\begin{lemma}[\nbrp is robust]
The number of records that are retired but not yet reclaimed is bounded.
\end{lemma}
\begin{proof}
The argument is similar to that in \lemref{nbr} in \chapref{chapnbr}.
Let $k$ be an upper bound on the number of records a thread reserves per operation, $p$ be the number of processes, and $h$ be the maximum limbo bag size at which a thread decides to reclaim (i.e., the \hw).
%
Let, \tid{r} be a \rl and \tid{j} be an arbitrary thread.
If \tid{j} is delayed (or crashes), it can reserve at most $k$ records in \tid{r}'s limbo bag.
A retired record can be present in only one limbo bag, so in the worst case a single thread can prevent only $k$ records from being reclaimed (in all limbo bags).
It follows that, in a system where $p-1$ threads can crash, those $p-1$ threads can prevent at most $k(p-1)$ records in total from being reclaimed.
Moreover, a \rl always reclaims when it hits the \hw, so a limbo bag contains at most $h$ records.

\end{proof}

\section{Evaluation}
\label{sec:nbrpeval}
We rigorously evaluate \nbr and \nbrp. 
To the best of our knowledge, this is the most extensive evaluation to date in terms of the number of reclamation schemes and variety of data structures evaluated.

\myparagraph{Setup:} We conducted our experiments on a quad-socket Intel Xeon Platinum 8160 machine with 192 threads, 384GB of RAM,  33MB of L3 cache per socket (total 132MB). The machine ran Ubuntu 20.04 with kernel 5.8 with GCC 9.3.0.
We implemented all algorithms in the Setbench~\cite{brown2020non} benchmark with \texttt{-O3} optimization flag and utilized \emph{jemalloc} as the memory allocator~\cite{evans2006scalable}. Our evaluation consisted of three types of experiments:
\begin{enumerate}
    \item [(E1):] {Evaluating \nbr(+) throughput with different workloads thread counts to understand scalability.}
    \item [(E2):] {Measuring throughput with varying data structure sizes to understand the impact of contention and cache misses.}
    \item [(E3):] {Evaluates peak memory usage of \nbrp with and without stalled threads to understand its memory behaviour
    .}
    
\end{enumerate}

\ignore{
\begin{table}[ht]
\centering
\begin{tabular}{lllll}
SMR                                             & \makecell{Bounds\\ Garbage?}  & Progress & \makecell{performance\\ summary} & Type \\
\hline
\hline
\rowcolor[rgb]{0.753,0.753,0.753}\nbrp        & YES                           & cond. lockfree    & high  & HYB\\ 
nbr                                             & YES                           & cond. lockfree    & medium & HYB\\
\rowcolor[rgb]{0.753,0.753,0.753} \debra         & NO                            & blocking          & high & EPOCH\\ 
\hp                                              & YES                           & lock-free         & low & \hp\\ 
\rowcolor[rgb]{0.753,0.753,0.753} \qsbr \& \rcu   & NO                            & blocking          & medium & EPOCH\\ 
\geibr                                          & NO                           & lock-free         & medium & HYB\\ 
\rowcolor[rgb]{0.753,0.753,0.753} \he            & NO                          & lock-free          & low & HYB\\ 
\wfe                                             & YES                           & wait-free         & medium & HYB\\ 
\rowcolor[rgb]{0.753,0.753,0.753} crystallineL  & YES                          & lock-free         & medium & HYB\\ 
crystallineW                                    & YES                          & wait-free         & medium & HYB\\

\end{tabular}
\caption{Reclamation algorithms used in the benchmark. We categorize an algorithm as medium if at least once it performed medium, as low if at least once it performed lowest, and high if always it performed above the medium category.}
\label{tab:smrtab}
\end{table}
}

\begin{table}
\fontsize{10pt}{10pt}\selectfont
\centering
\begin{tabular}{lllll}
\toprule
\textbf{SMR Algorithm}                                            & \textbf{\makecell{Bounds\\ Garbage?}}  & \textbf{Progress} & \textbf{Type} & \textbf{\makecell{Memory Layout\\ Changes?}} \\
\midrule
\rowcolor[rgb]{0.753,0.753,0.753}\nbrp\cite{singh2021nbr}        & Yes                           & cond. lock-free   & OS signals (relaxed HP) &  No\\ 
\nbr\cite{singh2021nbr}                                             & Yes                           & cond. lock-free   & OS signals (relaxed HP) & No\\
\rowcolor[rgb]{0.753,0.753,0.753} \debra\cite{brown2015reclaiming}         & No                            & blocking         & EPOCH & No\\ 
\hp\cite{michael2004hazard}                                              & Yes                           & lock-free        & \hp & No\\ 
\rowcolor[rgb]{0.753,0.753,0.753} \qsbr \& \rcu \cite{hart2007performance}   & No                            & blocking         & EPOCH & No\\ 
\geibr\cite{wen2018interval}                                          & No                           & lock-free         & EPOCH + \hp & Yes\\ 
\rowcolor[rgb]{0.753,0.753,0.753} \he \cite{ramalhete2017brief}            & No                          & lock-free          & EPOCH + \hp & Yes\\ 
\wfe \cite{nikolaev2020universal}                                             & Yes                           & wait-free        & EPOCH + \hp & Yes\\ 
\rowcolor[rgb]{0.753,0.753,0.753} crystallineL\cite{nikolaev2021crystalline}  & Yes                          & lock-free         & EPOCH + RC & Yes\\ 
crystallineW\cite{nikolaev2021crystalline}                                    & Yes                          & wait-free         & EPOCH + RC & Yes\\
\rowcolor[rgb]{0.753,0.753,0.753} VBR\cite{sheffi2021vbr}  & Yes                          & lock-free         & EPOCH & Yes\\ 
\bottomrule
\end{tabular}
\caption{SMRs used in benchmark. OS implies use of Operating System features, \hp: hazard pointers like reservation. EPOCH: use of EBR. RC: reference counting. The QSBR (threads announce quiescence at the end of operations) and RCU implementations are adapted from the IBR benchmark~\cite{wen2018interval}.
}
\label{tab:smrtab}
\end{table}

For each experiment type, we evaluate multiple memory reclamation algorithms applied to different data structures. Each configuration is run multiple times, and the results are averaged to produce a data point in a plot. In each run, all threads access a single data structure and use a single memory reclamation algorithm for a fixed period of time.

Data structures include lists, hash tables, and trees which exhibit varying memory access patterns. 
Specifically, we utilized the lazylist~(LL)~\cite{heller2005lazy}, Harris list~(HL)~\cite{harris2001pragmatic} and Harris-Michael list~(HMList)~\cite{michael2004hazard} for lists,
HMList chaining based hashtable (HMHT), for hash table, and external binary search tree of David et~al. (DGT)~\cite{david2015asynchronized} and Brown's relaxed (a,b)-tree (BABT)~\cite{brown2017techniques} for trees.

LL and DGT belong to the class of data structures that are naturally compatible with \nbr as they have a single \rdp and \wtp per operation. 
HL and BABT belong to the class of data structures with multiple \rdp and \wtp that start every \rdp from root. Thus, they are suitable for \nbr with careful separation of the phases as discussed in \secref{comptds}.
Finally, HMList and HMHT represent the semi-compatible data structures discussed in \secref{semicomptds}.
Each of these data structures is implemented with up to \textbf{twelve} different reclamation algorithms (as applicable), including \nbr, \nbrp and \textit{none} algorithm which is a leaky implementation. 
The reclamation algorithms are summarized in \tabref{smrtab}. 
\ignore
{The algorithms \qsbr, \rcu and \debra are epoch-based.
The algorithms \geibr, hazard eras (\he), and crystallineL represent the most recent hybrids of epoch-based and \hp-based approaches. 
We also evaluate algorithms that provide stronger wait-free progress guarantees like Wait Free Eras (\wfe) and crystallineW.
Besides this, the plots also include hazard pointers (\hp). A \textit{none} algorithm is included for reference, which is just the data structure implementation without any reclamation.
}

Trees do not include the lines for crystallineL and crystallineW. Although these algorithms worked correctly with lists, they occasionally crashed when used with trees in oversubscribed scenarios.
We may be using them wrong, but they appear to have the same usage limitations as HPs, and it is not clear how to use HPs with these trees~\cite{brown2015reclaiming}. 
Therefore, we decided to exclude them from the tree plots. Additionally, since \hp does not apply to HL~\cite{michael2004hazard}, it was omitted from the HL plots. Furthermore, VBR\cite{sheffi2021vbr} was not used in our experiments because it assumes a type-stable allocator that never frees memory to the OS, and a fair comparison would require forcing all algorithms to use memory pools. 

The reported results are obtained by averaging data from 5 timed trials, each lasting 5 seconds. The 
experiments were conducted with thread counts ranging from 1 to 384 threads on a system with 192 hardware threads and more than 91\% of data points had less than 5\% variance.
Before each execution, the data structure was pre-filled to half of its maximum size.

For each of (E1), (E2), and (E3) we measure throughput for three workloads, (1) \textit{update-intensive}: 100\% updates where 50\% of operations are inserts and the rest are deletes, (2) \textit{balanced}: 25\% of operations are inserts, 25\% are deletes and rest are searches, and (3) \textit{search-intensive}: 5\% of the operations are inserts, 5\% are deletes and the rest are searches.

\begin{figure}
\centering
     \begin{minipage}{\textwidth}
        \begin{subfigure}{\textwidth}
            \includegraphics[width=0.33\linewidth, height=6cm, keepaspectratio]{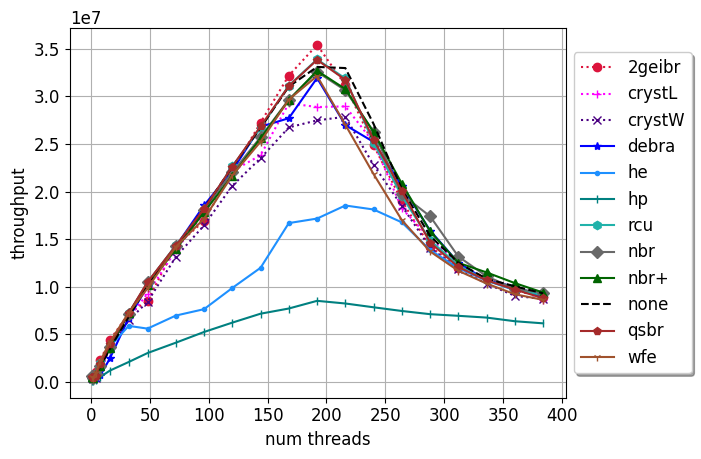}\hfill
            \includegraphics[width=0.33\linewidth, height=6cm, keepaspectratio]{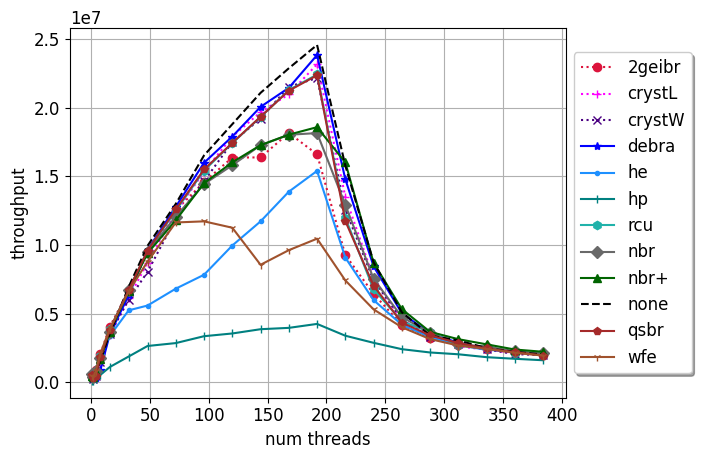}\hfill
            \includegraphics[width=0.33\linewidth, height=6cm, keepaspectratio]{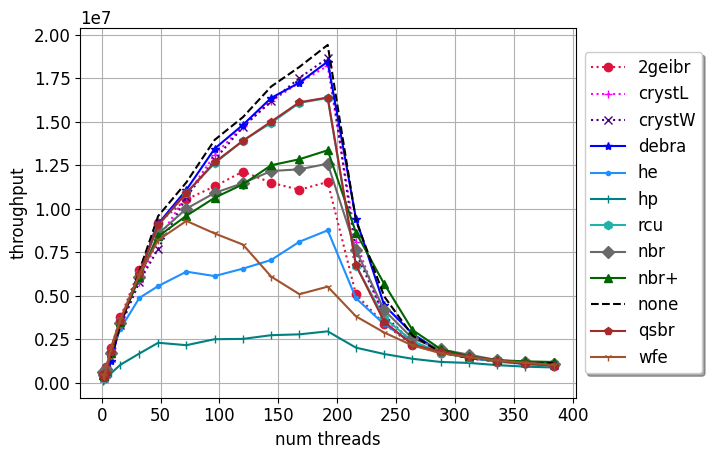}\hfill
            \caption{Lazylist (LL) throughput. Updates: Left: 10\%. Middle: 50\%. Right: 100\%. Max size:2000.}
            \label{fig:e1ll2000}
        \end{subfigure}
        \begin{subfigure}{\textwidth}
            \includegraphics[width=0.33\linewidth, height=6cm, keepaspectratio]{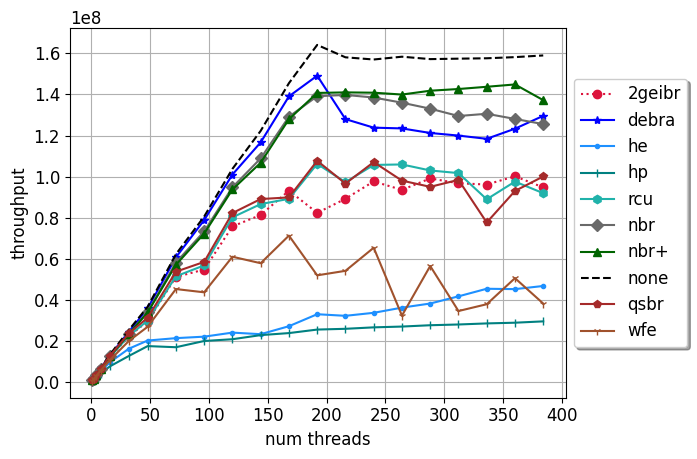}\hfill
            \includegraphics[width=0.33\linewidth, height=6cm, keepaspectratio]{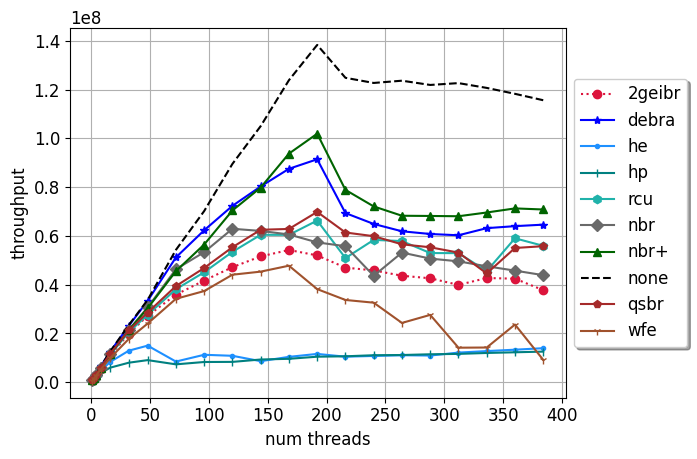}\hfill
            \includegraphics[width=0.33\linewidth, height=6cm, keepaspectratio]{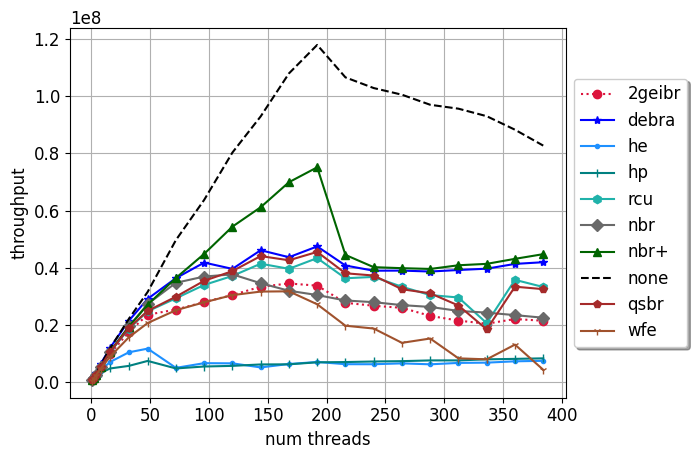}\hfill
            \caption{DGT throughput. Updates: Left: 10\%. Middle: 50\%. Right: 100\%. Max size:2000000.}
            \label{fig:e1dgt2000000}
        \end{subfigure}        
     \end{minipage}
    \caption{\textbf{E1:} Evaluation with data structures having a single read-write phase. Y axis: throughput in million operations per second. X-axis: \#threads.}
    \label{fig:e1ll}
\end{figure}

\begin{figure}
\centering
     \begin{minipage}{\textwidth}
        \begin{subfigure}{\textwidth}
            \includegraphics[width=0.33\linewidth, height=6cm, keepaspectratio]{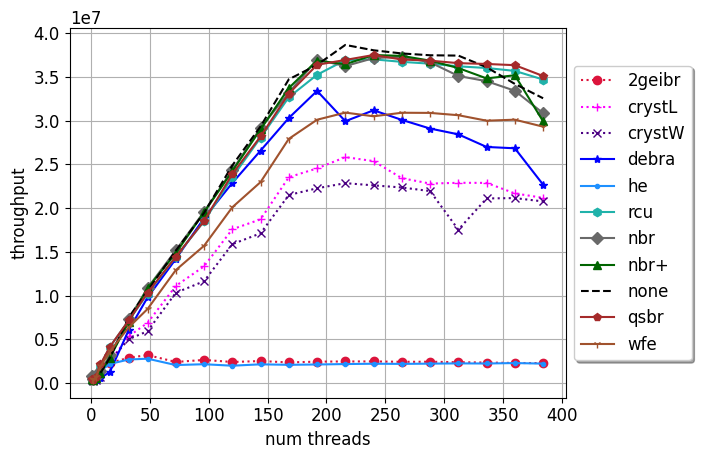}\hfill
            \includegraphics[width=0.33\linewidth, height=6cm, keepaspectratio]{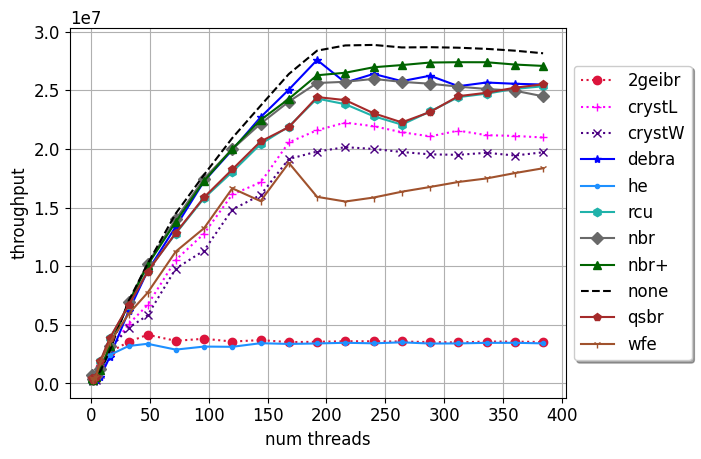}\hfill
            \includegraphics[width=0.33\linewidth, height=6cm, keepaspectratio]{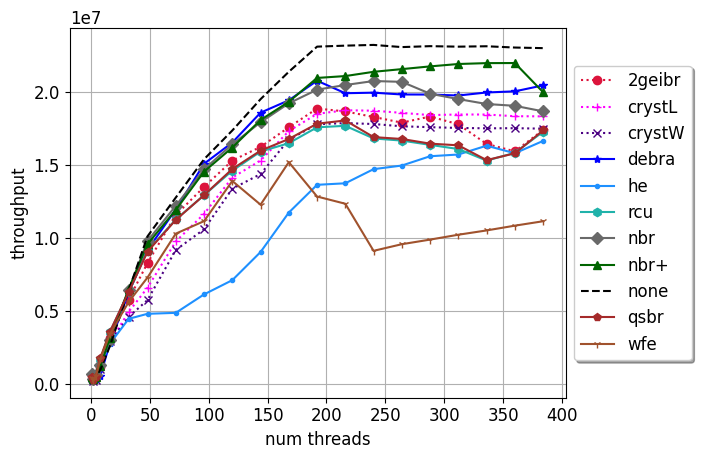}\hfill
            \caption{Harris list (HL) throughput. Updates: Left: 10\%. Middle: 50\%. Right: 100\%. Max size:2000.}
            \label{fig:e1hl2000}
        \end{subfigure}
        \begin{subfigure}{\textwidth}
            \includegraphics[width=0.33\linewidth, height=6cm, keepaspectratio]{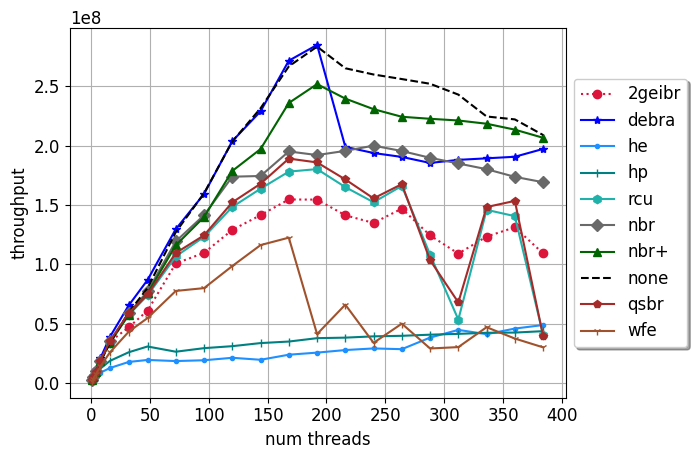}\hfill
            \includegraphics[width=0.33\linewidth, height=6cm, keepaspectratio]{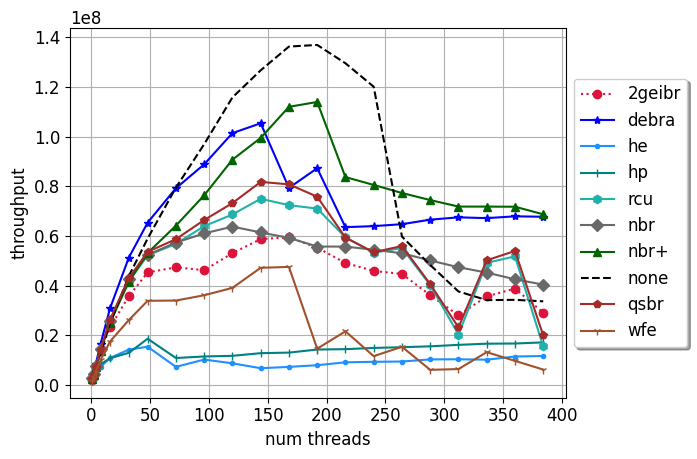}\hfill
            \includegraphics[width=0.33\linewidth, height=6cm, keepaspectratio]{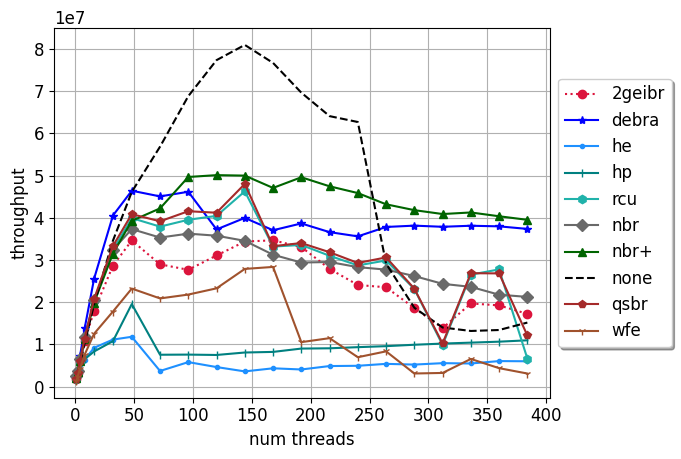}\hfill
            \caption{(a,b) tree (BABT) throughput. Updates: Left: 10\%. Middle: 50\%. Right: 100\%. Max size:2000000.}
            \label{fig:e1abtree2000000}
        \end{subfigure}
     \end{minipage}
    \caption{\textbf{E1:} Evaluation with data structures having multiple read-write phases. Y axis: throughput in million operations per second. X-axis: \#threads.}
    \label{fig:e1hl}
\end{figure}

\begin{figure}
\centering
     \begin{minipage}{\textwidth}
        \begin{subfigure}{\textwidth}
            \includegraphics[width=0.33\linewidth, height=6cm, keepaspectratio]{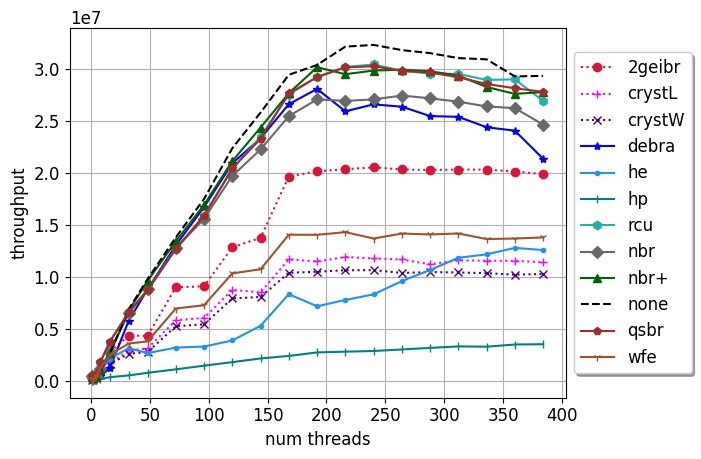}\hfill
            \includegraphics[width=0.33\linewidth, height=6cm, keepaspectratio]{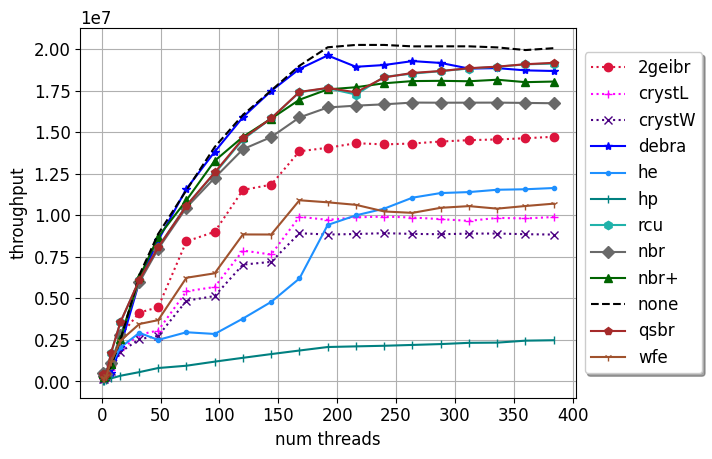}\hfill
            \includegraphics[width=0.33\linewidth, height=6cm, keepaspectratio]{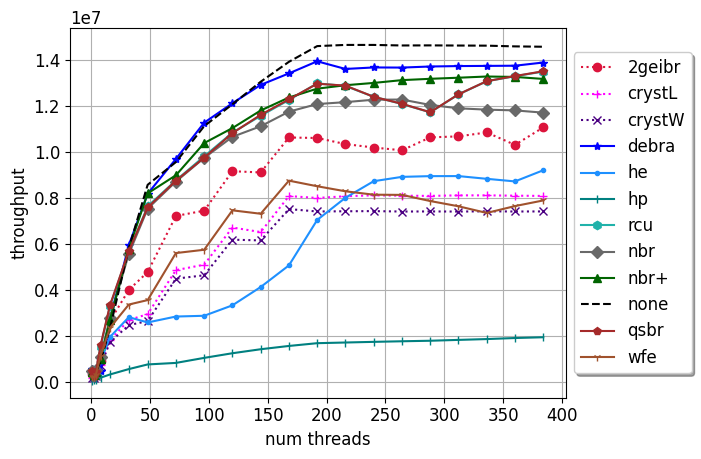}\hfill
            \caption{ Harris-Michael list (HMList) throughput. Updates: Left: 10\%. Middle: 50\%. Right: 100\%. Max size:2000.}
            \label{fig:e1hm2000}
        \end{subfigure}
        \begin{subfigure}{\textwidth}
            \includegraphics[width=0.33\linewidth, height=6cm, keepaspectratio]{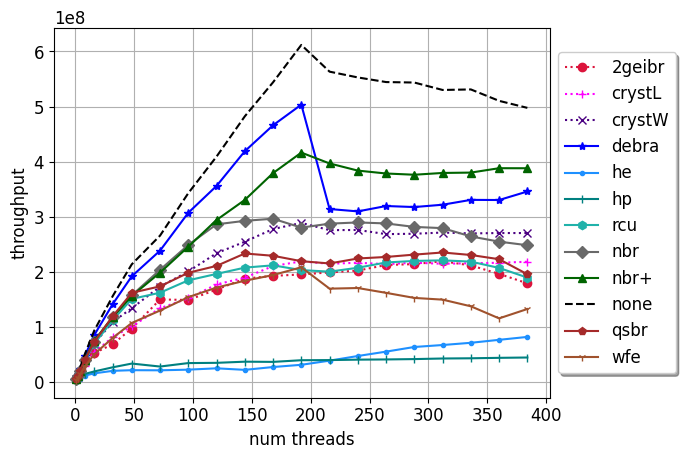}\hfill
            \includegraphics[width=0.33\linewidth, height=6cm, keepaspectratio]{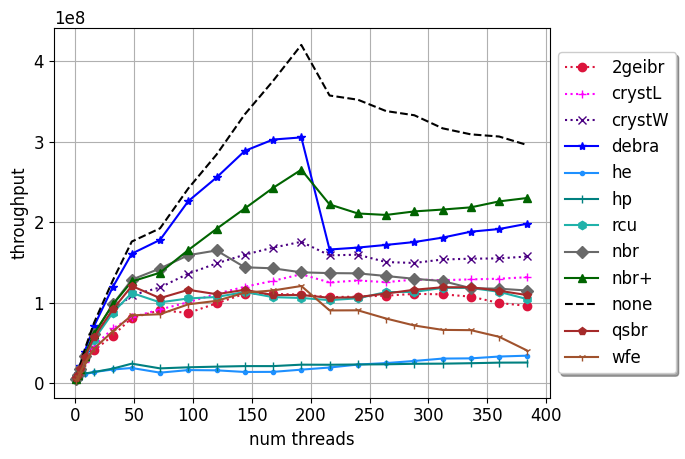}\hfill
            \includegraphics[width=0.33\linewidth, height=6cm,keepaspectratio]{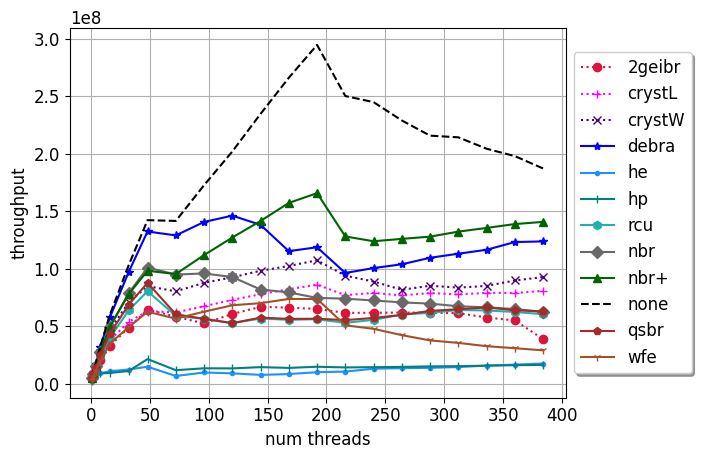}\hfill
            \caption{HMlist-based hashtable (HMHT) throughput. Updates: Left: 10\%. Middle: 50\%. Right: 100\%. load factor:6. Buckets:10K.}
            \label{fig:e1ht6}
        \end{subfigure}
     \end{minipage}
    \caption{\textbf{E1:} Evaluation with data structures that can be modified to restart from root/head. Y axis: throughput in million operations per second. X-axis: \#threads. In \nbr and \nbrp the HMlist restarts from root while for others it does not.}
    \label{fig:e1hm}
\end{figure}

\begin{figure}
\centering
     \begin{minipage}{\textwidth}
        \begin{subfigure}{\textwidth}
            \includegraphics[width=0.33\linewidth, height=6cm, keepaspectratio]{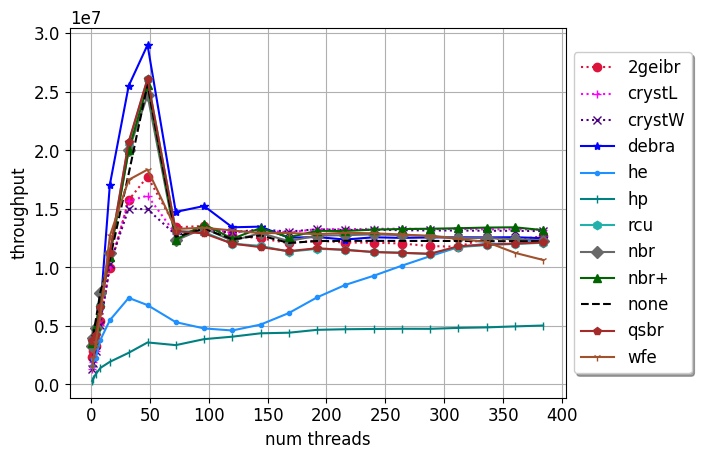}\hfill
            \includegraphics[width=0.33\linewidth, height=6cm, keepaspectratio]{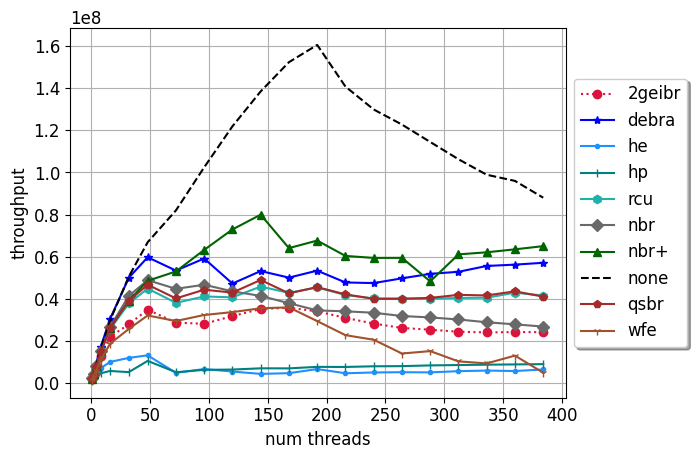}\hfill
            \includegraphics[width=0.33\linewidth, height=6cm, keepaspectratio]{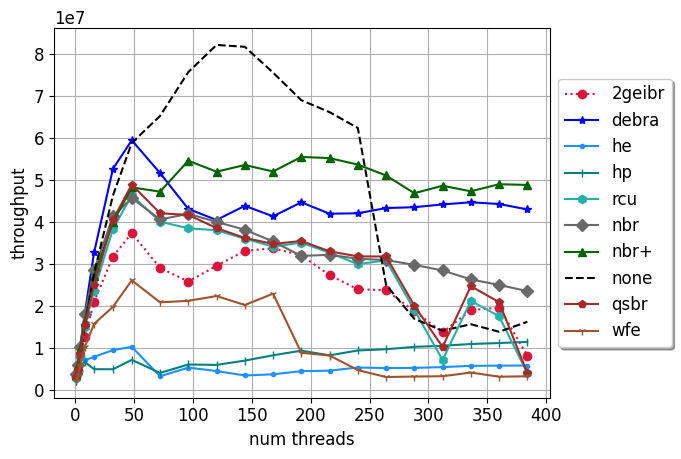}\hfill
            \caption{Left: Harris-Michael list size 200. Middle: DGT size 200K. Right: BABT size 200K.}
            \label{fig:e2small}
        \end{subfigure}
        \begin{subfigure}{\textwidth}
            \includegraphics[width=0.33\linewidth, height=6cm, keepaspectratio]{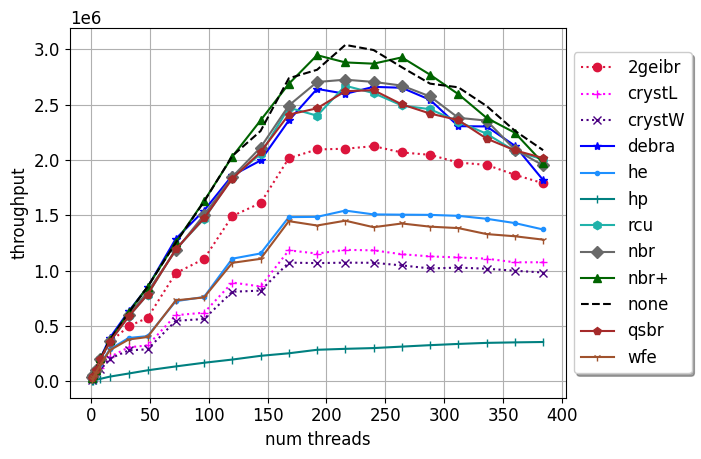}\hfill
            \includegraphics[width=0.33\linewidth, height=6cm, keepaspectratio]{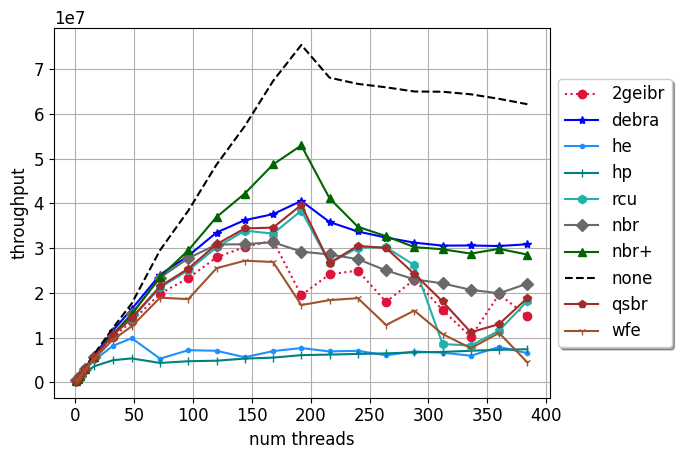}\hfill
            \includegraphics[width=0.33\linewidth, height=6cm, keepaspectratio]{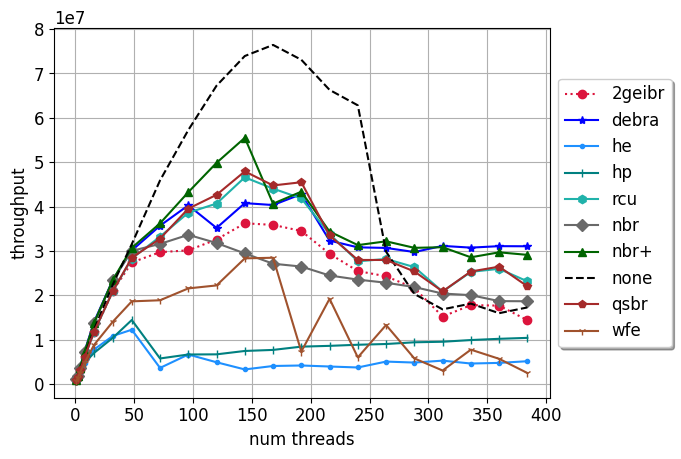}\hfill
            \caption{Left: Harris-Michael list size 20K. Middle: DGT size 20M. Right: BABT size 20M.}
            \label{fig:e2large}
        \end{subfigure}
     \end{minipage}
    \caption{\textbf{E2:} Throughput across different data structure sizes. Workload: 100\% updates. Y axis: throughput in million operations per second. X-axis: \#threads. Trees size of 20M exceeds LLC. Trees size of 200K fits LLC. Lists always fit in LLC.}
    \label{fig:e2}
\end{figure}

\subsection{Throughput}
\paragraph{Discussion of E1.}
\Figref{e1ll}, ~\ref{fig:e1hl}, ~\ref{fig:e1hm} show the throughput of all the data structures with varying workloads. In general, it can be observed that especially in balanced and update intensive workloads, \nbrp is similar to other reclamation techniques at lower thread counts, i.e up to 48 threads. Then at higher thread counts, from 48 to 192 threads, \nbrp is better than other reclamation techniques, except in the LL and HMList data structures. In oversubscribed scenarios, i.e., for thread counts greater than 192, \nbrp is competitive and sometimes outperforms the competition.

\ignore{
At a broader level, the reclamation techniques can be divided into three categories\textendash high, medium, and low performing. Except for the HL, \nbrp and \debra belong to the high-performing category, \hp and \he are in the low-performing category, and \geibr is in the medium category.
Others like \rcu, \qsbr, crystallineL and crystallineW in some cases perform high and in others perform medium.
The \wfe is in some cases medium and in others performs low.
Techniques like crystallineL, crystallineW, \geibr, \wfe, \he, and \hp have per read overhead which causes these algorithms to slow down, especially \hp and \he\footnote{We believe that \hp's poor performance is due to the high cost of \textit{mfence} used to publish the reservations which could be reduced by using more efficient \func{xchg} instructions to broadcast the reservations (refer section 11.5.1 in \url{https://www.amd.com/system/files/TechDocs/47414_15h_sw_opt_guide.pdf}).
We did that optimization separately in our experimentation and noticed that it did increase the absolute throughput for \hp but it remained significantly slower than \nbrp.}. 
On the other hand \nbr, \nbrp, \debra, \qsbr, and \rcu do not have per read overhead which explains why they belong to the high or middle-performing category. 
}

One interesting observation is the cross-over between \debra and \nbrp at the higher thread counts, particularly evident in the DGT, BABT and HMHT data structures. The slowdown of \debra at higher thread counts could be attributed to the \textit{delayed thread vulnerability}, where slow threads infrequently advance epochs, halting regular reclamation of limbo bags.
This results in the accumulation of a large number of retired records waiting to be reclaimed.
When the slow thread finally announces the latest epoch, all threads reclaim their large limbo bags, causing a \textit{reclamation burst} (many records being freed at once).
This burst harms overall throughput as reclamation bursts bottleneck the underlying allocator by increasing contention and triggering slow code paths as internal buffers in the allocator are quickly filled.
The probability of threads getting delayed increases as more threads get involved in high inter-socket and update-intensive computations. 
Furthermore, the thread count where \nbrp overtakes \debra is lower in the \textit{update-intensive} than the \textit{search-intensive} workloads.
This is because the overhead of burst reclamation sets in at lower thread counts for \textit{update-intensive} workloads.

\paragraph{Discussion of E2.}
In the second experiment, we evaluate all the reclamation algorithms with data structures of small and very large sizes. This serves two purposes. First, this illustrates the behavior under high contention at small sizes and low contention at large sizes. Second, this experiment also studies the impact of varying cache miss rates on throughput when data fits in the LLC (last level cache), and when it does not. 

\Figref{e2}, shows the HMList (left column), DGT (middle column), and BABT (right column). The first row (\figref{e2small}) depicts these data structures with sizes that fit in the LLC.
The HMList is of size 200 and both the trees are of size 200K.
The second row (\figref{e2large}) depicts these data structures with sizes that do not fit in LLC. The HMList is of size 20K and both the trees are of size 20M.
Note, the list, even at 20K size, fits in LLC and it is not possible to go beyond the size of 20K and complete the experiments in a reasonable amount of time.
So, to emulate the case of very low contention in lists we use the maximum size of 20K.
Additionally, in \Figref{e1ht6} and~\figref{hmht600} we include results for chaining hash tables with 60K buckets (fits in the LLC) and 6M buckets (exceeds the LLC), with load factors 6 and 600, respectively.



\begin{figure}
\centering
\includegraphics[width=\linewidth, height=6cm,keepaspectratio]{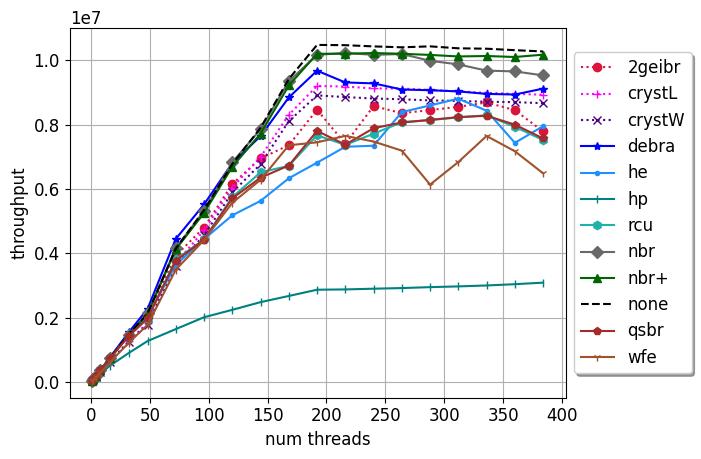}\hfill
\caption{\textbf{E2:} HMList-based Hash table with load factor 600 that does not fit in LLC.}
\label{fig:hmht600}
\end{figure}

Interestingly, \nbrp outperforms the competition, especially at the high thread counts. When combined with the analysis of E1, it is clear that \nbrp is fast [P1] and consistent [P4]. Specifically, \nbrp is comparable to other fast epoch-based algorithms like \debra, \qsbr, \rcu for the HMList at both the small and large sizes (see the left column in \Figref{e2small},~\ref{fig:e2large}).
In DGT (see the middle column in \Figref{e2small},~\ref{fig:e2large}), \nbrp is significantly faster at higher thread counts.
Similarly, in BABT (see the right column in \Figref{e2small},~\ref{fig:e2large}) \nbrp is faster at small sizes and comparable to other reclamation algorithms at large sizes. 
In the case of HMHT, \nbrp is comparable at small sizes (see the right column in \Figref{e1ht6}) as well as at large sizes (see~\figref{hmht600}).

Furthermore, experiment E2 also helps to explore two disparate usage scenarios for semicompatible data structures like HMList and the HMHT.
First, at a large size, for example 20K in HMList and a load factor of 600 in HMHT, we hypothesize restarts would be inexpensive due to low contention.
Second, at a small size, i.e. 200 in HMList and the load factor of 6 in HMHT, we assume that restarting from the head node will occur frequently due to high contention, therefore it should have high overhead and slowdown the data structures. 


For low contention, the \nbrp based implementation which enforces restarts is still faster than other reclaimers which do not restart, in HMlist and HMHT. Surprisingly, even for high contention, \nbrp outperforms other reclaimers. This suggests that, at least for semicompatible lists, \nbr's methodology is sufficiently fast and restarts in the data structure incur low overhead.

\subsection{Memory Usage}
\begin{figure}
\centering
    \includegraphics[width=0.49\linewidth, height=6cm, keepaspectratio]{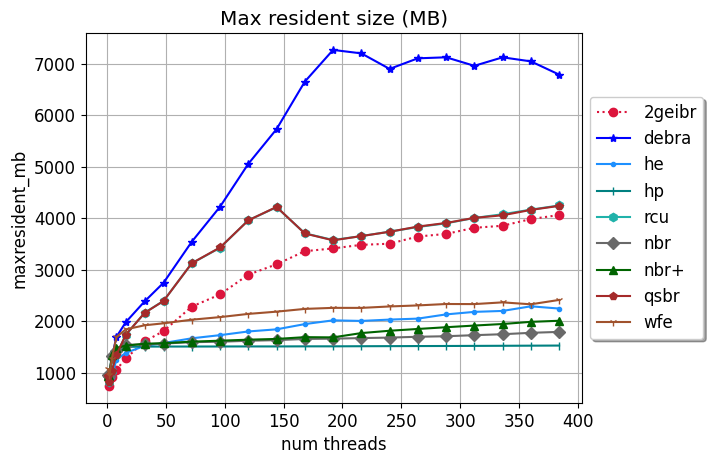}\hfill
    \includegraphics[width=0.49\linewidth, height=6cm, keepaspectratio]{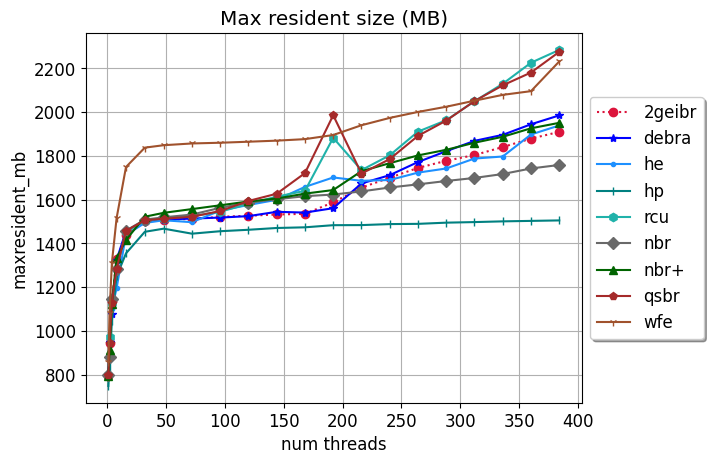}\hfill
    \caption{\textbf{E3:} Left: DGT with stalled threads. Right: DGT with no stalled threads. Max Size: 20M}
   \label{fig:e3}
\end{figure}

\paragraph{Discussion of E3.}
In our third type of experiment, we measure peak memory consumption of all reclaimers when a thread is stalled (see the middle column in \Figref{e3}) and when no thread is stalled (see the right column in \Figref{e3}) to establish that \nbr(+) bounds garbage.  

Each trial is run for 40 seconds.
For \debra, \qsbr, \rcu, \hp, \he, \nbr and \nbrp one thread is made to sleep within a data structure operation, for whole 40 seconds, imitating a stalled thread.
And, for \geibr, one thread is made to stall to maximize the interval of reservations such that maximum number of nodes are prevented from getting reclaimed, imitating pathological scheduling.

As expected, \debra, \geibr, \qsbr and \rcu do not have bounded garbage, show an increase in peak memory usage when a thread stalls (\Figref{e3}). In contrast, \nbr, \nbrp, \hp exhibit similar memory usage regardless of thread stalls.
In \he and \wfe memory usage increases marginally as a result of the stalled thread, showing some minor sensitivity to stalled threads. 

In summary, our experiments reveal that \nbr is fast, bounds garbage, and maintains the consistent performance across various scenarios. The code for our benchmark is publicly available at \url{https://gitlab.com/aajayssingh/nbr_setbench_plus}.

\newpage
\section{Summary}
\label{sec:nbrpconclusion}

In this chapter, we described the design and implementation of \nbrp: the second safe memory reclamation algorithm built using the neutralization paradigm.
We extensively evaluated both algorithms, \nbr introduced in \chapref{chapnbr}, and \nbrp introduced in this chapter, demonstrating that the neutralization algorithms are also fast and consistent in practice.

In the future, it will be interesting to test the behavior of our algorithms in a setting where our benchmark is run alongside several other data structures and unrelated processes on the same machine.

In terms of the objective we set for the neutralization paradigm in \chapref{chapnbr}, we have established three desirable traits: speed, bounded garbage, and consistent performance. In the next \chapref{chapappuse}, we discuss the remaining desirable traits, namely the applicability and ease of use of neutralization-based algorithms for many data structures.

\chapter{Practical Considerations and Guidelines}
\label{chap:chapappuse}


In the previous chapters (\chapref{chapnbr} and \chapref{chapnbrp}) we established that algorithms based on the neutralization paradigm are fast, consistent, and bound garbage, demonstrating three of the five key desirable properties that we set as our objective in \chapref{chapnbr}.
In this chapter, we discuss the remaining two properties: applicability and usability of the neutralization based algorithm.

The outline of this chapter is as follows.
In \secref{appl}, we start with a discussion on common patterns of data structures that influence whether NBR ( +) can be applied to them or not.
We also conduct a brief survey of eighteen popular data structures and report their applicability to several state of the art safe memory reclamation techniques, including NBR(+). 
Further, we discuss in detail the applicability of Hazard Pointers and NBR(+) to these data structures. 
In \secref{usb}, we demonstrate with an example how programmers can use NBR(+) in a data structure and also compare the ease of integration of two other state of the art reclamation algorithms. We end the chapter by presenting a brief set of rules that a programmer should keep in mind while using the POSIX signal interface. We hope these rules will further guide programmers to confidently use NBR(+).

\section{Applicability}
\label{sec:appl}

We study the applicability of NBR(+) to various data structures and categorize them into three broad categories: compatible, semi-compatible, and incompatible. Compatible data structures can be further classified into two types based on the pattern of read and write phases.


\subsection{Compatible data structures.}
\label{sec:comptds}

The first type of compatible data structures exhibit a single \rdp followed by a single \wtp. Examples include optimistic data structures such as lazylist~\cite{heller2005lazy} and DGT~\cite{david2015asynchronized}, where a sequence of reads identifies the target nodes (\rdp), followed by a few writes on those nodes (\wtp). Applying  \nbr(+) to these data structures is straightforward, involving invoking \texttt{checkpoint} and \Call{begin\rdp()}{} before the sequence of reads, and invoking \Call{end\rdp()}{} when the target nodes are identified and updates are executed. For read-only operations, \Call{end\rdp()}{} is invoked before the operation returns to reset the \texttt{restartable} flag and mark the start of the quiescent phase. 

The second type of compatible data structures involve a pattern of alternating \textit{read-write phases}. For example, Harris's lock-free list(HL)~\cite{harris2001pragmatic} shown in \algoref{knbr_eg} follows this pattern.
By examining our exposition on the HL, readers can gain insights into how  \nbr(+) can be applied to more sophisticated data structures that share the similar design pattern~\cite{brown2017techniques,ellen2010non,howley2012non,shafiei2013non,he2017deletion}. 

\begin{algorithm}
\small
    \caption{Demonstration that  \nbr(+) is simple to use with Harris list\cite{harris2001pragmatic} with multiple read/write phases $(\rdp \space\space \wtp)^+$ and other compatible data structures.}\label{algo:knbr_eg}
    \begin{algorithmic}[1]
    \Statex
    \Statex
\lstset{numbers=left,xleftmargin=4em,frame=single, xrightmargin=2em}
\begin{lstlisting}[]

Node* search(key, Node** left_node) {|\label{lin:ksrch}|
  Node *left_node_next, *right_node;
  search_again:
  do { 
      |\textcolor{red}{CHECKPOINT();}|    |\label{lin:ksrchchkp}|    
      |\textcolor{red}{begin\rdp{}();}|    |\label{lin:ksrchrdp}|    
      Node *t = head;
      Node *t_next = head.next;
      do{ |\label{lin:ksrchdo}|
          if(!is_marked_reference(t_next)){
            (*left_node) = t;
            left_node_next = t_next;
          }
          t = get_unmarked_reference(t_next);
          if (t == tail) break;
          t_next = t.next;
      }while(is_marked_reference(t_next) or (t.key<search_key)); |\label{lin:ksrchwhile}|
      right_node = t;
      |\textcolor{red}{end\rdp{}(left\_node, right\_node);}|   |\label{lin:ksrchwtp}| 
    
      if (left_node_next == right_node) 
        if ((right_node != tail) && is_marked_reference(right_node.next))
          goto search_again;
        else 
          return right_node; |\label{lin:ksrchret}| 
      // unlink and retire one or more marked nodes.
      if (|\textcolor{red}{CAS}|(&(left_node.next), left_node_next, right_node)) |\label{lin:auxupd}|
        if ((right_node != tail) && is_marked_reference(right_node.next))
          goto search_again;|\label{lin:secrdp}|
        else 
          return right_node;
  } while(true);
}
\end{lstlisting}
\Statex
\label{fig:alg3}
\algstore{part2}
\end{algorithmic}
\end{algorithm}

\begin{algorithm}
\small
\begin{algorithmic}[1]
\algrestore{part2}
\Statex
\Statex
\lstset{numbers=left,xleftmargin=4em,frame=single, xrightmargin=2em}
\begin{lstlisting}[]    
// note: search() reserves the node due to call to 
// |end\rdp()| before the search returns.
// no new unreserved nodes are accessed.
bool insert(key) {
  Node *right_node, *left_node;
  do{
      right_node = search (key, &left_node); |\label{lin:invsrch}|
      if((right_node!=tail) && (right_node.key==key)) 
        return false; 
      Node *new_node = new Node(key);
      new_node.next = right_node; 
      if (|\textcolor{red}{CAS}|(&(left_node.next), right_node, new_node)) |\label{lin:finwt}|
        return true;
  }while (true)
}
\end{lstlisting}
    \end{algorithmic}
\end{algorithm}

In HL, during an update or contains operation, a sequence of reads is performed to identify target nodes (\lineref{ksrchdo}-\ref{lin:ksrchwhile}), followed by auxiliary updates to unlink marked nodes from the list(\lineref{auxupd}). 
The sequence of reads is then restarted from the \textit{head} of the list, potentially repeating multiple times if marked nodes are encountered, resulting in a sequence of \textit{read-write phases}.

To apply  \nbr(+) to data structures with alternating \textit{read-write phases}, the following steps are followed. First, the \texttt{checkpoint} and \textsc{begin\rdp()} operations are invoked at the beginning of the sequence of reads. Then, the \textsc{end\rdp()} operation is called when auxiliary updates begin. After that, the \texttt{checkpoint} and \textsc{begin\rdp()} operations are invoked again when the sequence of reads restarts from the head of the list. Finally, the \textsc{end\rdp()} operation is invoked when the final intended update is executed. In the provided example, \lineref{ksrchrdp} and \lineref{ksrchwtp} indicate the start and end of \rdp and \wtp, respectively. For more detailed examples and discussions on incorrect placements of the \textsc{begin\rdp()} and \textsc{end\rdp()} functions a programmer should be aware of, refer to \secref{nbrrules}.

If a neutralization signal is received during the first \rdp, all threads will discard their shared node references and restart from the \texttt{checkpoint} corresponding to the first \textsc{begin\rdp()} invocation. Similarly, if the neutralization signal is received during the second \rdp, after an auxiliary update, the thread will discard all references acquired during this second \rdp and jump back to the \texttt{checkpoint} corresponding to the second \rdp.

It is crucial to note that each \rdp, including the second \rdp, starts from the head node of the data structure. As a result, there are no shared node references carried forward from the previous \rdp or \wtp (auxiliary update phase). This behavior ensures that each \rdp is treated as a new operation, disregarding any previously acquired node references. Thus, the requirement of acquiring all references to shared nodes during the current phase and discarding them (refer to \rdp rules in \secref{dsassmp}) when restarting from the corresponding checkpoint is met, ensuring the safety of reclamation.

\subsubsection{Limitation: restarting from the root.}
\label{sec:wtplimit}

If instead of restarting from head (or entry point), threads were to continue searching from a midpoint in the list, such as resuming from a shared node $R$ that was reserved during the previous \wtp, would introduce a risk of dereferencing a freed node. Although $R$ itself cannot be freed due to reservation, the nodes it points to may not be reserved and thus could be freed. Consequently, following any pointer from $R$ could lead to accessing a freed node, resulting in a crash.

Fortunately, there are many concurrent data structures in the literature that naturally restart from the head after auxiliary updates and exhibit alternating \textit{read-write phases}. These data structures, are suitable for integration with \nbr. Examples of such data structures include Harris' list~\cite{harris2001pragmatic}, Brown's lock-free ABTree, chromatic tree, AVL tree (B17)~\cite{brown2017techniques}, Natarajan et al.'s lock-free binary search tree~\cite{natarajan2014fast}, and several others~\cite{howley2012non,shafiei2013non,he2017deletion}, including the recently proposed elimination (a,b) Tree~\cite{srivastava2022elimination}. In our experiments in \secref{nbrpeval}, we utilized the Harris list and ABTree among these options.

\subsection{Semi-compatible data structures.}
\label{sec:semicomptds}
The need to restart from the entry point (head in lists or root in trees) at the beginning of each read phase presents a challenge in directly applying  \nbr(+) to data structures with patterns of alternating \textit{read-write phases} that do not restart from the root node. For example, the Harris-Michael list~\cite{michael2004hazard} and certain search trees~\cite{ellen2014amortized, drachsler2014practical, brown2020non, ramachandran2015castle}. However, it is possible to adapt  \nbr(+) for use with these data structures by modifying their operations to restart from the root after auxiliary updates. This adaptation, requires careful consideration as it may introduce trade-offs, potentially affecting progress guarantees, necessitating a new amortized complexity analysis, or introducing additional overhead.

\begin{figure}
\centering
            \includegraphics[width=0.33\linewidth, height=5cm, keepaspectratio]{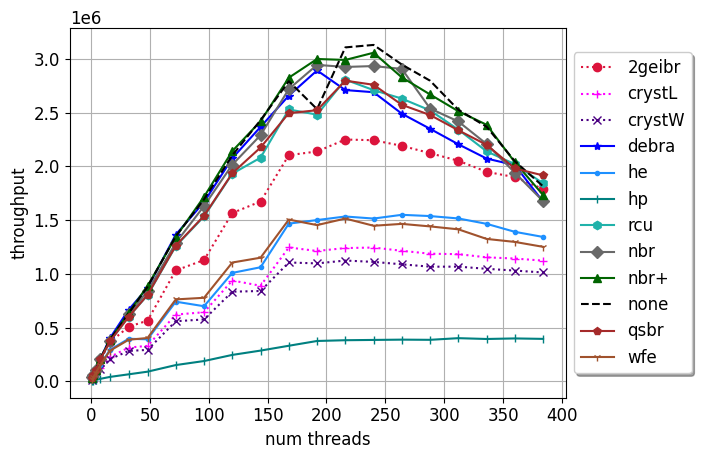}\hfill
            \includegraphics[width=0.33\linewidth, height=5cm, keepaspectratio]{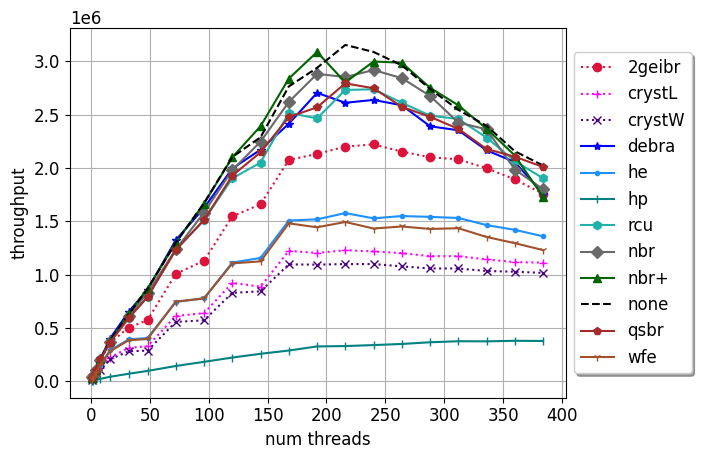}\hfill
            \includegraphics[width=0.33\linewidth, height=5cm, keepaspectratio]{figs/nbrfig/hm/throughput-hmlist-u50-sz20000.png}\hfill
\caption{Harris-Michael list.
Figure showing the impact of transforming Harris-Michael list to restart from root to apply  \nbr(+)  is negligible. In  \nbr(+) and \nbrp the HMlist restarts from root while for others it does not.
Left: 10\% updates. Middle: 50\% updates. Right: 100\% updates. Max size:20000.}
\label{fig:semicomp_hmlist}
\end{figure}

Nevertheless, we observe that modifying data structures to restart from the root after auxiliary updates is often a viable solution that does not significantly increase overhead. In scenarios with high contention, such as the Harris-Michael list, where multiple threads contend to unlink a same marked node, the majority of threads ($n-1$) would already restart from the root~\cite{michael2004hazard}. By enforcing a restart from the root for all threads, the number of threads requiring a restart increases by just one, resulting in $n$ threads restarting instead of $n-1$. This behavior aligns with the Harris list~\cite{harris2001pragmatic}, where all threads contending on the auxiliary CAS already restart from the root. In low contention scenarios, the impact of restarting from the root on performance is minimal.

Experimental results further support the notion that the cost of restarting from the root is negligible when adapting the Harris-Michael list to work with  \nbr(+) compared to using other reclamation techniques which use the list in its original form.~\figref{semicomp_hmlist} illustrates the analysis of the impact of restarting due to contention induced by varying workloads. Additionally,~\figref{e2small} \& \figref{e2large} (left column) provide an analysis of the impact of restarts resulting from contention induced by different data structures. These findings demonstrate that the overhead associated with restarting from the root in practical scenarios is generally low. Moreover, in search trees, given uniform node access distribution, contention is low, the performance difference between restarting from the root and continuing traversal from an ancestor is expected to be minimal due to shorter search paths.

By considering these factors, it is evident that modifying data structures to restart from the root after auxiliary updates allows for effective integration with  \nbr(+)  while maintaining reasonable performance characteristics.

\subsection{Incompatible data structures.}
%

Certain data structures, such as two concurrent relaxed-balance AVL trees~\cite{bronson2010practical}, require rotations after update operations to maintain height balance. These rotations may discover new nodes that were not accessed during the previous search phase, making it impractical to reserve nodes in advance for subsequent write phases. As a result, the  \nbr(+) technique does not apply to these data structures. Further, modifying their implementation to initiate rotations from the root would require extensive code changes which may require reproving progress or correctness.

Similarly, a recently introduced lock-free interpolation tree~\cite{brown2020non} periodically rebuilds its subtrees by repeatedly visiting and marking all nodes in old subtree which were not reserved beforehand to maintain balance. And, in order to visit and mark the nodes, it does not restart from the root, which violates the requirement of \nbr. Consequently,  \nbr(+) doesn't apply, also neither the DEBRA+ technique nor Hazard Pointers can be applied to this tree. At present, we are unaware of any SMR technique with bounded garbage that is compatible with the interpolation tree.

\subsection{Compatibility of  \nbr(+)  vs. other SMR algorithms}
\label{sec:apnappl}
%

\begin{table*}
\fontsize{8pt}{10pt}\selectfont
\centering
\begin{tabular}{lllllll}
\toprule
\textbf{Source}
& \textbf{Data structure}              & \textbf{Sync. type} & \textbf{ \nbr(+) }      & \textbf{EBR} & \textbf{DEBRA+} & \textbf{\makecell{HP/TS/IBR/HE\\/WFE/HY/QSense}} \\
\midrule
LL05\cite{heller2005lazy}             & linked list                 & opt. locks      & Yes       & Yes & No     & \makecell{No (reason,\\ similar to \cite{brown2015reclaiming})}                \\
\rowcolor[rgb]{0.753,0.753,0.753} HL01\cite{harris2001pragmatic}                                          & linked list                 & lock-free       & Yes       & Yes & *      & Yes                \\
HM04\cite{michael2004hazard}          & linked list                 & lock-free       & No        & Yes & *      & Yes                \\

\rowcolor[rgb]{0.753,0.753,0.753} \rowcolor[rgb]{0.753,0.753,0.753}DVY14a\cite{drachsler2014practical}    & partially external BST      & locks           & **        & Yes & No     & No \cite{brown2015reclaiming}                \\
EFRB10\cite{ellen2010non}             & external BST                & lock-free       & No       & Yes & *      & No \cite{brown2015reclaiming}                \\
\rowcolor[rgb]{0.753,0.753,0.753} NM14\cite{natarajan2014fast}                                            & external BST                & lock-free       & Yes       & Yes & *      & No \cite{brown2015reclaiming}                 \\
EFRB14\cite{ellen2014amortized}                                         & external BST                & lock-free       & No        & Yes & *      & No \cite{brown2015reclaiming}                \\
\rowcolor[rgb]{0.753,0.753,0.753} DGT15\cite{david2015asynchronized}                                      & external BST                & ticket locks    & Yes       & Yes & No     & \makecell{No (no marking,  \\cannot validate HP) {   }{ }}                \\
HJ12\cite{howley2012non}                                                & internal BST                & lock-free       & Yes       & Yes & *      & No (similar to \cite{brown2015reclaiming})                \\
\rowcolor[rgb]{0.753,0.753,0.753} RM15\cite{ramachandran2015fast}       & internal BST                & lock-free       & No        & Yes & No     & No (similar to \cite{brown2015reclaiming})                \\

BCCO10\cite{bronson2010practical}     & partially external AVL      & opt. locks      & No        & Yes & No     & Yes                \\
\rowcolor[rgb]{0.753,0.753,0.753} DVY14b\cite{drachsler2014practical}                                     & partially external AVL      & locks           & No        & Yes & No     & No \cite{brown2015reclaiming}                \\
HL17\cite{he2017deletion}             & external relaxed AVL tree   & lock-free       & Yes       & Yes & Yes    & No (similar to \cite{brown2015reclaiming})                 \\
\rowcolor[rgb]{0.753,0.753,0.753} B17b\cite{brown2017techniques}                                          & external AVL                & lock-free       & Yes       & Yes & Yes    & No \cite{brown2015reclaiming}                \\

S13\cite{shafiei2013non}              & patricia trie               & lock-free       & Yes       & Yes & *      & No \cite{brown2015reclaiming}                \\
\rowcolor[rgb]{0.753,0.753,0.753}BER14\cite{brown2014general}                                                                   & external chromatic tree     & lock-free       & Yes       & Yes & Yes    & No \cite{brown2015reclaiming}                \\
B17a\cite{brown2017techniques}        & external (a,b)-tree         & lock-free       & Yes       & Yes & Yes    & No \cite{brown2015reclaiming}                \\
\rowcolor[rgb]{0.753,0.753,0.753}BPA20\cite{brown2020non}                                                & external interpolation tree & lock-free       & No        & Yes & No     & No (similar to \cite{brown2015reclaiming})                \\
\bottomrule
\end{tabular}
\caption{Applicability of SMR algorithms. 
*compatible, but requires non-trivial data structure specific recovery code. 
**This is likely possible if code is restructured to reserve all relevant nodes before acquiring any locks.
This table was further updated in the work by Jung, Lee, Kim and Kang~\cite{jung2023applying}.  
}
\label{tab:appltab}
\end{table*}

This section surveys a carefully curated list of data structures, summarized in \tabref{appltab}, and reports whether it is compatible with  \nbr(+) and other SMR algorithms including EBR, DEBRA+ and HP (and variants of HP, including HE, IBR, WFE, ThreadScan, HY and QSense).



To safely apply  \nbr(+) the following requirements should be satisfied, as was discussed in \secref{wtplimit} and \secref{nbrreq}.


\begin{requirement}
\label{req:reqnbr}
Each \rdp in a data structure using  \nbr(+) should start from the root.
\end{requirement}

\begin{requirement}
\label{req:reqnbr2}
All references to shared pointers that could be accessed within a \wtp should be reserved before entering it.
\end{requirement}

Safely applying HPs to a data structure is more subtle.
Typically, accessing shared nodes in a data structure operation using the original form of HPs is a three step process where pointers to nodes are reserved in a hand-over-hand manner~\cite{michael2004hazard}. 
{\bf (1) announcing a hazard pointer:} requires a thread to save the shared node in a single-write multi-reader (SWMR) memory location.
{\bf (2) store-load memory order fence:} is necessary (on TSO) to ensure the reclaimers timely collect and skip freeing the announced hazard pointers.
{\bf (3) reachability validation:} is required to ensure the node was reachable from the root at the time it was announced to avoid announcing an already un-linked node that could be freed simultaneously.
When these three steps execute successfully, it can be said that the thread has {\em acquired} a hazard pointer.

\tabref{appltab} presents a survey of eighteen data structures, wherein \nbr{} applies to eleven, \debra{+} applies to four and \hp applies only to three. However this does not necessarily mean  \nbr(+) is strictly more applicable than \hp.

Qualitatively, \nbr{} does not apply to data structure implementations where read and write phases interleave in a way that leads to violation of shared node access requirements (referred to as~\reqref{reqnbr} and \reqref{reqnbr2}).
On the other hand, it is not clear how \hp should be applied to data structures where a sequence of logically deleted (or marked) nodes can be traversed~\cite{brown2015reclaiming}.
This scenario arises in an important and substantial class of well known data structures, including unbalanced binary trees~\cite{ellen2010non, natarajan2014fast, david2015asynchronized, ellen2014amortized, howley2012non, ramachandran2015fast, drachsler2014practical}, relaxed AVL trees~\cite{he2017deletion}, Chromatic trees~\cite{brown2014general}, B+trees\cite{braginsky2012lock}, ab-trees~\cite{brown2017techniques, srivastava2022elimination}, linked lists~\cite{heller2005lazy}, skip lists\cite{aksenov2023splay}, and Euler tour trees ~\cite{howley2012non}.
We have not studied the applicability of  \nbr(+)  to \textit{all} of these data structures, however several of them do appear in our survey.

Motivated in part by  \nbr(+) , Petrank et al.~\cite{sheffi2023era} formalized  \nbr(+) 's data structure template and requirements (\reqref{reqnbr} and \reqref{reqnbr2}), and gave a name to data structures that satisfy them: \textit{access-aware data structures}. 
This class of algorithms is the basis for the definition of what it means for an SMR algorithm to be \textit{widely applicable} in that paper. 
More precisely, an SMR algorithm that can be applied to all access-aware data structure implementations is said to be widely applicable.

 \nbr(+) is of course applicable to all access aware data structure implementations and is therefore widely applicable according to this definition.
On the other hand, HP and its many variants are not applicable to all access-aware data structures, and are thus not widely applicable. 
It appears that \debra{+} is also not widely applicable as it does not support lock-based access-aware data structure implementations. 

\subsubsection{Details of HP applicability}
Having discussed the requirements for ensuring safety and progress in a data structure, we next discuss the applicability of HP to the data structures in \tabref{appltab}.

In LL05\cite{heller2005lazy} searches are wait-free and updates use an optimistic locking pattern. If HP is applied to LL05 it is possible that a thread could repeatedly fail to acquire a HP on a node that is marked but not yet unlinked (because the thread that marked it is stalled or crashed). This breaks wait freedom for searches. 

DGT2015\cite{david2015asynchronized}, on the other hand, traverses the nodes in a synchronization-free manner (as in the lazy list) and uses version numbers and ticket locks to perform updates. However, since DGT15 doesn't use marking, it is not clear how HPs could be acquired safely.

In HJ12~\cite{howley2012non}, hazard pointer acquisition could fail indefinitely if a thread that marks a node fails before unlinking it. More specifically, threads that would try to help this failed thread to finish unlinking the node would need to acquire an HP on the marked node, and it is not clear how this HP could be acquired (without acquiring HPs to all nodes on the search path from the root). This could prevent all threads from making progress. Such issues were described in \cite{brown2015reclaiming}. Similarly, in HL17~\cite{he2017deletion} and the concurrent interpolation search tree of Brown et~al.~\cite{brown2020non}, it is not clear how a thread trying to acquire an HP on a node could verify that it is still reachable from the root. 

BCCO10\cite{bronson2010practical} uses lazy deletion in the sense that to delete a node with 2 children it converts it to a routing node (treating this as if the AVL tree is an external tree) and this routing node is deleted lazily at the time of next rebalancing step that involves it. Thus, it appears that we can use HP as one can leverage version based validation to know when a node is definitely reachable and if such a validation fails one can simply restart from the root. Note this doesn't theoretically impact progress guarantee of searches (unlike lazylist or DGT15) as they are already blocking due to hand over hand optimistic validation. However, in practice, HPs would necessitate restarts from the root when validation fails, whereas BCCO nominally attempts to continue traversal from the parent node in this case.
The logical ordering tree of DVY14\cite{drachsler2014practical} is based on BCCO10, but DVY14 does not use version numbers, so a search has no way to tell whether a node is currently in the tree. Thus, it is not clear how one could use HPs in DVY14.

\subsubsection{Details of  \nbr(+)  applicability}
Reasoning about the applicability of  \nbr(+) is comparatively easy, as one just needs to confirm that every \rdp restarts from root, and that it is possible to reserve all nodes before entering a \wtp (\reqref{reqnbr} and~\ref{req:reqnbr2}, respectively).
In other words, each thread should restart from root after \textit{helping} in search phase and no new pointers to shared nodes should be accessed in a \wtp.


\nbr, at first may appear to be easily applicable to EFRB10~\cite{ellen2010non} as it has lockless searches which either end in an update or helping and operations restart searches from root after helping. However, one helping step could lead to recursively more helping updates which could discover new pointers that could not be known before the first helping update. As a result such pointers can not be reserved before the first helping update (write phase), (violating \reqref{reqnbr2}).
Similarly,  \nbr(+) could not be applied to EFRB14~\cite{ellen2014amortized} as, after helping, an operation restarts from nearby ancestors and not from the root (violating \reqref{reqnbr}). Restarting from a nearby ancestor is crucial to achieve the amortized complexity result in that work.

In HJ12\cite{howley2012non}, the first lock-free internal BST, searches are sometimes required to help an ongoing update to avoid missing a node that is concurrently moved upwards in the tree, a situation that can arise when a search is concurrent with a two-child delete operation.
However, in HJ12, after helping, searches restart from the root (satisfying \nbr{'s} requirement that a \rdp should start from root), and all nodes required to do the helping update can be known beforehand through the descriptor (satisfying \nbr{'s} requirement that all nodes to be accessed in a \wtp should be reserved). Thus,  \nbr(+) can be used in HJ12.

Similarly, data structures designed using the tree update template of Brown et~al.~\cite{brown2017techniques}, for example lock-free chromatic trees, relaxed (a,b)-trees, relaxed AVL trees, and weak AVL trees HL17~\cite{he2017deletion} are compatible with \nbr.
In the template, operations do typically synchronization-free searches to find target node(s) (similar to our standard \rdp), then check whether they need to help (similar to entering \wtp during a search), and in the event that an operation $O$ helps another, $O$ typically restarts from the root.
Helping is implemented using descriptors which contains pointers to all nodes that would be required to execute the helping update, so  \nbr(+) could reserve all node pointers in the descriptor and enter \wtp for the helping update and then restart its search (\rdp) from the root.

The lock-free Patricia Trie (S13) \cite{shafiei2013non} too follows a pattern in which searches are synchronization-free, and updates may perform auxiliary helping using descriptors and then restart from the root. Thus,  \nbr(+) applies to S13 as we can know all nodes to be accessed in \wtp beforehand, and after helping, the operation can be restarted from the root, satisfying both \reqref{reqnbr} and ~\ref{req:reqnbr2}.  

DGT15\cite{david2015asynchronized} has sync-free searches followed by locking predetermined nodes for updates, which is similar to the \rdp followed by single \wtp pattern of LL05\cite{heller2005lazy}. Thus,  \nbr(+) applies to DGT15.

The unbalanced external BST of Natarajan and Mittal (NM14) \cite{natarajan2014fast} too has a pattern where each operation starts with synchronization-free search that returns a $SeekRecord$ object followed by possible modifications to the data structure.
During updates, the operation may possibly help by accessing nodes only pointed by a $SearchRecord$ object and then subsequently restart from root. 
Deletion consists of two modes: injection and cleanup. Both involve data structure modification, therefore, both should be done in \nbr{'s} \wtp. Additionally, it is possible that injection mode may succeed and subsequent cleanup may fail, in that case the operation is required to start from the root, satisfying the requirements for \nbr. 

The unbalanced external BST in DVY14\cite{drachsler2014practical} does a synchronization-free search and then executes updates using locks. But within the update phase it may obtain pointers to new nodes, for example, the successor, children or parent of a node, which may violate the \reqref{reqnbr2} of \nbr.
So, in order to apply  \nbr(+) to DVY14, one must modify DVY14 to perform all of the aforementioned reads of new node pointers in the update phase before the first lock is acquired, then validate after lock acquisitions that the values of those reads have not changed (and restart if validation fails). Note,  \nbr(+) would not work with the balanced variant of the DVY14 tree, which does bottom-up rebalancing without restarting from the root between rotations.

BCCO~\cite{bronson2010practical} Uses optimistic concurrency control (OCC) techniques to avoid locking nodes in searches as much as possible but occasionally locks and immediately unlocks a node as part of traversal, which suggests  \nbr(+)  cannot be used.
However, as described in~\cite{arbel2018getting}, this locking is an optional part of the BCCO algorithm (intended to improve fairness under heavy contention) and can simply be removed.
Unfortunately, this algorithm performs recursive bottom-up rebalancing without restarting from the root between rotations.
To use \nbr, one would need to restart from the root after each rotation, which would be a substantial algorithmic change.

In RM15~\cite{ramachandran2015fast}, an insert operation, does a sync-free search followed by a CAS on an appropriate child pointer. If the CAS fails, the pointer is re-read to find a descriptor, which helps complete the conflicting operation and the operation restarts from the root.
For safe application of \nbr, we must reserve both children of node, and if either child is actually a descriptor, we must reserve all pointers in the descriptor. At this point, if we see a descriptor, we might as well help it before attempting the CAS (if the CAS would be doomed to fail anyway). Unfortunately, this changes the progress argument (although we don't think it actually changes the progress property).
Delete in RM15 are the real problem for applying \nbr. after the inject part, the cleanup part requires obtaining and dereferencing new pointers without restarting from the root. modifying this algorithm to work with  \nbr(+) would require sweeping changes.


\subsection{Problem with application of HE and IBR}
In this section, we discuss why HE and IBR do not apply to data structures where a sequence of unlinked nodes could be traversed. We will use a lazylist~\cite{heller2005lazy} as an example.

\subsubsection{Problem with HE}
HE maintains a timestamp representing a global epoch (\textit{ge}) which is incremented occasionally. Each node is assigned the current \textit{ge} when it is first allocated and when it is subsequently retired:
These timestamps are stored in the node's \textit{birth\_epoch} and \textit{retire\_epoch} fields, respectively.
Each access to a new node follows three steps: 1) reserves the current \textit{ge} at a \textit{single writer multi reader} (SWMR) memory location, 2) creates a local reference to the node, and 3) validates that the reserved \textit{ge} has not changed since step 1. If the validation fails, the steps are continuously retried unless the validation succeeds, i.e. until the \textit{ge} reserved before the read in step 2 is the same as the \textit{ge} after the read. This ensures that a valid (reachable) node was reserved.
To free a node, a thread compares the node's \textit{birth\_epoch} and \textit{retire\_epoch} with the reserved time stamps of all threads. 
If any of the reserved epochs overlap with the interval between the birth\_epoch (be) and the retire\_epoch (re), then the node is not safe to free. Otherwise, the node is freed.
\\
\\
Now, we will trace a problematic execution in a lazylist. Consider a lazylist, $n1 \longrightarrow n2 \longrightarrow n3 \longrightarrow n4 \longrightarrow  n5$. 
Let some threads delete nodes n2, n3 and n4 in that order, such that the final list is $n1\longrightarrow n5$, and the unlinked nodes form a sequence $n2\longrightarrow n3 \longrightarrow n4$ and finally, the unlinked node n4 points to the linked node n5. Note, in such data structures it is possible to reach a correct node by following a sequence of unlinked nodes. Also, for lazylist only two SWMR slots per thread are needed to reserve nodes currently being accessed, similar to hazard pointers.

Assume, a thread T2 retires n2, n3 and n4, and each node has $be=re=0$. We will annotate the nodes with timestamps, as appropriate. For example, n2, n3 and n4 have \textit{be} and \textit{re} 0, so we denote them n2(be=0, re=0), n3(be=0, re=0), n4(be=0, re=0).
\begin{enumerate}
    \item Say, T1 has a private reference to n2 (\textit{ge}=0), where n2 (\textit{ge}=0) means the reference was acquired when \textit{ge} was 0 by reserving it. Now, T1 sleeps.
    \item Another thread T2 retires n2($be=re=0$), n3($be=re=0$) and n4($be=re=0$)
    \item Now, \textit{ge} changes to 1 and then to 2.
    \item T1 wakes up and creates a local reference $l_2 \gets n2.next$. Since the latest epoch is 2, therefore, it reserves epoch 2 at reservation slot 1.
    \item T1 again accesses $l_3 \gets n3.next$, and reserves epoch 2 at reservation slot 2.
    \item Now, let us say T2 reclaims n4. Since for n4, $be=re=0$, and the reserved epoch is 2, there is no reservation conflict during the scan. Thus, n4 is freed.
    \item T1 dereferences it local reference $l_3$ to node n4. Since, n4 has been freed already, T1 will segfault.
\end{enumerate}
\subsubsection{Problem with IBR}
IBR, similar to HE, uses per node \textit{be} and \textit{re} timestamps along with two reservations per thread to determine whether a node can safely be reclaimed.
The two reservations correspond to the \textit{ge} observed when starting an operation and the \textit{ge} observed at the most recent access to a node during an operation.
Thus the reservation is in form of an interval of time stamps. If any node's \textit{be} and \textit{re} overlaps with the reserved interval, then the node is said to be unsafe for freeing. The two ends of the interval are stored in two per thread SWMR locations, \textit{LR} (for lower end of the reservation interval) and \textit{UR} (for upper end of the reservation interval), respectively.
\\
\\
To understand why IBR does not apply to a lazylist, let's use the lazylist example again. Note, n2, n3 and n4 formed a sequence of unlinked nodes in the lazylist example. 
Let a thread T2 unlinks  n2 ($be=0, re=0$), n3($be=2, re=2$), and n4($be=4, re=4$), in that order. Let's assume the current \textit{ge=4}. But, before  n2 was unlinked a thread T1 already had a private reference to n2 with \textit{LR=0} and \textit{UR=0} (reserved interval of (0, 0)) and slept. Now let's say T2 frees n3 since its $be=re=2$ doesn't overlap with T1's reserved interval of (0,0). Next, if T1 wakes up and tries to access $n3 \gets n2.next$ it reserves the interval (0,4) and returns a reference to a now freed node n3. This leads to segfault if n2 is dereferenced. 


\begin{figure}
\begin{minipage}{.49\textwidth}
\begin{lstlisting}[basicstyle=\ttfamily,frame=tlrb,morekeywords={*, recl_start_op,recl_end_op, recl_retire,...},numbers=none, frame=L]{Name}
void OP_DEBRA()
{
|\colorbox{light-gray}{recl\_start\_op}|
RETRY:
    pred=head; 
    curr=pred.next;
    while (key |$\leq$| curr.key) {
        pred=curr;
        curr=cur.next;
    }
    
    if (key == curr.key) {
        return false;
    }
    
    lock(pred);
    lock(curr);
    if (!validate()) {
        unlock(pred); 
        unlock(curr);
        goto RETRY;
    }
    
    do update
    unlock(pred); 
    unlock(curr);
|\colorbox{light-gray}{recl\_end\_op \hspace{2mm}}|
}
\end{lstlisting}
\subcaption{DEBRA}
\label{fig:usbdebra}
\end{minipage}\hfill
\begin{minipage}{.49\textwidth}
\begin{lstlisting}[basicstyle=\ttfamily,frame=tlrb,morekeywords={*, start_op,end_op, recl_retire, upgrade_to_write_phase, save_for_write_phase,...},numbers=none, frame=L]{Name}
void OP_NBR+()
{
RETRY:

|\colorbox{light-gray}{CHECKPOINT()};|
|\colorbox{light-gray}{begin\rdp()};| //clears any reservations
    pred = head; 
    curr = pred.next;
    while(key |$\leq$| curr.key) {
        pred=curr;
        curr=cur.next;
    }

|\colorbox{light-gray}{end\rdp(pred, curr)};| //reserve pred and curr
    if (key == curr.key) {
        return false;
    }
    
    lock(pred); lock(curr);
    if (!validate()) {
        unlock(pred);
        unlock(curr);
        goto RETRY;
    }
    
    do update
    unlock(pred); 
    unlock(curr);
}
\end{lstlisting}
\subcaption{\nbr(+)}
\label{fig:usbnbr}
\end{minipage}
\caption{Complexity of using \debra and  \nbr(+) on a lazy list.  \nbr(+) is moderate in complexity.}
\label{fig:usb_debranbr}
\end{figure}

\begin{figure}
\begin{minipage}{.49\textwidth}
\begin{lstlisting}[basicstyle=\ttfamily,frame=tlrb,morekeywords={*, start_op,end_op, recl_retire, upgrade_to_write_phase, save_for_write_phase,...},numbers=none, frame=L]{Name}
void OP_NBR+()
{
RETRY:
|\colorbox{light-gray}{CHECKPOINT()};|
|\colorbox{light-gray}{begin\rdp()};| //clears any reservations
    pred = head; 
    curr = pred.next;
    while(key |$\leq$| curr.key) {
        pred=curr;
        curr=cur.next;
    }

|\colorbox{light-gray}{end\rdp(pred, curr)};| //reserve pred and curr
    if (key == curr.key) {
        return false;
    }
    
    lock(pred);
    lock(curr);
    if (!validate()) {
        unlock(pred); 
        unlock(curr);
        goto RETRY;
    }
    
    do update
    unlock(pred);
    unlock(curr);
}
\end{lstlisting}
\subcaption{\nbr(+)}
\label{fig:usbnbrp}
\end{minipage}\hfill
\begin{minipage}{.49\textwidth}
\begin{lstlisting}[basicstyle=\ttfamily,frame=tlrb,morekeywords={*, recl_start_op,recl_end_op, recl_retire, upgrade_to_write_phase, save_for_write_phase,...},numbers=none, frame=L]{Name}
void OP_HP()
{
RETRY:
    pred = head;
    curr = pred.next;
    |\colorbox{light-gray}{protect(curr) RETRY on fail}|;
    while (key |$\leq$| curr.key) {
        |\colorbox{light-gray}{unprotect(pred)}|;
        pred=curr; 
        curr=cur.next;
        |\colorbox{light-gray}{protect(curr) RETRY on fail}|;
    }
    if (key == curr.key) {
        |\colorbox{light-gray}{unprotect(pred)};| |\colorbox{light-gray}{unprotect(curr)}|;
        return false;
    }
    lock(pred); 
    lock(curr);
    if (!validate()) {
        |\colorbox{light-gray}{unprotect(pred)};| |\colorbox{light-gray}{unprotect(curr)}|;
        unlock(pred);
        unlock(curr);
        goto RETRY;
    }
    do update
    |\colorbox{light-gray}{unprotect(pred)};|
    |\colorbox{light-gray}{unprotect(curr)}|;
    unlock(pred);
    unlock(curr);
}
\end{lstlisting}
\subcaption{HP}
\label{fig:usbhp}
\end{minipage}
\caption{Complexity of using  \nbr(+) and \hp on a lazy list. \debra $<<$  \nbr(+) $<<$ \hp. HP is more complex.}
\label{fig:usb_nbrhp}
\end{figure}

\newpage
\section{Usability}
\label{sec:usb}

\figref{usb_debranbr} and \figref{usb_nbrhp} compares the difficulty of using \hp,  \nbr(+) and \debra in the insert operation of the lazy list of Heller et~al.~\cite{heller2005lazy}.
As \figref{usbhp} demonstrates, \hp is cumbersome to use because it requires a programmer to protect every record by \textit{announcing} hazard pointers, using a store/load fence or $xchg$ instruction to ensure that each announcement is visible in a timely manner by other threads, validating that the announced record is still safe before dereferencing it, and restarting if validation fails.
Programmers also need to unprotect records that they will no longer dereference, further increasing the need for intrusive code changes. 
On the contrary, applying  \nbr(+) to a data structure operation is, intuitively, similar to performing two-phased locking, in the sense that the primary difficulty revolves around identifying where the \wtp should begin, and which records it will access.

Recall the requirements discussed in \secref{nbrreq}: the programmer needs to invoke \func{CHECKPOINT}{} immediately followed by \func{begin\rdp} before the operation accesses its first shared record, in this example, at the start of the traversal for target records. 
Then s/he must invoke \func{end\rdp}, which marks the start of the write phase and reserves the nodes that will be accessed within the write phase, before modifying any shared records. The programmer passes the nodes required to be reserved as parameters to \func{end\rdp}. 
In this example, the write-phase begins just before the lock acquisition on pred.
If there are no modifications to be performed in an operation, for example, as in the \func{contains} operation of the lazy-list, then the programmer can simply invoke \func{end\rdp} before returning from the operation.

\debra is simplest as it requires programmers to invoke just two functions corresponding to the start and the end of a data structure operation 
(\figref{usbdebra}). 

In terms of programmer effort,  \nbr(+) finds a middleground between DEBRA and HPs. 
Although  \nbr(+) is slightly more involved than DEBRA, we believe that the benefits due to \nbr's bounded garbage property and better performance outweigh the extra effort of identifying which shared records will be modified by the \wtp and where in the code to invoke \func{end\rdp}. 

Just to give readers a quantitative view of the amount of programming effort needed to use HP and \nbr(+) we measured number of extra reclamation related lines of code needed to be written in our implementation of \func{insert(), delete() and contains()} methods for the lazylist and DGT. We observed that  \nbr(+) required only 10 extra lines of code in comparison to 30 extra lines of code needed to use HP.

As mentioned earlier, in \nbr, before a thread enters a \wtp, it must reserve all the records that will be accessed in the \wtp.
In some data structures it might not be possible to determine \textit{precisely} which records \textit{will} be accessed in the \wtp.
For example, in a tree, an operation may decide \textit{during} the write phase whether to modify the left or right child pointer. 
To apply  \nbr(+) in such a tree, one can simply reserve \textit{both} pointers. 
(Nevertheless there may be some data structures where it is infeasible to reserve \textit{all} of the records that might be accessed in a \wtp.)


\ignore{
    Not the complexity of interfaces but also comparison of per read overhead, per op overhead would be nice.
    Debra + recovery code to provide speed and robustness, nbr that too with out requiring user to provide complex recovery code. 
    
    IBR requires per node meta data in form of birth and retire epochs nbr doesnt.
    HP per read fence and validation making ds complex still slow.
    points from literature to call out performance drawbacks in state of art to and some lack robustness while some lack speed... etc..

    [Herlihy:]
    Herlihy first proposed transformation of sequential data structures into lf/wf form which are inspired by the optimisitic concurrency control mechanism of having a sync-free phase where operation read desired records and then have a validation based update phase. Such transformation is also referred as normalised form by petrank etal.
}

 
\section{Rules to correctly use sigsetjmp and siglongjmp}
\label{sec:nbrrules}

The subroutines \ssj~and \slj~help to perform non-local gotos to alter the flow of control within C/C++ programs. 
The two subroutines are used in pairs. Upon invocation, \ssj~saves the state of registers that includes stack pointer, frame pointer, and program counter into a buffer of type \var{jump\_buf}, a subsequent call to \slj, at some later time, resets the registers to the saved values. In other words, when a \slj~ is called, the program behaves as if it returned from its most recent call to \ssj, which saved the content of registers. 
Any modifications to existing variables between the \ssj~and \slj~may or may not be preserved.
Similarly, the execution of any system call could leave the program in an undefined state. More specifically, global variables will retain their values, and stack variables or variables in shared memory may or may not. We try to codify some rules below that should be sufficient (but not always necessary) to ensure that changes made between \ssj~and  a subsequent \slj~call are not accidentally left in an undefined state~\cite{page2014signals, unixprgsigs, ieeestdsig}.


\begin{enumerate}
    \item Do not modify any global variable.
    \item Do not modify any thread-local variables.
    \item Do not modify stack-allocated variables that have been defined prior to the call to \ssj.
    \item Do not perform any heap allocation or de-allocation.
    \item Do not perform I/O including writing to files, handling network packets, or syscalls.
    \item Do not return from a procedure that calls \ssj~ before the call to \slj.  
\end{enumerate}

To understand the pitfalls of not adhering to the above rules, we discuss the following scenarios.

\subsection{Scenario 1: Memory Leaks}
\begin{lstlisting}[xleftmargin=4em, xrightmargin=2em]
void func_scenario_a()
{
    RETRY: 
    begin_read(); 
    foo_t *tmp_foo = malloc(sizeof(foo_t)); // allocate temporary object
    do_some_reads(tmp_foo);
    free(tmp_foo); // de-allocate temporary object
    begin_write();
    do_some_writes();
} 
\end{lstlisting}
Suppose that the thread is neutralized during a call to \func{do\_some\_reads()} in line 6, then, the allocation on line 5 will be repeated, thus leaking the allocation of the foo\_t object. Rule 4) should prevent these kinds of scenarios.

One can get around this problem by preallocating the needed objects before the \func{sigsetjmp}, and freeing any unused ones at the end.
It is also possible to maintain a thread-local array of preallocated objects, which can be used between \func{sigsetjmp} and \func{siglongjmp}, and freed at the end.

\subsection{Scenario 2: Double freeing}
\begin{lstlisting}[xleftmargin=4em, xrightmargin=2em]
void func_scenario_b()
{ 
    foo_t *tmp_foo = malloc(sizeof(foo_t)); // allocate temporary object
    RETRY:
    begin_read(); 
    do_some_reads(tmp_foo);
    free(tmp_foo); // de-allocate temporary object
    do_some_more_reads(); 
    begin_write(); 
    do_some_writes(); 
}
\end{lstlisting}

Suppose that the thread is neutralized during \func{do\_some\_more\_reads()} in line 8, then, the tmp\_foo object will have been de-allocated. When the loop retries the read phase, the \func{free(tmp\_foo)} in line 7 will be re-executed, causing a double-free and possibly a crash of the application or data corruption. Rule 4) should prevent this kind of scenario.

Scenario 1 \& 2 together also mean that calling syscalls, file I/O and network I/O will also cause errors.

\subsection{Scenario 3: Update to thread local variables with higher scope}

\begin{lstlisting}[xleftmargin=4em, xrightmargin=2em]
thread_local int tls_a = 0;
thread_local int tls_b = 0;
void func_scenario_c()
{ 
    tls_a++; 
    RETRY: 
    begin_read(); 
    tls_b++; 
    do_some_reads(); 
    begin_write(); 
    do_some_writes(); 
    assert(tls_a == tls_b);
}
\end{lstlisting}
Suppose that the thread is neutralized during \func{do\_some\_reads()} in line 9, then, the increment is re-executed for \var{tls\_b} but not for \var{tls\_a}, causing application invariant to assert in line 12. Rule 1 \& 2) should prevent these kinds of scenarios.
However, note that such modifications can still be done, if one is just using a thread local variables to compute statistics, and one doesn't mind them being incremented repeatedly due to repeated \func{setlongjmp} calls. In such cases, approximate values are usually acceptable.

\subsection{Scenario 4: Updates to stack variables with higher scope}

\begin{lstlisting}[xleftmargin=4em, xrightmargin=2em]
void func_scenario_d()
{ 
    int stack_a = 0; int stack_b = 0; 
    stack_a++; 
    RETRY: 
    begin_read(); 
    stack_b++; 
    do_some_reads(); 
    begin_write(); 
    do_some_writes(); 
    assert(stack_a == stack_b); 
}
\end{lstlisting}

Suppose that the thread is neutralized during \func{do\_some\_reads()} in line 9, then, the increment is re-executed for \var{stack\_b} but not for \var{stack\_a}, causing application invariant to assert in line 12. Rule 3) should prevent these kinds of scenarios.
\section{Summary}
\label{sec:conclusion}

The chapter classifies data structures into three categories based on their compatibility with NBR: compatible, semi-compatible, and incompatible.

\begin{itemize}
    \item Compatible Data Structures: These structures~\cite{heller2005lazy, david2015asynchronized, harris2001pragmatic, brown2017techniques,howley2012non,shafiei2013non,he2017deletion, srivastava2022elimination, natarajan2014fast} can directly leverage NBR. The chapter details two types of compatible data structures and how NBR applies to them. Examples include optimistic data structures like lazy lists and lock-free lists.
    \item Semi-compatible Data Structures: These structures~\cite{drachsler2014practical, brown2020non, ramachandran2015castle, michael2004hazard} require modifications to work with NBR, potentially affecting performance or progress guarantees. The chapter discusses the trade-offs involved in adapting such structures.
    \item Incompatible Data Structures: Certain data structures~\cite{bronson2010practical, ellen2010non, ellen2014amortized, ramachandran2015fast, brown2020non}, due to their inherent operations, are not suitable for NBR. Examples include structures that involve rotations or rebuild subtrees without restarting from the root node.
\end{itemize}

In addition, we survey a list of eighteen popular safe memory reclamation techniques and discuss their applicability to NBR(+). We also compare the applicability of NBR(+) and several other state of the art safe memory reclamation algorithms to these data structures.  Furthermore, we discussed the relative complexity of using NBR(+), HP, and DEBRA with an example data structure code, concluding that NBR(+) is easier to use than HP and moderately complex to use than DEBRA.

With the end of this chapter, we end our discussion on the neutralization paradigm based safe memory reclamation algorithms: \nbr and \nbrp we presented in \chapref{chapnbr} and \chapref{chapnbrp}, respectively. We started with an objective to design a safe memory reclamation algorithms which satisfies five desirable properties together. This led us to propose the neutralization paradigm for safe memory reclamation, which yielded two algorithms, NBR and NBR+. 
In \chapref{chapnbr} and \chapref{chapnbrp} we showed that these algorithms are fast, consistent, and bound garbage. In this chapter, we demonstrated these algorithms are widely applicable and easy to use, thus concluding that our neutralization paradigm based safe memory reclamation algorithms satisfy the five desirable properties simultaneously.

\chapter{Reactive Synchronization Paradigm}
\label{chap:chaprsp}



Safe memory reclamation techniques that utilize per-read reservations, such as hazard pointers and hazard eras, often face significant overhead when traversing concurrent data structures. This is primarily due to the need to announce a reservation and enforce appropriate ordering before each read by making it globally visible using fences. In real-world read-intensive workloads, this overhead is amplified because, even if relatively little memory reclamation actually occurs, the full overhead of reserving records before use is still incurred while traversing data structures.

In this chapter, we propose a reactive synchronization paradigm by combining POSIX signals and delayed reclamation, introducing a publish-on-ping algorithm. This method eliminates the need to make reservations globally visible before use. Instead, threads privately track which records they are accessing and share this information on demand with threads that intend to reclaim memory. 
We present a family of publish-on-ping algorithms. Specifically, 
we apply the publish-on-ping algorithm to hazard pointers and hazard eras, resulting in HazardPointersPOP and HazardErasPOP algorithms, respectively.
Furthermore, the capability to retain reservations during traversals in data structure operations and publish them on demand facilitates the construction of a variant of hazard pointers (EpochPOP). This variant uses epochs to approach the performance of epoch-based reclamation in the common case where threads are not frequently delayed while retaining the robustness of hazard pointers.

The outline of this chapter is as follows. 
In \secref{popintroduction}, we describe the asymmetric overhead in state-of-the-art reclamation algorithms, motivating the need for the reactive synchronization paradigm discussed in \secref{secrsp}. We present an overview of the implementation, termed publish-on-ping (POP), in \secref{secpop}. In \secref{oralg}, we review hazard pointers, hazard eras, and epoch-based reclamation algorithms to present their respective POP-based variants in \secref{popalgos}. Then, in \secref{popcorr}, we discuss the correctness and robustness of the three POP algorithms. In \secref{popeval}, we evaluate the performance of the POP-based implementation of hazard pointers, hazard eras, and epoch-based reclamation. Finally, in \secref{pop-relatedwork}, we describe additional related work and conclude in \secref{popconclusion}. 

\section{Introduction}
\label{sec:popintroduction}


Several SMR algorithms are \textit{pointer-based} (also known as \textit{reservation-based}).
The most notable example, the celebrated hazard pointer (HP) algorithm~\cite{michael2004hazard}, has received a Dijkstra award, and is set to be included in the C++26 standard~\cite{michael2017hazard}.
Unfortunately, in HPs, 
whenever a thread encounters a new shared node, such as a node in a list, the thread must \textit{reserve} a pointer to the node by (1) storing it in a single-writer multi-reader (SWMR) slot in a shared array, (2) executing memory fence instruction(s) to \textit{publish} this reservation, making it visible to other threads (and preventing instruction reordering), and (3) re-reading a pointer from which this node was encountered to verify that the reserved node is still reachable at some time after the reservation was published.

The need to fence every time a new node is encountered can incur high overhead in linked data structures.
This overhead is even more pronounced in common read-intensive workloads wherein, even though memory reclamation may be infrequent, and relatively lightweight, the overhead it imposes on read-only operations is not. 

A more recent technique, called hazard eras (HEs)~\cite{ramalhete2017brief}, uses monotonically increasing global timestamps and reservation of timestamps, instead of pointers. 
Roughly speaking, the overhead in HEs is more coarse grained, as a timestamp represents many nodes. 
The use of timestamps reduces how often memory fences are needed.
For example, if a thread is about to reserve a timestamp that it has already reserved (because a new node being accessed corresponds to the same timestamp as a previously accessed node), the previously published timestamp can be reused (without additional fencing).
However, the overhead of HEs can still be substantial when the global timestamp changes frequently.

Another alternative is to use a fast epoch based reclamation (EBR) algorithm~\cite{hart2007performance,fraser2004practical,mckenney1998read}.
However, these techniques are not robust: A delayed thread can prevent all threads from reclaiming memory resulting in unbounded memory consumption. 
Subsequent reclamation techniques like NBR~\cite{singh2021nbr, singhTPDS2023NBRP}, VBR~\cite{sheffi2021vbr} and IBR~\cite{wen2018interval} provide a variety of robustness guarantees, but they have other tradeoffs in, for example, ease of integration, applicability to data structures, assumptions regarding whether memory pages are returned to the operating system, and so on.
NBR requires data structure operations to have a particular structure, consisting of read- and write-phases, with specific requirements in each phase. 
Optimistic techniques like VBR 
require a type preserving allocator, and cannot allow memory pages to be returned to the operating system (without tricks involving trapping and ignoring segmentation faults).
These techniques cannot be used as 
drop-in replacements for hazard pointers without making sacrifices in ease of use, portability, or applicability to certain data structures. 
Consequently, it is more desirable if we can eliminate this asymmetric overhead without sacrificing the desirable traits of the original safe memory reclamation algorithms.  

\subsubsection{Goal}
\textit{``To publish the reservations (references in hazard pointers or epochs in hazard eras) on demand to reclaimers to boost the performance of the state-of-the-art concurrent reclamation techniques"}

\section{Reactive Synchronization}
\label{sec:secrsp}

The key to safety from use-after-free errors in most pointer based techniques is that threads traversing a data structure eagerly reserve and publish (making the reservations visible to all threads) before accessing nodes during traversals, and threads reclaiming nodes must first scan all the reservation to avoid freeing reserved nodes.

It is highly pessimistic to eagerly publish reservations (and fence) 
just because there \textit{might be} some concurrent reclaimer thread about to free some nodes.
This is especially true considering that the majority of the time, the reserved nodes are not the ones being freed.
Even if there is very little (or no) reclamation in a given workload, traversals must pay this cost for every node visited.
For better scalability of the data structure, it is preferable to synchronize reclamation information (e.g., fence to publish reservations) exactly when a reclamation event is about to occur. 


We therefore posit the following: Is it possible for each traversal to \textit{publish} its reservation set exactly at the moment when a reclaimer is about to scan the reservations of all threads? Additionally, can this be achieved without compromising the simplicity and usability of the original algorithms? This strategy aims to remove the uneven overhead that safe memory reclamation algorithms impose on concurrent data structures. We term this method of desigining a safe memory reclamation as the \textit{reactive synchronization paradigm}. In the following section, we discuss the technique to implement such reactive synchronization.

\section{Publish on Ping}
\label{sec:secpop}

On POSIX compliant operating systems, threads can send signals to all threads in the system, and in response to
a signal, other threads execute a signal handler. This implies that a thread can provoke a reaction from other threads whenever needed. We leverage the communication pattern enabled by POSIX signals to allow threads to maintain reservations or epochs without publishing
them and publish the reservations or epochs on demand when signaled. 

In more detail, a thread that wants to reclaim memory sends signals to all other threads, and a thread receiving the signal publishes its reservations (which it had simply been tracking locally until this point), and issues a \textit{single} memory fence.
Once all threads have published their reservations, the reclaimer can scan them and free any nodes that are not reserved.
We name this technique as \textbf{Publish-On-Ping (POP)}.


The neutralization paradigm algorithms NBR(+) discussed in \chapref{chapnbr} and \ref{chap:chapnbrp} also utilized POSIX signals. However, the recipient's response to signals differs between POP and NBR(+). In NBR(+), a thread aiming to reclaim memory sends signals to all other threads, causing them to discard all pointers they hold and \textit{restart} their current traversals (unless they have already made changes to the data structure). In POP, a thread that wants to reclaim memory still sends signals to all other threads, but a thread that receives a signal does not need to restart its operation or alter its control flow. Consequently, POP overcomes a significant limitation in NBR, which required data structure operations to be carried out as a sequence of read and write phases, each read phase beginning its search from an entry point in the data structure. This constraint made it challenging to use the neutralization paradigm with certain helping-based data structures ~\cite{ellen2010non, ellen2014amortized, michael2004hazard}, as described in \secref{appl} of \chapref{chapappuse}.

In the following section, we illustrate how POP can be applied to a number of advanced safe memory reclamation algorithms. This yields a family of efficient and robust algorithms that significantly enhance their original versions, while maintaining their fundamental simplicity. Applying POP to hazard pointers results in the development of an algorithm called \textit{HazardPointerPOP}, and its application to hazard eras leads to the creation of an algorithm called \textit{HazardEraPOP}.
Perhaps surprisingly, EBR can also be incorporated into \textit{HazardPointerPOP}, as a sort of \textit{fast path}, resulting in an algorithm called \textit{EpochPOP} that achieves a performance similar to that of EBR, while retaining the same robustness guarantees.
The original state of the art algorithms are explained first, after which their improved versions based on POP are presented.


\ignore
{
\subsubsection{POP specific Terminology}
\textit{Minimize the terms to highlight publishing term....}
Typically, in \textit{deferred memory reclamation}, every thread has a local list, called \textit{retire list}, to which they add the \textit{retired} nodes until they are safe to be freed.
In the context of reclamation algorithms, a thread with a reference to a shared node in the data structure is termed a \textit{reader}, while a thread deleting an node from the data structure is called a \textit{reclaimer}. Readers can save a pointer or a timestamp, representing a collection of nodes, in their local or at shared memory locations. These are referred to as  \textit{reservations}.
These \textit{reservations} are made globally visible for reclaimers, referred to as \textit{publishing}.
Threads are considered to be in a \textit{quiescent state} between consecutive data structure operations. 
}


\section{Original Reclamation Algorithms}
\label{sec:oralg}
Deferred reclamation techniques exhibit a key characteristic: instead of immediately reclaiming a deleted data structure node, threads append them to their retire lists. This practice ensures that the retired nodes are safe to be freed and amortizes the reclamation overhead across multiple operations. When a retire list reaches a given threshold size, threads attempt to free the nodes, utilizing the underlying synchronization mechanism of the reclamation technique.

The synchronization mechanisms vary between techniques and ensure that the reclaimers only free \textit{safe} nodes and \textit{readers} do not access \textit{unsafe} nodes. For example, in hazard pointers, synchronization requires \textit{ readers} to reserve and publish nodes before accessing them, and \textit{reclaimers} to first scan reservations to identify all the \textit{unsafe} nodes, and subsequently freeing only the \textit{safe} nodes from their retire list.
Similarly, in epoch-based techniques, \textit{reclaimers} must wait until all threads have gone quiescent at least once since the node was retired before freeing it.

In this section, we revisit hazard pointers (HP)~\cite{michael2004hazard}, Hazard Eras ~\cite{ramalhete2017brief}, and a read-copy-update style (RCU)~\cite{hart2007performance} implementation of epoch based reclamation, hereafter, referred to as epoch based reclamation (EBR) unless otherwise specified.
HP is actively utilized in MongoDB and in Meta's folly open source C++ library, while the RCU is part of the Linux OS. 
Both techniques have recently been proposed for addition to the C++ standard library~\cite{michaelKenny2017proposed,michael2017hazard}.

Later we explain our publish-on-ping implementations of these algorithms.

\subsection{Hazard Pointers}
\begin{algorithm}
\small
    \caption{Hazard Pointer algorithm as in the benchmark by Wen, Izraelevitz, Cai, Beadle and Scott~\cite{wen2018interval}.}
    \label{algo:orighp}
    \begin{algorithmic}[1]
        \State \texttt{const reclaimFreq} \Comment{frequency of reclaiming retire list} \label{lin:var-orighpreclaimfreq}
        \State \texttt{thread\_local int tid} \Comment{current thread id}
        \State \texttt{list<T*> retireList [NTHREAD]} \label{lin:var-orighpretirelist}
        \State \texttt{atomic<T*> sharedReservations [NTHREAD][MAX\_HP]} \label{lin:var-orighpsharedreservations}  

        \Procedure{\texttt{T*} read}{\texttt{atomic<T*> \&ptrAddr, int slot}} \label{lin:orighp-proc-read}
            \Repeat
                \State \texttt{T* readPtr $\gets$ *ptrAddr}
                \State \texttt{sharedReservations[tid][slot] $\gets$ readPtr} \LComment{store load fence.}
            \Until{\texttt{readPtr $=$ *ptrAddr}} \label{lin:orighp-loopexit}
            \State \Return readPtr \label{lin:orighp-return}
        \EndProcedure
        \Statex
        \Procedure{retire}{\texttt{T* ptr}} \label{lin:proc-orighpretire}
            \State \texttt{myRetireList $\gets$ retireList[tid]}
            \State \texttt{myRetireList.append(ptr)}
            \If {\texttt{myRetireList.size() $\geq$ reclaimFreq}}
                \State {\Call{reclaimHPFreeable}{myRetireList}}
            \EndIf
        \EndProcedure    
    \Statex
        \Procedure{clear}{ } \label{lin:proc-orighpclearall}
            \For{$slot=0, \dots, MAX\_HP$}
                    \State \texttt{sharedReservations[tid][slot] $\gets$ NULL}\label{lin:orighpsetnull}
            \EndFor        
        \EndProcedure        
        \Procedure{reclaimHPFreeable}{myRetireList} \label{lin:proc-origreclaimhpfreeable}
                \LComment{collect all published reservations}
                \State \texttt{set<T*> collectedReservations $\gets$ \{\}} \label{lin:orighpcollectstart}
                    \ForAll {\texttt{<tid, slot> $\in$ sharedReservations[tid]}}
                        \State {\texttt{objPtr $\gets$ sharedReservations[tid][slot]}}
                        \State {\texttt{collectedReservations.insert(objPtr)}}
                    \EndFor \label{lin:orighpcollectend}
                \LComment{free all nodes not reserved}
                \ForAll {\texttt{objPtr $\in$ myRetireList}} \label{lin:orighpreclaimstart}
                    \If{\texttt{objPtr $\notin$ collectedReservations}}
                        \State {\texttt{free(objPtr)}}
                    \EndIf
                \EndFor \label{lin:orighpreclaimend}
        \EndProcedure

\end{algorithmic}
\end{algorithm}
\subsubsection{The Overview}
The key principle underlying the operations of Hazard Pointers (HP) involves a contract between \textit{readers} and \textit{reclaimers}. \textit{Readers} \textit{reserve} and \textit{publish} pointers to nodes currently being accessed in single-writer multi-reader (SWMR) locations for \textit{reclaimers}. 
\textit{Reclaimers}, in turn, ensure that they scan all SWMR locations to collect all reserved nodes and only free those whose reservations have not been \textit{published}. This ensures safety from use-after-free errors.

Crucial to publishing reservations in HP is that every read of a shared node pointer should execute a memory fence. Specifically, \textit{readers} should follow these steps to timely publish their reservation to a pointer: 1) save the pointer to the node being read, 2) execute a memory fence, 3) validate that the pointer saved in step 1 is still \textit{reachable}, if not, retry or abort the operation.
This ensures that the pointer was \textit{reachable} at the time it was reserved (in step 1) \textendash a crucial condition for the correct application of HPs~\cite{michael2004hazard}.
Without the memory fence, the reservation of the pointer in step 1 could be reordered after the validation of reachability step (step 3). This could lead to reservation of a node that has already been deleted, resulting in unsafe accesses in the future.

\subsubsection{Programmer's view of HP}
For programmers, a standard hazard pointers reclamation interface includes the following three main functions: 

(1) \Call{read}{ }: Used for every read of a new data structure node. \textit{Readers} use \Call{read}{} to ensure that threads \textit{reserve} and \textit{publish} the node in a shared location.
(2) \Call{clear}{ }: Used to remove reservations when threads exit the operation. 
(3) \Call{retire}{ }: Used during update operations, \textit{reclaimers} employ \Call{retire}{} to add the \textit{deleted} nodes to their retire lists. During the \Call{retire}{}, \textit{reclaimers} collect every other threads's published reservations. Subsequently, in an iterative fashion, they free those nodes from their retire list that were not present in their collected reservations.

\algoref{orighp} shows the psuedocode for one implementation of hazard pointers similar to that found in the benchmark used in~\cite{wen2018interval}. Threads append retired nodes to their \texttt{retireList} and use slots in  \texttt{sharedReservations} to reserve nodes currently being used. \texttt{MAX\_HP} is the maximum number of nodes a thread can reserve at a time. Typically, in concurrent data structure only a limited set of nodes are used at a given time so it suffices to have a constant number of slots to reserve.

The \Call{read}{} function takes the address of an atomic variable (usually the next field of an already reserved node) from which a new pointer to the next node has to be read (\lineref{orighp-proc-read}). It within a loop performs the above mentioned steps, i.e. reads the pointer from the address, atomically writes it to a reservation slot in its \textit{sharedReservations} array, then validates that the address still has the pointer that has been reserved, ensuring it was not concurrently changed. These steps repeat until the validation succeeds, following which the pointer to the reserved node is returned to the data structure operation. 

Within its \Call{retire}{} function (\lineref{proc-orighpretire}), threads append retired nodes to their \textit{retireList} and when it exceeds its size beyond \textit{reclaimFreq} it frees the nodes that are not reserved. Specifically, invoke the \Call{RECLAIMHPFREEABLE}{} function that collects all the reserved nodes by scanning all thread's reservation slots in \textit{sharedReservations} array (\lineref{orighpcollectstart} - \lineref{orighpcollectend}) and then frees all nodes that are not reserved (\lineref{orighpreclaimstart} - \lineref{orighpreclaimend}).

After a thread has completed its data structure operation, it can clear its reservation using \Call{clear}{} (\lineref{proc-orighpclearall}).

\subsubsection{The Problem}
The per read memory fences incur high overhead, leading to poor scaling of data structures. Our Perf tool analysis, on a Hazard-Michael list of size 100 nodes where 128 threads execute 50\% inserts and 50\% delete operations showed that searches approximately spend $\approx$ 50\% of CPU cycles on reading HPs, while searches in a leaky implementation of the same list only spend $approx$ 15\% of CPU cycles in reading the HPs.  

\subsection{Hazard Eras}
\subsubsection{The Overview}

Unlike hazard pointers, Hazard Eras~\cite{ramalhete2017brief} maintain a global monotonically increasing epoch variable. Each thread reserves the current value of this epoch when accessing a node, rather than reserving the node itself. Additionally, each node maintains its birth epoch and retire epoch, representing its lifespan during which it is reachable in a data structure. The key idea is that, before freeing a node, a thread compares the node's lifespan to the reserved epochs. If no thread has reserved an epoch that intersects with the node's birth and retire epochs, then no thread could hold a hazardous reference to the node, making it safe to free.

\algoref{orighe} shows the pseudocode of the hazard eras whose interface is similar to the hazard pointers. However, the difference is that hazard eras augments each node with two epoch fields and reserves eras in which a node is accessed instead of pointers to the nodes.

\begin{algorithm}
    \small
    \caption{Hazard Era ~\cite{ramalhete2017brief, wen2018interval}.}
    \label{algo:orighe}
    \begin{algorithmic}[1]
    \State \texttt{const reclaimFreq} \Comment{frequency of reclaiming retire list} 
    \State \texttt{thread\_local int tid} \Comment{current thread id}
    \State \texttt{list<T*> retireList [NTHREAD]} 
        \State \texttt{atomic<int> sharedReservations [NTHREAD][MAX\_HE]} 
        \State \texttt{atomic<int> epoch} \Comment{incremented periodically}
        \State \texttt{int collectedReservations[NTHREAD][MAX\_HE]$\gets$\{\}} 
        \Statex
        \Procedure{\texttt{T*} read}{\texttt{atomic<T*> \&ptrAddr, int slot}} \label{lin:orige-procread}
            \State \texttt{int oldEra $\gets$ sharedReservations[tid][slot]} 
            \While{\textit{True}}
                \State \texttt{T* readPtr $\gets$ *ptrAddr}
                \State \texttt{int newEra $\gets$ epoch}                
                \If{oldEra == newEra}
                    \Return readPtr
                \EndIf
            \State \texttt{sharedReservations[tid][slot] $\gets$ newEra} \LComment{store load fence.}
            \State \texttt{oldEra $\gets$ newEra} 
            \EndWhile
        \EndProcedure
    \Statex
        \Procedure{retire}{\texttt{T* ptr}} \label{lin:proc-origheretire}
            \State \texttt{myRetireList $\gets$ retireList[tid]}
            \State \texttt{myRetireList.append(ptr)}
            \State \texttt{ptr.retireEpoch $\gets$ epoch}
            \If {\texttt{myRetireList.size() $\geq$ reclaimFreq}}
                \State \texttt{epoch.fetch\_and\_add(1)}
                \State {\Call{reclaimHEFreeable}{myRetireList}}\label{lin:rechefreeable}
            \EndIf
        \EndProcedure    
    \Statex
        \Procedure{clear}{ } \label{lin:orighe-proc-clearall}
            \For{$slot=0, \dots, MAX\_HE$}
                    \State \texttt{localReservations[tid][slot] $\gets$ NONE}
            \EndFor        
        \EndProcedure        

    \Statex
        \Procedure{reclaimHEFreeable}{myRetireList} 
                \LComment{collect all published reservations}
                    \ForAll {\texttt{<tid, slot> $\in$ sharedReservations[tid]}} \label{lin:orighecollectstart}
                        \State {\texttt{reservedEra $\gets$ sharedReservations[tid][slot]}}
                        \State {\texttt{collectedReservations[tid][slot] $\gets$ reservedEra}}
                    \EndFor \label{lin:orighecollectend}
                \LComment{free all nodes not reserved}
                \ForAll {\texttt{objPtr $\in$ myRetireList}} \label{lin:orighereclaimstart}
                    \If{\Call{canFree}{objPtr}} \label{lin:orighecanfree}
                        \State {\texttt{free(objPtr)}}
                    \EndIf
                \EndFor \label{lin:orighereclaimend}
        \EndProcedure

\algstore{algnewhe}
\end{algorithmic}
\end{algorithm}

\begin{algorithm}
    \small
\begin{algorithmic} [1]                   
\algrestore{algnewhe}
        \Procedure{canFree}{objPtr}
            \ForAll {\texttt{<tid, slot> $\in$ collectedReservations[tid]}}
                \State {\texttt{reservedEra $\gets$ collectedReservations[tid][slot]}}
                \If{reservedEra $<$ objPtr.birthEra OR reservedEra $>$ objPtr.retireEra OR reservedEra == NONE}\label{lin:orighecanfreecomp}
                    \State \textit{continue}
                \EndIf
                \State{}
                \Return{False}
            \EndFor
            \State{}
            \Return{True}
        \EndProcedure
\end{algorithmic}
\end{algorithm}

The function \Call{read}{} takes the address of an atomic variable from which a new pointer to the next node must be read (\lineref{orige-procread}). Within a loop, it performs the following steps: it loads the previously reserved era, reads the pointer from the address, and compares the previously reserved epoch to the latest value of the epoch. If the epoch has changed, it reserves the latest epoch and attempts to fetch the pointer again. If the epoch has not changed, the latest epoch is already reserved, so it is not necessary to reserve it again, and the pointer to the node is returned. The data structure can safely dereference the node, as it cannot be freed while the epoch is reserved (assuming a retired node cannot be accessed from another node).

Within its \Call{retire}{} function (\lineref{proc-origheretire}), threads append retired nodes to their \textit{retireList}. When the list exceeds \textit{reclaimFreq} in size, it frees the nodes that are not reserved. Specifically, the \Call{ReclaimHEFreeable}{} function is invoked (at \lineref{rechefreeable}), which collects all reserved epochs by scanning all thread reservation slots in the \textit{sharedReservations} array (\lineref{orighecollectstart} - \lineref{orighecollectend}).
Later, at \lineref{orighecanfree}, \Call{canFree}{} is invoked, which checks whether the lifespan of the retired node intersects with any of the reserved epochs. If the lifespan of the node intersects with a reserved epoch, it skips freeing the node because the node might be accessed by the thread that reserved the epoch. Otherwise, no thread can access the node, and it is safely freed (\lineref{orighereclaimend}). 

After a thread has completed its data structure operation, it can clear its reservation using \Call{clear}{} (\lineref{orighe-proc-clearall}).

\subsubsection{The Problem}

On the reader's side, a fence corresponding to the write for reserving the epoch is incurred only if the global epoch changes during the pointer read to the node. However, the frequency of memory fences is inversely proportional to the frequency of global epoch updates, which tends to be higher with high contention or update rates. Reducing the rate of change of the global epoch decreases the frequency of memory fences but increases the memory footprint, and vice versa.

Hazard Eras may have a larger memory footprint than hazard pointers because an epoch may protect a set of nodes proportional to the size of the data structure. This approach loses the precision of the hazard pointers in bounding the number of unreclaimed nodes at a given time. Additionally, Hazard Eras require a change in the node's memory layout.

\subsection{Epoch Based Reclamation}
\label{sec:rcu}

\begin{algorithm}
    \small
\caption{Epoch Based Reclamation~\cite{wen2018interval}}
\label{algo:rcuebr}
    \begin{algorithmic}[1]
        \State \texttt{const reclaimFreq, epochFreq} \label{lin:var-reclaimFreq} \label{lin:var-epochFreq} 
        \State \texttt{atomic<int> epoch} \Comment{incremented periodically} \label{lin:epoch}
        \State \texttt{thread\_local int counter, tid} \Comment{tid is current thread id}
    \State \texttt{atomic<int> reservedEpoch[NTHREAD]} \label{lin:var-reservedEpoch}
    \State \texttt{list<T*> retireList [NTHREAD]} \label{lin:var-retireList}

    \Statex    
    \Procedure{startOp}{ } \label{lin:proc-startOp}
        \If{\texttt{0 == ++counter \% epochFreq}}
            \State {\texttt{epoch.fetch\_add(1)}} \label{lin:incrementEpoch}
        \EndIf
        \State {\texttt{reservedEpoch[tid] $\gets$ epoch}} \label{lin:reservEpoch}
    \EndProcedure        

    \Statex
        \Procedure{retire}{\texttt{T* ptr}}
            \State \texttt{myRetireList $\gets$ retireList[tid]}
            \State \texttt{myRetireList.append(ptr)}
            \State \texttt{ptr.retireEpoch $\gets$ epoch}
            \If {\texttt{$0 ==$ myRetireList.size() $\%$ reclaimFreq}}
                \State {\Call{reclaimEpochFreeable}{\texttt{myRetireList}}}
            \EndIf
        \EndProcedure    
        
    \Statex
        \Procedure{reclaimEpochFreeable}{\texttt{myRetireList}} \label{lin:proc-recEpochFree}
            \State { minReservedEpoch $\gets$ reservedEpoch.min()}
            \ForAll {\texttt{objPtr $\in$ myRetireList}}
                \If{\texttt{objPtr.retireEpoch} $<$ \texttt{minReservedEpoch}}
                    \State {\texttt{free(objPtr)}}
                \EndIf
            \EndFor
        \EndProcedure

    \Statex
    \Procedure{endOp}{ }
        \State {\texttt{reservedEpoch[tid] $\gets$ MAX}}
    \EndProcedure
    \end{algorithmic}
\end{algorithm}

\subsubsection{The Overview}
Epoch based reclamation (EBR) has multiple variants, such as those proposed by Harris ~\cite{harris2001pragmatic}, Fraser \cite{fraser2004practical}, RCU~\cite{hart2007performance}, and Brown's Debra~\cite{brown2015reclaiming}. The key concept across these algorithms is that threads execute data structure operations in progression of epochs.
\textit{Reclaimers} can free retired nodes in or before a given epoch $e$ once every thread has completed execution in epoch $e$ or earlier and has transitioned to a more recent epoch (greater than $e$). 
Implying that for nodes retired on or before epoch $e$, all threads have gone quiescent at least once, making the nodes \textit{safe} for reclamation. In RCU~\cite{michaelKenny2017proposed} terminology, while performing data structure operations, threads are assumed to be executing a read-side critical section. A \textit{reclaimer} can only free those nodes that are \textit{retired} before the beginning of the oldest read-side critical section.

\algoref{rcuebr}, shows an example implementation.
It maintains a monotonically increasing shared \func{epoch} variable (assuming that it never overflows). It is incremented after a certain number of operations to represent a progression of epochs. Threads utilize a \func{reservedEpoch} array, with one single-writer multi-reader slot for each thread to save and publish the epoch in which they are executing. A thread declares its entry into a read-side critical section by publishing the current value of \func{epoch} in its \func{reservedEpoch} slot. Declares the exit from the previous read-side critical section by publishing a maximum possible value MAX (assuming \func{epoch} can never be the same as this value).

When a thread decides to reclaim its \func{retireList}, it finds the minimum epoch reserved by a thread from the \func{reservedEpoch} array. Then it frees those nodes retired before the minimum reserved epoch. In the implementation, this is identified by comparing the \func{retireEpoch} of the retired nodes that are associated with the nodes at the time it is appended to its \func{retireList}.
Since frequent increments in shared \func{epoch} and freeing of \func{retireList} are overhead, these are amortized over multiple operations to boost performance.

\subsubsection{The Problem}
\label{sec:rcurobustness}
The main concern with EBR is that a thread could get stuck in a long-running data structure operation, leading to a delayed exit from an old read-side critical section due to arbitrary system-level reasons, such as page fault servicing or thread scheduling. These delays might be significant enough to cause the minimum declared epoch to lag far behind the current global \func{epoch}. This situation prevents all threads from freeing their retire lists, resulting in a drastic increase in system memory consumption and possibly leading to out-of-memory errors. This issue is commonly referred to as a lack of \textit{robustness}.

For example, in the given implementation, a thread stuck in a long-running data structure operation, after reserving a small epoch value, will prevent other threads from freeing all nodes retired on and after the reserved epoch, even though all other threads are at much larger, more recent epochs. Consequently, the \func{retireList} of all threads could continuously grow, leading to increasingly higher memory consumption.

\subsection{Summary: The Problem and the Solution}
Programmers face unattractive choices: choose robust but slow hazard pointers; opt for robust and fast reservation algorithms (e.g., \cite{dice2016fast, balmau2016fast}) but tolerate intrusiveness or less portable, hardware-dependent solutions; choose fast but not robust EBR; or opt for hazard eras, which are a middle ground between hazard pointers and EBR, compromising on the tighter upper bound on the memory footprint by trading off the number of memory fences during reads. In practice, hazard eras are still slow in many data structures, as shown in \chapref{chapnbrp}.

As a solution, we introduce versions of these algorithms based on publish-on-ping, which are fast and retain the original traits of these algorithms.

\section{Publish-on-Ping Algorithms}
\label{sec:popalgos}

\subsection{The Overview}
The key aspect of publish-on-ping is that threads can read new pointers and reserve them locally without immediate publishing to \textit{reclaimers}. This eliminates the need for costly memory fences per read. \textit{Reservations} are published only when a thread attempts to reclaim its retire list. In other words, when a thread wants to reclaim its retire list, it signals all threads to publish their local \textit{reservations} to shared locations. The \textit{reclaimer} then scans these shared locations to collect all \textit{reservations} and subsequently frees retired nodes in its retire list that are not reserved by other threads.

The publish-on-ping is implemented using a POSIX signal and a corresponding signal handler. \textit{Reclaimers} use a \func{pthread\_kill} call to signal all other threads in the system (\textbf{ping}). Other threads execute a signal handler to assign their local reservations to corresponding shared locations (\textbf{publish}).

In a nutshell, the publish-on-ping paradigm eliminates the memory fence overhead from the principal traversal path by allowing threads to publish reservations only when demanded by infrequent reclamation events through a simple and neat use of signaling. This elimination of overhead from the read path benefits data structures, especially in read-dominated workloads.

Moreover, the paradigm is equally effective for other related techniques that incur similar per-read memory fences during traversals, such as Hazard Eras (HE). We apply publish-on-ping to HE and observe performance improvements, as demonstrated in \secref{popeval}.

However, an important challenge is ensuring that \textit{reclaimers}, after sending pings to all threads, can provably establish a time when all \textit{reservations} are considered published. This is crucial for safely freeing nodes in the retire lists. We discuss this aspect in detail while describing our example HazardPointerPOP implementation in the next section, and it applies to all the publish-on-ping variants we propose.

\subsection{HazardPointersPOP}
Having discussed the fundamental concept of publish-on-ping, we now apply it to HP to eliminate the need for publishing reservations during traversals or reads of new nodes in data structures. The resulting algorithm is termed Hazard pointers publish-on-Ping (HazardPointersPOP). 

\begin{algorithm}
\small
    \caption{HazardPointerPOP: Hazard Pointer Publish-on-Ping.}
    \label{algo:pophp}
    \begin{algorithmic}[1]
    \State \texttt{const reclaimFreq} \Comment{frequency of reclaiming retire list} \label{lin:var-reclaimfreq}
    \State \texttt{thread\_local int tid} \Comment{current thread id}
    \State \texttt{list<T*> retireList [NTHREAD]} \label{lin:var-retirelist}
    \BeginBox[draw=black, dashed]
        \State \texttt{T* localReservations [NTHREAD][MAX\_HP]} \label{lin:var-localreservations}
        \State \texttt{atomic<T*> sharedReservations [NTHREAD][MAX\_HP]} \label{lin:var-sharedreservations} 
        \State \texttt{atomic<int> publishCounter [NTHREAD]} \label{lin:var-publishcounter} 
        \State \texttt{thread\_local int collectedPublishCounters [NTHREAD]} \label{lin:var-collectedpublishcounters} 
    \EndBox
        \Procedure{\texttt{T*} read}{\texttt{atomic<T*> \&ptrAddr, int slot}} \label{lin:pophp-proc-read}
            \Repeat
                \State \texttt{T* readPtr $\gets$ *ptrAddr}
                \BeginBox[draw=black, dashed]
                \State \texttt{localReservations[tid][slot] $\gets$ readPtr} \LComment{no store load fence needed.}
                \EndBox
            \Until{\texttt{readPtr $=$ *ptrAddr}} \label{lin:pophp-loopexit}
            \State \Return readPtr \label{lin:pophp-return}
        \EndProcedure
    \Statex
        \Procedure{retire}{\texttt{T* ptr}} \label{lin:proc-retire}
            \State \texttt{myRetireList $\gets$ retireList[tid]}
            \State \texttt{myRetireList.append(ptr)}
            \If {\texttt{myRetireList.size() $\geq$ reclaimFreq}}\label{lin:pophpsize}
                    \State \texttt{epoch.fetch\_and\_add(1)}
                \BeginBox[draw=black, dashed]
                    \State {\Call{collectPublishedCounters}{ }} \label{lin:collectpublishedcounters}
                    \State {\Call{\textbf{pingAllToPublish}}{ } } \label{lin:pingalltopublish}
                    \State {\Call{waitForAllPublished}{ }} \label{lin:waitforallpublished}
                \EndBox
                \State {\Call{reclaimHPFreeable}{myRetireList}}
            \EndIf
        \EndProcedure    
    \Statex
        \Procedure{clear}{ } \label{lin:proc-clearall}
            \For{$slot=0, \dots, MAX\_HP$}
                \BeginBox[draw=black, dashed]
                    \State \texttt{localReservations[tid][slot] $\gets$ NULL}\label{lin:setnull}
                    \LComment{no store load fence needed}
                \EndBox
            \EndFor        
        \EndProcedure        
\algstore{algpophp}
\end{algorithmic}
\end{algorithm}

\begin{algorithm} 
\small
\caption{HazardPointerPOP: Continued.}
\label{algo:contd}
\begin{algorithmic} [1]                   
\algrestore{algpophp}
        \Procedure{reclaimHPFreeable}{myRetireList} \label{lin:proc-reclaimhpfreeable}
                \LComment{collect all published reservations}
                \State \texttt{set<T*> collectedReservations $\gets$ \{\}} \label{lin:collectstart}
                    \ForAll {\texttt{<tid, slot> $\in$ sharedReservations[tid]}}
                        \State {\texttt{objPtr $\gets$ sharedReservations[tid][slot]}}
                        \State {\texttt{collectedReservations.insert(objPtr)}}
                    \EndFor \label{lin:collectend}
                \LComment{free all nodes not reserved}
                \ForAll {\texttt{objPtr $\in$ myRetireList}} \label{lin:reclaimstart}
                    \If{\texttt{objPtr $\notin$ collectedReservations}}
                        \State {\texttt{free(objPtr)}}
                    \EndIf
                \EndFor \label{lin:reclaimend}
        \EndProcedure

\BeginBox[draw=black, dashed]
    \Statex
        \Procedure{pingAllToPublish}{ } \label{lin:proc-pingalltopublish}
            \ForAll {\texttt{othertid $\neq$ tid} }
                \State {\texttt{pthread\_kill(othertid, \dots)}}
            \EndFor
        \EndProcedure
        \LComment{signal handler}
        \Procedure{publishReservations}{ } \label{lin:proc-publishreservations}
            \For{$ihp=0, \dots, MAX\_HP$}
                \State \texttt{sharedReservations[tid][ihp] $\gets$ localReservations[tid][ihp]}
            \EndFor
            \BeginBox[draw=blue, dotted]
            \State \texttt{publishCounter[tid] $\gets$ publishCounter[tid]+1} \label{lin:publishcounter}
            \EndBox
        \EndProcedure

    \Statex
        \Procedure{collectPublishedCounters}{ } \label{lin:proc-collectpublishedcounters}
            \ForAll {\texttt{tid}}
                \State {\texttt{collectedPublishCounters[tid] $\gets$ publishCounter[tid]}}
            \EndFor
        \EndProcedure    
    \Statex
        \Procedure{WaitForAllPublished}{ } \label{lin:proc-WaitForAllPublished}
            \LComment{establish all threads have executed signal handler}
            \ForAll {\texttt{othertid $\neq$ tid}} \label{lin:verifstart}
            \Repeat
            \Until{\texttt{publishCounter[othertid] $>$ collectedPublishCounter[othertid]}}
            \EndFor \label{lin:verifend}
        \EndProcedure    
\EndBox
    \end{algorithmic}
\end{algorithm}

\algoref{pophp} and \algoref{contd} show the implementation of the proposed HazardPointersPOP algorithm. From a programmer's perspective, it maintains the same interface as HP, making it a seamless drop-in replacement. However, unlike the eager publishing paradigm in HP, threads in HazardPointersPOP save nodes locally in their own \func{localReservations} array during \Call{read}{}. The \Call{read}{} procedure repeatedly reads the pointer to the node, saves it in a corresponding slot in \func{localReservations}, and then rereads the pointer. Assuming a system with a maximum of \func{NTHREADS} threads, each thread has a fixed number of slots represented by \func{MAX\_HP}. The loop exits only when the pointer remains unchanged on the second read, at which point the pointer is returned.

When a thread is about to go \textit{quiescent}, i.e., while exiting the data structure operation, it resets the local reservations by setting the corresponding slots to NULL (\lineref{setnull}) using \Call{clear}{}. These local reservations are published by writing to a corresponding shared array of single-writer multi-reader slots, called \func{sharedReservations} (\lineref{var-sharedreservations}).

Threads maintain a per-thread list, depicted as \func{retireList} (\lineref{var-retirelist}), to which they append retired nodes with a call to \Call{retire}{} within their update operations. When the list size exceeds a given threshold set in \func{reclaimFreq} (\lineref{pophpsize}), a \textit{reclaimer} invokes \Call{pingAllToPublish}{} to trigger the publishing of local reservations.

In \Call{pingAllToPublish}{} (\lineref{proc-pingalltopublish}), a \textit{reclaimer} uses \func{pthread\_kill()} to ping all threads. The threads receiving these pings execute a signal handler, called \Call{publishReservations}{} (\lineref{proc-publishreservations}). Within the signal handler, other threads write all their local reservations to the shared array. Once all threads complete execution of \Call{publishReservations}{}, the reservations become visible to the \textit{reclaimer}.

After every thread publishes its reservations, the \textit{reclaimer} can invoke \Call{reclaimHPFreeable}{} to free all the retired nodes that are not reserved (\lineref{proc-reclaimhpfreeable}). Specifically, within \Call{reclaimHPFreeable}{}, the \textit{reclaimer} collects all the reservations and frees those that are not in the collected reservations (\lineref{collectstart}-\lineref{reclaimend}).

An important aspect we deliberately delayed discussing for clarity is the challenge of determining when all threads have executed their \Call{publishReservations}{} so that the \textit{reclaimer} can safely proceed to free its \func{retireList} after the call to \Call{pingAllToPublish}{} returns.

To establish such a time, HazardPointerPOP employs the array \func{publishCounter} (\lineref{var-publishcounter}), where each slot is a counter that increases monotonically and belongs to a unique thread. When a thread completes publishing, it increments its slot in \func{publishCounter} (\lineref{publishcounter}). \textit{Reclaimers} observe every thread's \func{publishCounter} value before and after pinging all threads to ensure that every thread has finished publishing by comparing the previously read \func{publishCounter} values with the reread values.

Specifically, \textit{reclaimers} read every thread's \func{publishCounter} value into their thread-local \func{collectedPublishCounters} array using \Call{collectPublishedCounters}{} at \lineref{collectpublishedcounters}. Then, they invoke \Call{pingAllToPublish}{} (\lineref{pingalltopublish}), followed by a call to \Call{waitForAllToPublished}{} (\lineref{waitforallpublished}), wherein \textit{reclaimers} repeatedly reread every thread's \func{publishCounter} value and compare it with the previously observed values. They only exit the loop when all threads have incremented their \func{publishCounter} value at least once since the time the \textit{reclaimer} collected the initial \func{publishCounter} values (\lineref{verifstart} - \lineref{verifend}).

\subsubsection{\Call{WaitForAllPublished}{ } completes in finite time:}
The Publish-on-ping mechanism relies on the assumption that, upon being signaled, all threads finish executing their signal handlers within a finite time.
This is vital so that the \textit{reclaimer}, after signalling all threads, eventually exits its while loop.
In \chapref{chapnbr}, we experimentally established that the waiting time for a signal to be delivered is short
and techniques such as DEBRA+~\cite{brown2015reclaiming}, NBR~\cite{singh2021nbr}, PEBR~\cite{kang2020marriage} rely on a similar assumption. 






Infact, the utilization of POSIX signals gives us access to the underlying scheduler, which we leverage to relax the traditional asynchronous memory model. In the traditional asynchronous memory model, learning the execution state of a thread (i.e., whether it has terminated or is merely delayed) is not possible.
However, with POSIX signals, we can learn if a thread has terminated or is delayed arbitrarily.
Specifically, when a thread sends an IPI, and the receiver has terminated, the sender learns the execution state of the target thread as an error message is returned to it. This information can be used to eliminate the terminated thread from requiring to publish reservations, as a terminated thread cannot access shared memory. 
When a thread is delayed arbitrarily, like due to context switching, modern schedulers are reasonably reliable in ensuring that a signal handler will be executed within a bounded time, as shown in~\chapref{chapnbr}.

\subsubsection{Note on a Useful Property of HazardPointersPOP}
HazardPointersPOP enables threads to privately track reservations using a lightweight \Call{read}{} and publish them on demand with a single fence when required by a \textit{reclaimer}. In essence, even if threads are stuck in a long-running execution, one can ping the stalled thread to learn which node it might currently be accessing. This feature allows us to develop an efficient variant of hazard pointers that approaches the performance of EBR. Specifically, in common cases, threads follow a fast path similar to EBR algorithms and when threads suspect delays, the publish-on-ping mechanism is used to assist a thread who is unable to reclaim.  
In the \secref{secepochpop}, we introduce this variant of hazard pointers that incorporates EBR.

\subsection{HazardErasPOP}

We apply pubish-on-ping to HE that results in the algorithm called HazardErasPOP. 

\begin{algorithm}
\small
    \caption{HazardEraPOP: Hazard Era Publish-on-Ping.}
    \label{algo:pophe}
    \begin{algorithmic}[1]
    \State \texttt{const reclaimFreq} \label{lin:hevar-reclaimfreq}\Comment{frequency of reclaiming retire list} 
    \State \texttt{thread\_local int tid} \Comment{current thread id}
    \State \texttt{list<T*> retireList [NTHREAD]} \label{lin:hevar-retirelist}
        \State \texttt{int localReservations [NTHREAD][MAX\_HE]} 
        \State \texttt{atomic<int> sharedReservations [NTHREAD][MAX\_HE]} \label{lin:hevar-sharedreservations}
        \State \texttt{atomic<int> publishCounter [NTHREAD]} \label{lin:hevar-publishcounter}
        \State \texttt{thread\_local int collectedPublishCounters [NTHREAD]}
        \State \texttt{int collectedReservations[NTHREAD][MAX\_HE]$\gets$\{\}} 

        \Procedure{\texttt{T*} read}{\texttt{atomic<T*> \&ptrAddr, int slot}}
            \State \texttt{int oldEra $\gets$ localReservations[tid][slot]} 
            \While{\textit{True}}
                \State \texttt{T* readPtr $\gets$ *ptrAddr}
                \State \texttt{int newEra $\gets$ epoch}                
                \If{oldEra == newEra}
                    \Return readPtr \label{lin:hevalsucc}
                \EndIf
            \State \texttt{localReservations[tid][slot] $\gets$ newEra} \label{lin:lochereserve}\LComment{no store load fence needed.}
            \State \texttt{oldEra $\gets$ newEra} 
            \EndWhile
        \EndProcedure
    \Statex
        \Procedure{retire}{\texttt{T* ptr}} \label{lin:heproc-retire}
            \State \texttt{myRetireList $\gets$ retireList[tid]}
            \State \texttt{myRetireList.append(ptr)}
            \State \texttt{ptr.retireEpoch $\gets$ epoch}
            \If {\texttt{myRetireList.size() $\geq$ reclaimFreq}}
                \BeginBox[draw=black, dashed]
                    \State {\Call{collectPublishedCounters}{ }}\label{lin:hecollectpublishedcounters}
                    \State {\Call{\textbf{pingAllToPublish}}{ } } \label{lin:hepingalltopublish}
                    \State {\Call{waitForAllPublished}{ }} \label{lin:hewaitforallpublished}
                \EndBox
                \State {\Call{reclaimHEFreeable}{myRetireList}}
            \EndIf
        \EndProcedure    
    \Statex
        \Procedure{clear}{ } \label{lin:he-proc-clearall}
            \For{$slot=0, \dots, MAX\_HE$}
                    \State \texttt{localReservations[tid][slot] $\gets$ NONE} \label{lin:setnone}
                    \LComment{no store load fence needed}
            \EndFor        
        \EndProcedure 
        \Statex
        \Procedure{reclaimHEFreeable}{myRetireList} \label{lin:heproc-reclaimhefreeable}
                \LComment{Same as in \algoref{orighe}}
        \EndProcedure
    \Statex
        \Procedure{canFree}{objPtr}
            \LComment{Same as in \algoref{orighe}}
        \EndProcedure
\end{algorithmic}
\end{algorithm}

\algoref{pophe} shows the implementation of the proposed HazardErasPOP algorithm. From a programmer's perspective, it maintains the same interface as HP, HE, or HazardPointersPOP. However, unlike the eager publishing reservations paradigm in HE, threads in HazardErasPOP save the nodes (to reserve) locally in their \func{localReservations} array during \Call{read}{}.

The \Call{read}{} procedure fetches the pointer to the node currently being read, reads the value of the current epoch, and compares it with the previously saved epoch value in the slot at \func{localReservations}. 
The slot corresponds to the current pointer being fetched.
If the previously locally reserved epoch and matches the current epoch, i.e. the current epoch has not changed while the pointer to the node was being read, then the node is safe to be read, and the pointer to the node is returned. 
Otherwise, the new epoch value is locally reserved (in \func{localReservations} without publishing, at \lineref{lochereserve}) and the process is repeated until the global epoch remains unchanged while the value at the input atomic pointer containing the address of the node currently being read is successfully fetched (\lineref{hevalsucc}).

When a thread is about to go \textit{quiescent}, i.e., while exiting the data structure operation, it resets the local reservations by setting the corresponding slots to NONE (\lineref{setnone}) using \Call{clear}{}.
These local reservations are then published to a shared array of single-writer multi-reader slots, called \func{sharedReservations} (\lineref{hevar-sharedreservations}), when pinged by a reclaimer.

Threads maintain a per thread list, depicted as \func{retireList} (\lineref{hevar-retirelist}), to which they append the retired nodes with call to \Call{retire}{ } during their update operations.
When the list size exceeds a threshold set in \func{reclaimFreq} (\lineref{hevar-reclaimfreq}), a \textit{reclaimer} invokes \Call{pingAllToPublish}{ } to trigger the publishing of local reservations. This function is same as the one described for \algoref{contd}.


After every thread publishes its reservations, the \textit{reclaimer} can invoke \Call{reclaimHEFreeable}{ } to free all retired nodes that are not reserved (\lineref{heproc-reclaimhefreeable}). Specifically, within \Call{reclaimHEFreeable}{ }, the \textit{reclaimer} collects all reservations and frees those that are not in the collected reservations. This is similar to the original hazard era algorithm shown in \algoref{orighe}.

The functions \Call{publishReservations}{ }, \Call{collectPublishedCounters}{ }, 
\Call{pingAllToPublish}{ }, and \Call{waitForAllToPublished}{ } are the same as shown in \algoref{contd}.



\subsection{EpochPOP}
\label{sec:secepochpop}

\ignore{
\textbf{High level view of EpochPOP}
EpochPOP builds on HazardPointersPOP as follows.
In the common case, threads execute data structure operation largely as they would in epoch based reclamation.
A data structure operation begins by reading a global epoch number (a timestamp), and announcing the value it read in a SWMR per-thread slot in a shared array, then proceeds as usual.
Nodes unlinked from the data structure are stored in a per-thread list.
Periodically, threads increment the global epoch, and scan other threads' announcements to identify the oldest announced epoch. 
Threads can free nodes unlinked from the data structure in epochs older than the oldest announced epoch.
However, unlike traditional EBR, all threads also continue to execute the steps of the HazardPointersPOP algorithm, locally reserving nodes before accessing them. 
If threads announce the epochs sufficiently often, they are able to reclaim memory regularly to keep the size of the lists with unlinked nodes below a user-specified threshold, then the POP mechanism is not needed at all!
But, if thread delays hinder frequent announcements of epochs, preventing reclamation of the lists, then the POP mechanism is used to facilitate the reclamation.
}

\subsubsection{The overview}
In EpochPOP, threads operate in epochs, announcing their entry and exit from the \textit{quiescent state} using a global epoch variable, similar to the EBR algorithm. However, unlike EBR, threads privately track local reservations during traversals, utilizing lightweight reads of HazardPointerPOP. 

\textit{Reclaimers}, periodically scan all the announced epochs freeing nodes retired before the minimum announced epoch, provided they can keep the retire list sizes below a threshold. In rare case where thread delays prevent a \textit{reclaimer} from freeing its retired nodes, the \textit{reclaimer} employs publish-on-ping to force all threads to publish their current reservations, emptying the retire list while excluding the reserved nodes.  

\begin{algorithm}
\small
\caption{EpochPOP}
\label{algo:rcupophp}
    \begin{algorithmic}[1]
    \State \texttt{const reclaimFreq, epochFreq, C} \label{lin:rcuvar-epochfreq}
    \State \texttt{atomic<int> epoch} \label{lin:rcuvar-epoch}
    \State \texttt{thread\_local int tid, counter}
    \State \texttt{atomic<int> reservedEpoch[NTHREAD]} \label{lin:rcuvar-reservedepoch}
    \State \texttt{list<T*> retireList [NTHREAD]} \label{lin:rcuvar-retirelist}
    \State \texttt{T* localReservations [NTHREAD][MAX\_HP]}
    \State \texttt{atomic<T*> sharedReservations [NTHREAD][MAX\_HP]} 
    \State \texttt{atomic<int> publishCounter [NTHREAD]}
    \State \texttt{int collectedPublishCounters [NTHREAD]}
    \Statex    
    \Procedure{startOp}{ }
        \If{\texttt{0 == ++counter \% epochFreq}}
            \State {\texttt{epoch.fetch\_add(1)}} \label{lin:incepoch}
        \EndIf
        \State {\texttt{reservedEpoch[tid] $\gets$ epoch}}
    \EndProcedure    
    
    \Statex
        \Procedure{\texttt{T*} read}{\texttt{atomic<T*> \&ptrAddr, int slot}} \label{lin:proc-read}
            \Repeat
                \State \texttt{T* readPtr $\gets$ *ptrAddr}
                \BeginBox[draw=black, dashed]
                \State \texttt{localReservations[tid][slot] $\gets$ readPtr}
                \EndBox
            \Until{\texttt{readPtr $=$ *ptrAddr}} \label{lin:loopexit}
            \State \Return readPtr \label{lin:return}
        \EndProcedure    
    \Statex
        \Procedure{retire}{\texttt{T* ptr}}
            \State \texttt{myRetireList $\gets$ retireList[tid]}
            \State \texttt{myRetireList.append(ptr)}
            \State \texttt{ptr.retireEpoch $\gets$ epoch}

            \If {\texttt{$0 ==$ myRetireList.size() $\%$ reclaimFreq}}
                \State {\Call{reclaimEpochFreeable}{\texttt{myRetireList}}} \label{lin:reclaimepochfreeable}
                \If{\texttt{myRetireList.size() $\ge$ C*reclaimFreq}} \label{lin:pophpstylebegin}
                \BeginBox[draw=black, dashed]
                    \State {\Call{collectPublishedCounters}{ }} 
                    \State {\Call{\textbf{pingAllToPublish}}{ } } 
                    \State {\Call{waitForAllPublished}{ }}
                \EndBox
                \State {\Call{reclaimHPFreeable}{myRetireList}} \label{lin:rcureclaimhpfreeable}
                
                \EndIf
            \EndIf
        \EndProcedure    
\algstore{algepochpop}
\end{algorithmic}
\end{algorithm}       

\begin{algorithm}   
\small
\caption{EpochPOP: Continued.}
\label{algo:epochpopcontd}
\begin{algorithmic} [1]                   
\algrestore{algepochpop}
        \Procedure{reclaimEpochFreeable}{\texttt{myRetireList}}
            \State { minReservedEpoch $\gets$ reservedEpoch.min()}
            \ForAll {\texttt{objPtr $\in$ myRetireList}}
                \If{\texttt{objPtr.retireEpoch} $<$ \texttt{minReservedEpoch}}
                    \State {\texttt{free(objPtr)}}
                \EndIf
            \EndFor
        \EndProcedure

    \Statex
        \Procedure{reclaimHPFreeable}{myRetireList}
                \State \texttt{set<T*> collectedReservations $\gets$ \{\}}
                    \ForAll {\texttt{<tid, slot> $\in$ sharedReservations[tid]}}
                        \State {\texttt{objPtr $\gets$ sharedReservations[tid][slot]}}
                        \State {\texttt{collectedReservations.insert(objPtr)}}
                    \EndFor
                \ForAll {\texttt{objPtr $\in$ myRetireList}}
                    \If{\texttt{objPtr $\notin$ collectedReservations}}
                        \State {\texttt{free(objPtr)}}
                    \EndIf
                \EndFor
        \EndProcedure
    \Statex
    \Procedure{endOp}{ }
        \State {\texttt{reservedEpoch[tid] $\gets$ MAX}}
        \State {\Call{clear}{ }} \label{lin:rcuclearall}
    \EndProcedure
    \end{algorithmic}
\end{algorithm}

\subsubsection{Description of algorithm}

The example implementation in \algoref{rcupophp} and \algoref{epochpopcontd} describes the EpochPOP building upon the EBR implementation presented in \secref{rcu}. All line references in this section refer to \algoref{rcupophp} and \algoref{epochpopcontd}.
Each thread maintains a \func{retireList} to collect the retired nodes (\lineref{rcuvar-retirelist}), \func{reservedEpoch} to announce the current epoch in (\lineref{rcuvar-reservedepoch}), and a monotonically increasing \func{epoch} variable (\lineref{rcuvar-epoch}). 
The \func{epoch} is incremented periodically using the value of \func{epochFreq} \lineref{rcuvar-epochfreq}.

Similar to the original EBR, a thread announces its transition to the quiescent state by reserving the current epoch in the appropriate slot of \func{reservedEpoch} with a call to \Call{startOp}{}. To access new nodes of the data structure within an operation, the thread uses \Call{read}{}, privately reserving them without immediate publishing (using a fence), similar to HazardPointerPOP. Threads enter the \textit{quiescent state} by announcing a maximum possible epoch through \Call{endOp}{} and clearing the local reservations (\lineref{rcuclearall}).

\textit{Reclaimers} can continue freeing their retire lists as they normally do in EBR unless a thread delay is suspected. During \Call{retire}{} calls, threads append retired nodes to their retireList by associating the current epoch as their \func{retireEpoch}. When the list reaches a threshold size, they invoke \Call{reclaimEpochFreeable}{} (\lineref{reclaimepochfreeable}). In this procedure, they find the minimum epoch reserved by a thread from the \func{reservedEpoch} array and then free the nodes retired before the minimum reserved epoch.

If a \textit{reclaimer} suspects delays, it invokes a robust reclamation process similar to HazardPointersPOP during the \Call{retire}{} call. In the example implementation, if, after attempting EBR-style reclamation (\lineref{reclaimepochfreeable}), the thread finds that its \func{retireList} is still not empty (e.g., more than half of the \func{retireList} remains unreclaimed), it assumes that some threads reserving an older epoch have been delayed. This triggers reclamation in the HazardPointerPOP style, where the \textit{reclaimer} pings all threads to publish their reservations and then frees its \func{retireList} using \Call{reclaimHPFreeable}{} (\lineref{rcureclaimhpfreeable}).

Privately tracking reserved pointers is crucial for threads trying to reclaim nodes, in the event where a delayed thread prevents freeing of it retired nodes. This approach helps reclaimers to precisely determine nodes that a stuck thread is currently accessing, allowing the \textit{reclaimer} to safely free nodes in its retire list and skip a bounded set of reserved nodes. 

One might argue that simply pinging all threads and waiting until they complete executing their signal handlers would be sufficient. 
However, the issue with this approach is that if a thread is stuck in an erroneous control flow, such as an infinite loop, the \textit{reclaimer} will not learn anything about the nodes that could be accessed by the stuck threads. 

Alternatively, one might suggest publishing the currently announced epoch of the stuck thread instead of reserving pointers and publishing them when signalled.
With this approach, again, for the stuck thread in an erroneous loop, although, the reclaimer will learn about a range of nodes which might be safe to be freed, but this information is no better than what the reclaimer already knows by scanning the announced epoch in EBR mode.
This is because the announced epoch remains unchanged unless the stuck thread starts a new operation or restarts to publish reserve the latest epoch.
Thus, the reclaimer can only free those nodes that were retired before the minimum published epoch and may have to leave an unbounded number of nodes unreclaimed.  

The question of determining what nodes are safe to be freed compelled earlier techniques using signals, such as  Debra+~\cite{brown2015reclaiming} and NBR~\cite{singh2021nbr} to forcibly change control flow, forcing threads to restart. However, EpochPOP does not need to alter control flow of threads, thanks to its ability to track reservations privately and publish them on demand.

\section{Correctness}
\label{sec:popcorr}
\subsubsection{Correctness and Progress in HazardPointersPOP}

\begin{assumption}
\label{asm:sigasm}
Threads publish their reservations in a bounded time after being pinged.    
\end{assumption}
We experimentally verified in \chapref{chapnbr} that this assumption is true.

\begin{property}
    (Safety) HazardPointersPOP is safe from use-after-free errors. 
\end{property}
In order to prove HazardPointersPOP is safe, we need to establish that any \textit{reclaimer} will not free a node which other threads could subsequently access.

$Wlog.$, by the way of contradiction, let us assume, a thread $T1$ frees a node $n$ at a time $t1$ which is subsequently accessed by another thread $T2$ at a later time $t2$ ($t1 < t2$). 
In order to access $n$, $T2$ must successfully \textit{reserve} (not publish) it at an earlier time $t2'$, such that $t2' < t2$.

Similarly, to free $n$, $T1$ retires, then pings all threads to publish their reservations, and then for a bounded time waits to ensue that all threads complete publishing their reservations at a time $t1'$, such that $t1' < t1$. In this case, $T1$ pings $T2$ to publish its reservation to $n$ and $T1$ see that the reservation to $n$ is published at time $t1'$.

Now, two cases arise. 

First, $t1' < t2'$, i.e., $T2$ published its reservation for $T1$ before $T2$ reserved $n$ at $t2'$. This case can only happen if $n$ was already retired. 
In this case, since $n$ was already retired by $T1$, $T2$ will fail validation while reserving $n$. 
Thus, $T2$ cannot access $n$ without successfully reserving (HP requirement that threads reserve before access).

In the second case, $t2' < t1'$, i.e $n$ was successfully reserved before $T1$ requested $T2$ to publish the reservations.
Note, in this case, $T2$ would have published the reservation which $T1$ will scan and skip freeing $n$ as it is guaranteed to find it in the reservation list of $T2$ (\asmref{sigasm}). Hence, HazardPointersPOP is immune to use-after-free errors.  

\begin{property}
    (Liveness) HazardPointersPOP is robust.
\end{property}
A thread accumulates $r$ nodes in its retire list of which $N*H$ nodes may be reserved, where $N$ is the number of threads and $H$ is the maximum number of reservations that $N$ threads could hold at a given time. This implies, at a given time, a thread could free at least $r - N*H$ nodes, and at most $N*H$ nodes may not be freed.
Since $N*H$ is a constant, the amount of unreclaimed garbage per thread is always constant. Therefore, popHP is robust. 

\subsubsection{Correctness and Progress in HazardErasPOP}

\begin{property}
    (Safety) HazardErasPOP is safe from use-after-free errors. 
\end{property}
In order to prove that HazardErasPOP is safe, we need to establish that any \textit{reclaimer} will not free a node that other threads could subsequently access.

$Wlog.$, by the way of contradiction, let us assume that a thread $T1$ frees a node $n$ at a time $t1$ which is subsequently accessed by another thread $T2$ at a later time $t2$ ($t1 < t2$). 

In order to access $n$, $T2$ must successfully protect it at an earlier time $t2'$, such that $t2' < t2$. 
Similarly, to free $n$, $T1$ retires, then pings all threads to publish their reservations, and then waits, for a bounded time, to ensue that all threads complete publishing their reservations at a time $t1'$, such that $t1' < t1$.

Now, two cases arise. 

First, $t1' < t2'$, that is, $T2$ published all its reservations for $T1$ to see before $T2$ reserved $n$ at $t2'$. This is only possible if $T2$ reserved the retired $n$, which is not possible because $T2$ will fail to reserve $n$ at $t2'$ because $T1$ before pinging would have changed the global epoch.

In the second case, $t2' < t1'$, that is, $n$ was successfully reserved before $T1$ requested $T2$ to publish the reservations.
Note, in this case, $T2$ would have published the reservation of the epoch value at the time the pointer to $n$ is read, and $T1$ will scan the reservation and is guaranteed to find the reserved epoch. Now, when $T1$ attempts to free $n$ it will notice that the reserved epoch is less or equal to the retire epoch, i.e. lifespan of $n$ overlaps with the reserved epoch and there will not free $n$.
Hence, HazardErasPOP is immune to use-after-free errors.

\begin{property}
    (Liveness) HazardErasPOP is robust.
\end{property}
HazardErasPOP retains the original robustness property of Hazard Eras. A reserved epoch could only prevent reclamation of the nodes whose lifespan intersects with the epoch. All nodes allocated after or retired before the reserved epoch can continue to be reclaimed. 

\subsubsection{Correctness and Progress in EpochPOP}

\begin{property}
    (Safety) EpochPOP is free from use-after-free errors. 
\end{property}
EpochPOP, mostly runs classic EBR synchronization between \textit{readers} and \textit{reclaimers}.
In scenarios where no delayed threads are detected, \textit{reclaimers} only free nodes whose retire epoch indicates that they were retired before the oldest announced epoch across all threads.
The success of the aforementioned condition indicates that all threads have gone quiescent at least once since the node (which a \textit{reclaimer} wishes to free) was retired. Consequently, no thread could hold a reference to this node.

When a delayed thread is detected, a \textit{reclaimer} takes the following steps:
\begin{enumerate}
    \item Signals all threads to timely publish the reservations they were maintaining all along (\asmref{sigasm}).
    \item Waits in a bounded loop to ensure that all reservations are published before proceeding to free its retire list.
    \item Scans the reservations and frees only the nodes that are not reserved.
\end{enumerate}
In this way, EpochPOP ensures that no use-after-free errors occur.

\begin{property}
    (Liveness) EpochPOP is robust.
\end{property}
This is ensured by the ability of the algorithm to maintain reservations while traversing the data structures and detect delayed threads. This allows threads to ensure continuous reclamation of nodes in its retire list, only skipping a bounded set of reserved nodes across all threads.

\section{Evaluation}
\label{sec:popeval}
We implemented HazardPointersPOP, a publish on ping implementation of hazard pointers; HazardErasPOP, a publish on ping version of hazard eras, and EpochPOP. These implementations were evaluated within the setbench benchmark\textendash a benchmark also used in techniques including DEBRA~\cite{brown2015reclaiming}, NBR~\cite{singh2021nbr} and TokebEBR~\cite{kim2024token}. The performance of these techniques was then compared with traditional hazard pointers (HP), hazard eras (HE), EBR and a dummy implementation with no reclamation (NR).
Each of these memory reclamation schemes was applied to three distinct data structures: Harris-Michael list~(HML)~\cite{michael2004hazard}, HML chaining based hashtable (HMLHT), and the external binary search tree of David et~al. (DGT)~\cite{david2015asynchronized}.

\paragraph{Experimental Setting.}
All experiments were conducted on an Intel Xeon Platinum 8160 machine, with 132 MB of L3 Cache and 384 GB of DRAM. The machine is equipped with 4 NUMA nodes, each hosting 24 cores running at 2.1 GHz with 2-way hyperthreading, amounting to a total of 48 logical threads per node and 192 logical threads across all 4 sockets.

Our benchmark, compiled using \texttt{std=c++17} with \texttt{-O3} optimization, was executed on Ubuntu 20.04 with kernel version 5.8.0-55. All experiments were performed with \texttt{numactl --interleave=all}, using the MiMalloc allocator.
Threads were pinned in a pattern such that first all 24 physical cores were filled within a socket, followed by pinning to the other 24 logical threads on the same socket, before moving on to the other sockets.

\paragraph{Allocator Choice.}
Recently, Brown et. al.~\cite{kim2024token}, identified a negative performance interaction between deferred memory reclamation schemes and jeMalloc~\cite{evans2006scalable}, a widely used allocator for evaluating reclamation techniques.
This issue stems from contention on free lists, leading to suboptimal scaling of reclamation techniques on large-scale machines. MiMalloc~\cite{leijen2019mimalloc} addresses this problem by employing a multilevel sharding approach for free lists. To ensure that our evaluation remains unaffected by these challenges, we use MiMalloc as our benchmark.

\paragraph{Experimental Methodology.}
In each trial of our experiment, threads prefill the data structure up to half of the maximum fixed key range. Subsequently, they enter an execution phase, performing data structure operations for 5 seconds. During this phase, a randomly chosen insert, delete, or contains operation is repeatedly invoked with a key randomly selected from the given key range in the data structure.

In our experiments, we use read-intensive workloads with 90\% contains, 5\% inserts, and 5\% deletes, as well as update-only workloads with 50\% inserts and 50\% deletes. At the end of the 5 seconds of the execution phase, the measured execution throughput (number of operations per second) and peak memory consumption in bytes are reported. Our plots depict throughput averaged over 5 trials, each run with a thread sequence of 1, 4, 8, 16, 32, 64, 128, and 190. In particular, we did not observe any significant variation between trials, the standard deviation between trials being below 3\%.

For trees, the key range is set to $[0, 2 \times 10^6]$, for lists, we use a key range of $[0, 2 \times 10^3]$, and for hash tables, the key range is $[0, 6 \times 10^6]$ with $10^6$ buckets and a load factor of 6 keys per bucket. Each \textit{reclaimer} uses a maximum retired bag size of 32000 nodes before attempting reclamation. Techniques using epochs increments it every $NTHREADS\times EPOCHFREQ$ allocation, where $EPOCHFREQ$ is a constant (set to 100 in our experiments).  

\begin{figure}
\centering
     \begin{minipage}{\textwidth}
        \begin{subfigure}{\textwidth}
            \includegraphics[width=0.33\linewidth, keepaspectratio]{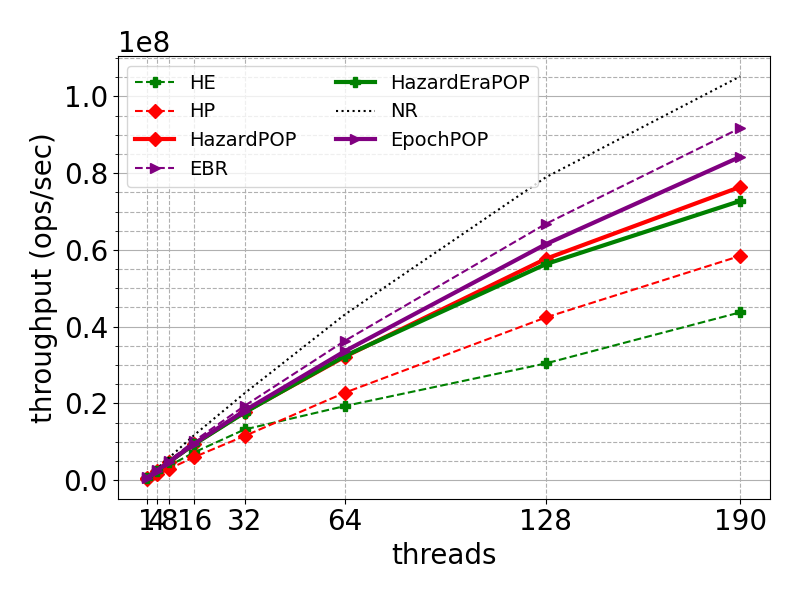}\hfill
            \includegraphics[width=0.33\linewidth, keepaspectratio]{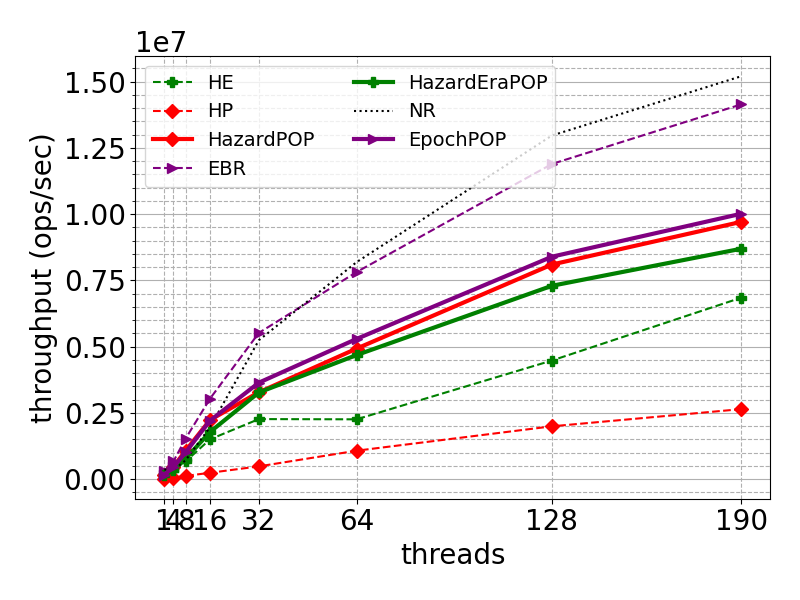}\hfill
            \includegraphics[width=0.33\linewidth, keepaspectratio]{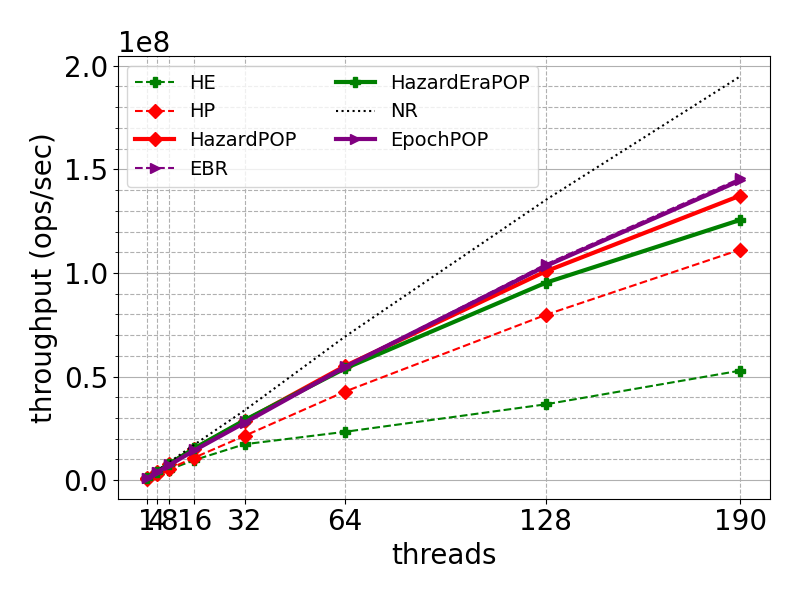}\hfill
            \caption{update heavy: 50\% inserts and 50\% deletes. }
            \label{fig:exp50p}
        \end{subfigure}
        \begin{subfigure}{\textwidth}
            \includegraphics[width=0.33\linewidth, keepaspectratio]{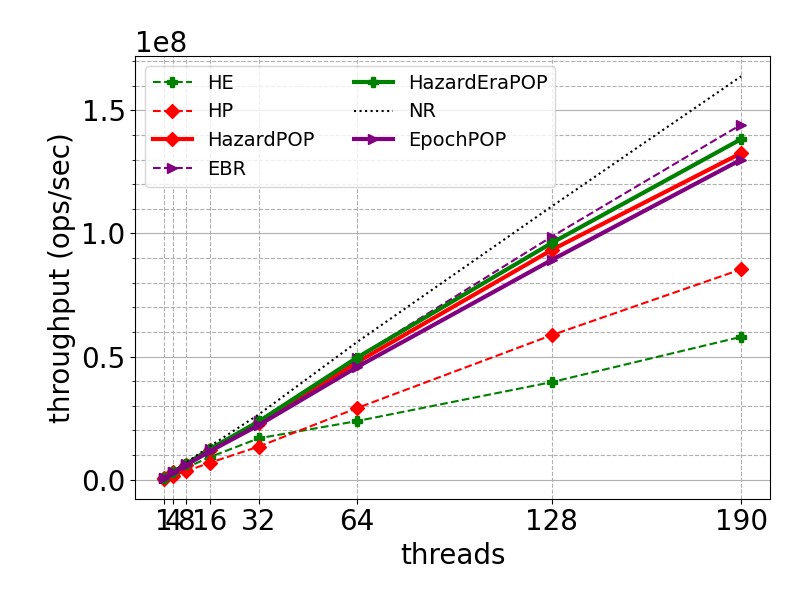}\hfill
            \includegraphics[width=0.33\linewidth, keepaspectratio]{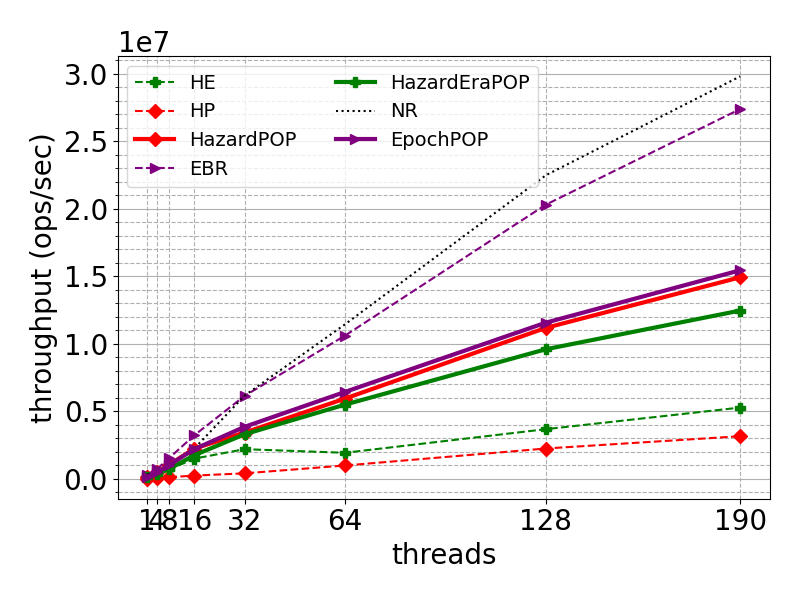}\hfill
            \includegraphics[width=0.33\linewidth, keepaspectratio]{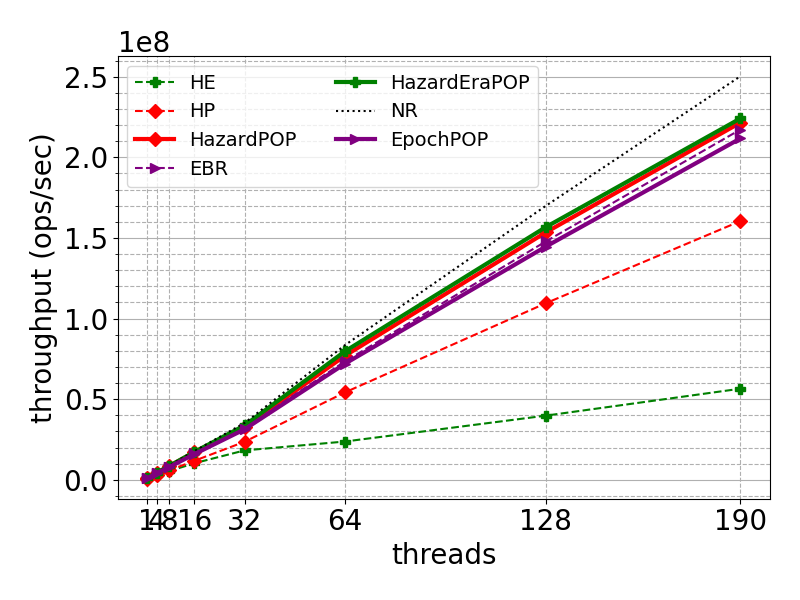}\hfill
            \caption{read heavy: 5\% inserts, 5\% deletes and 90\% contains.}
            \label{fig:exp10p}
        \end{subfigure}
     \end{minipage}
    \caption{EXP1: Throughput across different data structure sizes. Left: DGT, Center: Harris-Michael list (HML), right: Hash table with HML based chaining (HMLHT). Y-axis: throughput in millions of operations per second. X-axis: \#threads.}
    \label{fig:e1}
\end{figure}

We design our experiments with two objectives. First (EXP1, see \figref{e1}) is to evaluate the performance of our proposed publish-on-ping implementations HazardPointersPOP, HazardErasPOP and EpochPOP. Second (EXP2, \figref{pop-exp2}) aims to assess the peak memory consumption of publish-on-ping variants with and without stalled threads. For clarity, the publish-on-ping reclaimers are represented using solid lines in the plots.

\paragraph{Discussion.}

\begin{figure*}
    \centering
    \includegraphics[width=0.49\linewidth, keepaspectratio]{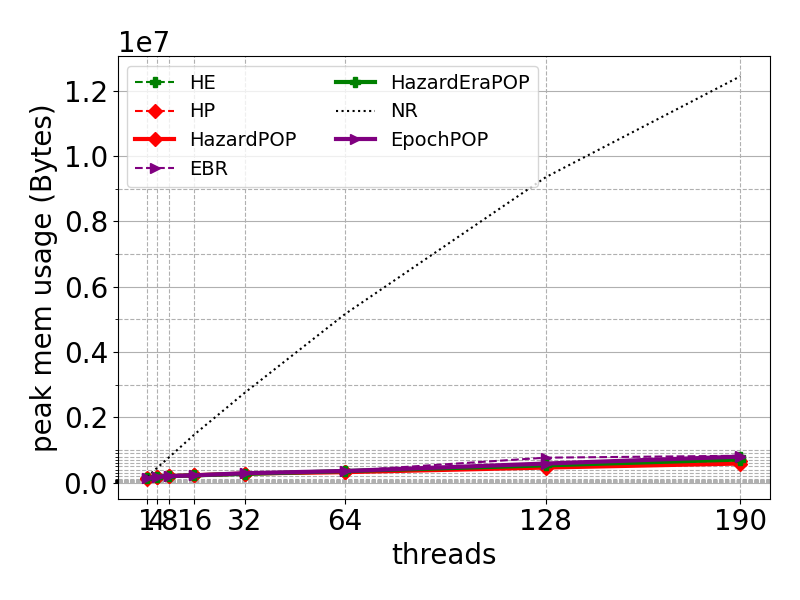}\hfill
    \includegraphics[width=0.49\linewidth, keepaspectratio]{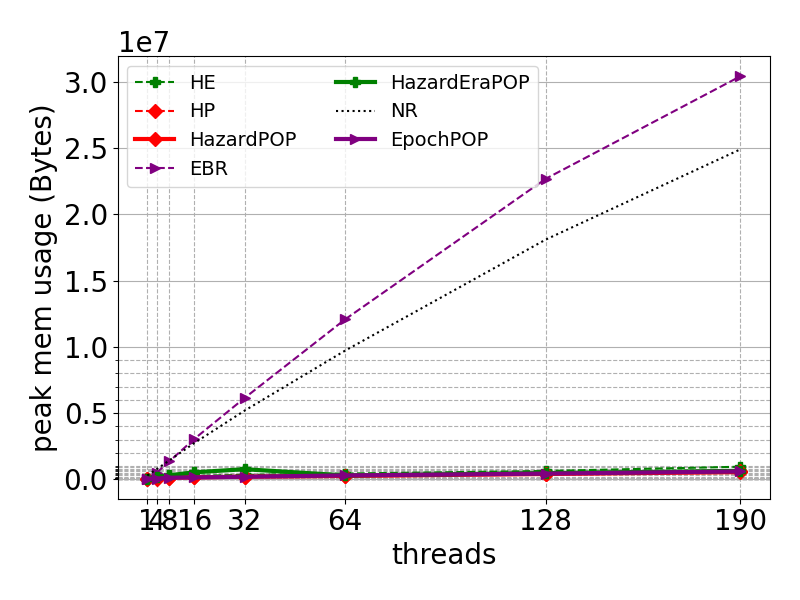}\hfill
    \caption{Exp2: Peak Memory Usage for DGT. Left: with no Stalled threads. Right: with stalled threads.}
    \label{fig:pop-exp2}
\end{figure*}

\figref{e1} depicts EXP1. We organize the plots in a grid where rows represent workload type and columns represent data structures. The top row shows throughput (operations per second) for the update-only workload, and the bottom row shows these data structures with read-intensive workload. The left column depicts DGT, the middle column is for HML and the rightmost column is for HMLHT.

In all data structures, the impact of reduction in overhead due to the elimination of publishing of reservation with every pointer access translates to improvement in overall throughput and scalability. For example, in \figref{exp50p}, HazardPointersPOP improves HP by up to 1.2$\times$ in DGT and HMLHT, while for HML we see an improvement of up to 3.8$\times$; HazardErasPOP improves HE by 1.5$\times$ in DGT and HML, whereas by 2.5$\times$ in HMLHT. In read-intensive workloads (\figref{exp10p}), the performance improvement is amplified, as the overhead of publishing reservations is reduced further with infrequent reclamation. Therefore, HazardHP improves HP by up to 1.4$\times$ in DGT and HMLHT, whereas for HML we see up to 5$\times$ improvement; HazardErasPOP improves HE by 2.4$\times$ in DGT and HML, whereas by 3$\times$ in HMLHT.

On average, the EpochPOP algorithm incurs up to 1.5$\times$ the overhead of EBR across both workloads for HML. This increase is primarily due to the additional \texttt{read()} function in EpochPOP, which keeps local reservations and adds to the already high pointer-chasing overhead in the lists. Additionally, the cost of robustness due to signaling also contributes to the overhead, both of which are absent in EBR. However, in data structures with minimal traversal overhead (low pointer chasing), such as DGT and HMLHT, EpochPOP performs similarly to EBR.

\begin{figure}
\centering
     \begin{minipage}{\textwidth}
        \begin{subfigure}{\textwidth}
            \includegraphics[width=0.33\linewidth, height=6cm, keepaspectratio]{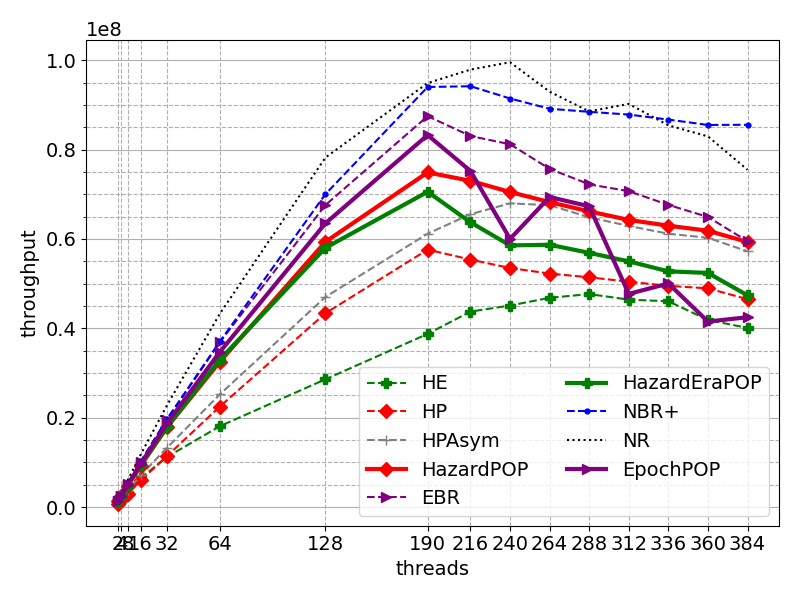}\hfill
            \includegraphics[width=0.33\linewidth, height=6cm, keepaspectratio]{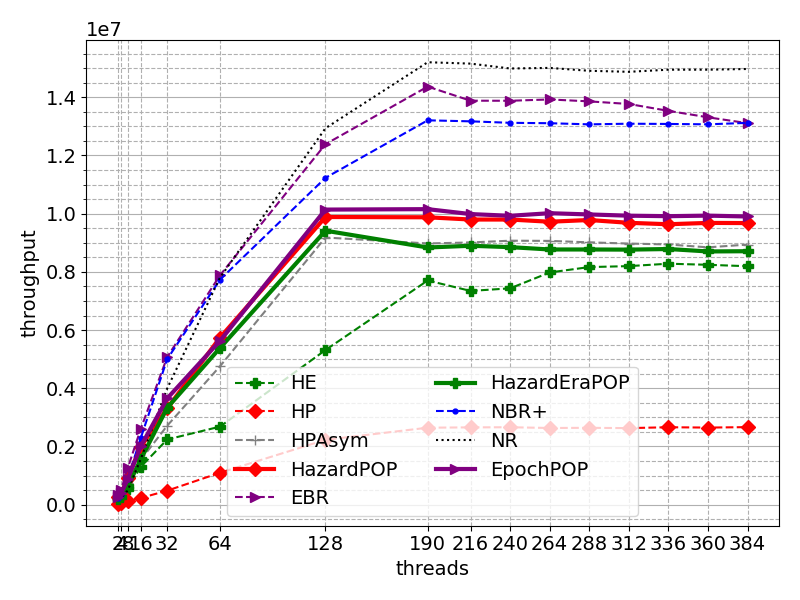}\hfill
            \includegraphics[width=0.33\linewidth, height=6cm, keepaspectratio]{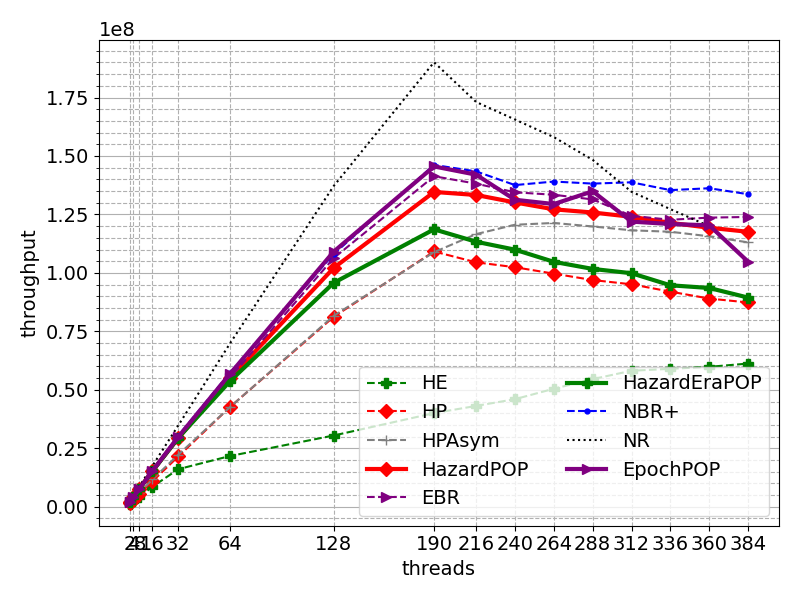}\hfill
            \caption{update heavy: 50\% inserts and 50\% deletes. }
            \label{fig:reb-exp50p}
        \end{subfigure}
        \begin{subfigure}{\textwidth}
            \includegraphics[width=0.33\linewidth, height=6cm, keepaspectratio]{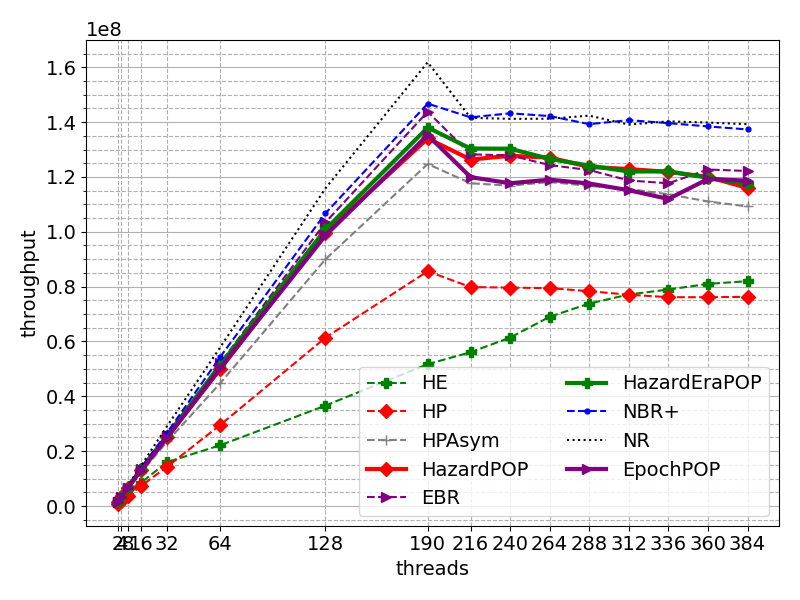}\hfill
            \includegraphics[width=0.33\linewidth, height=6cm, keepaspectratio]{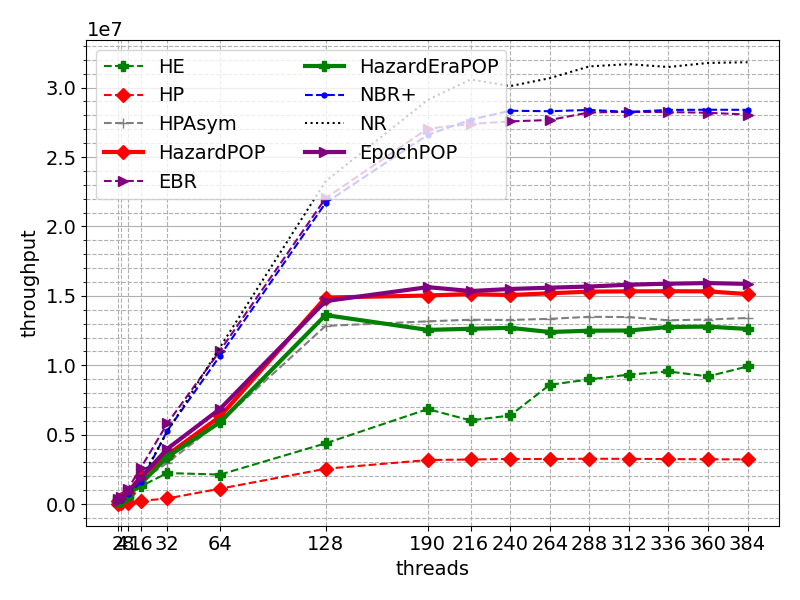}\hfill
            \includegraphics[width=0.33\linewidth, height=6cm, keepaspectratio]{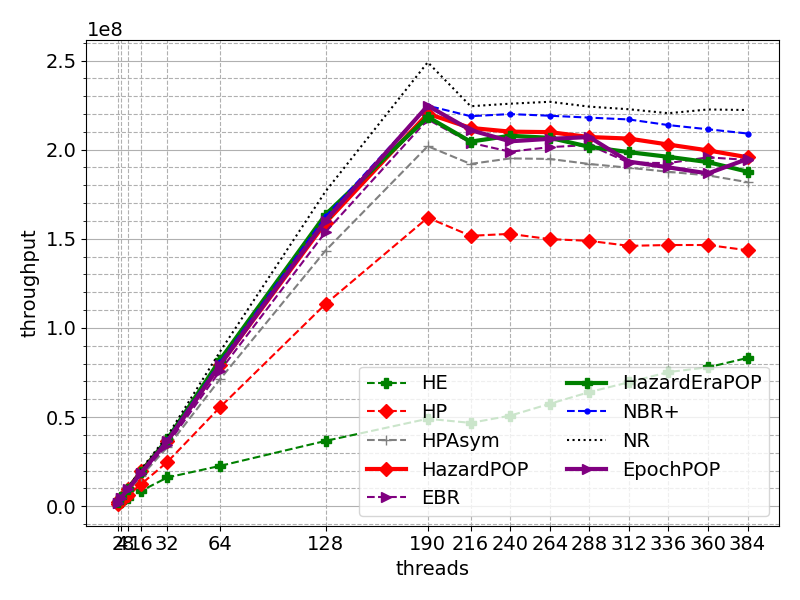}\hfill
            \caption{read heavy: 5\% inserts, 5\% deletes and 90\% contains.}
            \label{fig:reb-exp10p}
        \end{subfigure}
     \end{minipage}
    \caption{EXP1: Throughput across different data structure sizes. Left: DGT, Center: Harris-Michael list (HML), right: Hash table with HML based chaining (HMLHT). Y-axis: throughput in millions of operations per second. X-axis: \#threads.
    HPAsym: is the asymmetric fence based implementation of HP in Linux. NBR+: Neutralization-based reclamation. As shown in the plots, our HazardPOP technique is mostly faster than HPAsym.  }
    \label{fig:reb-e2}
\end{figure}

In \figref{pop-exp2}, we measure peak memory consumption of all the reclaimers without (on the left) and with (on the right) stalled threads. In the plot with stalled threads (on the right), we sleep one thread for the duration of the experiment to simulate a delayed thread. As expected, in absence of stalled threads (left plot), all reclaimers including EBR are able to reclaim memory regularly contributing similar low peak memory except the baseline NR, which never reclaims memory. On the other hand, in the presence of stalled thread the lack of robustness in EBR shows up as exponentially increasing peak memory usage, whereas publish-on-ping techniques maintain their lower peak memory usage, like other robust techniques.

\subsubsection{Additional Experiments with Oversubscription and Asymmetric Fences}

In this section, we evaluate the performance of POP algorithms against highly optimized hazard pointers with asymmetric fences, particularly under conditions of oversubscription. Figure ~\ref{fig:reb-e2} presents the EXP1 results, comparing Linux's asymmetric-fencing-based implementation of Hazard Pointers (HPAsym) with our NBR+ approach discussed in \chapref{chapnbrp}. The figure also illustrates the performance of POP-based techniques when oversubscribed (with 190 threads).

\section{Additional Related Work}
\label{sec:pop-relatedwork}

In this section we exclusively focus on the deferred pointer-based techniques which aim to eliminate or reduce the overhead on the traversal path.

Drop the Anchor (DTA)~\cite{braginsky2013drop} integrates epoch-based reclamation (EBR) and Hazard Pointers (HP). It relies primarily on EBR and enhances robustness against rare stalled threads by employing a recovery mechanism. This mechanism utilizes periodically published hazard pointers. During recovery, the process duplicates the range of nodes reachable from the stalled thread into the data structure, enabling other threads to resume reclamation without having to worry about freeing nodes reachable from the stalled threads.

Another approach assumes an alternative memory model called temporally bounded total store order (TBTSO)~\cite{morrison2015temporally} and applies it to HPs to guarantee that reserved pointers will be published to all threads within a bounded time, facilitating safe memory reclamation. In \cite{balmau2016fast}, Balmau et. al. employ context switches triggered by periodic scheduling of auxiliary processes per core to timely flush hazard pointers.
In these techniques, reclaimers wait for a set interval of time to pass since an node is retired to ensure that reservations to the node (if any) would be published, reclaiming the node if it is not reserved. 
These techniques incur overhead to periodically publish reservations, even when threads might not be reclaiming.
Dice et. al.~\cite{dice2016fast} advocate using the write-protection feature of memory pages that triggers a global barrier to facilitate timely publication of hazard pointers before threads reclaim. However, it could block threads trying to reserve a node if a reclaimer stalls after write protecting the page the node resides on.   



Several subsequent techniques have aimed to reduce or eliminate the need for memory fence during traversals in HP.
For example, Cadence~\cite{balmau2016fast} utilizes system-level memory fences triggered by context switches. 
This allows \textit{readers} to defer publishing reservations (using explicit memory fence) until a context switch occurs. \textit{Reclaimers}, to free a retired node, wait for a context switch to occur since the node was retired. This waiting period ensures that the reservations to the node, if any, become visible and the \textit{reclaimer} can free it if it is not found to be reserved.

Cadence employs auxiliary threads pinned to each core to force the context switches at regular intervals. These auxiliary threads compete with primary worker threads that execute data structure operations and assume that rescheduling always triggers a memory fence. This mechanism requires the instrumentation of data structure nodes with a timestamp to determine whether a global memory fence has occurred since the node was retired. The overhead resulting from auxiliary threads enforcing a global memory fence, though amortized over multiple operations, is incurred even if the threads do not reclaim.

Dice et al. in ~\cite{dice2016fast} discuss a technique to eliminate memory fences from traversals by leveraging write-protect feature in modern operating systems, which enforces a global memory barrier.
\textit{Reclaimers} briefly write-protect all memory pages associated with hazard pointers and then remove the protection, ensuring a global memory fence is executed. This makes all the hazard pointers visible before attempting to free retired nodes.
However, this technique requires all threads to block in an interrupt handler until the write protection is revoked. In the same paper, the authors also discuss a hardware extension suggesting the addition of a dedicated store buffer unit along with two new instructions to maintain hazard pointers.

\section{Summary}
\label{sec:popconclusion}

In this chapter, we have discussed the reactive synchronization paradigm that implements a publish-on-ping algorithm for reclamation, utilizing POSIX signals to expedite reservation-based techniques, such as hazard pointers and hazard eras. 
Unlike other signal-based techniques, such as Debra + and our previous technique NBR(\chapref{chapnbr}), our approach does not require a reclamation-triggered change in control flow in data structure operations. It can seamlessly serve as a drop-in replacement for hazard pointers\textendash a technique set to be included in the C++26 standard library.

Furthermore, we integrate epochs alongside the publish-on-ping variant of hazard pointers. The resulting algorithm, EpochPOP, not only significantly reduces the performance gap with epoch-based reclamation but also matches it in some cases. In the future, it would be worth exploring whether other synchronization problems exhibiting such an asymmetric safe memory reclamation-like pattern can leverage our approach.

Our publish-on-ping implementations of hazard pointers and hazard eras, applied to the Harris-Michael list, an external binary search tree, and a chaining hash table, demonstrate notable performance improvements, ranging from 1.2X to 5X and 1.5 to 3X faster (respectively) compared to the hazard pointer and hazard era implementations taken from the benchmark published with interval-based memory reclamation. EpochPOP is at most 1.6X slower for the list, but performs similarly to epoch-based reclamation for the tree and hash table.

\ignore{
\begin{itemize}
    \item Add parallel jemalloc experiments.
    \item Highlight Asymmetric fence exp and oversubscription experiments.
    \item In family of algorithms include IBR, QSBR as well.
    \item A pseudocode of usability. Example list and HP and PoPHP shown init. May be copy from NBR+ listing and modify for POP.
    \item Applicability add POP column too to talk about applicability. may be HP POPHP, EBR - POPEBR, HE POPHE etc.---> POPEBR can have wider app if HP++ technique is deployed. 
    \item MODEL: like the HE full paper, I should use reader for a thread that could deref a pointer. And separate it from writer and reader in NBR. More neat and natural copy of terms from literature.
\end{itemize}
}

\chapter{Hardware-Software Co-design Paradigm}
\label{chap:chapirp}


In the previous \chapref{chaprsp}, we presented the reactive synchronization paradigm to design safe memory reclamation algorithms that helps to eliminate the asymmetric synchronization overhead that existing reclamation algorithms incur on concurrent data structures.

In this chapter, we present the third paradigm of leveraging hardware-software co-design to safely reclaim memory in concurrent data structures. Essentially, it uses cache-level events that precede use-after-free errors and exposes these events to programmers through novel memory access instructions that programmers use to design their data structures. 
 
The resulting technique named Conditional Access is easy to program, incurs no extra space at the programming level, is comparable to well-known epoch-based techniques in terms of speed, and at the same time achieves immediate memory reclamation, enabling immediate memory reuse to keep memory footprint comparable to their sequential counterparts. In addition, the use of the proposed memory access instructions saves data structure operations from incurring costly cache misses, which are the main concern for the performance of concurrent software on modern architectures.

The outline of this chapter is as follows. 
First, we introduce the problem with current techniques and their impact on real systems in \secref{caintroduction}. This sets the stage for our observation about hardware-level events during reclamation, which motivated the proposal of the hardware-software co-design paradigm discussed in \secref{hwcodesign}. In \secref{hbrspec}, we present a high-level overview of our Conditional Access (CA) algorithm and its key components, along with a proposed hardware implementation.

In \secref{applyin-to-ds}, we explore how programmers can integrate Conditional Access instructions into optimistic data structures, using a stack and a lazy linked list as examples. This is followed by a proof of correctness and a discussion on the progress guarantees of Conditional Access in \secref{cacorr}.

In \secref{caeval}, we evaluate the performance and memory efficiency of Conditional Access by implementing it with multiple data structures and comparing it with several state-of-the-art reclamation algorithms. Additional related work is discussed in \secref{carelated}, with a summary of the chapter in \secref{casummary}.

\section{Introduction}
\label{sec:caintroduction}

\paragraph{Motivation.}

Current \smr(SMR)~\cite{brown2015reclaiming, michael2004hazard, alistarh2014stacktrack, alistarh2017forkscan, alistarh2018threadscan, balmau2016fast, cohen2018every, nikolaev2019hyaline, wen2018interval, detlefs2002lock, hart2007performance, nikolaev2020universal, singh21nbr} algorithms used in many optimistic data structures \textit{delay reclamation} and free nodes in batches to trade-off space in favor of high performance and safety.
When the batches are too small, the data structure's \textit{throughput} suffers due to more overhead from frequent reclamation. On the other hand, when the batches are too large, 
though the reclamation overhead is amortized due to reduced frequency of reclamation, 
the occasional freeing of large batches causes long program interruptions and  dramatically increases \textit{tail latency} for data structure operations~\cite{kim2024token}.%

Larger batch sizes also increase the \textit{memory footprint} of applications, which makes memory utilization and allocation challenging in virtualized environments~\cite{vmwareunderstanding}.
For example, increased memory footprints of virtual machines (VMs) or processes due to large batch sizes preclude the host machine from taking advantage of \textit{Memory Overcommitment} where the available dynamic memory could otherwise be shared amongst multiple VM instances (or other processes).

Besides requiring programmers to find an acceptable batch size (i.e., \textit{reclamation frequency}), most fast epoch-based or epoch reservation-based SMR algorithms \cite{brown2015reclaiming, nikolaev2019hyaline, wen2018interval, nikolaev2021crystalline} also have to determine an optimal increment frequency of a global timestamp (sometimes referred to as \textit{epoch frequency}). The values of these parameters in tandem influence the time and space efficiency of SMR algorithms.
Choosing an optimal value for these parameters can be quite challenging since they vary depending on the type of data structure, workload, and machine characteristics.

From a security perspective, the extended lifetime of nodes could be exploited for carrying out Denial of Service attacks, where many objects to be freed are accumulated, causing out-of-memory (OOM) errors\cite{MemDisagg19}. 
Such DoS attacks were reported in RCU implementations in the Linux kernel, where it was shown that a large number of deferred objects (or nodes) can be generated by a malicious user performing file open-close operations in a loop\cite{mckenney2006extending}.

In summary, the current paradigm of delayed reclamation (also known as deferred or batch reclamation) in non-blocking data structures has significant drawbacks, primarily due to the extended lifetime of retired nodes. Batch reclamation seems to be an unavoidable compromise to maintain the performance of the associated data structures. Without this approach, the overhead introduced by reclamation algorithms could significantly degrade performance, making them impractical. In this chapter, we explore whether it is possible to avoid deferred reclamation while still achieving high performance.

\subsubsection{Goal}

\textit{``To recycle retired nodes immediately, similar to sequential data structures, and yet achieve the same speed as state-of-the-art reclamation algorithms that perform batch reclamation."}   

\subsubsection{Observation}
We examine how a typical concurrent memory reclamation algorithm works on concurrent data structures using first principles.
For example, consider a lazy list\cite{heller2005lazy}.
Operations on a lazy list consist of a synchronization-free \emph{traversal phase} followed by an atomic \emph{update phase}.
During a delete operation, threads mark and then unlink the target node. They defer freeing the unlinked node by saving the node in the thread's retired list. This node cannot be reclaimed unless all threads coordinate to confirm that they have discarded their local references to the node, which is managed by the reclamation algorithm.

Fundamentally, a concurrent reclamation algorithm ensures that use-after-free errors do not occur. This typically requires all threads to coordinate. 
Threads that may have reference to a node learn whether it is safe to access, i.e., it has not been deleted since they last acquired the reference to the node. Alternatively, a thread attempting to reclaim a node learns that no other thread holds a reference to the node, allowing it to safely free the node. \textbf{We observe that this information can be inferred from the cache-level events that precede a store operation marking a node.}

Consider a hypothetical concurrent execution of threads $T1$ on core $C1$ and $T2$ on core $C2$ in a lazy list. Say, $T1$ executes \textit{delete(N)} and $T2$ executes \textit{lookup(N)}. 
Now, at the moment $T1$ arrives at N, $C1$ fetches N's address in its L1 cache in \emph{shared(S)} state.
Similarly, when $T2$ arrives at N, $C2$ also gets N's address in its L1 cache in the \emph{shared(S)} state.
Later, when $T1$ marks N for deletion, $C1$ upgrades N's cache line to \emph{modified(M)} state, \emph{invalidating(I)} cache lines containing N at other cores.
In this case, $T2$'s cache line containing N in \emph{shared(S)} state is \emph{invalidated(I)}. 
If $T2$ could somehow harness the cache invalidation message, it can learn that $T1$ has marked N for deletion. Therefore, $T2$ can simply avoid accessing the probably freed node N without needing redundant coordination between all threads at the application level to obtain the same information.

We introduce a straightforward hardware extension that captures cache-level events and makes information about potential use-after-free errors accessible to programmers through novel conditional memory access instructions. These instructions allow programmers to safely interact with shared nodes in their data structures. If a use-after-free error is detected, programmers, rather than simply incurring a segmentation fault, can decide how to handle it, such as by restarting the data structure operation. 
A detailed explanation of our method is provided in \secref{hbrspec}, followed by examples illustrating how programmers can use our technique with various classes of data structures.

\section{Hardware-Software Co-design Paradigm}
\label{sec:hwcodesign}
A \uaf error can be considered a special case of a read-write data race, where a shared memory location is accessed after it has been freed by a different thread.
In modern systems with coherent caches, \uaf errors are always preceded by events in the cache coherence protocol.
Consider a traditional MESI (Modified, Exclusive, Shared, Invalid) protocol: To store a value at a location \texttt{X} that is currently in the shared state, a core $C$ first invalidates copies of \texttt{X} at other cores by sending \textit{invalidation} messages to all other cores.
Upon receiving such a message, a core \textit{invalidates} its copy of that location and responds with an acknowledgement message.
Once $C$ has received acknowledgements from all other cores, it has exclusive access to \texttt{X}.
A thread that reads \texttt{X} after it is freed will respond to these invalidation messages before reading \texttt{X}.
This reveals that, at the level of the coherence protocol, readers are aware that the memory location they are trying to access may have been concurrently modified. 
Moreover, a subsequent read of \texttt{X} must begin with a cache miss\textemdash ~an avoidable overhead if the information about concurrent modification could be harnessed.

However, read-write data races and use-after-free errors are indistinguishable at the architectural level, complicating the identification of use-after-free errors based solely on cache coherence events. To address this, one could interpret invalidation messages as potential indicators of use-after-free errors, although with the possibility of false positives. By monitoring these messages and exposing them through specialized memory access instructions, we can facilitate safe memory reclamation.

The programmer in data structure operations can then use these memory access instructions to perform memory accesses that they suspect \textit{might} result in \uaf errors. Underneath these instructions, 
the hardware can \emph{tag} the corresponding cache lines, indicating to the hardware that \textit{invalidation} of such a cache line is an event of interest.
Similarly, reclaimers can do store before freeing a node, so they can be sure to trigger a cache event that revokes other threads' access to that node.
The reclaimer can then immediately free the node.
Any thread that has tagged this node before it was freed, and subsequently tries to access the tagged node, will observe that the corresponding cache line has been \textit{invalidated}, and is potentially freed. 

In summary, readers and reclaimers can leverage cache coherence events to prevent use-after-free errors. This approach, which integrates hardware and software mechanisms, is referred to as the hardware-software co-design paradigm for safe memory reclamation.

\section{Conditional Access}
\label{sec:hbrspec}
Inspired by the recent \textit{Memory Tagging} proposal of Alistarh et~al.~\cite{alistarh2020memory}, we propose a hardware mechanism called \textit{Conditional Access}, which tracks a set of addresses accessed by a thread and monitors invalidation messages for those addresses.
It allow readers to efficiently determine whether a node they are trying to access has been freed.
As we will see when we explain \ca in more detail, \ca offers cread(): a core instruction to our immediate memory reclamation technique which has no equivalent instruction in Memory Tagging. \ca, unlike Memory Tagging, has no explicit AddTag instruction.
From the perspective of implementation, \ca requires no changes to the underlying coherence protocol, whereas Memory Tagging’s Invalidate and Swap (IAS) Instruction requires changes to the coherence protocol as this single instruction can invalidate many (potentially non-contiguous) remote cache lines (potentially spanning many pages). 

While traditional SMR algorithms typically ensure that reclaimers defer reclamation until a node is entirely inaccessible to readers, our approach permits immediate freeing by reclaimers. This method sidesteps previously mentioned issues related to delayed reclamation, achieves an ideal memory footprint similar to that of sequential data structures by eliminating garbage nodes, and maintains high throughput.


\textit{Conditional Access} exposes the cache-level events that pertain to a potential use-after-free error to the programmers using new hardware memory access instructions. 
Specifically, programmers use a \textit{tag} instruction to start tracking a location for invalidation messages. A location being tagged is valid if it has not been deleted already. Validity of a node being tagged is established by verifying that a set of previously tagged locations have not been invalidated.
Subsequent accesses to the address use \textit{conditional read} (\texttt{cread}) or \textit{conditional write} (\texttt{cwrite}) instructions. These conditional instructions allow a thread to read or write a new location \textit{only if} a set of programmer defined set of \textit{tagged} locations
has not changed since they were previously read.


\textit{Conditional Access} is ideal for implementing data structures for which 
one can prove a read is safe if a small set of previously read locations have not changed since they were last read.
For example, in a linked list that sets a \textit{marked bit} in a node before deleting it, if a thread reads the next pointer of an unmarked node, and at some later time its marked bit and next pointer have not been changed, then it is still safe to dereference its next pointer.
In such a data structure, once a node is unlinked and marked, it can \textit{immediately} be freed, since doing so will merely cause subsequent \texttt{cread}s or \texttt{cwrite}s on the node to fail, triggering a \textit{restart} (an approach common to many popular SMRs~\cite{michael2004hazard, cohen2018every, wen2018interval, cohen2015automatic}).

\textit{Conditional Access} can be thought of as a generalization of load-link/store-conditional (LL-SC) where the load is also conditional, and the store can depend on many loads.
Whereas an LL effectively tags a location, and an SC untags the location,
in \textit{Conditional Access}, locations are not automatically untagged when they are written.
So, multiple \texttt{cread}s and \texttt{cwrite}s can be performed on the same set (or a dynamically changing set) of tagged locations.

\textit{Conditional Access} also has some similarities to a restricted form of transactional memory.
However, whereas hardware transactional memory (HTM) is increasingly being disabled due to security concerns, we believe \textit{Conditional Access} can be implemented more securely.
For example, since a thread becomes aware of concurrent updates to its tagged nodes only when it performs a \texttt{cread} or \texttt{cwrite} and then checks a status register, we can avoid some timing attacks that are made possible by the immediacy of aborts (as a result of conflicting access by the other threads) in current HTM implementations.

Much of the information needed to efficiently implement \textit{Conditional Access} is already present in modern cache coherence protocols.
We propose a simple extension where tagging is implemented at the L1 cache level without requiring changes to the coherence protocol. 
At a high level, each L1 cache line has an associated \textit{tag} (a single bit). 
This \textit{tag} is set by a \texttt{cread} on any location in that cache line, and unset by an \texttt{untagOne} on a location in that cache line or an \texttt{untagAll} instruction.

Each core tracks invalidations of its own tagged locations. 
In SMT architectures, where $k$ hyperthreads share a core, each hyperthread tracks invalidations of its own tagged locations.
The extensions we require to the cache, and between the cache and processor pipeline, are a strict subset of those needed to implement HTM. This strongly suggests that \ca implementations can be practical and efficient.

\ignore{
\ca enables memory footprints similar to those of sequential data structures.
This is desirable in modern data centers, to save costs related to memory over-allocation and to facilitate Memory Overcommitment~\cite{vmwareunderstanding}. 
Further, immediate reclamation can help in avoiding exploits that use the extended lifetime of unlinked objects in delayed reclamation algorithms to leak private data.
It also has the potential to prevent denial of service attacks in which threads induce a schedule that causes batches of unreclaimed memory to grow unboundedly, leading to out-of-memory errors.
Such attacks have been reported in RCU implementations in the Linux kernel~\cite{mckenney2006extending}.
}

\subsection{Design}
\label{sec:seccadesign}

This section outlines the memory access instructions that form the programming interface in \ca, along with the abstract structures that support it.

\paragraph{Additional Storage:}
\begin{itemize}
    \item [(A)] Each core \textit{tags} addresses it wishes to monitor for invalidation requests. These tagged addresses can be abstractly represented by \hbrset.
    \item [(B)] Each core also maintains an \hbrbit, initially clear, which is set when access is revoked for any address in its \hbrset. For simplicity, we assume that \hbrset's capacity is unbounded in this section. Efficiently approximating this set is the subject of Section~\ref{sec:hwimpl}.
\end{itemize}

\paragraph{Remote Events:}
For each entry in a core $C$'s \hbrset, the hardware is required to detect whether any other core has invalidated that cache line since $C$ tagged it.
If another core invalidates this cache line, 
the hardware must set $C$'s \hbrbit.

\paragraph{Memory Access Instructions:} 
\begin{itemize}
    \item [\textbf{(1)}] \textbf{\func{\crd addr, dest}:}
Similar to a \texttt{load} instruction, \crd updates register \texttt{dest} with the value at the address in register \texttt{addr}, but with two key differences: \textit{tagging} and \textit{conditional access}\footnote{For simplicity of
  presentation, we do not parameterize \crd by the number of bytes to read
  from memory, or consider different addressing modes.  In a practical
  system, several opcodes will be needed for these purposes.}. 
More specifically, \crd atomically
checks if \texttt{addr} is in \hbrset, and if not, adds it to \hbrset. 
It also checks if \hbrbit is set, and if so, skips the load, and updates some other processor state, such as a flag register, to indicate that there may have been a \uaf error.
In this case, we say the \crd{} has \textit{failed}.
Otherwise, it loads the value at \texttt{addr} into \texttt{dest}, indicating that the memory access was safe.
In this case, we say the \crd{} has \textit{succeeded}. 

\item [\textbf{(2)}] \textbf{\func{\cwr addr, v}:}
Unlike \crd, 
\cwr does not update \hbrset.
Atomically: 
\cwr checks if the \hbrbit is set or \texttt{addr} is not in the \hbrset, in which case the store is skipped and a processor flag is set to indicate that the \cwr has \textit{failed} (suggesting there may have been a \uaf error).
Otherwise, it stores \texttt{v} at \texttt{addr}, and we say the \cwr has \textit{succeeded}.

It is worth discussing here why \cwr fails when it executes on an \texttt{addr} which is not in the \hbrset.
This design decision rules out uses where programmers may invoke \cwr before invoking a \crd (or in other words before first tagging a location). This helps to avoid tagging during a \cwr, which could incur significant delays if the access misses in the L1 Cache, making it easier to avoid tricky time-of-check to time-of-use (TOCTOU) issues. In particular we would prefer to avoid scenarios where a \cwr misses in the L1, waits for the data, and takes exclusive ownership of the line, only to discover that the \hbrbit has been set during the wait (e.g., because some of the other previously tagged cache lines could have been invalidated), thus eventually failing the \cwr. By requiring \crd to be performed first, we move the high latency parts of this operation into a shared mode access, potentially reducing invalidations and coherence traffic.

\item [\textbf{(3)}] \textbf{\func{\utag{addr}}:}
The \utag{} instruction does not access memory.
Its purpose is to allow the programmer to remove an address from the \hbrset.
If \texttt{addr} is not in \hbrset, \utag{} has no effect. 
Once an address is removed from a core's \hbrset, subsequent remote invalidations of the address will \textit{not} set the core's \hbrbit.

\item [\textbf{(4)}] \textbf{\func{\utagall}:}
\utagall 
clears the \hbrset and unsets the value of \hbrbit.
It is intended to be used in two cases: (1) when a \crd{} or \cwr{} fails, at which point a data structure operation will need to be retried; and (2) before returning from a successful data structure operation.

Note, for SMT architectures with multiple hardware threads \textit{additional storage} and \textit{remote events} are required per hardware thread, instead of per core.

\end{itemize}
\subsection{Hardware Implementation}
\label{sec:hwimpl}
\ca can be implemented by a straightforward extension of existing caches, such that modifications are only introduced between a processor and its primary cache, e.g., the L1 data cache.
Based on our prototyping on a multicore simulator, we believe these changes are a strict subset of those required to implement HTM, which implies \ca is practical and efficient to implement.

Implementing \ca requires integrating \hbrset and \hbrbit as hardware components.
The proposed instructions use those data structures to track relevant invalidation messages, which are generated by the underlying cache coherence protocol.  

\begin{itemize}
    \item [(A)]The \hbrset can be approximated by adding one \texttt{tag} bit to each cache line of a core's L1 data cache. 
This is similar to how hardware transactional memory approximates its read and write sets.
\item [(B)] The \hbrbit, which tracks the invalidations of the addresses in a thread's \hbrset, requires adding one bit for each core.
One way the \hbrbit could be implemented is by adding it to the condition code or flag registers of the host architecture (e.g., EFLAGS on x86).
\end{itemize}

Note, in SMT architectures, where k hyperthreads share a core, each hyperthread will track which of its cache lines are tagged and track invalidations of its tagged locations. For instance, on a 2-way SMT architecture, two \emph{tag} bits and two \hbrbit{s}, one for each hardware thread, will be required.

Given these changes, we can now harness the cache coherence protocol to detect unsafe accesses.
When a \crd adds an address to the \hbrset, it loads that line into the
cache and sets the \textit{tag bit} for that line. 
The subsequent departure of the cache line from the cache could indicate a potential use-after-free error.
%
There are two ways in which the line can subsequently depart the cache:
remote invalidation or a local associativity conflict. 
In either case, the cache must notify the hardware thread that its \hbrbit
must be set, so that its subsequent \crd{} or \cwr{} will fail. 
For remote invalidations, doing so must be atomic with acknowledging the remote request.
For associativity conflicts, doing so must be atomic with fetching new data
from the memory hierarchy.
There are cases when \crd{} and \cwr{} can fail spuriously. We discuss these cases later to keep the exposition simple.
The atomicity requirements for \utag and \utagall are simpler: they cannot
be reordered with respect to loads and stores by the same hardware thread.
Furthermore, \utagall must clear the \hbrbit for future operations.

Besides the aforementioned two ways, in SMT architectures, a thread's \hbrbit can be set upon a write to any of the \emph{tagged} shared cache lines by another thread (in the case of hyperthreading) or on a context switch.
Setting the bit on a context switch is more straightforward to implement since it enables the operating system to avoid keeping track of invalidations on behalf of switched-out thread. 
These properties provide a foundation for the \ca to be used in multiuser systems.

Intuitively, tagging in \ca facilitates a kind of local protection of shared memory locations that does not trigger any additional coherence traffic. This is contrary to popular paradigms like hazard pointers~\cite{michael2004hazard} or other reservation-based~\cite{wen2018interval} techniques, which always trigger global cache traffic between threads.

The \hbrset size is bounded by the associativity of the cache and therefore \hbrset could overflow.
This would lead to eviction of tagged addresses (in \hbrset), causing \hbrbit to be set.
This, in turn, could lead to spurious failures of subsequent \crd{s} or \cwr{s} which could stall progress.
However, in practice it is not an issue because in most cases the \hbrset is small.
Our experiments (\secref{caeval}) show associativity does not have any significant impact on progress for the workloads we consider.

\section{Usage Requirements and Applicability}
\label{sec:applyin-to-ds}
In this section we discuss how \ca can be used
to achieve safe memory reclamation with optimistic data structures such as lists~\cite{heller2005lazy} and external binary search trees~\cite{ellen2010non}.
Operations of many such data structures have a search phase consisting of multiple reads, wherein a thread continuously traverses the next fields of nodes until it has visited a set of nodes it is interested in, where the operation eventually takes effect.
After reaching the nodes of interest, the operation may perform zero 
or more writes.
For example, in a linked list, a thread might traverse multiple links to find a predecessor and current node where an operation should take effect.

For ease of exposition we assume that each node fits in a single cache line, and a cache line contains only one node.
Thus, adding a node to a core's \hbrset implies adding a cache line containing the node to the core's \hbrset.
\subsection{Requirements}
We start by stating the following high-level directives required for all data structures to be able to correctly use \ca.


(DI) \textbf{Replace and Analyse:} 
\textit{Replace}: all read/write accesses to nodes that can be freed should be substituted by the corresponding \crd{} and \cwr{} instructions. This enables \ca to tag a node and monitor it for concurrent modification and notify programmers by updating a flag register.
\textit{Analyze}: If a \crd{} or \cwr{} fails, the operation should immediately \utagall and retry. A failed instruction implies a node could have been concurrently freed, therefore any future access will not be safe. 
Note, a \crd{} and \cwr{} could also fail spuriously due to various hardware reasons, such as cache eviction due to cache associativity. A more detailed discussion appears where we discuss progress in \secref{cacorr}.

(DII) \textbf{Validate Reachablility:}
A node is tagged when it is first \crd{}. 
In order to ensure that the tagged node is valid, it should be verified it was reachable in the data structure after the fact.

We now demonstrate how these directives can be applied to use \ca in different classes of optimistic data structures. Depending upon the data structures, DI could be partially relaxed, as we will see in the example of a lazy list or DII may not be needed as we will see in the example of a lock free stack. The lazy list requires some more rules which are detailed in the \secref{ll}.

\subsection{Usage in Data Structures with Single Writes}
Data structures with a single write in their update phase include some list based stacks~\cite{treiber1986systems} and queues~\cite{michael1996simple}, both of which we have implemented.
For the purpose of illustration, we will consider a list-based unbounded lock-free stack. In such a stack, a \func{push} operation involves reading a \texttt{top} pointer, allocating a new node for a key value to be pushed, and then doing a Compare-and-Swap (CAS)
to set the node as new \texttt{top}. Likewise, a \func{pop} operation consists of reading the \texttt{top}, and then atomically setting the \texttt{top} to its next node. After a \func{pop}, the unlinked node cannot be freed if a concurrent thread might still access it.


The original operations of the stack could be upgraded to enable \ca by simply replacing every read with \crd{} and the CAS with \cwr{} (DI). Then the \func{pop} operation could immediately free the unlinked node as shown in  \algoref{stack}.
Note, in our pseudocode \texttt{CAFAIL} is set when a \crd{} or \cwr{} fails. This is similar to updating a flag register.

\noindent
\textbf{Linearizability of the upgraded operations}.
The correctness follows from the fact that the \texttt{top} is read using \crd{}, which adds it to the corresponding thread's \hbrset, upon which the thread starts monitoring for any subsequent modifications to \texttt{top}.
Since the top itself is never deleted it is guaranteed to be always in the data structure at the time it is added to the \hbrset (DII). 
At the beginning of an operation the \hbrbit is clear and it is only set when the thread receives an invalidation request for the \texttt{top} when it is modified elsewhere. 
Later, the thread attempts to change the \texttt{top} using \cwr{} which atomically checks the \hbrbit for any interfering memory access. 
It fails if the \hbrbit is set.  This causes the thread to remove the top from its \hbrset, clear the \hbrbit using \utagall, and then retry the operation.
Otherwise, the thread succeeds by changing \texttt{top} to another node.
The push and pop operations can be linearized on the successful \cwr{} at \lineref{pushlp} and \lineref{poplp} in \algoref{stack}, respectively.
\begin{algorithm}[!htbp]
    \caption{Using \ca with unbounded lock free stack. Lines annotated as LP are the linearization points. 
    }
    \label{algo:stack}
    \footnotesize
    \begin{algorithmic}[1]
        \State type Node \{Key key, Node *next\}
        \State class Stack \{Node *top\}
        \State \#define \texttt{CA\_CHECK} \IfThenNoS{\texttt{CAFAIL}}{\texttt{untagAll(); goto retry;}}
        \Statex
        \Procedure{push}{key}
        	\State newtop = new node(key);
            \State retry:
            \State t $\leftarrow$ \crdp{top}; \texttt{CA\_CHECK}
            \State newtop->next = t;
            \State \cwrp{\&top}{newtop}; \texttt{CA\_CHECK} \label{lin:pushlp}\Comment{LP}
        \EndProcedure
        \Statex
        \Procedure{pop}{ }
            \State retry:
            \State t $\leftarrow$ \crdp{top}; \texttt{CA\_CHECK}
            \If {NULL == t}
                \State \texttt{untagAll(); and return;}
            \EndIf
            \State \cwrp{\&top}{t->next}; \texttt{CA\_CHECK} \label{lin:poplp}\Comment{LP}
            \State free(t) \label{lin:stackfree}
        \EndProcedure
    \end{algorithmic}
\end{algorithm}

Note, the call to free at \lineref{stackfree} is safe because
whenever a core C1 modifies the \texttt{top} (either for push or pop), any other core C2 having access to \texttt{top} will fail its \cwr{} because C1 will invalidate C2's tag by setting its \hbrbit.


\ca is ABA-safe despite the fact that it allows immediate reuse of freed objects.
Suppose a thread T1, in order to insert a new node, reads an address A from the \texttt{top} into a local variable $t$.  Then just before it executes a CAS to set \texttt{top} to $t$'s next, some other thread T2 removes A by setting a node at address B as the new \texttt{top}, frees A, and then pushes a new node at this recycled address A, making it the new \texttt{top}.
Now T1 would succeed its CAS (based on address comparison) as the \texttt{top} still contains the address A, which matches the expected address stored in its local variable $t$.  Thus it incorrectly succeeds when it should have failed\textemdash a typical ABA case.
\ca prevents this error as \cwr{}, unlike a CAS, is not based on comparing two values. Instead, it relies on the underlying cache invalidation messages to detect that a location has been modified since it was last read.




\subsection{Usage in Data Structures having Multiple Writes with Locks}
\label{sec:ll}
Another category of linked concurrent data structures have update operations wherein threads optimistically traverses a sequence of nodes ending in multiple updates within a critical section guarded by locks. One example is the lazy list~\cite{heller2005lazy}.

Operations of data structures with such design patterns could be upgraded to enable the proposed technique's \smr using the following broad guidance:
\begin{enumerate}
    \item In the search phase, use DI to replace all reads with \crd{}. Use \utag{} to remove previously traversed nodes from current thread's \hbrset when they are no longer required to prove that a node, to be accessed in the future, is reachable in its data structure at the time it is tagged.
    If a \crd{} fails during the traversal then do \utagall{} and retry the search.
    
    For read only operations this will suffice. Update operations require  the following steps:
    \item Use try locks 
    designed using \crd/\cwr{} (\algoref{lock}) to lock all the nodes identified at the end of the search. This marks the beginning of a critical section to execute the updates atomically. 
    If lock acquisition fails on any of the nodes then unlock the previous nodes (if any), do \utagall and retry the operation. 
    \item Within the critical section use normal writes to execute intended updates.
    This is safe because the nodes are guarded by the critical section and therefore cannot be concurrently updated or reclaimed (partial relaxation of DI). 

    \item If the update is a delete, mark the node before unlinking it. This helps reading threads to satisfy the requirement in DII, i.e. they do not tag a node which has been already unlinked from the data structure. 
    
    \item Finally, unlock any locked nodes and execute \utagall before exiting the operation. Note that unlock may use regular stores instead of \cwr{},  since locked nodes cannot be freed by other threads.
\end{enumerate}

By the way of example of a lazylist (in reference to \algoref{list}) we will demonstrate how we can easily upgrade it to use \ca.

Using D1, all the regular reads are replaced by \crd{s} in searches, as shown in \algoref{list}.
As explained in the specification, the \crd{s} atomically: add a node to current thread's \hbrset, if it is not in it already, check that the \hbrbit is clear, and complete a normal read, if the condition succeeds.
However, if the \crd{} fails (when \texttt{CAFAIL} is set) then it could be the case that a subset of the nodes in the \hbrset have been modified (potentially deleted) since they were last accessed, therefore it may not be safe to read them. 
In such a case the \hbrset is emptied, the \hbrbit is unset using \utagall, and the search is retried. Otherwise it continues and eventually stops when some \pred and \curr nodes of interest are found (\lineref{locateret} in \algoref{list}).



Note, if we do not untag previous nodes during searches, then since \crd{s} tag nodes, we will have all the nodes in the search path added to a thread's \hbrset.
This could cause \crd{s} to fail when any node (relevant to current access or not) in the search path is modified, which forces operations to retry repeatedly. In other words, certain updates will be serialized, as if threads acquired a global lock, which will inhibit concurrency.
As a remedy to this problem, threads can untag previous nodes using \utag and are only required to keep two consecutive nodes tagged at any given time, which is equal to the number of nodes required to carry out updates, much like hand-over-hand locking.

Furthermore, in order to guarantee searches are safe we need to ensure that at the time a node is tagged it is reachable in the list 
(DII). To see why, assume a case where a thread accesses content of an arbitrary node using \crd{}. During this \crd{}, atomically: a cache line containing the node will be tagged and then its content will be loaded. Now, if the node was already marked before it was tagged by the \crd{}, then a subsequent \crd{} would succeed even though the node is marked (logically deleted), which is not safe as the node could be reclaimed (a \uaf error). 
This is resolved by validating that a node is not marked, immediately after the \crd that tagged it.
If the node is found to be marked, validation fails and the corresponding operation untags all nodes and retries. 
For example, in the lazy list a node is first tagged during \crd at \lineref{tagged} in \algoref{list}, due to a \func{validate()} invoked from \lineref{tagged2},
~\ref{lin:tagged3}, or~\ref{lin:tagged4}. If \func{validate()} returns False due to the node being marked then the operations untags all nodes and retries.
This way DII is satisfied.

\begin{algorithm}[t]
    \caption{\ca based lock. Precondition: the node containing the lock field should be \crd{} so that it is tagged.}
    \label{algo:lock}
    \footnotesize
    \begin{algorithmic}[1]
        \Procedure{tryLock}{bool *lock}
            \State lockVal $\leftarrow$ \crdp{lock};
            \IfThen{\texttt{CAFAIL} or 1 $==$ lockVal}{return False;}
            \Statex
            \State \cwrp{lock}{1};
            \IfThenNoS{\texttt{CAFAIL}}{return False;}
            \State return True;  
        \EndProcedure

        \Statex

        \Procedure{unlock}{bool *lock}
            \State *lock $\leftarrow$ 0; \Comment{safe as a node can only be mutated by owner.}
        \EndProcedure
    \end{algorithmic}
\end{algorithm}

\begin{algorithm}
\small
    \caption{Using \ca with lazylist~\cite{heller2005lazy}. 
    Lines annotated as LP are the linearization points.
    }
    \label{algo:list}
    \begin{algorithmic}[1]
        \State type Node \{Key key, lock, mark, Node *next\}
        \State class Lazylist \{Node *head\}
        \State \#define \texttt{CA\_CHECK} \IfThenNoS{\texttt{CAFAIL}}{\texttt{untagAll(); goto retry;}}
        \Statex
        \Procedure{validate}{Node *node} \label{lin:validatep}
            \State isMarked $\leftarrow$ \crdp{node->mark}; \IfThenNoS{\texttt{CAFAIL}}{return False;} \label{lin:tagged}
            \IfThenElse{\texttt{isMarked}}{return False;}{return True;}
        \EndProcedure
        \Statex
        \Procedure{locate}{Key key}
            \State retry:
            \State pred $\leftarrow$ \crdp{head}; \texttt{CA\_CHECK} \label{lin:tagged1}
            \State \Call{validate}{pred}; \IfThenNoS{$False$}{\texttt{untagAll(); goto retry;}} \label{lin:tagged2}
            \State curr $\leftarrow$ \crdp{pred->next}; \texttt{CA\_CHECK} 
            \State \Call{validate}{curr}; \IfThenNoS{$False$}{\texttt{untagAll(); goto retry;}} \label{lin:tagged3}
            \State currkey $\leftarrow$ \crdp{curr->key};  \texttt{CA\_CHECK}
            \While{currkey < key}
                \State \utagp{pred};
                \State pred $\leftarrow$ \crdp{curr}; \texttt{CA\_CHECK}
                \State curr $\leftarrow$ \crdp{curr->next}; \texttt{CA\_CHECK}
                \State \Call{validate}{curr}; \IfThenNoS{$False$}{\texttt{untagAll();goto retry;}} \label{lin:tagged4}
                \State currkey $\leftarrow$ \crdp{curr->key};  \texttt{CA\_CHECK}
            \EndWhile
            \State return $\langle$ pred, curr, currkey $\rangle$; \label{lin:locateret}
        \EndProcedure
        \Statex
        \Procedure{contain}{key}
            \State $\langle$ pred, curr, currkey $\rangle$ $\leftarrow$ locate(key); \Comment{LP:When currkey was read.}
            \State \utagall{\texttt{()}};
            \State return (currkey == key);
        \EndProcedure
        \Statex        
    \algstore{ca-part1}
    \end{algorithmic}
\end{algorithm}
        
\begin{algorithm}
\small
    \caption{\algoref{list} continued.
    }
    \begin{algorithmic}[1]
    \algrestore{ca-part1}
        \Procedure{insert}{key} \label{lin:insproc}
            \State retry:
            \State $\langle$ pred, curr, currkey $\rangle$ $\leftarrow$ locate(key); \Comment{$LP$ when insert fails.}
            \IfThen{currkey == key}{\utagall{}; \& return False;}
            \If{False == tryLock(\&pred->lock)} \Comment{attempt locking pred.} \label{lin:ilockpred}
                \State \utagall{} \& retry;
            \EndIf
            \If{False == tryLock(\&curr->lock)} \Comment{attempt locking curr.} \label{lin:ilockcurr}
                \State unlock(\&pred->lock);
                \State \utagall{} \& retry;
            \EndIf
            \State node $\leftarrow$ new Node(key, curr);
            \State pred->next $\leftarrow$ node; \Comment{$LP$ when insert succeeds.} \label{lin:insadd}
            \State unlockAll() and \utagall{}; 
            \State return True;
        \EndProcedure
        \Statex
        \Procedure{delete}{key} \label{lin:delproc}
            \State retry:
            \State $\langle$ pred, curr, currkey $\rangle$ $\leftarrow$ locate(key); \Comment{$LP$ when delete fails.}
            \IfThen{currkey != key}{\utagall{}; \& return False;}
            \If{False == tryLock(\&pred->lock)} \label{lin:dlockpred} \Comment{attempt locking pred.}
                \State \utagall{}; \& retry;
            \EndIf
            \If{False == tryLock(\&curr->lock)} \label{lin:dlockcurr}\Comment{attempt locking curr.}
                \State unlock(\&pred->lock);
                \State \utagall{}; \& retry;
            \EndIf
            \State curr->mark $\leftarrow$ true; \Comment{$LP$ when delete succeeds.} \label{lin:delmark}
            \State pred->next $\leftarrow$ curr->next; 
            \State unlockAll(); \& \utagall{}; 
            \State \texttt{free(curr)};
            \State return True;
        \EndProcedure        
    \end{algorithmic}
\end{algorithm}        

One may further ask, what if the node was marked (already logically deleted) and also freed before it is tagged? In that case a subsequent \crd{} could succeed as its \hbrbit will not be set since no update will occur after the node was tagged. This could cause a \uaf error.
However, this cannot happen because, in order to free the node, a reclaimer has to unlink it by modifying the next field of its predecessor, which is already in the thread's \hbrset. Thus, if the predecessor node is modified the thread's \hbrbit will be set and the \crd will fail, preventing unsafe access. This invariant is maintained during a search that eventually yields a \pred and \curr that were reachable in list at the time they were tagged.

Later, before starting the updates, locks on the \pred and \curr nodes are acquired (\lineref{ilockpred} \&~\ref{lin:ilockcurr} for \func{insert()} and \lineref{dlockpred} \&~\ref{lin:dlockcurr} for \func{delete()}). However, it may happen that after the search returns the nodes and before the locks are acquired some thread may delete these nodes. In that case if the lock is accessed with normal reads and writes then the thread may attempt acquiring lock on a freed node which could lead to undefined behaviour: unlike \crd{s}, regular reads do not have the ability to check whether the nodes have been modified. Thus, to resolve this issue we provide \crd/\cwr{} based try locks which only acquire the lock on a node if it has not been modified (deleted) concurrently.

\algoref{lock}
depicts the implementation of this lock.
It has a precondition that the node containing the lock field should have been previously accessed using \crd{} so that it gets added to its thread's \hbrset{}, enabling a \crd/\cwr{} to verify through \hbrbit whether the node has been modified since then. 
In further detail, a thread does a \crd{} on the lock variable. If it sets \texttt{CAFAIL}, the node of which the lock is part might have been deleted; if it returns $1$, it means that lock is busy. In both the cases, the lock acquisition fails. Otherwise, a thread proceeds to acquire the lock by setting the lock field to 1 using a \cwr{}, which again checks if the node containing the lock field has not been modified (possibly deleted). 
If the check succeeds it writes 1 to the lock field and returns \texttt{True}, indicating that lock acquisition is successful. Otherwise if the \cwr{} fails (by setting \texttt{CAFAIL}) it returns \texttt{False} indicating that lock acquisition has failed, and the operation which invoked the lock untags all nodes and retries.

The insert operation(\lineref{insproc}, \algoref{list}), first executes \texttt{locate}, which returns tagged \pred and  \curr nodes along with \texttt{currkey} (key field of \curr). If the key to be inserted is already present in the list then the operation returns false. Otherwise, the key is not present and needs to be inserted.
To insert the key, first the \ca based trylocks on the \pred and \curr nodes are acquired, then a new node is created and inserted between pred and curr. Following that all locks are released, nodes are untagged, and then the operation returns true.

Note, because the \pred and \curr nodes are already locked, no other thread could ever modify them without acquiring lock first. Therefore, validation to check whether the node has been concurrently freed is not needed. This allows us to use normal reads and writes instead of \crd/\cwr{} within the critical section. 

The delete operation(\lineref{delproc}, \algoref{list}) invokes \texttt{locate}, which returns tagged \pred and \curr nodes along with currkey. If the key to be deleted is not present in the list then the operation untags the nodes and returns false.
Otherwise, similar to the insert operation, it acquires the trylocks on both nodes, does a write on the \curr node to set its \texttt{mark} field, unlinks the \curr node, unlocks and untags both the nodes, frees the \curr node and then returns true.


\subsection{Usage in lock-free data structures with multiple writes }


The main complication of using \ca with lock-free data structures with multiple writes is that if used naively 
an aborted data structure operation could leave the data structure in an inconsistent state.
This would happen if earlier writes succeed, a later write fails, and the operation then attempts to untagAll and retry. 
It would require a more complex variant of the \cwr{} primitive that allows validating and updating multiple locations at the same time.
Supporting \ca based immediate reclamation for these data structures may seem to be a natural extension of the current work.  
While the hardware required for the primitive that operates atomically on multiple locations will likely overlap with that of hardware transactional memory, ensuring both progress and correctness would require careful design, so we leave this extension beyond the scope of the present work.

However, some relatively simple lock-free data structures with multiple writes, such as the Harris list~\cite{harris2001pragmatic} or the Harris-Michael list~\cite{michael2004hazard} can potentially be modified to use \ca. For example, in the Harris list, while deleting, if after marking the unlinking step fails, then instead of retrying one could exit the operation. Later, any search could unlink the node as it traverses the list. Similarly, in the Michael\&Scott queue~\cite{michael1996simple}, the second CAS to set the tail is a clean-up CAS. If it fails, instead of retrying, a thread can exit the operation and later enqueues will fix the tail pointer, so \ca appears to be applicable.


\section{Correctness}
\label{sec:cacorr}

If all the aforementioned rules are followed to enable \ca in the lazylist then the list is linearizable and all access in it are safe. \texttt{Contains} or unsuccessful \texttt{inserts} and \texttt{deletes}, which behave like \texttt{contains}, can be linearized at the time when the key of \curr was read. Whereas, successful \texttt{inserts} and \texttt{deletes} can be linearized when a new node is linked (\lineref{insadd}) or a node is marked (\lineref{delmark}), respectively. Also, because \crd{s} never dereference an unreachable (unlinked) node, \uaf errors do not occur. Therefore the lazy list with \ca is safe. The following section discusses the correctness in detail.


\begin{assumption}
List operations are implemented using the provided guidelines to enable \ca based immediate reclamation.
\begin{enumerate}
    \item Replace and Analyse (DI) rule is followed. 
    \item Validate Reachability (DII) is followed. In order to ensure that only reachable and unmarked nodes are tagged, marked field of a node is validated after it is tagged.
    \item In updates, normal locks are replaced with \ca based try locks.
    \item Before freeing a thread must write on the node (i.e update the mark field).
\end{enumerate}
\label{asm:wellbehave}
\end{assumption}

For \ca based programs to be correct, all data structure operations should do shared memory reads and writes using \crd{} and \cwr{}, a thread is required to write on a node before freeing it, and when a node is tagged it should be reachable and unmarked. 
Additionally, for \textbf{DI} to work we assume that the underlying hardware provides cache coherency mechanism like MSI, MESI or other such equivalent mechanisms. 


\begin{property}
\label{prop:cAccessAtomic}
A node is tagged the first time its content is \crd{} and it stays tagged until explicitly untagged using \utag{} or \utagall.
\end{property}

\begin{property}
\label{prop:writeinvalidates}
When any field of a node is written the underlying cache coherency mechanism invalidates all the cached copies of the location at other threads by setting their \hbrbit to true, which stays true unless unset by \utagall.
\end{property}


\begin{claim}
\label{clm:cAccessSucc}
A conditional access i.e. a \crd or \cwr on a node only succeeds if none of the tagged nodes have been invalidated since they were tagged.
\end{claim}

\begin{lemma}
\label{lem:correcttag}
A node's predecessor is unmarked and its next field points to the node at the time it is tagged.
\end{lemma}
This implies that a thread will never \crd{} next field (or discover new nodes) from a marked node. 
\begin{proof}
We prove the lemma by induction on the sequence of nodes in the \ca based lazylist. \\
\noindent
\textbf{Base case:} Initially when the search starts it does a \crd on \head, thus loads the address of head node in a local pointer \pred. 
Subsequently, the head node's mark field is \crd{} which atomically: tags the head node and loads the value of its mark field. Since the head pointer is never changed and the head node is never marked both the head node is valid when it was tagged. And any subsequent \crd{} on the fields of head node would fail if any of the fields of the head node change. 
Later, the head node's next field is \crd{} and saved in a local pointer \curr followed a \crd{} on the mark field of the \curr node. The latter \crd{} tags the curr node and then loads the value of its mark field.

So at this point, when the \curr node is tagged its predecessor \pred is guaranteed to be not marked because a head node is never marked and its next field is guaranteed to be unchanged (still points to \curr) because otherwise the \crd on \curr nodes's mark field would have failed by virtue of \propref{cAccessAtomic}, \propref{writeinvalidates}, and \clmref{cAccessSucc}. Also note that the \curr node's mark field is verified to be unmarked. So, both the \pred and \curr nodes are unmarked and \pred points to \curr for the base case.
\\
\noindent
\textbf{Induction Hypothesis:} Say, this is true for any arbitrary \pred and \curr nodes in search path of a thread. That is both the \pred and \curr nodes are unmarked and \pred points to \curr. Now we shall prove that at the time a node next to \curr, say \suc, is tagged then the \curr is still unmarked and points to \suc. \\
\noindent
\textbf{Induction Step:}  
In order to be able to access \suc to tag it, a thread first does a \crd{} on its predecessor node that is \curr to load its next field. If the \crd{} succeeds implies that \curr node hasn't changed since it was tagged. Meaning neither it has been marked or its next field has changed. Therefore, when marked field of \suc is \crd to be tagged its predecessor is unmarked and points to \suc. Otherwise, the \crd would fail by virtue of the \propref{cAccessAtomic}, \propref{writeinvalidates}, and \clmref{cAccessSucc}. 

Hence, the lemma is true for an arbitrary node.
\end{proof}

\begin{theorem}[\ca is Safe]
\label{thm:hbrsafe}
Threads cannot access reclaimed nodes in \ca based list.
\end{theorem}
\begin{proof}
Let us assume the contrary. That is thread could access reclaimed nodes. This may happen in \ca based lazylist in following two scenarios:
\begin{enumerate}
    \item A thread reclaims a node without setting its mark field to true. In this case, readers which have the node tagged will not receive an invalidation message and thus will not have their \hbrbit set. As a result the readers, unaware of the deletion, will succeed subsequent \crd{} which could cause \uaf error. But this is not possible because we assume that reclaimer's rule is followed.
    \item A thread tags a node after it has been marked so that subsequent \crd{} or \cwr{} will not receive any invalidation request for the node. And, therefore will not have their \hbrbit set. Which in turn would cause conditional accesses to succeed. This case too is not possible because of the \lemref{correcttag}.  
\end{enumerate}
Hence, both the the scenarios contradict our assumption implying use of \ca in lazylist is safe.
\end{proof}







\begin{theorem}[\ca is ABA free] Conditional accesses are immune to ABA problem.
\label{thm:abasafe}
\end{theorem}
\begin{proof}
Assume, that the \ca based lazylist have ABA problem. Therefore, there exists a case where a memory location \textit{m} is conditionally accessed by a thread $T_i$ (A), and then \textit{m} is modified by another thread $T_j$ (B), followed by another conditional memory access on \textit{m} by $T_i$ (A), such that the second memory access succeeds without recognising the fact that \textit{m} was changed since its first access. 

Since, whenever a memory location is first \crd{} it is tagged, $T_i$ tags \textit{m} when it does its first \crd{} on \textit{m}. Later if \textit{m} is modified by $T_j$ then by virtue of the cache coherence mechanism it invalidates all the remote copies of \textit{m}, in process invalidating the tags (\propref{writeinvalidates}). 
Therefore, when $T_i$ does \crd{} again on \textit{m} it will find that \textit{m} has been invalidated causing it to fail (\clmref{cAccessSucc}). Since, a \crd{} cannot fail and succeed at the same time, our initial assumption that the second \crd{} on \textit{m} by $T_i$ succeeds even when another thread modifies it after $T_i$ first accessed it is wrong. Therefore, ABA issue cannot occur.

\end{proof}


\paragraph{Progress:}
If \ca instructions are implemented in hardware such that spurious failures due to associativity evictions or interrupts are eliminated and conditional accesses could only fail due to real data conflicts then searches for \ca based data structures could be guaranteed to be lock-free. Because, a failure of a conditional access would imply that some thread successfully executed an update on a conflicting address and thus at least one thread is guaranteed to make progress.

On the other hand, if the implementation of proposed \ca semantics could not guarantee freedom from spurious failures then a fallback technique could be used. 

\paragraph{facilitating progress.}
As noted conditional accesses are vulnerable to spurious failures due to limits on hardware resources, like overflow of tagged cachelines because of associativity evictions. To reduce the probabilty of such spurious failures we propose \utag{} which helps in maintaining only a minimal set of tagged references by untagging previous locations which will not be required to establish correctness of searches.

\section{Evaluation}
\label{sec:caeval}

We prototype \ca (CA) using the Graphite multicore simulator~\cite{miller2010graphite}.  Our modifications were restricted to the L1 data cache level; we did not change the cache coherence protocol.
Graphite is configured to use a directory based MSI cache coherency protocol with a private 32K L1 and a shared inclusive 256K L2 cache. Each cacheline is 64 bytes and each thread runs on a dedicated simulated core with a basic branch prediction mechanism and an out-of-order memory subsystem.

We evaluate the scalability and memory efficiency of CA using microbenchmarks that stress test the lazy list and an external binary search tree (extbst).  We also use stack and hash table microbenchmarks to evaluate CA at different contention levels.
The keys in the lazy list, stack and hash table range from 0 to 1K; the extbst keys range from 0 to 10K. The hash table has 128 buckets,
where each bucket is a lazy list.  

Each of these data structures are made to use the following safe memory reclamation techniques: a leaky implementation(\textbf{none}:), \ca (\textbf{ca}), the 2geibr variant of IBR (\textbf{ibr}), \textbf{rcu}, quiescent state based reclamation (\textbf{qsbr}), hazard pointers(\textbf{hp}), and hazard eras (\textbf{he}). CA reclaims each deleted node immediately and requires no other parameters. The other reclamation schemes were configured to attempt reclamation after every 30 successful remove operations (\textit{reclamation frequency}). For epoch based schemes (ibr, rcu, qsbr and he) the epoch were configured to change after every 150 allocations (\textit{epoch frequency}). These values are the default in the IBR benchmark\cite{wen2018interval}.

Each trial in each experiment prefills its data structure to 50\% full and executes 3K operations per thread. The number of threads varies from 1 to 32. Each time a thread invokes a data structure operation, it randomly chooses an operation with a random key. In our experiments threads choose insert or delete with equal probability of 0\%, 5\% or 50\%, allowing us to run experiments with 0\% (read only), 10\% and 100\% updates, respectively.
Because the insert and delete probabilities are equal in all our workloads the data structure size remains roughly constant, storing half the elements in the key range.
For each workload configuration we report the average of three runs. There was no significant variance across the runs.

\begin{figure*}[ht]
     \begin{minipage}{\textwidth}
        \begin{subfigure}{\textwidth}
            \includegraphics[width=0.33\linewidth, height=6cm, keepaspectratio]{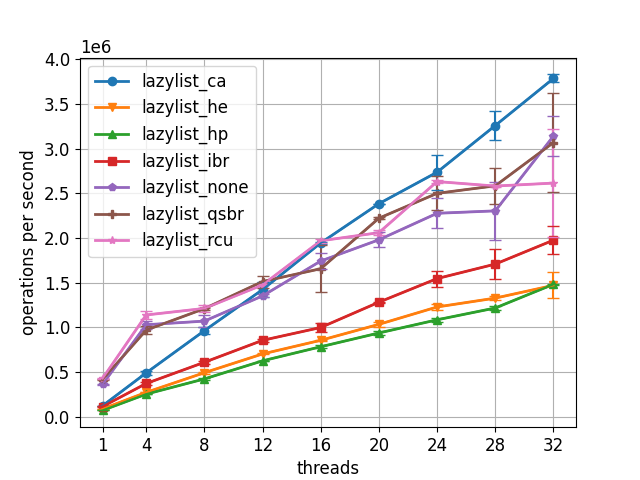}\hfill
            \includegraphics[width=0.33\linewidth, height=6cm, keepaspectratio]{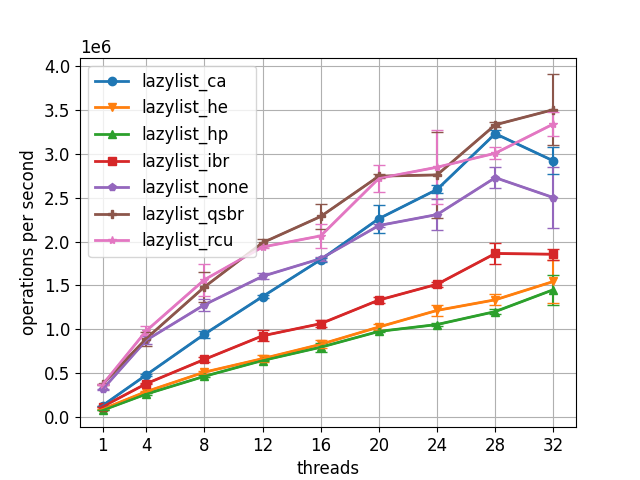}\hfill
            \includegraphics[width=0.33\linewidth, height=6cm, keepaspectratio]{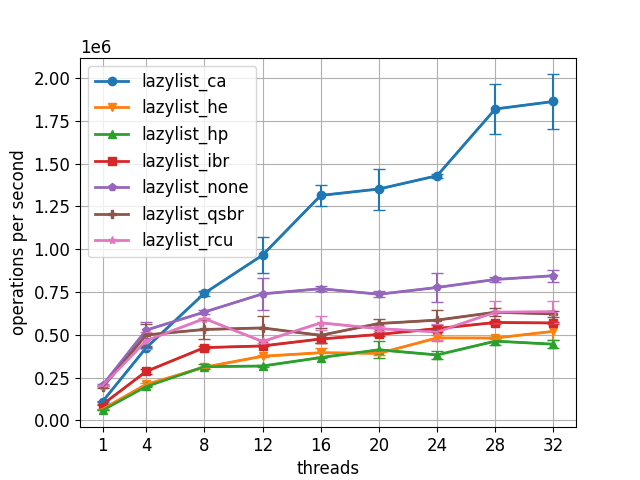}\hfill
            \label{fig:list}
        \end{subfigure}
        \begin{subfigure}{\textwidth}
            \includegraphics[width=0.33\linewidth, height=6cm, keepaspectratio]{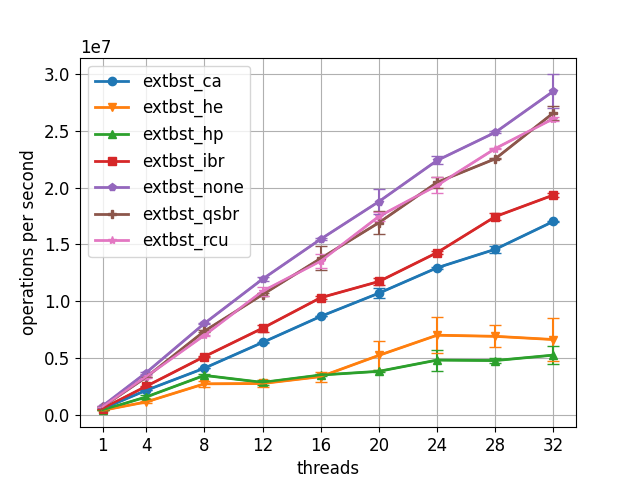}\hfill
            \includegraphics[width=0.33\linewidth, height=6cm, keepaspectratio]{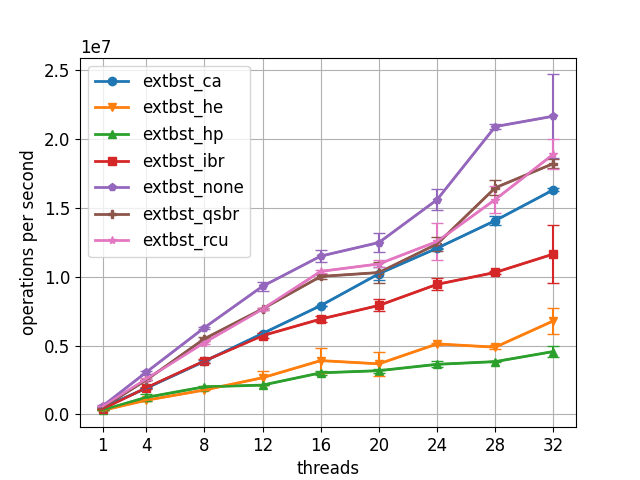}\hfill
            \includegraphics[width=0.33\linewidth, height=6cm, keepaspectratio]{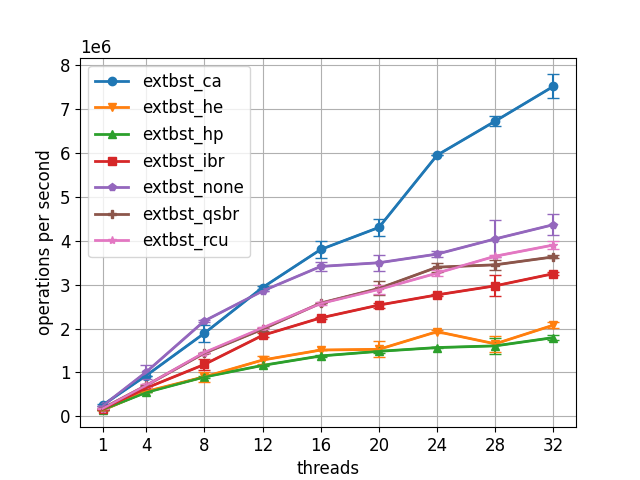}\hfill
            \label{fig:extbst}
        \end{subfigure}        
     \end{minipage}

    \caption{Evaluation of throughput. Y axis: throughput. X axis: \#threads. Left: 0i-0d. Middle: 5i-5d. Right: 50i-50d. (Top Row) Lazy linked-list, size:1K. (Bottom Row) External BST, size:10K. 
    }
    \label{fig:exp1}
\end{figure*}

\begin{figure*}[ht]
     \begin{minipage}{\textwidth}
        \begin{subfigure}{\textwidth}
            \includegraphics[width=0.33\linewidth, height=6cm, keepaspectratio]{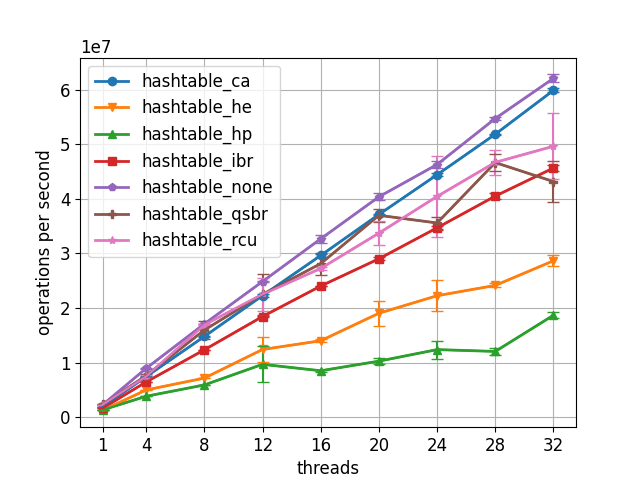}\hfill
            \includegraphics[width=0.33\linewidth, height=6cm, keepaspectratio]{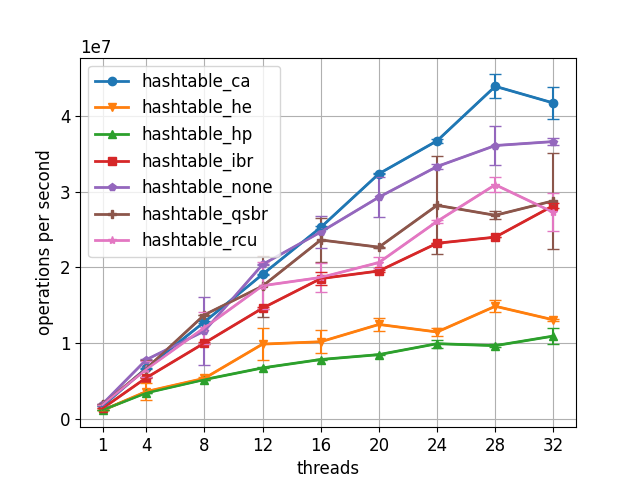}\hfill
            \includegraphics[width=0.33\linewidth, height=6cm, keepaspectratio]{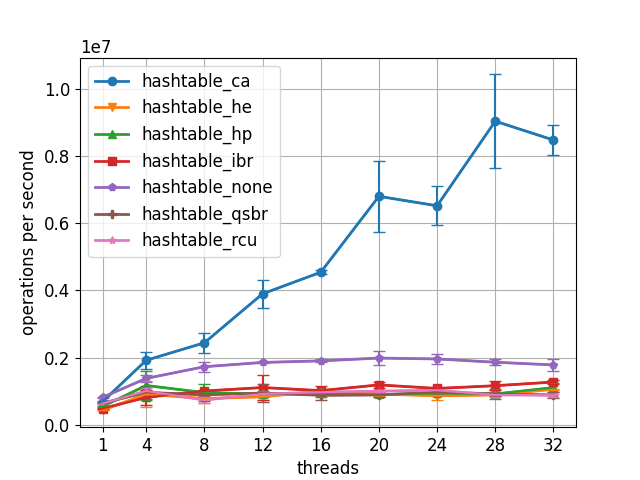}\hfill
            \label{fig:hashtable}
        \end{subfigure}
        \begin{subfigure}{\textwidth}
            \includegraphics[width=0.33\linewidth, height=6cm, keepaspectratio]{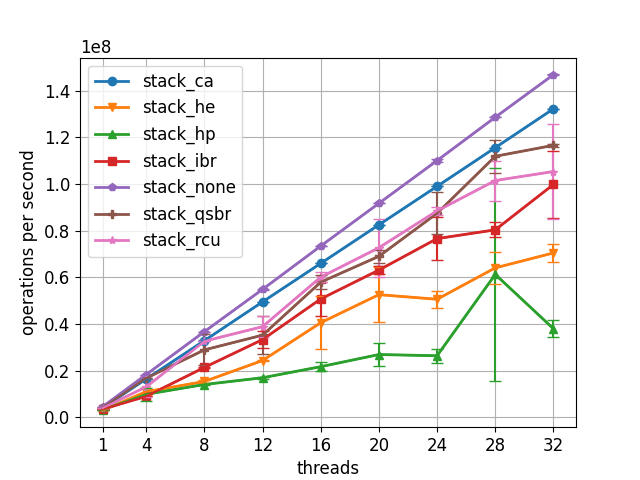}\hfill
            \includegraphics[width=0.33\linewidth, height=6cm, keepaspectratio]{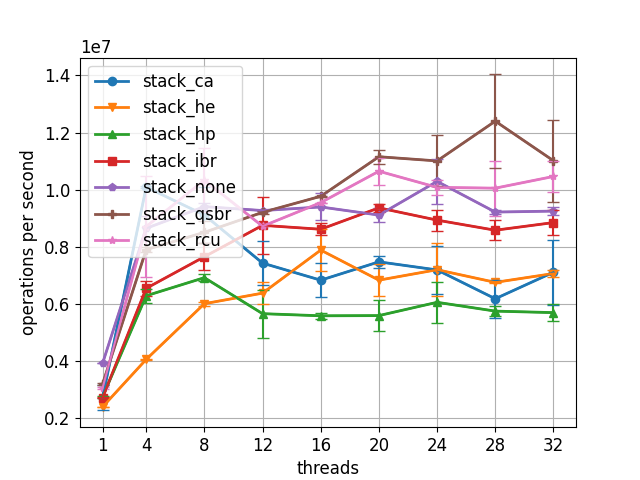}\hfill
            \includegraphics[width=0.33\linewidth, height=6cm, keepaspectratio]{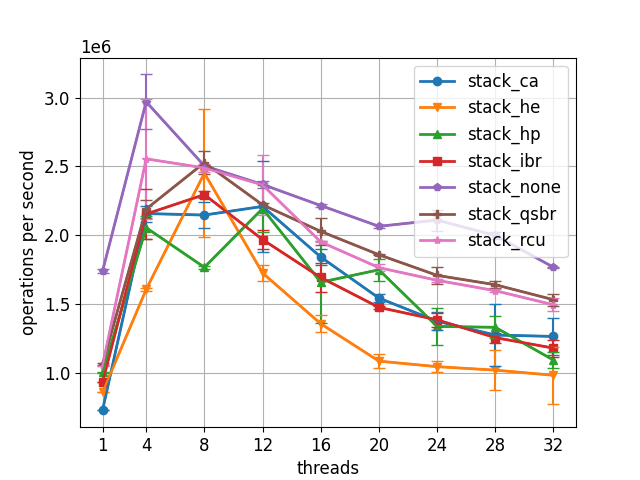}\hfill
            \label{fig:stack}
        \end{subfigure}
     \end{minipage}
    \caption{Evaluation of throughput. Y axis: throughput. X axis: \#threads. Left: 0i-0d. Middle: 5i-5d. Right: 50i-50d. (Top Row) Chaining Hash table, \#buckets:128, MAX size 128K. (Bottom Row) Stack.
    }
    \label{fig:exp2}
\end{figure*}

Throughout the experiments in \figref{exp1} and \figref{exp2}; \texttt{hp}, \texttt{he} and \texttt{ibr} are generally slower than the other algorithms.
This mainly can be attributed to high per-read overheads, as these algorithms have read/write fences to access or update reservations and epochs, respectively. 
Additionally, these algorithms have reclamation overhead which requires scanning of reservations to determine which records are safe to free.
In general this results in poor cache behaviour and high operation latency.

On the other side, \texttt{rcu} and \texttt{qsbr} have no per-read overhead.  Their main overhead arises from their reclamation events, where batches of retired objects are freed after scanning the epochs of all the processes.  This is amortized over multiple operations. As a result these algorithms are faster and perform similar to the baseline \texttt{none}, across workloads and data structures. 

In read-only workloads, CA is comparatively slower than \texttt{rcu}, \texttt{qsbr} and \texttt{none}.  This is due to the increased latency: checking the \hbrbit after each \crd{} increases the instruction count.  Since there are no conflicts, these checks are superfluous.
However, in workloads with updates, CA is closer to or faster than \texttt{rcu}, \texttt{qsbr} and \texttt{none}. 
It even outperforms these algorithms in high contention scenarios (i.e., high updates and high thread counts). This is due to the fact that CA avoids read-write fences for both readers and reclaimers. 
CA brings additional benefits.  Immediate reclamation improves cache and TLB locality, especially relative to \texttt{none}; it discovers failures earlier than other algorithms, which enables it to restart without wasting as much work; and it avoids some cache miss latencies.
All these contribute to low latency and higher throughput.  The low cost of cache misses is due to a property that unlike regular reads, in \crd{s} the impact of cache misses remains confined to its core \cite{alistarh2020memory}. We explain this in following paragraphs.

In data structure operations with normal reads and writes, all threads that share memory locations experience latency due to cache misses. Consider a lazy list, and suppose that thread T1 is about to acquire locks on its \pred and \curr nodes.  Suppose that another thread T2 has read the \pred and is about to re-read it.
At this point, at the cache level, both T1 and T2 will have copies of the cache lines corresponding to these addresses in the shared state.
When T1 acquires a lock on \pred it does a write.  This triggers coherence traffic: all other readers of \pred must invalidate their copies of the cache line (here T2).  When T2 reads \pred again it will suffer a cache miss as its copy of the cache line is invalid. In order to serve the cache miss:
\begin{itemize}
    \item T2 triggers a cache level transaction to fetch the latest copy of the cache line and waits for a response.
    \item T1, which has the cache line in M state may be forced to write its copy of the cache line back to the memory hierarchy, and also supply it to T2. 
\end{itemize}

This wastes T2's compute time, because T2 will ultimately see that the line has changed, necessitating that it restart its operation.  If it had not waited, it could have already restarted and executed multiple instructions. 
Furthermore, since T1 acquired the lock on \pred it is likely to write to \pred again.  T2's request caused the line to downgrade from M to S in T1's cache, so a subsequent write by T1 will need to begin with an ownership request that causes T2 to re-invalidate the line. Such frequent downgrading to shared state and upgrading to modified state interferes with the gains made by write buffering and makes it difficult to hide the cache latency. These overheads worsen with increases in contention on shared locations.

On the other hand, in data structures designed using \crd{} and \cwr{}, T2's second \crd{} will fail validation and retry its operation by detecting that the line is no longer present, \emph{without requesting a new copy of the line}.  Unlike the aforementioned issues with regular reads/writes, CA allows T2 to skip requesting the value of \pred, which prevents global cache traffic.  This avoids read latency for T2, and also helps T1 to avoid a cache state upgrade transactions. In other words, unlike regular read/write based data structures, the impact of failure to access cache lines in data structures with CA remains confined to a local core~\cite{alistarh2020memory}. 

Thus, in all the data structure implementations, the lower L1 data cache latencies that result from the aforementioned properties allow CA to be as good as the other algorithms, if not faster, when contention is high.  In read only workloads, CA is slower or comparable to the baseline (none) and other fast algorithms (\texttt{qsbr} and \texttt{rcu}) mainly due to overhead of its higher instruction count. 

\begin{figure}
\centering
\includegraphics[width=0.5\textwidth]{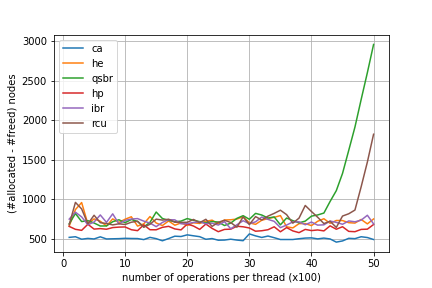}
\caption{Memory Consumption: Shows Number of nodes allocated but not yet freed for a list of size $\sim$500. Insert and delete percent is 50 each. For 16 threads.}
\label{fig:memconsume}
\end{figure}

\figref{memconsume} looks at memory overheads: For each of the reclamation schemes, we measure the number of nodes that were allocated but not yet freed (Y axis) during execution of the lazy list data structure after every 1000 operations (X axis).  This test exposes the amount by which the memory footprint of a data structure increases when paired with different reclamation schemes. For this experiment, we use a lazy list with values in the range of 0 to 1000, initially pre-filled with 500 nodes.  The experiment has 16 threads operating on the list.  During the measured part of the experiment all threads execute insert and delete operations with equal probability of 50\% (100\% update workload). Each thread runs 5000 ops. 

In the ideal case, at any time during the experiment the list size should be roughly 500 due to the workload characteristics, and the number of nodes deleted but yet not freed should be zero.  The CA scheme has a consistent reading of roughly 500 nodes that are allocated but not freed; these are the nodes that are still reachable in the list. This confirms that we are achieving immediate reclamation and keeping the memory footprint low.  Since the other reclamation schemes defer memory reclamation of deleted nodes by collecting them in a local retired list, the number of nodes that have not been reclaimed increases, which in turn leads to increased memory footprint.  This intuition is verified in the chart, as hp, he, ibr, rcu and qsbr all report a higher number of un-reclaimed nodes It is worth mentioning that, since in qsbr and rcu a delayed thread could prevent reclamation of all threads, the number of unreclaimed nodes could increase without bound.  Had we run the experiment for longer, the number of unreclaimed nodes for these schemes would be expected to balloon as soon as any thread context switched.

Alternatively, the performance boost can also be due to the approach to design \ca algorithm that takes into account the lower level layers that could mitigate the negative interaction that might occur at the lower level in the software stack, for example, similar to the one we reported in ~\cite{kim2024token}.

\section{Limitations and Miscellaneous Discussion}
\label{sec:calim}

In this section, we discuss some limitations and their possible solutions along with some questions regarding our approach that could arise in a reader's mind.
\begin{enumerate}
    \item \ca assumes that a single cacheline does not contain multiple nodes to avoid accidental untagging due to cacheline sharing (this is also an assumption in previous work like Memory Tagging by Alistarh et al.~\cite{alistarh2020memory}).
    In order to support multiple nodes in the same cache line one could use a Counting Bloom filter to make sure nodes are not untagged by accident. Alternatively, one could use more than one bit per cache line. E.g., 2 bits per cache line would support two data structure nodes.

    In many cases in concurrent data structures it is actually preferable to pad nodes so that only one node appears in each cache line---sometimes to avoid false sharing, and other times to avoid cache line crossings (individual nodes that span more than one cache line because of odd node sizes, such as 48 bytes)~\cite{arbel2018getting}.
    
    \item 
    One might wonder whether, in a multiprogrammed system, context switches between threads could create an ABA problem on the tagged addresses themselves. However, this can be easily handled by clearing tags (or simply setting the \hbrbit, which achieves the similar effect) on a context switch. 
    This essentially aborts operations that are concurrent with a context switch, and so, in principal, could threaten progress, but context switches are typically infrequent, happening on a much longer time scale than individual concurrent data structure operations.
    Similarly, the \hbrbit is set if any cache line corresponding to a data structure node is evicted.
    
    \item 
    From a security standpoint, additional study is necessary to verify that \ca remains resistant to meltdown/spectre vulnerabilities related to speculative side-channel attacks.
    Nonetheless, we posit that \ca exhibits a lower attack bandwidth compared to hardware transactional memory (HTM). 
    Unlike HTM, \ca does not have immediacy of aborts. That is, in \ca a thread, as far as true data conflict is concerned, fails mainly due to its own steps in programs and it must check if the \crd{}/\cwr{} it issued has failed, and must manually clean up and restart its operation.
    This should be slower than the way that HTM immediately aborts and jumps to a handler on a memory conflict. 
    Consequently, we expect \ca to not offer meltdown-level bandwidth for speculation-based side-channel attacks.

    \item  Memory access instructions in \ca are prone to spurious failures that may affect progress guarantees. The failures due to associativity conflicts could be avoided if tagging were implemented with a TCAM. (Note that TCAMs have been used previously, e.g., in IBM's transactional memory implementations, to avoid associativity conflicts.) If progress is still a concern, one could add a fallback path, taking inspiration from transactional lock elision~\cite{rajwar2001speculative, afek2015amalgamated}. Our expectation is that the fallback path should execute less frequently than in HTMs (especially if we avoid associativity conflicts by using a TCAM).
    Some possible approaches to a fallback path include lock-based execution, as in TLE, and epoch-based reclamation which we would like to explore as part of the future work.

\end{enumerate}

There are similarities between \ca and the \textit{Optimistic Access} line of algorithms (see~\secref{oaparadigm}), but we believe that immediate reclamation has not yet been efficiently solved. Software solutions like VBR~\cite{sheffi2021vbr}, OA~\cite{cohen2015efficient}, AOA~\cite{cohen2015automatic}, and FA~\cite{cohen2018every} essentially assume that a thread never calls free, and also that memory is type stable – if address A of type T is freed, a subsequent malloc will only return A if it will be used as type T. This way, one can reuse nodes immediately without segfaulting. (Alternatively, one could trap segfaults, but that has downsides like complicating program debugging.) \ca does not require this assumption.
After conditional access was proposed, a custom lock-free allocator was designed~\cite{moreno2023releasing} that resolves the issue with the ability of these techniques to invoke free. 

Furthermore, unlike \textit{Optimistic Access} style algorithms like VBR~\cite{sheffi2021vbr}, \ca also does not need wide CAS, and does not need to change the memory layout of objects to augment mutable fields with version stamps (which increases the size of objects, possibly having a negative impact on performance).
As an extension to \ca it would be interesting to explore whether the techniques in \ca could be used to accelerate the \textit{Optimistic Access} line of software solutions and, in particular, to determine whether data structures in normalized form could be adapted for use with \ca.

\section{Additional Related Work} 
\label{sec:carelated}
\subsection{Discussion on Reclamation Techniques}
In this section, we focus on techniques which could provide immediate reclamation~\cite{zhou2017hand, sheffi2021vbr} and therefore are most closely related to \ca. \Chapref{chapsurvey} contains a detailed survey of the safe memory reclamation techniques.

Zhou et al.\cite{zhou2017hand} make use of a sequence of short hardware transactions which execute in hand over hand fashion to design concurrent data structures that retain the property of immediate memory reclamation. 
The technique relies on augmenting the data structure with a table of metadata, which can be a source of false conflicts.  Consequently, it does not appear to be as general as \ca.  Moreover, we found that the frequent starting and committing of transactions for read-only operations introduced significant latency.


VBR~\cite{sheffi2021vbr} attaches metadata to each mutable field of each node in a concurrent data structure.  
It also requires a type preserving allocator, where unlinked nodes can never be returned to the operating system.
Threads can detect use-after-free errors through the per-field metadata, which is updated atomically with the corresponding field.
While VBR can support immediate reclamation, it is most efficient when it waits until it has a batch of nodes to reclaim in a single operation.

On the other hand, \ca does not require any metadata to achieve safe reclamation and only makes use of the implicit book-keeping of the underlying cache-coherence protocol. 
In addition, since it does not make any assumptions about the number of threads present in the system, it is fully adaptive~\cite{herlihy2003space}. Furthermore, whereas HTM can accelerate timing-based attacks by leveraging the immediacy with which a thread is aborted upon a memory conflict~\cite{Lipp2018meltdown}, in particular, transaction rollbacks could lead to data leaks~\cite{intelTAA}, we believe \ca is less risky, since threads must poll to learn of remote coherence events.

\subsection{Discussion on Similar Synchronization Techniques}

\ca is inspired by, but quite different from, the Memory Tagging proposal of Alistarh et al.~\cite{alistarh2020memory}. Perhaps the most significant difference is that \ca solves the safe memory reclamation problem (and moreover offers immediate reclamation), in addition to providing useful synchronization primitives for designing concurrent data structures. In contrast, Memory Tagging does not address the memory reclamation problem, and it requires a data structure designer to rely on separate safe memory reclamation algorithms, which come with their own tradeoffs.

In reference to the programming interface, \ca offers \crd, which is critical to our immediate memory reclamation technique. The \crd instruction has no equivalent instruction in Memory Tagging, and it is not clear how one could implement \crd using memory tagging.
We also streamlined tagging by integrating it into \crd, whereas Memory Tagging requires a programmer to use an explicit \textit{AddTag} instruction before reading.

From an implementation standpoint, \ca does not require changing the underlying coherence protocol, whereas Memory Tagging’s Invalidate and Swap (IAS) instruction does, as this single instruction can invalidate many (potentially non-contiguous) remote cache lines (potentially spanning many pages). Moreover, \ca requires only 1 bit per cache line (2 bits per cache line in case of 2-way hyperthreading), whereas Memory Tagging needs to additionally maintain a set of addresses to invalidate with IAS.

It is worth noting that at the outset \ca may appear similar to HTM with early release (ER)~\cite{sonmez2007unreadtvar, skare2006early}. Possibly, one could achieve many aspects of our work by using HTM with early release. However, this would introduce various downsides.
HTM defaults to putting all reads and writes into the read/write sets. This includes the stack, the allocator, library code, etc. In data structures, many reads and writes would need to be released, which would increase the instruction count significantly. This could reduce performance and might yield a less convenient interface than \ca.

Practically, some commercial HTMs have a region based (not per access) \textit{disable tracking} feature, for instance, Intel's new TSXLDTRK~\cite{intelSDM}, and IBM's TSUSPEND/TRESUME~\cite{le2015transactional}, but this does not \textit{release} load tracking of already-read locations, it only prevents tracking of future accesses. This is different from \ca's proposed \texttt{untag} instruction which allows the release of any previously accessed location. Among Early Release proposals, we are not aware of any that release writes, although AMD's 2008 ASF proposal allowed per-access decisions about whether or not to track~\cite{christie2010evaluation}. However, ASF remains unimplemented. Additionally, we have not experimented with TSXLDTRK, but TSUSPEND suffers from relatively high overhead.

Unlike HTM, \ca does not need a write-set at all, which admits a simpler implementation in hardware, as well as simpler conflict tracking and resolution.
We think \ca solves an important problem for optimistic data structures with less hardware (as demonstrated in \secref{hwimpl}). Our hope is that our hardware-software codesign approach to \ca will enable the concurrent data structure community to discover novel and efficient solutions to existing concurrency problems.

\section{Summary}
\label{sec:casummary}

In this chapter, we introduced \ca, a hardware extension that enables concurrent data structures to reclaim memory immediately, without introducing new inter-thread coordination.
\ca reconciles the trade-off between speedup and immediate reclamation by leveraging the underlying cache coherence mechanism which readily provides the information about the probable event where a thread might attempt to access a concurrently freed node, therefore eliminating the need for redundant and suboptimal bookkeeping at the application level.

\ca is fast.
Unlike state-of-the-art safe memory reclamation algorithms, it does not require tuning to achieve high performance, and is tailored to the needs of modern optimistic data structures. In particular, it allows immediate reuse of memory without compromising speedup of the data structures, which is desirable in the universal memory allocator of the FreeBSD kernel.

We prototype \ca on Graphite, which is an open source multicore simulator, and implement a benchmark comprised of multiple state-of-the-art memory reclamation techniques and data structures.

To date, we have used \ca for simple non-blocking data structures, as well as optimistic lock-based data structures.
In the future, it would be interesting to determine whether \ca\ can also be used for more complex lock-free data structures.
We also believe that there are exciting opportunities at the interface between \ca and non-volatile main memory technologies.

\ignore
{

TODO:
\subsubsection{System and Architectural Background}
\textbf{HW}:Discuss MSI, MESI cache protocol and behaviour of a shared memory access and underlying working which helps readers to understand the working of Conditional Access.
this moves to a bg section in CA chapter.

-- \ajay{A figure of a NUMA architecture system showing processor, cache levels.}

\noindent
\textbf{OS}: 

--Signals (NBR chapter specific)

--allocators (General) -- explain how memory is alloctaed and reused etc...

--virtual memory (general)

Explain memory Subsystem and how programs use reuse memory from software level to Hw level. The info reader need to understand SMRs and put them in context. Then Signals and Cache explanation moves to specific chapters.

\textbf{Memory Ordring and COnsistency \& discussion on memory model}: To help readers understand the HP reodering and Placement of of fences in NBR and CA. Chapter 2.2 in Trev Book.
Talk of C++ ordering and memory model.

\textbf{Describe Atomic instructions}: FAA, CAS. used in my algorithms.
\textbf{Describe Linerizability, Non blocking progress, NBDS, and optimistic data structure types}
}


\chapter{Conclusion}
\label{chap:chapconclusion}

This dissertation contributes to advances in the state of the art in safe memory reclamation for solving use-after-free errors in concurrent data structures using non-blocking mechanisms, with three major objectives.

Designing a fast safe memory reclamation algorithm that has a bounded memory footprint has been a challenge. 
In practice, it has been even more challenging to design a fast and bounded memory footprint reclamation algorithm that consistently preserves these two properties 
across various levels of contention and types of workloads, while also avoiding alterations to the memory layout of data structures, not needing particular compiler or architectural support, and being applicable to a broad array of data structures. Our first contribution in this dissertation addresses this problem.

We adopt an approach that, while similar in spirit to the optimistic family of reclamation algorithms, differs significantly in practice. At an abstract level, unlike state-of-the-art methods where readers either pessimistically reserve references or wait for references to become safe to free, our method allows readers to access memory with only local reservations (no costly fence and validations like reservations in HP) 
and eliminates the need for reclaimers to wait. Instead, reclaimers enforce coordination between threads to ensure safety.

Concurrent data structures typically have a long read phase followed by a write phase. The key idea is to permit readers to access memory locations without reservations, while preparing them to discard their references if a reclaimer is concurrently reclaiming those locations. If readers cannot discard their references, they must reserve them. Rather than waiting for nodes to become safe for reclamation, reclaimers enforce reclamation by sending signals to all threads. In response, readers either discard their references and restart or confirm that their references are reserved and thus cannot be reclaimed. We refer to this paradigm as neutralization.

We use this paradigm to implement NBR, presented in \chapref{chapnbr}, a neutralization-based algorithm, followed by a highly optimized version, NBR+, presented in \chapref{chapnbrp}. These algorithms address the long-standing challenge of designing fast, practical, and bounded memory footprint safe reclamation algorithms that perform consistently across various contention levels and system settings, are easy to use, and apply to a wide range of data structures, as shown in \chapref{chapappuse}. Furthermore, our integration of these techniques with multiple data structures demonstrates their simplicity and wide applicability compared to other hybrid approaches.

A current limitation of neutralization-based algorithms is that they require programmers to distinguish between the read and write phases of a data structure. Although this is straightforward for most data structures, Sheffi and Petrank~\cite{sheffi2023era} classified these as "access-aware" data structures, and more complex structures may require expert skills. The technique does not apply to data structures where the read phase does not begin from an entry point in the data structure, such as head in lists or root in trees. Recently, Kim, Jung, and Kang explored a hazard pointer-based checkpointing technique that we believe could potentially address this limitation~\cite{kim2024expediting}. 
As part of the future extension of this work, it will be interesting to further extend applicability of NBR and also evaluate the algorithms with user space signals made available through the introduction of \texttt{senduipi} instruction on Intel's Sapphire Rapid architectures.

Presented in \chapref{chaprsp}, our second major contribution is the introduction of a paradigm for enabling reactive synchronization in the design of safe memory reclamation algorithms. This paradigm specifically addresses the issue of uneven synchronization overhead imposed by state-of-the-art safe memory reclamation algorithms on the data structures with which they are paired.
In practice, threads frequently read shared-memory locations and rarely modify them. Delete events, which can cause use-after-free errors, are even rarer. Given this observation, it is an overhead on the entire system to publish reclamation-related synchronization information (references in hazard pointers and eras in hazard eras) eagerly for every read. It is more efficient to publish this information only when a reclamation event occurs.

To achieve this, we use signals again, but to enable readers to publish reclamation-related information only when a reclamation event occurs, which we call the "publish on ping" algorithm. We designed a family of publish-on-ping reclamation algorithms using publish-on-ping with hazard pointers and hazard eras. The most important application of publish-on-ping is to allow hazard pointers to function as epoch-based reclamation mechanisms almost all the time. When a thread is unable to reclaim, it pings all threads to publish their reservations, which they have been silently maintaining. To the best of our knowledge, this is the first approach where the fast path and slow path run concurrently without switching between the two; in fact, a thread can reclaim using the slow path, while another operates in the fast path.

One advantage of publish on ping is that it does not interfere with applicability and usability of the original algorithms, and yet allows for drastic improvements in their performance. 

Recently proposed user-space signals have demonstrated a performance improvement of up to 10× over POSIX signals~\cite{useripilinux}. The first hardware implementation of these user-space signals is available on Intel's Sapphire Rapids machines. As part of the future work, it would be interesting to study how much more speedup the neutralization-based algorithms and publish-on-ping based algorithms can achieve using these fast user-space signals.
Additionally, in the EpochPOP variant, inspired by Erez Petrank's suggestion, we will explore reducing the number of signals by selectively signaling threads that appear to be stuck.

In our third and final contribution presented in \chapref{chapirp}, we examine the issue of safe memory reclamation from the first principles and observe events at the architectural level of the system stack. We note that much of the synchronization required for safe memory reclamation already occurs at the architectural level. Therefore, exposing these events to programmers through new hardware instructions is beneficial to efficiency, as redundant synchronization is eliminated. To the best of our knowledge, this is the first time in the literature of safe memory reclamation algorithms that such a hardware-software co-design paradigm has been proposed.

A surprising benefit of this codesign is that it not only eliminates redundant synchronization at the programmer level, but also enables safe memory reclamation to achieve a sequential data structure-like ideal memory footprint, effectively eliminating deferred reclamation. The resulting algorithm, called Conditional Access, enables immediate reclamation and can offer significant advantages in modern data center environments, particularly in terms of performance, resource allocation, and security. 
In the future, it would be interesting to apply Conditional Access to more data structures, starting with those in normalized form~\cite{timnat2014practical}.
In addition, we would like to prototype Conditional Access on FPGAs and validate the benefits shown on the graphite simulator.

\bibliographystyle{plain}
\cleardoublepage 
\phantomsection  
\renewcommand*{\bibname}{References}

\addcontentsline{toc}{chapter}{\textbf{References}}

\bibliography{uw-ethesis.bib}

\begin{thebibliography}{100}

\bibitem{afek2014cb}
Yehuda Afek, Haim Kaplan, Boris Korenfeld, Adam Morrison, and Robert~E Tarjan.
\newblock The cb tree: a practical concurrent self-adjusting search tree.
\newblock {\em Distributed computing}, 27(6):393--417, 2014.

\bibitem{afek2015amalgamated}
Yehuda Afek, Alexander Matveev, Oscar~R Moll, and Nir Shavit.
\newblock Amalgamated lock-elision.
\newblock In {\em Distributed Computing: 29th International Symposium, DISC
  2015, Tokyo, Japan, October 7-9, 2015, Proceedings 29}, pages 309--324.
  Springer, 2015.

\bibitem{aksenov2023splay}
Vitaly Aksenov, Dan Alistarh, Alexandra Drozdova, and Amirkeivan Mohtashami.
\newblock The splay-list: A distribution-adaptive concurrent skip-list.
\newblock {\em Distributed Computing}, pages 1--24, 2023.

\bibitem{alistarh2020memory}
Dan Alistarh, Trevor Brown, and Nandini Singhal.
\newblock Memory tagging: Minimalist synchronization for scalable concurrent
  data structures.
\newblock In {\em Proceedings of the 32nd ACM Symposium on Parallelism in
  Algorithms and Architectures}, pages 37--49, 2020.

\bibitem{alistarh2014stacktrack}
Dan Alistarh, Patrick Eugster, Maurice Herlihy, Alexander Matveev, and Nir
  Shavit.
\newblock Stacktrack: An automated transactional approach to concurrent memory
  reclamation.
\newblock In {\em Proceedings of the Ninth European Conference on Computer
  Systems}, pages 1--14, 2014.

\bibitem{alistarh2017forkscan}
Dan Alistarh, William Leiserson, Alexander Matveev, and Nir Shavit.
\newblock Forkscan: Conservative memory reclamation for modern operating
  systems.
\newblock In {\em Proceedings of the Twelfth European Conference on Computer
  Systems}, pages 483--498, 2017.

\bibitem{alistarh2018threadscan}
Dan Alistarh, William Leiserson, Alexander Matveev, and Nir Shavit.
\newblock Threadscan: Automatic and scalable memory reclamation.
\newblock {\em ACM Transactions on Parallel Computing (TOPC)}, 4(4):1--18,
  2018.

\bibitem{anderson2021concurrent}
Daniel Anderson, Guy~E Blelloch, and Yuanhao Wei.
\newblock Concurrent deferred reference counting with constant-time overhead.
\newblock In {\em Proceedings of the 42nd ACM SIGPLAN International Conference
  on Programming Language Design and Implementation}, pages 526--541, 2021.

\bibitem{anderson2022turning}
Daniel Anderson, Guy~E Blelloch, and Yuanhao Wei.
\newblock Turning manual concurrent memory reclamation into automatic reference
  counting.
\newblock In {\em Proceedings of the 43rd ACM SIGPLAN International Conference
  on Programming Language Design and Implementation}, pages 61--75, 2022.

\bibitem{apachehbase}
{Apache HBase Team}.
\newblock {\em Apache HBase® Reference Guide}.
\newblock Apache.
\newblock Version 4.0.0-alpha-1-SNAPSHOT, chapter:81.

\bibitem{arbel2018getting}
Maya Arbel-Raviv, Trevor Brown, and Adam Morrison.
\newblock Getting to the root of concurrent binary search tree performance.
\newblock In {\em Proceedings of the 2018 USENIX Conference on Usenix Annual
  Technical Conference}, USENIX ATC '18, page 295–306, USA, 2018. USENIX
  Association.

\bibitem{balmau2016fast}
Oana Balmau, Rachid Guerraoui, Maurice Herlihy, and Igor Zablotchi.
\newblock Fast and robust memory reclamation for concurrent data structures.
\newblock In {\em Proceedings of the 28th ACM Symposium on Parallelism in
  Algorithms and Architectures}, pages 349--359, 2016.

\bibitem{ben2022lock}
Naama Ben-David, Guy~E Blelloch, and Yuanhao Wei.
\newblock Lock-free locks revisited.
\newblock In {\em Proceedings of the 27th ACM SIGPLAN Symposium on Principles
  and Practice of Parallel Programming}, pages 278--293, 2022.

\bibitem{blelloch2020concurrent}
Guy~E Blelloch and Yuanhao Wei.
\newblock Concurrent reference counting and resource management in wait-free
  constant time.
\newblock {\em arXiv preprint arXiv:2002.07053}, 2020.

\bibitem{braginsky2013drop}
Anastasia Braginsky, Alex Kogan, and Erez Petrank.
\newblock Drop the anchor: lightweight memory management for non-blocking data
  structures.
\newblock In {\em Proceedings of the twenty-fifth annual ACM symposium on
  Parallelism in algorithms and architectures}, pages 33--42, 2013.

\bibitem{braginsky2012lock}
Anastasia Braginsky and Erez Petrank.
\newblock A lock-free b+ tree.
\newblock In {\em Proceedings of the twenty-fourth annual ACM symposium on
  Parallelism in algorithms and architectures}, pages 58--67, 2012.

\bibitem{bronson2010practical}
Nathan~G Bronson, Jared Casper, Hassan Chafi, and Kunle Olukotun.
\newblock A practical concurrent binary search tree.
\newblock {\em ACM Sigplan Notices}, 45(5):257--268, 2010.

\bibitem{brown2017techniques}
Trevor Brown.
\newblock Techniques for constructing efficient lock-free data structures.
\newblock {\em arXiv preprint arXiv:1712.05406}, 2017.

\bibitem{brown2014general}
Trevor Brown, Faith Ellen, and Eric Ruppert.
\newblock A general technique for non-blocking trees.
\newblock In {\em Proceedings of the 19th ACM SIGPLAN symposium on Principles
  and practice of parallel programming}, pages 329--342, 2014.

\bibitem{Brown:2014}
Trevor Brown, Faith Ellen, and Eric Ruppert.
\newblock A general technique for non-blocking trees.
\newblock In {\em Proceedings of the 19th ACM SIGPLAN Symposium on Principles
  and Practice of Parallel Programming}, PPoPP '14, pages 329--342, 2014.
\newblock Full version available from \url{http://tbrown.pro}.

\bibitem{brown2020non}
Trevor Brown, Aleksandar Prokopec, and Dan Alistarh.
\newblock Non-blocking interpolation search trees with doubly-logarithmic
  running time.
\newblock In {\em Proceedings of the 25th ACM SIGPLAN Symposium on Principles
  and Practice of Parallel Programming}, pages 276--291, 2020.

\bibitem{brown2015reclaiming}
Trevor~Alexander Brown.
\newblock Reclaiming memory for lock-free data structures: There has to be a
  better way.
\newblock In {\em Proceedings of the 2015 ACM Symposium on Principles of
  Distributed Computing}, pages 261--270, 2015.

\bibitem{cha2001cache}
Sang~Kyun Cha, Sangyong Hwang, Kihong Kim, and Keunjoo Kwon.
\newblock Cache-conscious concurrency control of main-memory indexes on
  shared-memory multiprocessor systems.
\newblock In {\em VLDB}, volume~1, pages 181--190, 2001.

\bibitem{christie2010evaluation}
Dave Christie, Jae-Woong Chung, Stephan Diestelhorst, Michael Hohmuth, Martin
  Pohlack, Christof Fetzer, Martin Nowack, Torvald Riegel, Pascal Felber,
  Patrick Marlier, et~al.
\newblock Evaluation of amd's advanced synchronization facility within a
  complete transactional memory stack.
\newblock In {\em Proceedings of the 5th European conference on Computer
  systems}, pages 27--40, 2010.

\bibitem{cohen2018every}
Nachshon Cohen.
\newblock Every data structure deserves lock-free memory reclamation.
\newblock {\em Proceedings of the ACM on Programming Languages},
  2(OOPSLA):1--24, 2018.

\bibitem{cohen2015automatic}
Nachshon Cohen and Erez Petrank.
\newblock Automatic memory reclamation for lock-free data structures.
\newblock {\em ACM SIGPLAN Notices}, 50(10):260--279, 2015.

\bibitem{cohen2015efficient}
Nachshon Cohen and Erez Petrank.
\newblock Efficient memory management for lock-free data structures with
  optimistic access.
\newblock In {\em Proceedings of the 27th ACM symposium on Parallelism in
  Algorithms and Architectures}, pages 254--263, 2015.

\bibitem{correia2021orcgc}
Andreia Correia, Pedro Ramalhete, and Pascal Felber.
\newblock Orcgc: automatic lock-free memory reclamation.
\newblock In {\em Proceedings of the 26th ACM SIGPLAN Symposium on Principles
  and Practice of Parallel Programming}, 2021.

\bibitem{CVE202426602}
sched/membarrier: reduce the ability to hammer on sys membarrier.
\newblock Available from MITRE, {CVE-ID} CVE-2024-26602., February~19 2024.

\bibitem{david2015asynchronized}
Tudor David, Rachid Guerraoui, and Vasileios Trigonakis.
\newblock Asynchronized concurrency: The secret to scaling concurrent search
  data structures.
\newblock {\em ACM SIGARCH Computer Architecture News}, 43(1):631--644, 2015.

\bibitem{mprotectoverhead}
Mathieu Desnoyers.
\newblock {A}lternative to signals/sys\_membarrier() in liburcu.
\newblock Retrieved online, 12 March, 2015.
\newblock [Accessed 15-08-2024].

\bibitem{desnoyers2009low}
Mathieu Desnoyers.
\newblock {\em Low-impact operating system tracing}.
\newblock PhD thesis, {\'E}cole Polytechnique de Montr{\'e}al, 2009.
\newblock chapter 6.

\bibitem{membarrierSystemwide}
Mathieu Desnoyers.
\newblock sys\_membarrier(): system-wide memory barrier (x86) [{L}{W}{N}.net]
  --- lwn.net.
\newblock \url{https://lwn.net/Articles/640239/}, April, 2015.
\newblock [Accessed 15-08-2024].

\bibitem{detlefs2002lock}
David~L Detlefs, PaulA Martin, Mark Moir, and GuyL Steele~Jr.
\newblock Lock-free reference counting.
\newblock {\em Distributed Computing}, 15(4):255--271, 2002.

\bibitem{dice2016fast}
Dave Dice, Maurice Herlihy, and Alex Kogan.
\newblock Fast non-intrusive memory reclamation for highly-concurrent data
  structures.
\newblock In {\em Proceedings of the 2016 ACM SIGPLAN International Symposium
  on Memory Management}, pages 36--45, 2016.

\bibitem{dice2016refined}
Dave Dice, Alex Kogan, and Yossi Lev.
\newblock Refined transactional lock elision.
\newblock {\em ACM SIGPLAN Notices}, 51(8):1--12, 2016.

\bibitem{drachsler2014practical}
Dana Drachsler, Martin Vechev, and Eran Yahav.
\newblock Practical concurrent binary search trees via logical ordering.
\newblock In {\em Proceedings of the 19th ACM SIGPLAN symposium on Principles
  and practice of parallel programming}, pages 343--356, 2014.

\bibitem{dragojevic2011power}
Aleksandar Dragojevi{\'c}, Maurice Herlihy, Yossi Lev, and Mark Moir.
\newblock On the power of hardware transactional memory to simplify memory
  management.
\newblock In {\em Proceedings of the 30th annual ACM SIGACT-SIGOPS symposium on
  Principles of distributed computing}, pages 99--108, 2011.

\bibitem{ellen2014amortized}
Faith Ellen, Panagiota Fatourou, Joanna Helga, and Eric Ruppert.
\newblock The amortized complexity of non-blocking binary search trees.
\newblock In {\em Proceedings of the 2014 ACM symposium on Principles of
  distributed computing}, pages 332--340, 2014.

\bibitem{ellen2010non}
Faith Ellen, Panagiota Fatourou, Eric Ruppert, and Franck van Breugel.
\newblock Non-blocking binary search trees.
\newblock In {\em Proceedings of the 29th ACM SIGACT-SIGOPS symposium on
  Principles of distributed computing}, pages 131--140, 2010.

\bibitem{evans2006scalable}
Jason Evans.
\newblock A scalable concurrent malloc (3) implementation for freebsd.
\newblock In {\em Proc. of the bsdcan conference, ottawa, canada}, 2006.

\bibitem{fakhoury2024nova}
Ramy Fakhoury, Anastasia Braginsky, Idit Keidar, and Yoav Zuriel.
\newblock Nova: Safe off-heap memory allocation and reclamation.
\newblock In {\em 27th International Conference on Principles of Distributed
  Systems (OPODIS 2023)}. Schloss-Dagstuhl-Leibniz Zentrum f{\"u}r Informatik,
  2024.

\bibitem{fatourou2019persistent}
Panagiota Fatourou, Elias Papavasileiou, and Eric Ruppert.
\newblock Persistent non-blocking binary search trees supporting wait-free
  range queries.
\newblock In {\em The 31st ACM Symposium on Parallelism in Algorithms and
  Architectures}, pages 275--286, 2019.

\bibitem{fraser2004practical}
Keir Fraser.
\newblock Practical lock-freedom.
\newblock Technical report, University of Cambridge, Computer Laboratory, 2004.

\bibitem{gidenstam2008efficient}
Anders Gidenstam, Marina Papatriantafilou, H{\aa}kan Sundell, and Philippas
  Tsigas.
\newblock Efficient and reliable lock-free memory reclamation based on
  reference counting.
\newblock {\em IEEE Transactions on Parallel and Distributed Systems},
  20(8):1173--1187, 2008.

\bibitem{guerraoui2016optimistic}
Rachid Guerraoui and Vasileios Trigonakis.
\newblock Optimistic concurrency with optik.
\newblock {\em ACM SIGPLAN Notices}, 51(8):1--12, 2016.

\bibitem{harris2001pragmatic}
Timothy~L Harris.
\newblock A pragmatic implementation of non-blocking linked-lists.
\newblock In {\em International Symposium on Distributed Computing}, pages
  300--314. Springer, 2001.

\bibitem{hart2007performance}
ThomasE Hart, PaulE McKenney, Angela~Demke Brown, and Jonathan Walpole.
\newblock Performance of memory reclamation for lockless synchronization.
\newblock {\em Journal of Parallel and Distributed Computing},
  67(12):1270--1285, 2007.

\bibitem{he2017deletion}
Meng He and Mengdu Li.
\newblock Deletion without rebalancing in non-blocking binary search trees.
\newblock In {\em 20th International Conference on Principles of Distributed
  Systems (OPODIS 2016)}, 2017.

\bibitem{heller2005lazy}
Steve Heller, Maurice Herlihy, Victor Luchangco, Mark Moir, William~N Scherer,
  and Nir Shavit.
\newblock A lazy concurrent list-based set algorithm.
\newblock In {\em International Conference On Principles Of Distributed
  Systems}, pages 3--16. Springer, 2005.

\bibitem{herlihy1991wait}
Maurice Herlihy.
\newblock Wait-free synchronization.
\newblock {\em ACM Transactions on Programming Languages and Systems (TOPLAS)},
  13(1):124--149, 1991.

\bibitem{herlihy2007simple}
Maurice Herlihy, Yossi Lev, Victor Luchangco, and Nir Shavit.
\newblock A simple optimistic skiplist algorithm.
\newblock In {\em Structural Information and Communication Complexity: 14th
  International Colloquium, SIROCCO 2007, Castiglioncello, Italy, June 5-8,
  2007. Proceedings 14}, pages 124--138. Springer, 2007.

\bibitem{herlihy2005nonblocking}
Maurice Herlihy, Victor Luchangco, Paul Martin, and Mark Moir.
\newblock Nonblocking memory management support for dynamic-sized data
  structures.
\newblock {\em ACM Transactions on Computer Systems (TOCS)}, 23(2):146--196,
  2005.

\bibitem{herlihy2003space}
Maurice Herlihy, Victor Luchangco, and Mark Moir.
\newblock Space-and time-adaptive nonblocking algorithms.
\newblock {\em Electronic Notes in Theoretical Computer Science}, 78:260--280,
  2003.

\bibitem{herlihy1990linearizability}
Maurice~P Herlihy and Jeannette~M Wing.
\newblock Linearizability: A correctness condition for concurrent objects.
\newblock {\em ACM Transactions on Programming Languages and Systems (TOPLAS)},
  12(3):463--492, 1990.

\bibitem{howley2012non}
Shane~V Howley and Jeremy Jones.
\newblock A non-blocking internal binary search tree.
\newblock In {\em Proceedings of the twenty-fourth annual ACM symposium on
  Parallelism in algorithms and architectures}, pages 161--171, 2012.

\bibitem{ieeestdsig}
The open~group IEEE~CS.
\newblock Ieee standard for information technology--portable operating system
  interface (posix(tm)) base specifications, issue 7 - redline.
\newblock {\em IEEE Std 1003.1-2017 (Revision of IEEE Std 1003.1-2008) -
  Redline}, pages 1--6900, 2018.

\bibitem{intelTAA}
Intel.
\newblock Intel transactional synchronization extensions (intel tsx)
  asynchronous abort.
\newblock Technical Report CVE-2019-11135, Intel, 2019.

\bibitem{intelSDM}
Intel.
\newblock Intel®64 and ia-32 architectures software developer’s manual.
\newblock Technical Report Document Number: 252046-070, Intel, Dec, 2022
  (online).

\bibitem{ISOCPP2020}
{ISO}.
\newblock {\em {ISO/IEC 14882:2020 Information technology --- Programming
  languages --- C++}}.
\newblock International Organization for Standardization, Geneva, Switzerland,
  December 2020.

\bibitem{jones2023garbage}
Richard Jones, Antony Hosking, and Eliot Moss.
\newblock {\em The garbage collection handbook: the art of automatic memory
  management}.
\newblock CRC Press, 2023.

\bibitem{jung2024concurrent}
Jaehwang Jung, Jeonghyeon Kim, Matthew~J Parkinson, and Jeehoon Kang.
\newblock Concurrent immediate reference counting.
\newblock {\em Proceedings of the ACM on Programming Languages},
  8(PLDI):151--174, 2024.

\bibitem{jung2023applying}
Jaehwang Jung, Janggun Lee, Jeonghyeon Kim, and Jeehoon Kang.
\newblock Applying hazard pointers to more concurrent data structures.
\newblock In {\em Proceedings of the 35th ACM Symposium on Parallelism in
  Algorithms and Architectures}, pages 213--226, 2023.

\bibitem{kang2020marriage}
Jeehoon Kang and Jaehwang Jung.
\newblock A marriage of pointer-and epoch-based reclamation.
\newblock In {\em Proceedings of the 41st ACM SIGPLAN Conference on Programming
  Language Design and Implementation}, pages 314--328, 2020.

\bibitem{kim2024token}
Daewoo Kim, Trevor Brown, and Ajay Singh.
\newblock Are your epochs too epic? batch free can be harmful, 2024.

\bibitem{kim2024expediting}
Jeonghyeon Kim, Jaehwang Jung, and Jeehoon Kang.
\newblock Expediting hazard pointers with bounded rcu critical sections.
\newblock In {\em Proceedings of the 36th ACM Symposium on Parallelism in
  Algorithms and Architectures}, pages 1--13, 2024.

\bibitem{kung1981optimistic}
Hsiang-Tsung Kung and John~T Robinson.
\newblock On optimistic methods for concurrency control.
\newblock {\em ACM Transactions on Database Systems (TODS)}, 6(2):213--226,
  1981.

\bibitem{le2015transactional}
Hung~Q Le, Guy~L Guthrie, Dan~E Williams, Maged~M Michael, Brad~G Frey,
  William~J Starke, Cathy May, Rei Odaira, and Takuya Nakaike.
\newblock Transactional memory support in the ibm power8 processor.
\newblock {\em IBM Journal of Research and Development}, 59(1):8--1, 2015.

\bibitem{leijen2019mimalloc}
Daan Leijen, Benjamin Zorn, and Leonardo de~Moura.
\newblock Mimalloc: Free list sharding in action.
\newblock In {\em Programming Languages and Systems: 17th Asian Symposium,
  APLAS 2019, Nusa Dua, Bali, Indonesia, December 1--4, 2019, Proceedings 17},
  pages 244--265. Springer, 2019.

\bibitem{leis2016art}
Viktor Leis, Florian Scheibner, Alfons Kemper, and Thomas Neumann.
\newblock The art of practical synchronization.
\newblock In {\em Proceedings of the 12th International Workshop on Data
  Management on New Hardware}, pages 1--8, 2016.

\bibitem{levanoni2001fly}
Yossi Levanoni and Erez Petrank.
\newblock An on-the-fly reference counting garbage collector for java.
\newblock In {\em Proceedings of the 16th ACM SIGPLAN conference on
  Object-oriented programming, systems, languages, and applications}, pages
  367--380, 2001.

\bibitem{Lipp2018meltdown}
Moritz Lipp, Michael Schwarz, Daniel Gruss, Thomas Prescher, Werner Haas,
  Anders Fogh, Jann Horn, Stefan Mangard, Paul Kocher, Daniel Genkin, Yuval
  Yarom, and Mike Hamburg.
\newblock Meltdown: Reading kernel memory from user space.
\newblock In {\em 27th {USENIX} Security Symposium ({USENIX} Security 18)},
  2018.

\bibitem{MemDisagg19}
Ling Liu, Wenqi Cao, Semih Sahin, Qi~Zhang, Juhyun Bae, and Yanzhao Wu.
\newblock Memory disaggregation: Research problems and opportunities.
\newblock In {\em 2019 IEEE 39th International Conference on Distributed
  Computing Systems (ICDCS)}, pages 1664--1673, 2019.

\bibitem{massalin1992lock}
Henry Massalin and Calton Pu.
\newblock A lock-free multiprocessor os kernel.
\newblock {\em ACM SIGOPS Operating Systems Review}, 26(2):108, 1992.

\bibitem{mckenney2013rcu}
Paul~E McKenney, Silas Boyd-Wickizer, and Jonathan Walpole.
\newblock Rcu usage in the linux kernel: One decade later.
\newblock {\em Technical report}, 2013.

\bibitem{mckenney2006extending}
Paul~E McKenney, Dipankar Sarma, Ingo Molnar, and Suparna Bhattacharya.
\newblock Extending rcu for realtime and embedded workloads.
\newblock In {\em Ottawa Linux Symposium, pages v2}, pages 123--138, 2006.

\bibitem{mckenney1998read}
Paul~E McKenney and John~D Slingwine.
\newblock Read-copy update: Using execution history to solve concurrency
  problems.
\newblock In {\em Parallel and Distributed Computing and Systems}, volume
  509518, 1998.

\bibitem{useripilinux}
Sohil Mehta.
\newblock User interrupts – a faster way to signal.
\newblock Retrieved online, September, 2021.
\newblock [Accessed 15-08-2024].

\bibitem{michael2002high}
Maged~M Michael.
\newblock High performance dynamic lock-free hash tables and list-based sets.
\newblock In {\em Proceedings of the fourteenth annual ACM symposium on
  Parallel algorithms and architectures}, pages 73--82, 2002.

\bibitem{michael2004hazard}
Maged~M Michael.
\newblock Hazard pointers: Safe memory reclamation for lock-free objects.
\newblock {\em IEEE Transactions on Parallel and Distributed Systems},
  15(6):491--504, 2004.

\bibitem{michael1996simple}
Maged~M Michael and Michael~L Scott.
\newblock Simple, fast, and practical non-blocking and blocking concurrent
  queue algorithms.
\newblock In {\em Proceedings of the fifteenth annual ACM symposium on
  Principles of distributed computing}, pages 267--275, 1996.

\bibitem{michael1995correction}
Maged~M Michael and Michael~Lee Scott.
\newblock {\em Correction of a memory management method for lock-free data
  structures}.
\newblock University of Rochester, Department of Computer Science, 1995.

\bibitem{michael2017hazard}
Maged~M Michael, Michael Wong, Paul McKenney, Arthur O’Dwyer, and David
  Hollman.
\newblock Hazard pointers: Safe resource reclamation for optimistic
  concurrency.
\newblock Technical report, Technical Report P0233R3. C++ SG14 Working Group,
  2017.

\bibitem{michaelKenny2017proposed}
Maged~M Michael, Michael Wong, Paul McKenney, Geoffrey Romer, Andrew Hunter,
  Arthur O’Dwyer, David~S Hollman, JF~Bastien, Hans Boehm, David Goldblatt,
  et~al.
\newblock Proposed wording for concurrent data structures: Hazard pointer and
  read-copy-update (rcu), 2017.

\bibitem{miller2010graphite}
Jason~E Miller, Harshad Kasture, George Kurian, Charles Gruenwald, Nathan
  Beckmann, Christopher Celio, Jonathan Eastep, and Anant Agarwal.
\newblock Graphite: A distributed parallel simulator for multicores.
\newblock In {\em HPCA-16 2010 The Sixteenth International Symposium on
  High-Performance Computer Architecture}, pages 1--12. IEEE, 2010.

\bibitem{moir2018concurrent}
Mark Moir and Nir Shavit.
\newblock Concurrent data structures.
\newblock In {\em Handbook of Data Structures and Applications}, pages
  741--762. Chapman and Hall/CRC, 2018.

\bibitem{moreno2023releasing}
Pedro Moreno and Ricardo Rocha.
\newblock Releasing memory with optimistic access: A hybrid approach to memory
  reclamation and allocation in lock-free programs.
\newblock In {\em Proceedings of the 35th ACM Symposium on Parallelism in
  Algorithms and Architectures}, pages 177--186, 2023.

\bibitem{morrison2015temporally}
Adam Morrison and Yehuda Afek.
\newblock Temporally bounding tso for fence-free asymmetric synchronization.
\newblock {\em ACM SIGARCH Computer Architecture News}, 43(1):45--58, 2015.

\bibitem{natarajan2014fast}
Aravind Natarajan and Neeraj Mittal.
\newblock Fast concurrent lock-free binary search trees.
\newblock In {\em Proceedings of the 19th ACM SIGPLAN symposium on Principles
  and practice of parallel programming}, pages 317--328, 2014.

\bibitem{nikolaev2019hyaline}
Ruslan Nikolaev and Binoy Ravindran.
\newblock Hyaline: fast and transparent lock-free memory reclamation.
\newblock In {\em Proceedings of the 2019 ACM Symposium on Principles of
  Distributed Computing}, pages 419--421, 2019.

\bibitem{nikolaev2020universal}
Ruslan Nikolaev and Binoy Ravindran.
\newblock Universal wait-free memory reclamation.
\newblock In {\em Proceedings of the 25th ACM SIGPLAN Symposium on Principles
  and Practice of Parallel Programming}, pages 130--143, 2020.

\bibitem{nikolaev2021crystalline}
Ruslan Nikolaev and Binoy Ravindran.
\newblock Crystalline: Fast and memory efficient wait-free reclamation, 2021.

\bibitem{nikolaev2024family}
Ruslan Nikolaev and Binoy Ravindran.
\newblock A family of fast and memory efficient lock-and wait-free reclamation.
\newblock {\em Proceedings of the ACM on Programming Languages},
  8(PLDI):2174--2198, 2024.

\bibitem{java22}
{Oracle}.
\newblock Java platform, standard edition, core libraries.
\newblock
  \url{https://docs.oracle.com/en/java/javase/22/core/heap-and-heap-memory.html},
  April 2024.
\newblock Release 22, F87196-02.

\bibitem{page2014signals}
Rule Page.
\newblock Signals (sig).
\newblock {\em The CERT{\textregistered} C Coding Standard: 98 Rules for
  Developing Safe, Reliable, and Secure Systems}, page 333, 2014.

\bibitem{cassandra}
Vijay Parthasarathy.
\newblock {\em Learning Cassandra for Administrators}.
\newblock Packt Publishing Birmingham, UK, 2013.
\newblock page 12.

\bibitem{prokopec2012concurrent}
Aleksandar Prokopec, Nathan~Grasso Bronson, Phil Bagwell, and Martin Odersky.
\newblock Concurrent tries with efficient non-blocking snapshots.
\newblock In {\em Proceedings of the 17thACM SIGPLAN symposium on Principles
  and Practice of Parallel Programming}, pages 151--160, 2012.

\bibitem{rajwar2001speculative}
Ravi Rajwar and James~R Goodman.
\newblock Speculative lock elision: Enabling highly concurrent multithreaded
  execution.
\newblock In {\em Proceedings. 34th ACM/IEEE International Symposium on
  Microarchitecture. MICRO-34}, pages 294--305. IEEE, 2001.

\bibitem{ramachandran2015castle}
Arunmoezhi Ramachandran and Neeraj Mittal.
\newblock Castle: fast concurrent internal binary search tree using edge-based
  locking.
\newblock {\em ACM SIGPLAN Notices}, 50(8):281--282, 2015.

\bibitem{ramachandran2015fast}
Arunmoezhi Ramachandran and Neeraj Mittal.
\newblock A fast lock-free internal binary search tree.
\newblock In {\em Proceedings of the 2015 International Conference on
  Distributed Computing and Networking}, pages 1--10, 2015.

\bibitem{ramalhete2017brief}
Pedro Ramalhete and Andreia Correia.
\newblock Brief announcement: Hazard eras-non-blocking memory reclamation.
\newblock In {\em Proceedings of the 29th ACM Symposium on Parallelism in
  Algorithms and Architectures}, pages 367--369, 2017.

\bibitem{unixprgsigs}
Kay~A. Robbins and Steven Robbins.
\newblock {\em Practical UNIX programming: a guide to concurrency,
  communication, and multithreading}.
\newblock Prentice-Hall, Inc., USA, 1995.

\bibitem{umasmrfreebsd}
Jeffery Roberson.
\newblock Freebsd universal memory allocator.
\newblock uma\_smr.c, \url{https://reviews.freebsd.org/D22586}.

\bibitem{scott2013shared}
Michael~Lee Scott and Trevor Brown.
\newblock {\em Shared-memory synchronization}.
\newblock Springer, 2013.

\bibitem{shafiei2013non}
Niloufar Shafiei.
\newblock Non-blocking patricia tries with replace operations.
\newblock In {\em 2013 IEEE 33rd International Conference on Distributed
  Computing Systems}, pages 216--225. IEEE, 2013.

\bibitem{sheffi2021vbr}
Gali Sheffi, Maurice Herlihy, and Erez Petrank.
\newblock Vbr: Version based reclamation.
\newblock In {\em Proceedings of the 33rd ACM Symposium on Parallelism in
  Algorithms and Architectures}, SPAA '21, page 443–445, New York, NY, USA,
  2021. Association for Computing Machinery.

\bibitem{sheffi2023era}
Gali Sheffi and Erez Petrank.
\newblock The era theorem for safe memory reclamation.
\newblock In {\em Proceedings of the 28th ACM SIGPLAN Annual Symposium on
  Principles and Practice of Parallel Programming}, pages 435--437, 2023.

\bibitem{shi2023optiql}
Ge~Shi, Ziyi Yan, and Tianzheng Wang.
\newblock Optiql: Robust optimistic locking for memory-optimized indexes.
\newblock {\em Proceedings of the ACM on Management of Data}, 1(3):1--26, 2023.

\bibitem{singh2023IPDPS}
A.~Singh, T.~Brown, and M.~Spear.
\newblock Efficient hardware primitives for immediate memory reclamation in
  optimistic data structures.
\newblock In {\em 2023 IEEE International Parallel and Distributed Processing
  Symposium (IPDPS)}, pages 112--122, Los Alamitos, CA, USA, may 2023. IEEE
  Computer Society.

\bibitem{singh2020nbr}
Ajay Singh, Trevor Brown, and Ali Mashtizadeh.
\newblock Nbr: Neutralization based reclamation, 2020.

\bibitem{singh2021nbr}
Ajay Singh, Trevor Brown, and Ali Mashtizadeh.
\newblock Nbr: neutralization based reclamation.
\newblock In {\em Proceedings of the 26th ACM SIGPLAN Symposium on Principles
  and Practice of Parallel Programming}, pages 175--190, 2021.

\bibitem{singh21nbr}
Ajay Singh, Trevor Brown, and Ali Mashtizadeh.
\newblock {\em NBR: Neutralization Based Reclamation}, page 175–190.
\newblock Association for Computing Machinery, New York, NY, USA, 2021.

\bibitem{singh2023efficient}
Ajay Singh, Trevor Brown, and Michael Spear.
\newblock Efficient hardware primitives for immediate memory reclamation in
  optimistic data structures.
\newblock {\em arXiv preprint arXiv:2302.12958}, 2023.

\bibitem{singhTPDS2023NBRP}
Ajay Singh, Trevor~Alexander Brown, and Ali~José Mashtizadeh.
\newblock Simple, fast and widely applicable concurrent memory reclamation via
  neutralization.
\newblock {\em IEEE Transactions on Parallel and Distributed Systems},
  35(2):203--220, 2024.

\bibitem{skare2006early}
Travis Skare and Christos Kozyrakis.
\newblock Early release: Friend or foe.
\newblock In {\em Workshop on Transactional Memory Workloads}, volume~77, 2006.

\bibitem{sonmez2007unreadtvar}
Nehir S{\"o}nmez, Cristian Perfumo, Srdjan Stipic, Adrian Cristal, Osman~S
  Unsal, and Mateo Valero.
\newblock unreadtvar: Extending haskell software transactional memory for
  performance.
\newblock {\em Trends in Functional Programming}, 8:89--114, 2007.

\bibitem{srivastava2022elimination}
Anubhav Srivastava and Trevor Brown.
\newblock Elimination (a, b)-trees with fast, durable updates.
\newblock In {\em Proceedings of the 27th ACM SIGPLAN Symposium on Principles
  and Practice of Parallel Programming}, pages 416--430, 2022.

\bibitem{sundell2004efficient}
H{\aa}kan Sundell.
\newblock {\em Efficient and practical non-blocking data structures}.
\newblock Department of Computer Engineering, Chalmers University of
  Technology, 2004.

\bibitem{timnat2014practical}
Shahar Timnat and Erez Petrank.
\newblock A practical wait-free simulation for lock-free data structures.
\newblock {\em ACM SIGPLAN Notices}, 49(8):357--368, 2014.

\bibitem{treiber1986systems}
R~Kent Treiber.
\newblock {\em Systems programming: Coping with parallelism}.
\newblock International Business Machines Incorporated, Thomas J. Watson
  Research, 1986.

\bibitem{turek1992locking}
John Turek, Dennis Shasha, and Sundeep Prakash.
\newblock Locking without blocking: making lock based concurrent data structure
  algorithms nonblocking.
\newblock In {\em Proceedings of the eleventh ACM SIGACT-SIGMOD-SIGART
  symposium on Principles of database systems}, pages 212--222, 1992.

\bibitem{valois1995lock}
John~D Valois.
\newblock Lock-free linked lists using compare-and-swap.
\newblock In {\em Proceedings of the fourteenth annual ACM symposium on
  Principles of distributed computing}, pages 214--222, 1995.

\bibitem{vmwareunderstanding}
ESX VMware.
\newblock Understanding memory resource management in vmware esx 4.1.
\newblock {\em VMware ESX 4.1 Documentation, EN-000411-00}, 2010.

\bibitem{wang2018building}
Ziqi Wang, Andrew Pavlo, Hyeontaek Lim, Viktor Leis, Huanchen Zhang, Michael
  Kaminsky, and David~G Andersen.
\newblock Building a bw-tree takes more than just buzz words.
\newblock In {\em Proceedings of the 2018 International Conference on
  Management of Data}, pages 473--488, 2018.

\bibitem{wen2018interval}
Haosen Wen, Joseph Izraelevitz, Wentao Cai, H~Alan Beadle, and Michael~L Scott.
\newblock Interval-based memory reclamation.
\newblock {\em ACM SIGPLAN Notices}, 53(1):1--13, 2018.

\bibitem{zhou2017hand}
Tingzhe Zhou, Victor Luchangco, and Michael Spear.
\newblock Hand-over-hand transactions with precise memory reclamation.
\newblock In {\em Proceedings of the 29th ACM Symposium on Parallelism in
  Algorithms and Architectures}, pages 255--264, 2017.

\end{thebibliography}





\end{document}